\documentclass[letterpaper]{article} 
\usepackage{aaai25}  
\usepackage{times}  
\usepackage{helvet}  
\usepackage{courier}  
\usepackage[hyphens]{url}  
\usepackage{graphicx} 
\usepackage{amsmath, amsthm, amssymb, amsbsy, mathtools}

\usepackage{hyperref}

\usepackage[capitalize]{cleveref}
\usepackage{natbib}  
\usepackage{caption} 

\usepackage{algorithm}
\usepackage{algorithmic}

\usepackage{multirow}
\usepackage[table]{xcolor}
\usepackage{booktabs}

\usepackage{nicematrix}
\usepackage{makecell}
\mathcode`*=\string"8000
\begingroup
\catcode`*=\active
\xdef*{\noexpand\textup{\string*}}
\endgroup
\usepackage{array}
\newcolumntype{M}[1]{>{\centering\arraybackslash}m{#1}}
\newcolumntype{N}{@{}m{0pt}@{}}

\usepackage[section]{placeins}

\usepackage{newfloat}
\usepackage{listings}
\DeclareCaptionStyle{ruled}{labelfont=normalfont,labelsep=colon,strut=off} 
\floatstyle{ruled}
\newfloat{listing}{tb}{lst}{}
\floatname{listing}{Listing}
\newcommand{\ones}{\mathbf{1}}

\newcommand{\RR}{\mathbb{R}}

\newcommand{\tatonnement}{{t\^atonnement}}
\newcommand{\Tatonnement}{{T\^atonnement}}
\newcommand{\lbp}{\tilde{p}}
\newcommand{\subGD}{SubGD}
\newcommand{\ck}[1]{{\color{purple}[*ck: #1 *] }}
 
\newcommand{\tianlong}[1]{{\color{teal}[* #1 *]}}

\DeclareMathOperator*{\argmax}{arg\,max}
 
\DeclareMathOperator*{\Argmax}{\textnormal{Argmax}}

\newcommand{\cD}{\mathcal{D}} 

\DeclarePairedDelimiterX{\norm}[1]{\lVert}{\rVert}{#1}
\DeclarePairedDelimiterX{\En}[1]{\lVert}{\rVert}{#1} 
\DeclarePairedDelimiterX{\inp}[2]{\langle}{\rangle}{#1, #2} 

\newtheorem{theorem}{Theorem}

\newtheorem{lemma}{Lemma}
\newtheorem{corollary}{Corollary}

\newtheorem{fact}{Fact}

\crefname{algocf}{alg.}{algs.}
\crefname{fact}{Fact}{Facts}

\newtheorem{innercustomlem}{Lemma}
\newenvironment{customlem}[1]
  {\renewcommand\theinnercustomlem{#1}\innercustomlem}
  {\endinnercustomlem}

\newtheorem{innercustomthm}{Theorem}
\newenvironment{customthm}[1]
  {\renewcommand\theinnercustomthm{#1}\innercustomthm}
  {\endinnercustomthm}

\newtheorem{innercustomcrl}{Corollary}
  \newenvironment{customcrl}[1]
    {\renewcommand\theinnercustomcrl{#1}\innercustomcrl}
    {\endinnercustomcrl}

\title{On the Convergence of T\^atonnement for Linear Fisher Markets}
\author{
    Tianlong Nan\textsuperscript{\rm 1}, 
    Yuan Gao\textsuperscript{\rm 2}, 
    Christian Kroer\textsuperscript{\rm 1}}
\affiliations{
	\textsuperscript{\rm 1}Columbia University\\ \textsuperscript{\rm 2}Microsoft\\
    tianlong.nan@columbia.edu, gaoyuan@microsoft.com, christian.kroer@columbia.edu
}

\begin{document}

\maketitle

\begin{abstract}
    T\^atonnement is a simple, intuitive market process where prices are iteratively adjusted based on the difference between demand and supply. Many variants under different market assumptions have been studied and shown to converge to a market equilibrium, in some cases at a fast rate. However, the classical case of linear Fisher markets have long eluded the analyses, and it remains unclear whether t\^atonnement converges in this case. We show that, for a sufficiently small step size, the prices given by the t\^atonnement process are guaranteed to converge to equilibrium prices, up to a small approximation radius that depends on the stepsize. To achieve this, we consider the dual Eisenberg-Gale convex program in the price space, view t\^atonnement as subgradient descent on this convex program, and utilize last-iterate convergence results for subgradient descent under error bound conditions. In doing so, we show that the convex program satisfies a particular error bound condition, the quadratic growth condition, and that the price sequence generated by t\^atonnement is bounded above and away from zero. 
    We also show that a similar convergence result holds for tâtonnement in quasi-linear Fisher markets. Numerical experiments are conducted to demonstrate that the theoretical linear convergence aligns with empirical observations.
\end{abstract}


\section{Introduction} 
\label{sec:intro}

\emph{Market equilibrium} (ME) is a central solution concept in economics.
It describes a steady state of the market in which supply is equal to demand.
In recent year, it has found various applications in resource allocation in large-scale online and physical market places due to its fairness, efficiency in allocating items, and tractability under various real-world-inspired settings.
Since ME was first studied by L\'eon Walras in 1874, extensive work has been undertaken to establish the existence and uniqueness of ME in various market models~\citep{walras1874elements,arrow1954existence,cole2017convex}.
In this paper, we focus on the \emph{Fisher market} setting, a well-known market model involving $m$ divisible items which are to be allocated among $n$ buyers.
Each item has a fixed supply (usually assumed to be $1$) and each buyer is endowed with a fixed budget of money.
In a Fisher market, an equilibrium is a set of prices for the items, and an allocation of items to buyers, such that each buyer is allocated the optimal bundle subject to the prices and her budget, and each item is allocated at its supply.

Alongside the development of market settings and equilibrium solution concepts, market equilibrium computation, by itself or as part of an end-to-end approach to real-world resource allocation problems, has seen a growing interest in various research communities in economics, computer science, and operations research~\citep{scarf1967computation,kantorovich1975mathematics,daskalakis2009complexity,bateni2022fair,babaioff2019fair,othman2010finding,cole2017convex,birnbaum2011distributed,kroer2019computing,kroer2021computing,gao2020first,gao2021online}.
This is especially so in the past two decades, as the emergence of internet marketplaces, including advertising systems, require approximating or maintaining equilibrium variables in large-scale market settings. 

{\Tatonnement} is a simple, natural, and decentralized price-adjustment process where prices of items are adjusted to 
reflect buyers' demands: the price of a good increases if the demand exceeds the supply under current prices, and decreases if the supply exceeds the demand.
It can also be viewed as an algorithm for computing market equilibria. 
Intuitively, if the {\tatonnement} process is guaranteed to converge to a market equilibrium, then it lends credence to market equilibrium as a solution concept for the economic allocation of goods.
Indeed, for various market settings, including Fisher markets with 
non-linear \emph{constant-elasticity-of-substitution} (CES) utilities, the convergence of {\tatonnement} has been established and is well-understood~\citep{cole2008fast,cheung2019tatonnement,avigdor2014convergence}.
A prominent exception is the simplest, and arguably most prevalent, setting of linear Fisher markets (LFM), where buyers have linear (additive) utilities over items.
It has long been known that, for LFMs, {\tatonnement} does not converge to an equilibrium in the strictest sense, that is, in terms of a vanishing metric or norm; see, e.g., \citet{cole2019balancing} and \cref{sec:conv-analysis} for further discussion and examples. However, whether {\tatonnement} demonstrates any form of \emph{approximate} convergence remains largely unexplored.

One important fact about an LFM
is that its market equilibria are captured by the well-known Eisenberg-Gale (EG) convex program:
\begin{equation}
    \max_{x\in \mathbb R^{n\times m}}\, \sum_{i=1}^n B_i \log \langle v_i, x_i \rangle \ \ {\rm s.t.}\ \sum_{i=1}^n x_i \leq \ones \label{eq:eg-primal} \tag{EG-Primal}, 
\end{equation}
where $B_i \in \RR, v_i \in \RR^m, x_i \in \RR^m$ denotes budget, valuation vector, and allocation vector for buyer $i$, respectively, in an LFM composed of $n$ buyers and $m$ items.  
It turns out that {\tatonnement} in the LFM is equivalent to the subgradient descent method applied to the dual program of~\cref{eq:eg-primal} (See, e.g.,~\citet{cheung2013tatonnement}). 
This can be shown as follows: 
by taking the Lagrangian of~\cref{eq:eg-primal}, we obtain that the dual problem to EG is to minimize the following function
\begin{equation}
    \begin{aligned}
    \varphi(p) &= \max_{x \geq 0}\, \mathcal{L}(x, p) \\ 
    &= \max_{x \geq 0}\, \sum_{i=1}^n \left( B_i \log\inp{v_i}{x_i} - \sum_{j=1}^m p_j x_{ij} \right) + \sum_{j=1}^m p_j, 
    \label{eq:linear-dual-lagrangian}
    \end{aligned}
\end{equation}
which is a max-type nonsmooth convex function w.r.t. $p$ by construction.
Every subgradient of $\varphi(p)$ is of the form $\mathbf{1} - \sum_i x^*_i(p)$, where $x^*_i(p)$ is a maximizer of $\max_{x_i \geq 0}\, (B_i \log\inp{v_i}{x_i} - \sum_{j=1}^m p_j x_{ij})$. 
By first-order optimality conditions, for any $p$, $x^*_i(p)$ is the demand of buyer $i$ under $p$. 
Thus, {\tatonnement} (in the form of $p^{t+1} \gets p^t + \eta ( \sum_i x^*_i(p^t) - \mathbf{1})$, where $\eta \in \RR_{>0}$) is equivalent to applying subgradient descent on $\min_{p \in \RR^m_{\geq 0}} \varphi(p)$.
We will use this relation to establish a linear convergence rate for {\tatonnement} in LFMs.

\paragraph{Contributions.} We consider a classic form of the {\tatonnement} process in LFMs, which updates prices as follows: 
\begin{equation}
    p^{t+1} \gets p^t + \eta z^t, \quad \forall\, t, 
    \label{eq:intro-ttm}
\end{equation}
where $z$ denotes excess demand at time $t$. 
We give the first positive answers to the question of whether {\tatonnement} converges in LFMs: we show that the prices generated by the process converge at a fast linear rate (i.e. exponentially fast) to a small neighborhood around the true equilibrium prices. 
Moreover, we show that the size of this neighborhood shrinks to $0$ as the stepsize $\eta$ goes to $0$. 
Thus, for any approximation level $\epsilon > 0$, we can find an $\eta > 0$ such that {\tatonnement} converges linearly to an $\epsilon$-approximate equilibrium. 
Our results are achieved by leveraging the connection between {\tatonnement} and subgradient methods, and from there showing that the dual of the EG convex program satisfies the quadratic growth condition around its optimum. Moreover, we must then show bounded subgradients, which we establish by showing that the {\tatonnement} process is guaranteed to be bounded above and away from zero for its entire trajectory, assuming a sufficiently small stepsize.
As is explained in \cref{app:related-work}, many existing works on {\tatonnement} use some modified version of the process where prices are generally kept away from zero either through projection to a region bounded away from zero, or by artificially upper-bounding excess demands. 
Our results show that for LFMs, such modifications are not necessary as long as stepsizes are sufficiently small.
Finally, we extend our analysis framework to an important variant of LFMs: 
quasi-linear Fisher markets. For such markets, we again show that {\tatonnement} achieves a linear convergence rate to approximate equilibrium.


\paragraph{Notation.}
Throughout this paper, we use $\| \cdot \|$ to denote Euclidean distance and 
$\inp{\cdot}{\cdot}$ to denote the standard inner product. 
We use $\Pi_\mathcal{X}(x)$ to denote the Euclidean projection of $x$ onto set $\mathcal{X}$. 
We use $\| \cdot \|_1$ to denote the $\ell_1$ norm. 
We use $\Argmax$ to denote the set of values that reach the maximum. 
The subdifferential of a function $f$ at $x$ is denoted as $\partial f(x)$. 
We use $\nabla f$ to denote the differential of $f$ if $f$ is smooth.
We denote $[m] = \{ 1, \ldots, m \}$. 

\section{Related Work} \label{app:related-work}







\paragraph{Convergence of {\tatonnement}.}
Initially, economists studied the convergence of {\tatonnement} in continuous time.
There, \citet{arrow1959stability} proved that a continuous-time version of {\tatonnement} converges to ME for markets satisfying \emph{weak gross substitutability} (WGS). 
In the remainder of the literature review, we focus on the discrete-time version of {\tatonnement}, which is the focus of our paper.
\citet{codenotti2005market} considered a discrete version of {\tatonnement} 
and proved that it converges to an approximate equilibrium in polynomial time for any exchange economy satisfying WGS. 
Though their {\tatonnement} updates are additive and similar to ours, 
they need to apply {\tatonnement} on a transformed market and 
use price projection in order to restrict prices to a region that is bounded and bounded away from zero. 
Secondly, they assume that demand is unique, which is not true for LFMs.
\citet{cole2008fast} considered a multiplicative {\tatonnement} with artificially upper-bounded excess demands, 
and showed polynomial convergence for it on non-linear CES markets satisfying WGS ($0 \leq \rho < 1$). 
\citet{cheung2012tatonnement} extended this analysis to some non-WGS markets. 


\citet{cheung2019tatonnement} proved that the entropic {\tatonnement} 
is equivalent to generalized gradient descent with KL divergence, and thus 
is guaranteed to reach $\epsilon$ fraction of initial distance to ME measured in 
a potential function in $\mathcal{O}(\frac{1}{\epsilon})$ iterations 
for the Leontief Fisher market, 
and $\mathcal{O}(\log{\frac{1}{\epsilon}})$ iterations 
for the complementary CES Fisher market ($-\infty < \rho < 0$). 
\citet{avigdor2014convergence,cheung2012tatonnement} considered multiplicative {\tatonnement} and showed convergence to an approximate equilibrium for Fisher markets with nested non-linear CES-type utilities.

\begin{table*}[t]
    \centering
    \setlength{\extrarowheight}{0.5mm} 
    \begin{tabular}{cccc} 
        \toprule 
            \textbf{T\^atonnement} & \textbf{Papers} & \textbf{Market Type} & \textbf{Result} \\
        \midrule
            \makecell{Additive t\^atonnement$^{\ast}$ \\ $p^{t+1}_j = \Pi_\Delta\left( p^t_j + \eta z^t_j \right)$} & 
            \citet{codenotti2005market} & 
            \makecell{Exchange market \\ WGS, unique demand} & 
            \makecell{Poly-time \\ to approx ME} \\ 
            \hline
            \makecell{Multiplicative t\^atonnement \\ $p^{t+1}_j = p^t_j( 1 + \eta z^t_j )$} & 
            \citet{avigdor2014convergence} & 
            \makecell{Nested CES-Leontief \\ $\rho \in (-\infty, 0) \cup (0, 1)$} & 
            \makecell{$\mathcal{O}(\frac{1}{\epsilon^3})$} \\ 
            \hline 
            \multirow{2}{*}{\makecell{Multiplicative t\^atonnement$^{\ast}$ \\ $p^{t+1}_j = p^t_j( 1 + \eta \min\{ z^t_j, 1 \} )$}} & 
            \citet{cole2008fast} & 
            \makecell{Exchange market \\ WGS} & 
            \makecell{$\mathcal{O}(E\log{\frac{1}{\epsilon}})$} \\ 
            \cline{2-4}
                                 & 
            \citet{cheung2012tatonnement} & 
            \makecell{CES ($-1 < \rho \leq 0$) \\ Nested non-lin. CES} & 
            \makecell{$\mathcal{O}(\log{\frac{1}{\epsilon}})$} \\ 
            \hline 
            \multirow{3}{*}{\makecell{Entropic t\^atonnement$^{\ast}$ \\ $p^{t+1}_j = p^t_j e^{\eta \min\{ z^t_j, 1 \}}$}} & 
            \citet{cheung2019tatonnement} & 
            \makecell{CES ($-\infty \leq \rho < 0$)} & 
            \makecell{$\mathcal{O}(\log{\frac{1}{\epsilon}})$ \\ $\mathcal{O}(\frac{1}{\epsilon})\,\textnormal{if}\,\rho = -\infty$} \\ 
            \cline{2-4}
                                 & 
            \multirow{2}{*}{\citet{goktas2023tatonnement}} & 
            \multirow{2}{*}{\makecell{CCNH}} & 
            \multirow{2}{*}{\makecell{$\mathcal{O}(\frac{1 + E^2}{\epsilon})$}} \\ 
                                 & 
            &   &  \\ 
            \hline 
            \makecell{Entropic t\^atonnement$^{\ast}$ \\ $p^{t+1}_j = \Pi_\Delta\big( p^t_j e^{\eta \min\{ z^t_j, 1 \}} \big)$} & 
            \citet{cole2019balancing} & 
            \makecell{Linear \\ Large market assum.} & 
            \makecell{Linear conv. \\ to approx ME} \\ 
            \hline 
            \rowcolor{green!70} 
            \makecell{Additive t\^atonnement \\ $p^{t+1}_j = p^t_j + \eta z^t_j$} & 
            This work & 
            Linear & 
            \makecell{Linear conv. \\ to approx ME} \\ 
        \bottomrule              
        \hline
    \end{tabular}
    \caption{Comparison of results of different {\tatonnement} methods. 
                 The superscript $^{\ast}$ denotes a modified version of the update rule. 
                 Note that only the last two papers consider linear Fisher markets. 
                 \citet{cole2008fast} and \citet{goktas2023tatonnement} consider a broader class of utilities, but their parameter $E$ is infinity for linear utilities.
                 $\Pi_\Delta(\cdot)$ denotes projection onto a bounded region. 
                 $E$ denotes the upper bound of the elasticity of demand.
        }
    \label{tab:comparison}
\end{table*}

\citet{cole2019balancing} studied convergence to approximate equilibrium for the LFM, 
using a modified entropic {\tatonnement} process. 
In particular, they considered the update rule $p^{t+1} \gets p^t e^{\eta \min\{ z^t,\, 1 \}}$, 
where $p^t$ and $z^t$ denote the price and excess demand at time $t$. 
They showed a linear convergence rate under what they call a ``large market assumption,'' which 
requires that for buyers with linear or close-to-linear utilities, the spending (of budget) on a single item does not vary too much as prices change. 
This is unnatural for LFMs as the elasticity of demand can be infinity in the LFM. 
Moveover, their quality of the approximation is proportional to a particular squared parameter in their assumption, 
while our results do not even require any such assumption. 

\citet{goktas2021consumer} studied the dual of an Eisenberg-Gale-like convex program from the perspective of the expenditure minimization problem (EMP). 
\citet{goktas2023tatonnement} showed an $\mathcal{O}((1 + E^2) / T)$ convergence for an entropic {\tatonnement} process on Fisher markets with \emph{concave, continuous, nonnegative, and $1$-homogeneous}
\footnote{A utility function $u: \RR^m_{\geq 0} \rightarrow \RR_{\geq 0}$ is $\alpha$-homogeneous if $u(a x) = a^\alpha u(x)$ for any $a \in \RR_{\geq 0}$ and $x \in \RR^m_{\geq 0}$.} 
(CCNH) utilities, where $E$ denotes the upper bound of the elasticity of demand. 
Note that in the case of LFM, the elasticity of demand is unbounded, and thus this result gives an infinite upper bound. 

\citet{fleischer2008fast} study the convergence of the \emph{averaged} iterates generated by a {\tatonnement}-like algorithm in a class of Fisher markets that include linear, Leontief, and CES utilities. 
Note that we focus on the \emph{last} iterate, which is the only type of convergence that allows us to say whether the day-to-day prices are converging.
Secondly, last-iterate convergence implies average-iterate convergence. 

We summarize the above related work on the convergence of {\tatonnement} in~\cref{tab:comparison}. 
The table provides an overview of the different {\tatonnement} methods considered in the literature, along with a high-level description of convergence rates.

\paragraph{Convergence of subgradient descent methods.}
Subgradient descent methods (SubGDs) have been extensively studied since their introduction in the 1960s.
Historically, SubGD was only known to have an $O(1/\sqrt{T})$ convergence rate for the averaged iterate, where $T$ is the number of iterations. 
Moreover, this convergence was in terms of the optimality gap, and not in terms of the distance to the optimal solution. For {\tatonnement}, we are interested in the latter, as we want to know how quickly the prices converge to the market equilibrium prices.
Recently, \citet{zamani2023exact} showed a nearly $O(1/\sqrt{T})$ optimal rate for the \emph{last} iterate, but again only in terms of optimality gap.
It is known, however, that SubGD can achieve a faster convergence rate when the objective satisfies various notions of \emph{error-bound} conditions (along with a bounded gradient condition). 
For prior work on SubGDs under error-bound conditions, see~\citet{johnstone2020faster} and references therein.
Particularly relevant to us, SubGD is known to have linear convergence under the \emph{quadratic growth condition} (QG)~\citep{nedic2001convergence,karimi2016linear}, in terms of the distance to the optimal solution.
This condition also allows one to connect convergence in optimality gap to convergence in iterates. We will show that this condition holds for LFMs, and thus SubGD has linear convergence in our setting.
\section{Linear Fisher markets and {\Tatonnement}} 
\label{app:lin-fim-eg-cp}


\paragraph{Linear Fisher markets.}
We consider a linear Fisher market with $n$ buyers and $m$ divisible items. 
Each buyer $i$ has a budget $B_i > 0$ and valuation $v_{ij} \geq 0$ for each item $j$.
Denote $B = \sum_{i \in [n]} B_i$ and $B_{\min} = \min_{i \in [n]} B_i$.
We denote $v_i = (v_{i1}, \dots, v_{im})$ for all $i$ and 
$v$ as the valuation matrix whose $i$-th row is $v_i$.
An \emph{allocation} (or \emph{bundle}) of items $x_i \in \RR_{\geq 0}^m$ gives buyer $i$ a utility of $u_i(x_i) = \langle v_i, x_i \rangle$.
Without loss of generality, we assume the following:
\begin{itemize}
    \item There is a unit supply of each item $j$.
    \item $\|v_i\|_1 = 1$ for all $i$.
    \item $v \in \RR^{n\times m}_{\geq 0}$ does not have
    all-zero rows or columns. In other words, each item $j$ is wanted by at least one buyer, and each buyer $i$ wants at least one item.
\end{itemize}

Given a set of prices $p \in \RR^m_{\geq 0}$ for the items, the \emph{demand set} of buyer $i$ is defined as 
\begin{align}
    \label{eq:demand-def}
    \mathcal{D}_i(p) := \Argmax_{x_i} \left\{\langle v_i, x_i \rangle: \langle p, x_i \rangle \leq B_i,\, x_i \geq 0 \right\}. 
\end{align}
We use $d_i(p) \in \mathcal{D}_i(p)$ to denote an arbitrary demand vector of buyer $i$ under $p$. 
We use $z_j(p) = \sum_{i \in [n]} d_{ij}(p) - 1$ to denote the \emph{excess demand} for item $j$ under $p$. 
The \emph{excess supply} is $- z_j(p)$. 

In the LFM (or its variants), buyers only buy 
items with \emph{maximum bang-per-buck} (MBB). 
We call such items the MBB items for a buyer,
i.e., $j$ is an MBB item for buyer $i$ if $\frac{v_{ij}}{p_j} \geq \frac{v_{ij'}}{p_j'}$ for all $j' \in [m]$. 


A Fisher market is said to reach a \emph{market equilibrium} if allocations 
$x^*_i\in \RR^m_{\geq 0}, \; \forall \; i \in [n]$ and prices $p^* \in \RR^m_{\geq 0}$ 
satisfy 
\begin{itemize}
    \item Buyer optimality: $x^*_i \in \mathcal{D}_i(p^*)$ for all $i \in [n]$,
    \item Market clearance: $\sum_{i=1}^n x^*_{ij} = 1$ for all $j \in [m]$.
\end{itemize}

\paragraph{{\Tatonnement}.}
We consider the classic, discrete-time {\tatonnement} dynamics that adjust prices at time $t$ as follows: 
\begin{equation}
    \begin{array}{rll}
        \bullet & \text{Pick}\; x_i^t \in \cD_i(p^t) \quad & \forall\; i \in [n] \\ [3pt]
        \bullet & z^t_j = \sum_{i = 1}^n x_i^t - 1 \quad & \forall\; j \in [m] \\ [3pt]
        \bullet & p^{t+1}_j = p^t_j + \eta^t z^t_j \quad & \forall\; j \in [m].
    \end{array} 
    \label{eq:general-linear-ttm}
\end{equation}
Typically, many authors have considered some form of projection to ensure $p^{t+1} \geq 0$, or even such that it is strictly bounded away from zero.
We will show that for LFMs, it is possible to select $\eta^t$ such that the prices are ensured to be strictly bounded away from zero without projection.
This is critical for our later convergence results.
We will focus on the case where stepsizes are constant, i.e. $\eta^t = \eta$ for some $\eta>0$.

As shown in~\cref{sec:intro}, \cref{eq:general-linear-ttm} is equivalent to applying subgradient descent method on the following non-smooth convex minimization problem: 
\begin{equation}
    \min\nolimits_{p \in \RR^m_{\geq 0}} \quad \varphi(p). 
    \label{pgm:general-nonsmooth-convex-program}
\end{equation}
In particular, starting with some feasible point $p^0 \in \RR^m_{\geq 0}$, 
at time $t$, we update $p$ as follows:
\begin{equation}
    \begin{array}{rl}
        \bullet & \text{Pick}\; g(p^t) \in \partial \varphi(p^t) \\ [3pt]
        \bullet & p^{t+1} = \Pi_{\RR^m_{\geq 0}}\left( p^t - \eta g(p^t) \right).
    \end{array}
    \label{eq:ttm-subgradient-update}
\end{equation}

\section{Convergence of {\Tatonnement}} 
\label{sec:conv-analysis}
In this section, we establish our main result: showing that the {\tatonnement} process converges to an arbitrarily small neighborhood of the equilibrium price vector at a linear rate, where the size of the neighborhood is determined by the stepsize. 
To do this, we show that the corresponding {\subGD} guarantees such a linear convergence. 
Specifically, we first show that the prices generated by the {\tatonnement} dynamics are lower bounded by positive constants if stepsizes are small enough (without any other modifications to {\tatonnement}). 
Because prices are bounded away from zero, the demands are all upper bounded, and in turn this implies that  {\tatonnement} is equivalent to {\subGD} with bounded subgradients. This is crucial, because bounded subgradients are required in all last-iterate convergence results for {\subGD} that we are aware of.
Next, we establish a quadratic growth (QG) property of the non-smooth EG dual objective function. 
Finally, we obtain the desired linear convergence by a standard analysis for {\subGD} under QG and bounded subgradients.
Then, we show an example where the prices cycle in a small neighborhood of the equilibrium prices, exactly as predicted by our theory once we hit the neighborhood implied by the stepsize. 

\paragraph{Price lower bounds.}
The expression for $\varphi(p)$ in~\cref{eq:linear-dual-lagrangian} can be simplified into $\sum_{j=1}^m p_j - \sum_{i=1}^n B_i \log \big( \min_j \frac{p_j}{v_{ij}} \big)$ up to a constant, which yields the following dual formulation of~\cref{eq:eg-primal}:
\begin{align}
    \min_{p \in \RR^m_{\geq 0}} \, \varphi(p) = \sum_{j=1}^m p_j - \sum_{i=1}^n B_i \log \left( \min_{k \in [m]} \frac{p_k}{v_{i k}} \right). 
    \label{eq:linear-convex-minimization}
\end{align}
Here, the two expressions for $\varphi$ differ only by a constant. 
As a result, {\tatonnement} in an LFM is equivalent to SubGD on $\min_{p \in \mathbb{R}^m_{\geq 0}} \varphi(p)$. 
In~\cref{app:sec:equiv-ttm-sgd}, we show how to directly derive this well-known equivalence in terms of $\varphi(p)$.

As we mentioned, to ensure that the {\tatonnement} process
is well-defined and converges to a market equilibrium, some past work (e.g.~\citet{cole2019balancing}) assumed an artificial lower bound on prices. This is to ensure that we do not encounter unbounded demands (which occur as prices tend to zero).
Next, we show that such artificial lower bounds are not necessary: the {\tatonnement} process itself guarantees that prices are bounded away from zero, as long as our stepsize is not too large. 
This implies a global upper bound on demands at each step. 
For the proof of this result and other omitted proofs in this section, see~\cref{app:sec:proofs-conv-analysis}.

\begin{lemma}
    Let $\lbp$ be an $m$-dimensional vector where 
    \begin{equation}
        \lbp_j = \frac{B_{\min}}{4m} \max_i  \frac{v_{ij}}{\norm{v_i}_\infty}, \quad \mbox{ for all } j \in [m]. 
        \label{eq:def-tilde-p}
    \end{equation}
    Assume that we adjust prices in the LFM with~\cref{eq:general-linear-ttm}, 
    starting from an initial price vector $p^0 \geq \lbp$, 
    with any stepsize $\eta < \frac{1}{2m} \min_j \lbp_j$. 
    \setlength{\leftmargini}{25pt} 
    Then, we have for all $t \geq 0$, 
    \begin{itemize}
        \item[(i)] $p^t_j \geq \lbp_j - 2m\eta$ for all $j \in [m]$; 
        \item[(ii)] $\En{z(p^t)} \leq \frac{B}{\min_j \tilde{p}_j - 2m\eta} + m$; 
        \item[(iii)] $p^t_j \leq \left( 1 + \frac{\eta}{\tilde{p}_j - 2m\eta} \right) B$ for all $j \in [m]$.
    \end{itemize}
    \label{lem:lower-bounds-prices-upper-bounds-demands-upper-bounds-prices}
\end{lemma}

Note that demands in LFMs are non-unique for a given set of prices, and thus there are multiple trajectories that the prices can take. Our analysis does not impose any assumptions on the choices of demands and works for all possible trajectories.
We also show that the lower bound on price in \cref{lem:lower-bounds-prices-upper-bounds-demands-upper-bounds-prices} is tight, in terms of its order dependence on $m$.
\begin{lemma}
    There is an instance for which $p^t = \lbp - (m-1)\eta$ for some item at some time step $t$. 
    \label{lem:lower-bounds-prices-tightness}
\end{lemma}


\paragraph{Quadratic growth of the dual EG objective.}
To guarantee a fast convergence rate for (sub)gradient descent methods, 
we need the objective function to satisfy some form of metric regularity condition. 
See~\citet[Appendix A]{karimi2016linear} for a comprehensive comparison of different types of regularity conditions. 
Here, we focus on the \emph{strong convexity} (SC) and \emph{quadratic growth} (QG) conditions. 
A differentiable convex function $f$ is said to be $\mu$-strongly convex if 
\begin{equation*}
    f(y) \geq f(x) + \inp{\nabla f(x)}{y - x} + \frac{\mu}{2} \En{y - x}^2 \; \forall\; x, y \in \text{dom } f. 
\end{equation*}
A closed convex function $\varphi$ is said to have the $\alpha$-quadratic growth property if 
\begin{equation*}
    \varphi(x) - \varphi(\Pi_{\mathcal{X}^*}(x)) \geq \alpha \En{x - \Pi_{\mathcal{X}^*}(x)}^2 \quad \forall\; x \in \text{dom } \varphi, 
\end{equation*}
where $\mathcal{X}^*$ is the set of the minima of $\varphi$.
It is easy to see that $\mu$-strong convexity implies $\frac{\mu}{2}$-quadratic growth.
Next, we show that $\varphi(p)$ in~\cref{eq:linear-convex-minimization} satisfies QG.

\begin{lemma}
    The convex function 
    \begin{align*}
        \varphi(p) =  \sum_{j = 1}^m p_j - \sum_{i = 1}^n B_i \log{\left( \min_{k \in [m]} \frac{p_k}{v_{i k}} \right)} 
        \label{eq:sdeg-obj} 
    \end{align*} 
    satisfies the quadratic growth condition with modulus 
    $\alpha = 
    \min_j \frac{p^*_j}{2 \big(1 + \frac{\eta}{\tilde{p}_j - 2m\eta} \big)^2 B^2}$, 
    where $p^*$ denotes the unique equilibrium price vector, and $\tilde{p}$ is defined in~\cref{lem:lower-bounds-prices-upper-bounds-demands-upper-bounds-prices}. 
    \label{lem:qg}
\end{lemma} 

\begin{proof}
Let $\bar{p}_j = (1 + \frac{\eta}{\tilde{p}_j - 2m\eta} ) B$ denote the price upper bound for item $j$. 
First, we construct an auxiliary function 
\begin{equation*}
    h(p) = \sum_{j = 1}^m p_j - \sum_{j = 1}^m \sum_{i = 1}^n p^*_j x^*_{ij} \log{ \left( \frac{p_j}{v_{ij}} \right) }, 
\end{equation*} 
where $(p^*, x^*)$ is a pair of equilibrium prices and allocations. 
Since $h(p)$ is strictly convex, it has a unique minimum. 
Note that, since $\nabla_{p_j} h(p) = 1 - \frac{p^*_j}{p_j}\sum_i x^*_{ij} = 1 - \frac{p^*_j}{p_j}$ for all $j \in [m]$, 
by first-order optimality conditions we have that $p^*$ is the minimum of $h(p)$. 
Then, we have 
\begin{align*} 
    \varphi(p) &= \sum_{j = 1}^m p_j - \sum_{i = 1}^n B_i \log{ \left( \min_{k \in [m]} \frac{p_k}{v_{i k}} \right) } \nonumber \\ 
    &= \sum_{j = 1}^m p_j - \sum_{i = 1}^n \sum_{j = 1}^m p^*_j x^*_{ij} \log{ \left( \min_{k \in [m]} \frac{p_k}{v_{i k}} \right) } \tag{$\sum_{j = 1}^m p^*_j x^*_{ij} = B_i$} \nonumber \\ 
    &\geq \sum_{j = 1}^m p_j - \sum_{j = 1}^m \sum_{i = 1}^n p^*_j x^*_{ij} \log{ \left( \frac{p_j}{v_{ij}} \right) } = h(p). 
\end{align*} 
Also, note that 
\begin{align*} 
    \varphi(p^*) &= \sum_{j = 1}^m p^*_j - \sum_{j = 1}^m \sum_{i: x^*_{ij} > 0} p^*_j x^*_{ij} \log{ \left( \min_{k \in [m]} \frac{p^*_k}{v_{i k}} \right) } \nonumber \\ 
    &= \sum_{j = 1}^m p^*_j - \sum_{j = 1}^m \sum_{i: x^*_{ij} > 0} p^*_j x^*_{ij} \log{ \left( \frac{p^*_j}{v_{i j}} \right) } = h(p^*), 
\end{align*} 
where the second equality follows because $x^*_{ij} > 0$ implies ${p^*_j}/{v_{ij}} = \min_{k \in [m]} {p^*_k}/{v_{ik}}$.  

Since $p$ is upper bounded by $\bar{p}$, 
$h$ is strongly convex with modulus $\mu = \min_j {p^*_j}/{\bar{p}_j^2}$ (see \cref{app:sec:strong-convexity} for derivation). 
Since strong convexity implies the quadratic growth condition, we obtain 
$h(p) - h(p^*) \geq (\min_j {p^*_j}/{2\bar{p}_j^2}) \En{p - p^*}^2$ for all $p \leq \bar{p}$. 
Since $\varphi(p) \geq h(p)$ for all $p$, and $\varphi(p^*) = h(p^*)$ at the unique minimum $p^*$, 
\begin{equation*}
    \varphi(p) - \varphi(p^*) \geq h(p) - h(p^*) 
    \geq \min_j \frac{p^*_j}{2\bar{p}_j^2} \En{p - p^*}^2, \; \forall\, p \leq \bar{p}. 
\end{equation*}
Therefore, $\varphi$ is a convex function satisfying the QG condition with modulus $\alpha = \min_j \frac{p^*_j}{2\bar{p}_j^2}$. 
\end{proof}

\paragraph{Approximate linear convergence.}
Next we show that {\tatonnement} converges at a linear rate, due to our results on bounded prices and quadratic growth.
The proof of this theorem is typical for subgradient descent methods under these conditions, 
and easily follows from our structural results on {\tatonnement} and the EG dual. We defer it to~\cref{app:sec:proofs-conv-analysis}.
\begin{theorem}
    Assume that we adjust prices in the linear Fisher market with~\cref{eq:general-linear-ttm}, 
    starting from an initial price vector $p^0 \geq \tilde{p}$, 
    with any stepsize $\eta < \frac{1}{2m} \min_j \tilde{p}_j$, where $\tilde{p}_j$ is defined          
    in~\cref{lem:lower-bounds-prices-upper-bounds-demands-upper-bounds-prices}. 
    Then, we have 
    \begin{equation*}
        \En{p^t - p^*}^2 
        \leq (1 - 2\eta \alpha)^t \En{p^0 - p^*}^2 + e, \quad \text{for all } t \geq 0,
    \end{equation*}
    where $\alpha$ is defined in~\cref{lem:qg} 
    and 
    $
        e = 
        \frac{\eta}{2\alpha} \big( \frac{B}{\min_j \tilde{p}_j - 2m\eta} + m \big)^2$.
    
    \label{thm:linear-convergence}
\end{theorem}


\cref{thm:linear-convergence} shows that the price vector $p^t$ converges linearly to an $\epsilon$ neighborhood around the equilibrium price vector $p^*$, where the size of the neighborhood is determined by the stepsize $\eta$.
Suppose instead that we want to get arbitrarily close to $p^*$. Next, we show a corollary of \cref{thm:linear-convergence} for this case.
\begin{corollary}
    For a given $\epsilon > 0$, 
    let 
    $
        \eta \leq 
        \min\left\{ \frac{\min_j \tilde{p}_j}{4m}, 
        \frac{2 \epsilon \min_j p^*_j}{9 B^2 (\frac{2B}{\min_j \tilde{p}_j} + m)^2} \right\}
    $.
    Then, starting from any $p^0 \geq \tilde{p}$ such that $\sum_{j = 1}^m p^0_j = B$, 
    {\tatonnement} with stepsize $\eta$ generates a price vector $p$ 
    such that $\En{p - p^*}^2 \leq \epsilon$ 
    in $\mathcal{O}(\frac{1}{\epsilon}\log{\frac{1}{\epsilon}})$ iterations.
    If $\epsilon \geq \frac{9 B^2 \min_j \tilde{p}_j}{8m \min_j p^*_j} (\frac{2B}{\min_j \tilde{p}_j} + m)^2$, 
    the time complexity is on the order of $\mathcal{O}(\log{\frac{1}{\epsilon}})$.
    \label{crl:time-complexity}
\end{corollary}

Note that convergence measured via $\En{p - p^*}^2$ implies convergence to an approximate equilibrium.
This is explained in~\cref{app:sec:more-detailed-discussion}. 
We also discuss the possibility of showing convergence of {\tatonnement} without considering regularity conditions (e.g., QG) in~\cref{app:sec:more-detailed-discussion}.

\paragraph{A non-convergence example and discussion.} \label{subsec:non-conv-example-discussion}
The non-convergence of {\tatonnement} under linear utilities has been discussed in previous work, e.g., \citet[Section 5]{cole2008fast}, \citet[Section 6]{cheung2019tatonnement}, \citet[Section 6]{avigdor2014convergence}.
\citet{cole2019balancing} gives a simple market instance (Example 1 in that paper) showing that a \emph{multiplicative} form of {\tatonnement} results in prices cycling between off-equilibrium prices.
Here, we provide a similar, simple instance to demonstrate the non-convergence of the additive form of {\tatonnement} \cref{eq:general-linear-ttm} under \emph{any} stepsize $\eta>0$.
Consider a single buyer with a \emph{unit} budget and two items with valuations $v_{11} = v_{12} = 0.5$. 
The initial prices are chosen as $p^0 = (0.5, 0.5 + \eta)$.
It can be easily verified that the equilibrium prices $p^* = (0.5, 0.5)$.
For a small $\eta>0$, {\tatonnement} gives 
\[ p^t = \begin{cases}
    (0.5+\eta, 0.5), & t = 1, 3, 5, \dots\\
    (0.5, 0.5+\eta), & t = 2, 4, 6, \dots
\end{cases} \]
and, as can be verified via straightforward calculations, $p_1^{2t+1} \rightarrow 0.5$, $p_1^{2t} \rightarrow 0.5+\eta$.
Note, though, that the region of cycling around the equilibrium prices is proportional to $\eta$, which is consistent with \cref{thm:linear-convergence}. 

Although prices cycle with a fixed stepsize (as expected since $e$ in~\cref{thm:linear-convergence} is proportional to the stepsize), they do converge to the exact equilibrium prices with adaptive stepsizes. See~\cref{app:sec:more-detailed-discussion} for more details on this point. 

\section{Quasi-Linear Fisher Markets} 
\label{sec:ql}
The quasi-linear Fisher market (QLFM) is an extension of the linear Fisher market which considers buyers who have outside value for money. 
Formally, 
given a set of prices $p = \{ p_j \}_{j\in [m]} \geq 0$ on all items, 
the \emph{demand} of buyer $i$ is  
\begin{align} 
    \cD^q_i(p) = \argmax_{x_i} \left\{\inp{v_i - p}{x_i}: \inp{p}{x_i} \leq B_i,\, x_i \geq 0 \right\}. 
    \nonumber
\end{align}
In the QLFM, 
if the price for every item is greater than the value to the buyer, the buyer prefers not to spend her money. 
Thus, unlike for the LFM, buyers may not spend all their money in equilibrium. 
However, as in the LFM, buyers only spend money on their MBB items,
and the market clears in equilibrium. 

Recently, there has been increasing interest in the QL Fisher market, due to its relation to budget management in internet advertising markets. 
\citet{conitzer2022pacing} showed that the \emph{first price pacing equilibrium} (FPPE) concept, which captures steady-state outcomes of typical budget pacing algorithms for a first-price auction setting,
is equivalent to market equilibrium in the QLFM.
Consequently, FPPE can be computed efficiently by solving EG-like convex programs. 

As shown in~\citet{chen2007note,devanur2016new}, market equilibria in the QLFM can be captured by 
the following pair of EG-like primal and dual problems: 
\begin{align}
    \begin{aligned}
        \max_{x \in \RR^{n \times m}_{\geq 0}, u, y \in \RR^n_{\geq 0}} \quad & \sum\nolimits_i \left( B_i \log{(u_i)} - y_i \right) \\ 
        \textnormal{s.t.} \quad & u_i \leq \inp{v_i}{x_i} + y_i \quad \forall\, i \in [n] \\ 
        & \sum\nolimits_i x_{ij} \leq 1 \quad \forall\, j \in [m],
    \end{aligned}
    \tag{ql-P}
    \label{pg:ql-P}
\end{align}
\begin{align}
    \begin{aligned}
        \min_{\beta \in \RR^n_{\geq 0}, p \in \RR^m_{\geq 0}} \quad & \sum\nolimits_j p_j - \sum\nolimits_i B_i \log{\beta_i} \\ 
        \textnormal{s.t.} \quad & v_{ij} \beta_i \leq p_j \quad \forall\, i \in [n], j \in [m] \\ 
        & \beta_i \leq 1 \quad \forall\, i \in [n]. 
    \end{aligned}
    \tag{ql-D}
    \label{pg:ql-D} 
\end{align}


\paragraph{Equivalence between {\tatonnement} and SubGD.} 
Similar to~\cref{eq:linear-dual-lagrangian}, the dual problem can be expressed as 
minimizing a max-type nonsmooth convex function $\varphi^q(p)$ as follows: 
\begin{align*}
    \max_{x, y \geq 0} \sum_{i=1}^n & \Big( B_i \log{\left(\inp{v_i}{x_i} + y_i\right)} - y_i - \sum_{j=1}^m p_j x_{ij} \Big) + \sum_{j=1}^m p_j.
    \label{eq:quasi-linear-lagrangian}
\end{align*}
By the first-order optimality conditions, 
we have that for any $p$, 
any maximizer $\bar{x}$ of the inner maximization problem corresponds to a set of demands under price $p$,
and the corresponding maximizer $\bar{y}$ denotes leftover budgets under $p$.
Since any subgradient of $\varphi^q(p)$ is given by $\mathbf{1} - \sum_{i=1}^n \bar{x}_i$, 
the subgradient equals the excess supply under price $p$.

The EG dual problem for QLFM can be simplified as 
\begin{equation}
    \min_{p \in \RR^m_{\geq 0}} \varphi^q(p) =  \sum_{j = 1}^m p_j - \sum_{i = 1}^n B_i \log{\left( \min\left\{ \min_{k \in [m]} \frac{p_k}{v_{i k}}, 1 \right\} \right)}. 
    \label{pg:mincvx-quasi-linear} 
\end{equation} 
Again, the two expressions for $\varphi^q(p)$ only differ by a constant.
Thus, we have $\partial \varphi^q(p) = -z^q(p)$ for any $p$. 
Consequently, {\tatonnement} in the QLFM is equivalent to SubGD on~\cref{pg:mincvx-quasi-linear}. 


\paragraph{Convergence analysis.}

By leveraging the same framework of analysis as for LFMs, 
we show a linear convergence rate and corresponding time complexity for {\tatonnement} in QLFMs.
We show these results in~\cref{thm:linear-convergence-ql} and~\cref{crl:time-complexity-ql}, 
and defer all the analyses and proofs to~\cref{app:sec:quasi-linear-proofs}.

\begin{theorem}
    Let $\pmb{b}(j)$ be the buyer with the minimum index such that
    \begin{equation*}
        \pmb{b}(j) \in \Argmax_i \frac{v_{i j}}{\|v_i\|_\infty}
    \end{equation*}
    and $v_{\min} = \min_j v_{\pmb{b}(j) j}$.
    Assume that we adjust prices in the QLFM with~\cref{eq:general-linear-ttm}, 
    starting from any initial price $p^0 \geq \hat{p}$, with stepsize $\eta < \frac{1}{2m} \min_j \hat{p}_j$, where 
    $\hat{p}$ is an $m$-dimensional price vector where 
    $\hat{p}_j = \min\left\{ \frac{B_{\min}}{4m}, \frac{v_{\min}}{2} \right\} \max_i \frac{v_{ij}}{\En{v_i}_\infty}, \;\forall\, j \in [m]$.
    Then, we have 
    \begin{equation}
        \En{p^t - p^*}^2 
        \leq (1 - 2\eta \alpha)^t  
            \En{p^0 - p^*}^2 + e \quad \text{for all } t \geq 0, 
    \end{equation}
    where $\alpha = \min_j 
    \frac{p^*_j}{2 \big(\min\{ \max_i v_{i j} , B \} + \frac{\eta B}{\hat{p}_j - 2m\eta} \big)^2}$ and 
    $e = 
        \frac{\eta}{2\alpha} \left( \frac{B}{\min_j \hat{p}_j - 2 m\eta} + m \right)^2$.
    \label{thm:linear-convergence-ql}
\end{theorem} 

\begin{corollary}
    For a given $\epsilon > 0$, 
    let 
    $\eta \leq 
        \min\{ \frac{\min_j \hat{p}_j}{4m}, 
        \frac{2 \epsilon \min_j p^*_j}{9 B^2 (\frac{2B}{\min_j \hat{p}_j} + m)^2} \}$.
    Then, starting from any $p^0 \geq \hat{p}$ such that $\sum_j p^0_j = B$, 
    {\tatonnement} with stepsize $\eta$ generates a price vector $p$ 
    such that $\En{p - p^*}^2 \leq \epsilon$ 
    in $\mathcal{O}(\frac{1}{\epsilon}\log{\frac{1}{\epsilon}})$ iterations.
    If $\epsilon \geq \frac{9 B^2 \min_j \hat{p}_j}{8m \min_j p^*_j} (\frac{2B}{\min_j \hat{p}_j} + m)^2$, 
    the time complexity is on the order of $\mathcal{O}(\log{\frac{1}{\epsilon}})$.
    \label{crl:time-complexity-ql}
\end{corollary}

\section{Numerical Experiments} \label{sec:experiments}
In this section, we demonstrate the convergence of {\tatonnement} for the LFM and QLFM through numerical experiments.
To compute the ``true'' equilibrium prices, we run proportional response (PR) dynamics~\citep{zhang2011proportional} for LFM, and the quasi-linear variant of PR~\citep{gao2020first}, which gives highly accurate solutions for our instances.
We use the \emph{squared error norm}, $\| p^t - p^* \|^2$, as the measure of convergence. 
For each market class, instance size, and step size, we run {\tatonnement} for a large number of iterations until the error residuals do not decrease further. 



\paragraph{Synthetic  data.}
We first test {\tatonnement} on randomly generated instances. 
We draw budgets $\{ B_i \}_{i \in [n]}$ from the $[0, 1]$ uniform distribution and valuations $\{ v_{i, j} \}_{i \in [n], j \in [m]}$ from various distributions, 
and then normalize them to $\sum_{j \in [m]} v_{ij} = 1$ for each $i$.
In the LFM, we also normalize the budgets to $\sum_{i \in [n]} B_i = 1$.
In~\cref{fig:single-instance-error-norms-linear,fig:single-instance-error-norms-ql}, we plot $\|p^t - p^*\|^2$ across iterations for {\tatonnement} with different stepsize in the LFM and QLFM, respectively.
The transparent dotted lines represent the values of $\|p^t - p^*\|^2$ for all iterations, and the solid lines represent those values for uniformly spaced iterations.
See the rest of results for more instance sizes and distributions in~\cref{app:sec:additional-experiments}.

\begin{figure}[t]
    \centering
    \includegraphics[scale=.235]{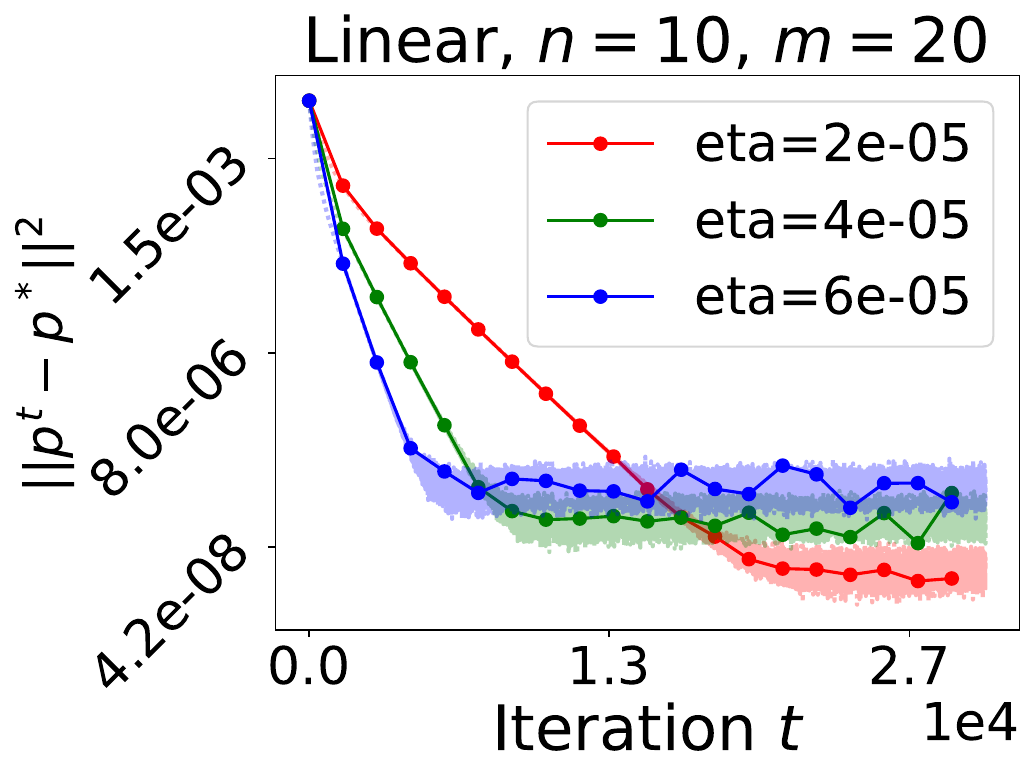}
    \includegraphics[scale=.235]{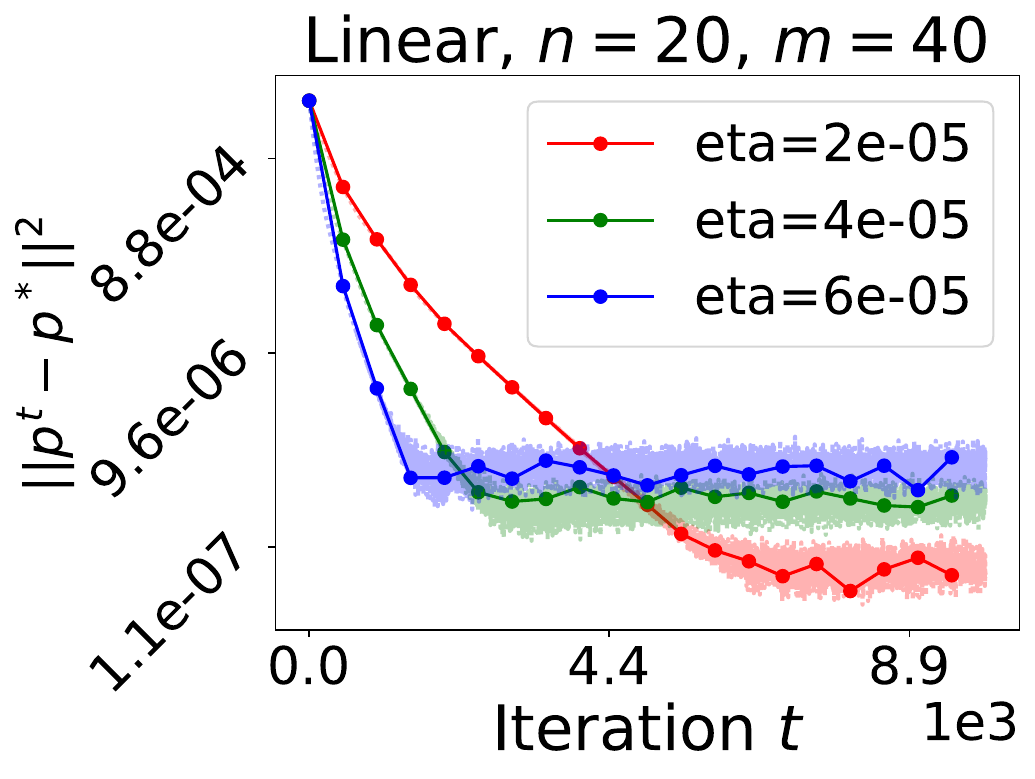}

    \includegraphics[scale=.235]{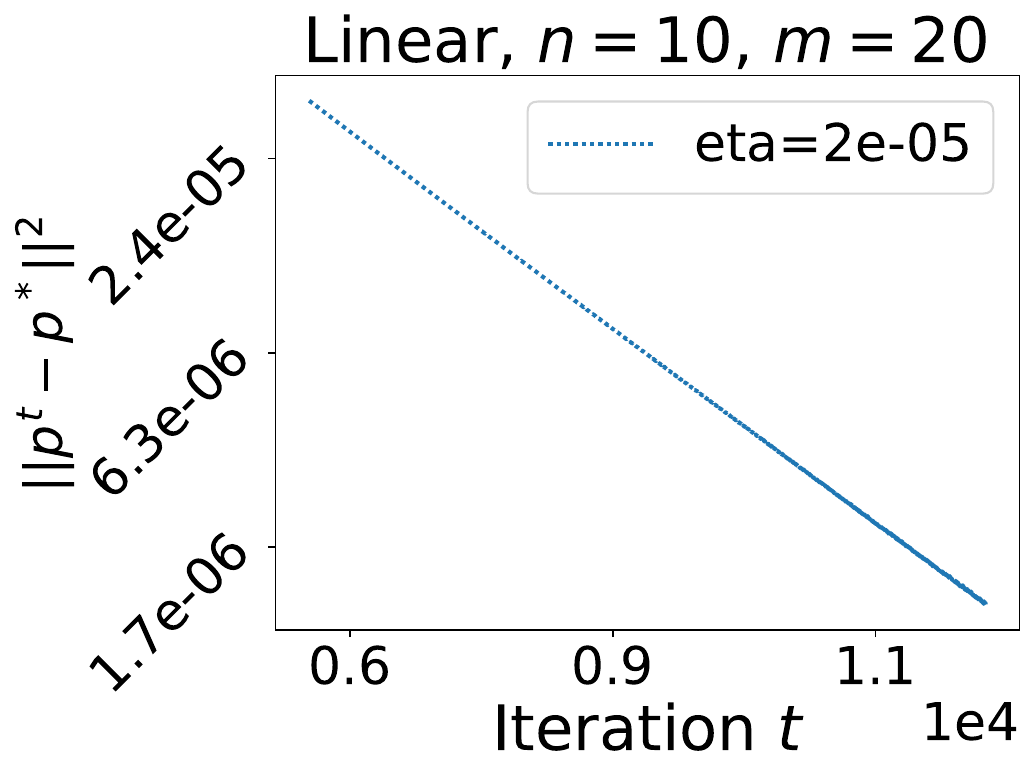}
    \includegraphics[scale=.235]{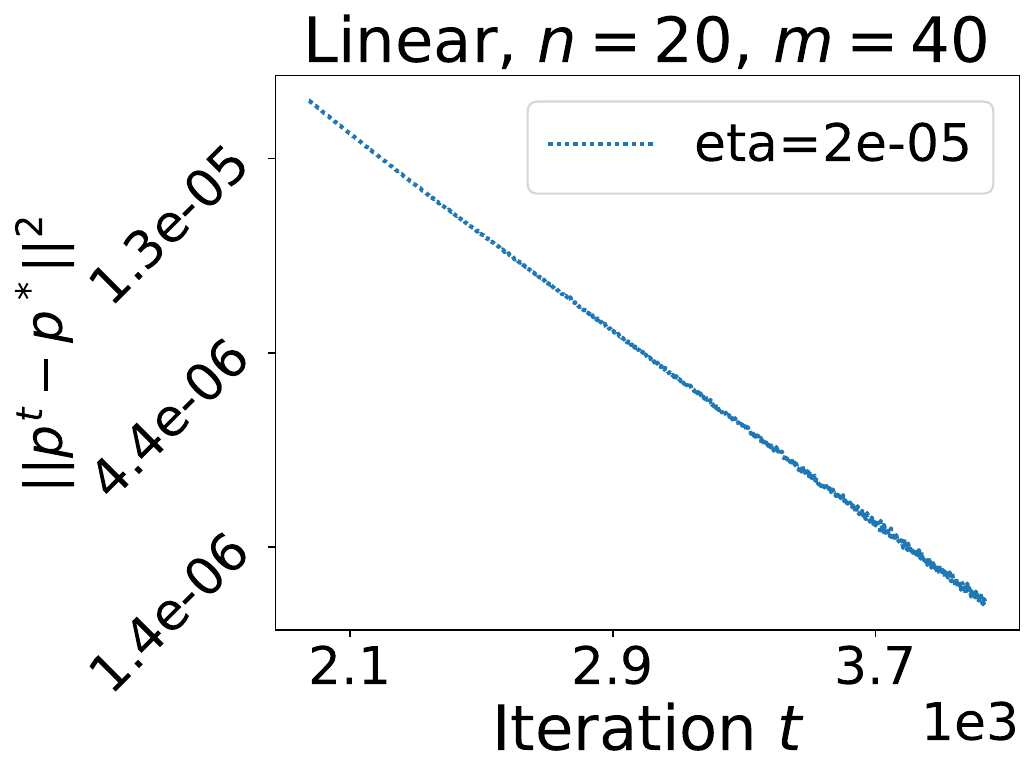}
    \caption{Convergence of squared error norms on random generated instances ($v$ is generated from the exponential distribution with scale $1$) of different sizes under linear utilities. 
    The bottom row zooms in on partial iterations to better show the initial linear convergence.
    }
    \label{fig:single-instance-error-norms-linear}
\end{figure}

\begin{figure}[t]
    \centering
    \includegraphics[scale=.235]{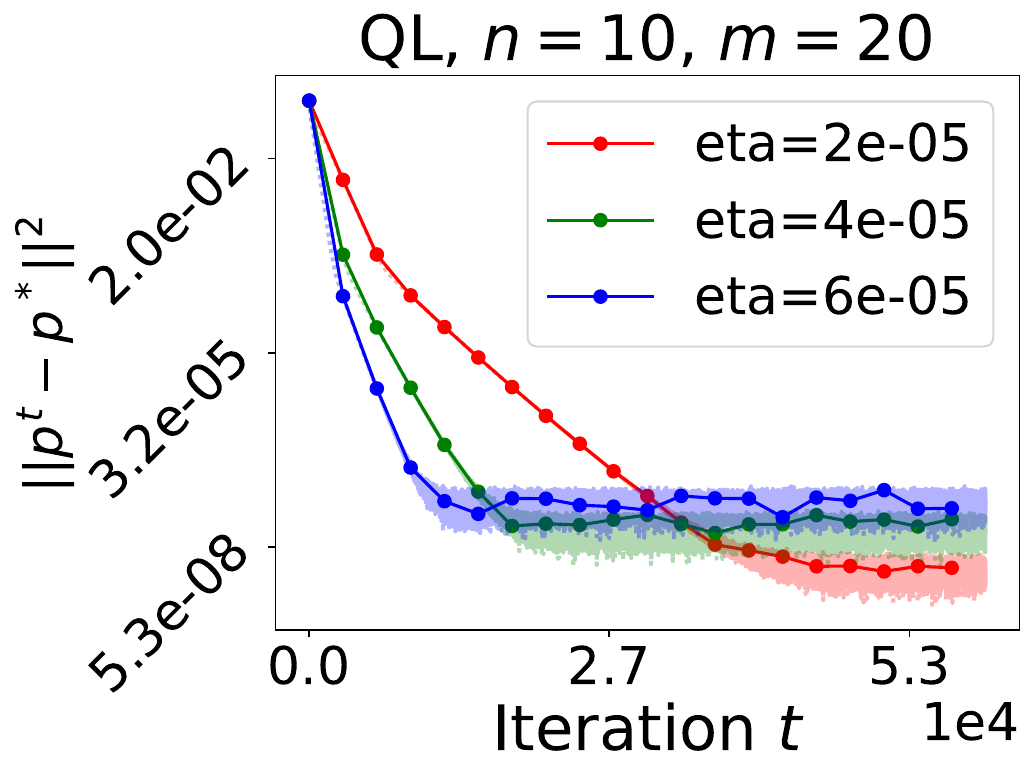}
    \includegraphics[scale=.235]{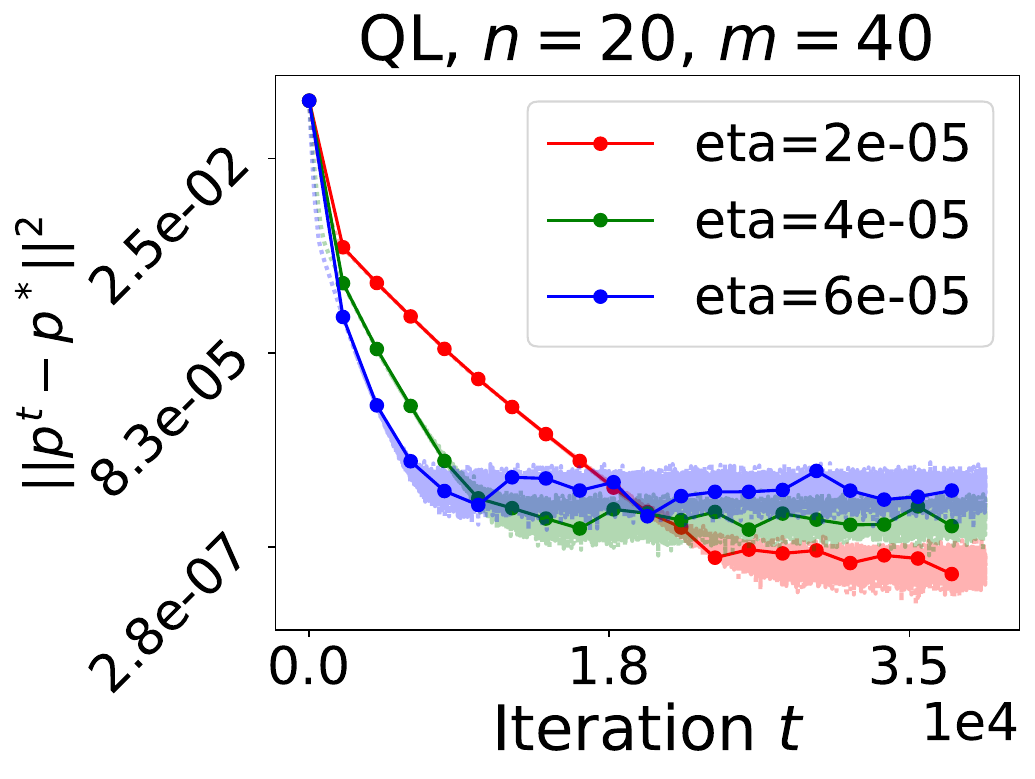}
    
    \includegraphics[scale=.235]{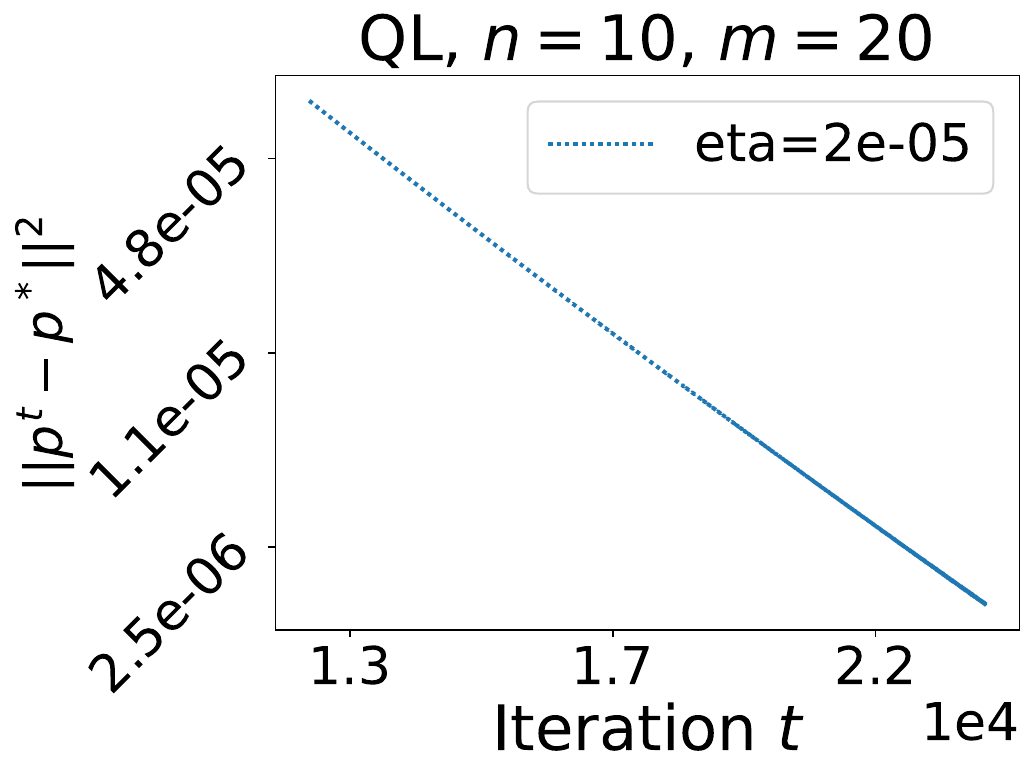}
    \includegraphics[scale=.235]{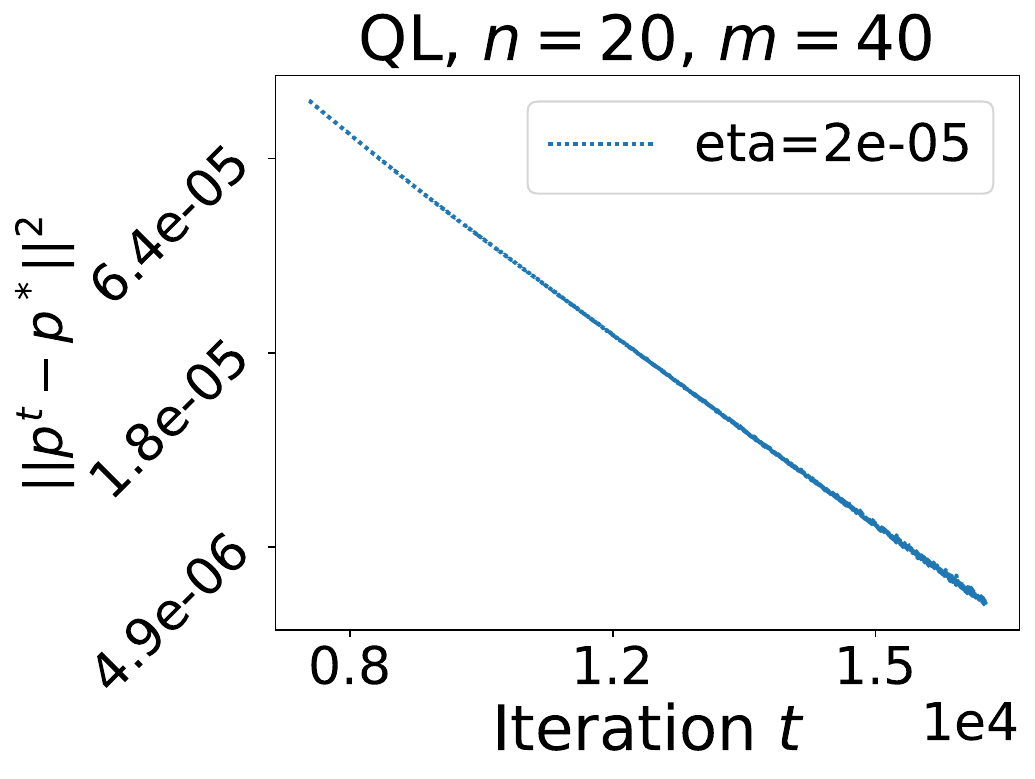}
    \caption{Convergence of squared error norms on random generated instances 
    of different sizes under QL utilities.
    }
    \label{fig:single-instance-error-norms-ql}
\end{figure}

\paragraph{Real data.}

We next test {\tatonnement} on a large-scale instance constructed from a movie rating dataset~\citep{Dooms13crowdrec,nan2023fast}.
In this instance, we regard each movie as an item and each rater as a buyer. 
The valuation of a buyer for an item is the rating she gives to the movie.
See a full description of this instance in~\cref{app:sec:additional-experiments}.
We run {\tatonnement} in both linear and quasi-linear settings and plot the results in~\cref{fig:real-data-experiments}.
\begin{figure}[t]
    \centering
    \includegraphics[scale=.235]{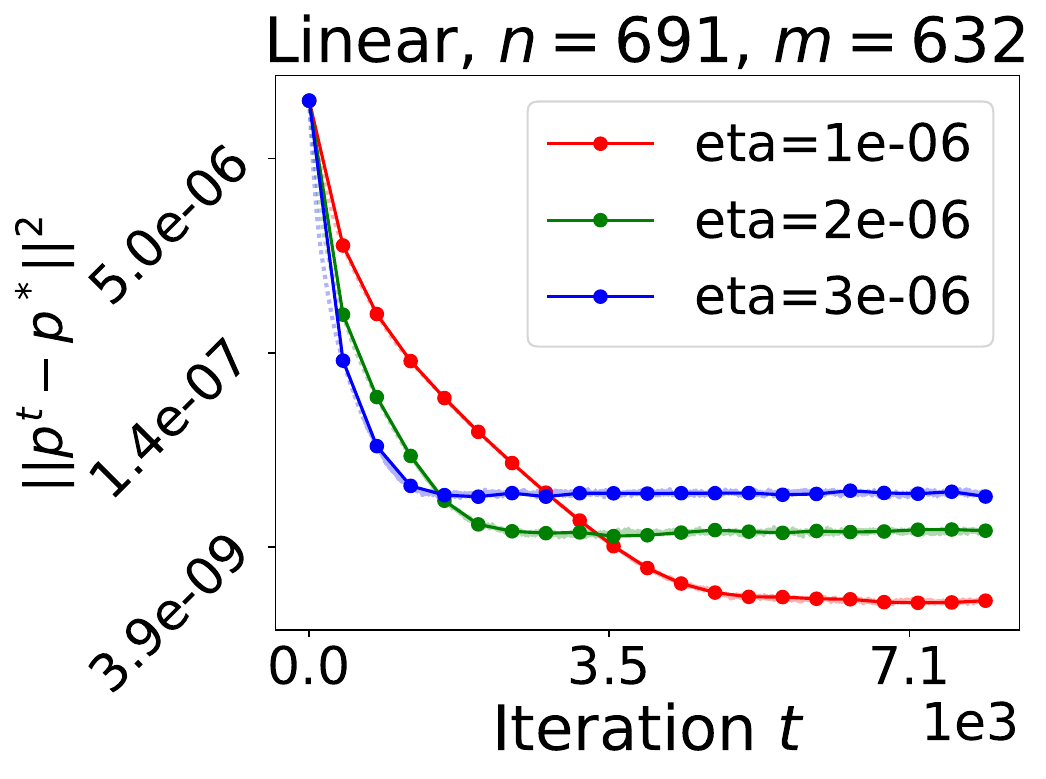}
    \includegraphics[scale=.235]{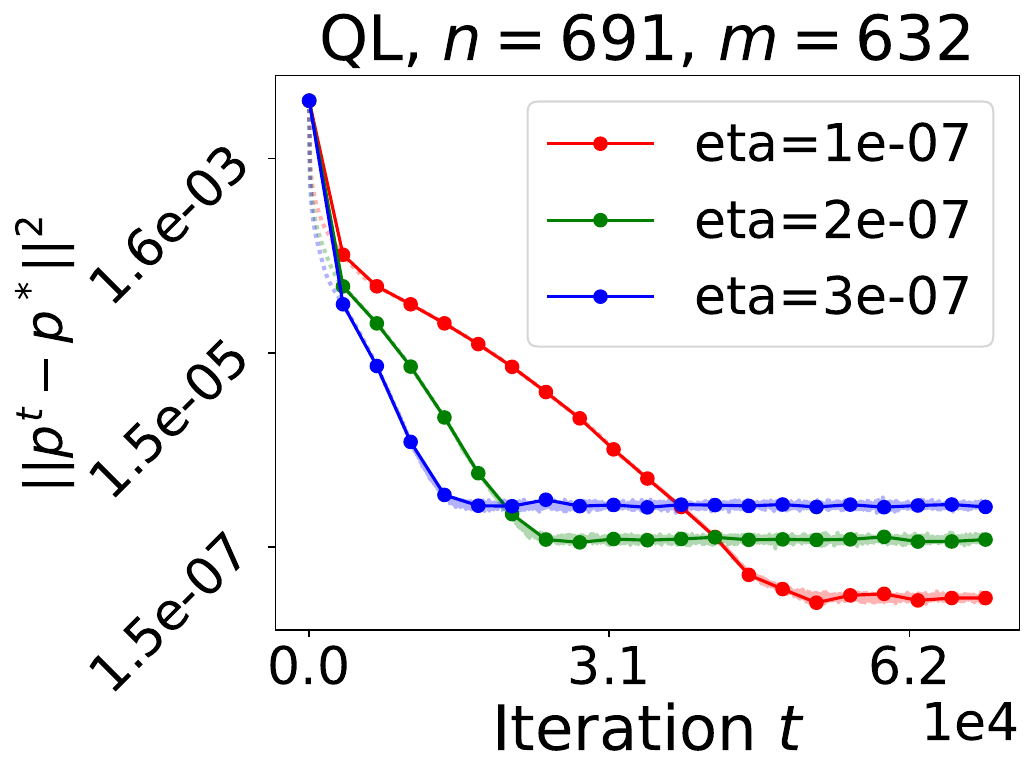}

    \includegraphics[scale=.235]{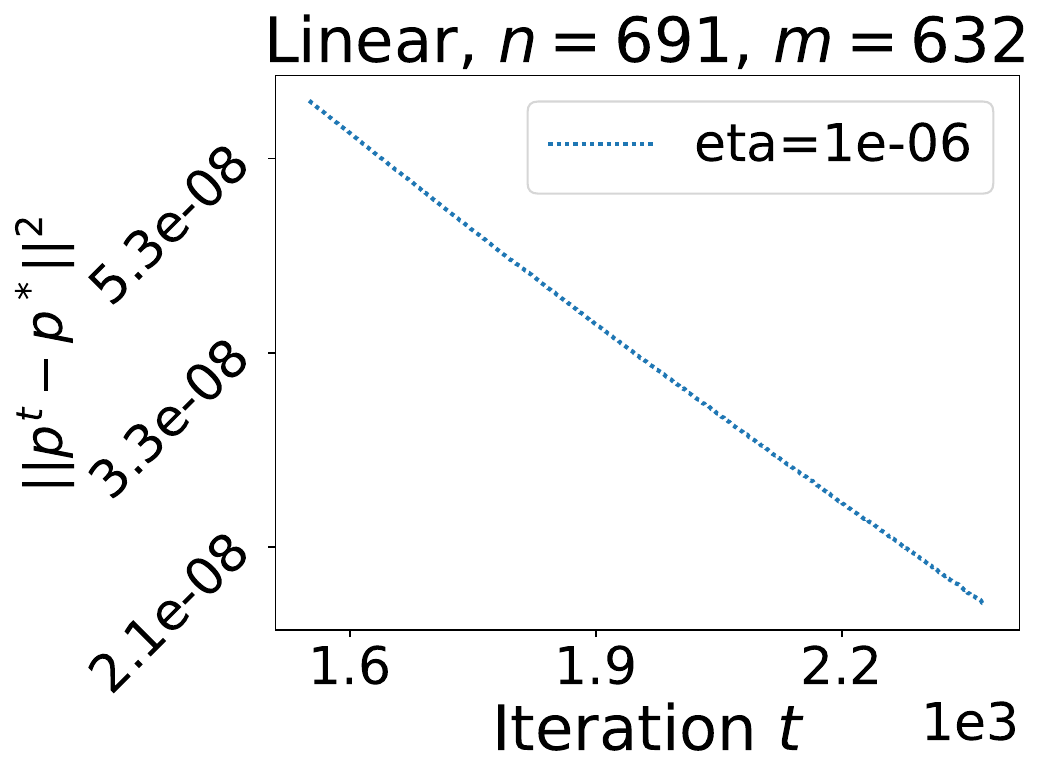}
    \includegraphics[scale=.235]{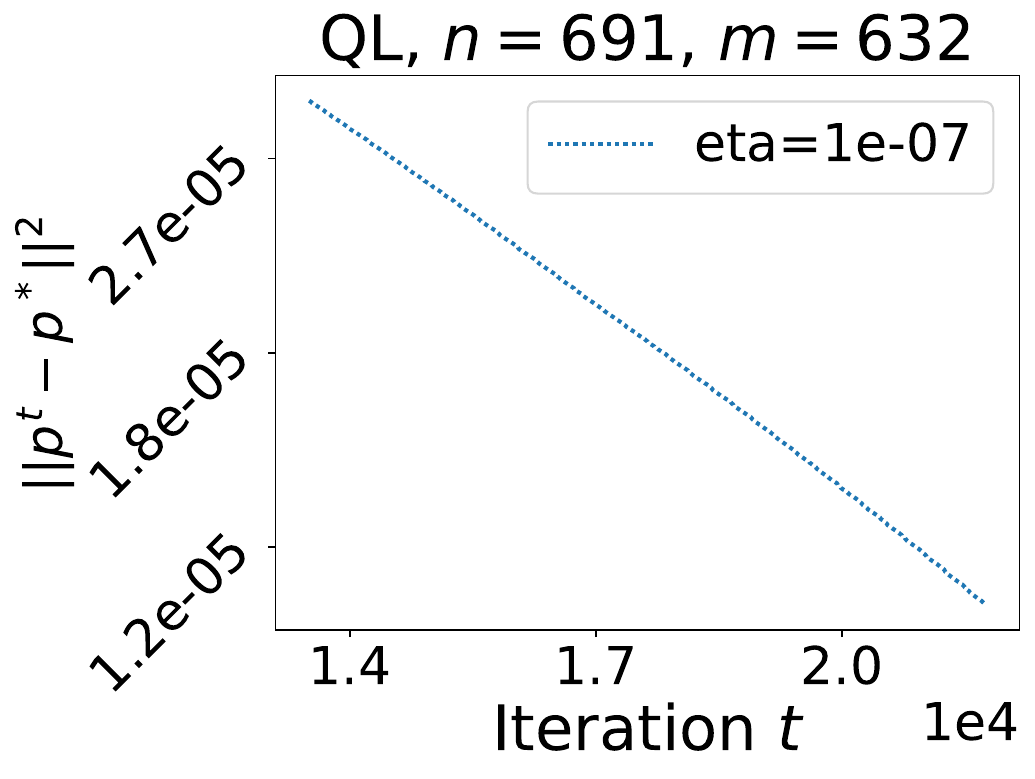}
    \caption{Convergence of squared error norms on our instances generated from a movie rating dataset.
    }
    \label{fig:real-data-experiments}
\end{figure}



\paragraph{Discussion.} 
As can be seen from~\cref{fig:single-instance-error-norms-linear,fig:single-instance-error-norms-ql,fig:real-data-experiments}, 
under both linear utilities and QL utilities, 
{\tatonnement} with all stepsizes converges to a small neighborhood around the equilibrium prices for all instances.
Moreover, larger stepsizes lead to faster initial convergence, 
but higher final error levels.
This is consistent with our theoretical results: 
when $\eta \ll \min_j \tilde{p}_j - 2m\eta$, 
then $\alpha \approx \min_j {p^*_j}/{2 B^2}$ is approximately a constant. 
Thus, $1 - 2\eta\alpha$ linearly decreases as $\eta$ increases;
and the final error level $e$ 
increases as $\eta$ increases.

In~\cref{fig:single-instance-error-norms-linear,fig:single-instance-error-norms-ql,fig:real-data-experiments}, 
We also zoom into the initial time steps of {\tatonnement}. 
As can be seen, 
the lines exhibit an approximately linear convergence pattern.
This demonstrates that our linear convergence results align with the practical performance of {\tatonnement}.
That is to say, 
we discover the ``true'' underlying convergence rate of t\^atonnement prior to hitting the oscillating neighborhood.

We also compare the performance of different variants of {\tatonnement} (e.g., those in~\cref{tab:comparison}) in~\cref{app:sec:compare-tatonnement}.

\section*{Ethical Statement}
There is no ethical concern in this work.

\section*{Acknowledgments}
This research was supported by the Office of Naval Research awards N00014-22-1-2530 and N00014-23-1-2374, and the National Science Foundation awards IIS-2147361 and IIS-2238960. The authors would like to thank the anonymous reviewers for their helpful comments and suggestions. 




\bibliography{refs}

\newpage
\appendix
\onecolumn

\FloatBarrier
\newpage



\begin{center}
    \huge \textbf{Appendices}
    \vspace{12pt}
\end{center}

\section{Direct equivalence between {\Tatonnement} and SubGD on CPF}
\label{app:sec:equiv-ttm-sgd}

In this section, we show a direct equivalence between {\tatonnement} process described in the main texts, without considering the first-order optimality of the EG-like programs. 

\begin{lemma}
    Let 
    \begin{equation*}
        f(x) = \min_j \mathbf{a}_j^\top x + b_j. 
    \end{equation*}
    Then, the subdifferential of $f(x)$ is 
    \begin{equation}
        \partial f(x) = \textnormal{conv}\left(\left\{ \mathbf{a}_{j^*} \left|\, \mathbf{a}_{j^*}^\top x + b_{j^*} = \min_j \mathbf{a}_j^\top x + b_j \right. \right\}\right).
    \end{equation}
\end{lemma}

\paragraph{Linear Fisher market.} 
Denote $\frac{1}{\mathbf{v}_{ij}}$ is an $m$-dimensional vector whose 
$j$-th entry is $\frac{1}{v_{ij}}$ and all other entries are $0$.
Let $\phi_i(p) = \min_{j \in [m]} \frac{p_j}{v_{ij}} = \min_{j \in [m]} \frac{1}{\mathbf{v}_{ij}}^\top p$. 
Denote $J^*_i(p) = \left\{ j^* \left|\, \frac{p_{j^*}}{v_{ij^*}} = \min_{j \in [m]} \frac{p_j}{v_{ij}} \right. \right\}$.
It follows that 
\begin{equation*}
    \partial \phi_i(p) = \text{conv}\left( \left\{ \frac{1}{\mathbf{v}_{ij}} \Bigg|\, j \in J^*_i(p) \right\} \right)
    = \left\{ \left( \frac{\lambda_{i1}}{v_{i1}}, \ldots, \frac{\lambda_{im}}{v_{im}} \right) 
    \left|\, \sum_{j=1}^m \lambda_{ij} = 1, \lambda_{ij} \geq 0, \lambda_{ij} = 0 \text{ if } j \notin J^*_i(p) \right. \right\}. 
\end{equation*}

Because $\varphi = \sum_j p_j - \sum_i B_i \log{\phi_i(p)}$, by the chain rule of subgradient methods we have 
\begin{align}
    \partial \varphi(p) 
    &= \mathbf{1} - \sum_i \frac{B_i}{\phi_i(p)} \partial \phi_i(p) \nonumber \\ 
    &= \mathbf{1} - \sum_i \left\{ \frac{B_i}{\phi_i(p)} \left( \frac{\lambda_1}{v_{i1}}, \ldots, \frac{\lambda_m}{v_{im}} \right) \left|\, \sum_{j=1}^m \lambda_j = 1, \lambda_j \geq 0, \lambda_j = 0 \text{ if } j \notin J^*_i(p) \right. \right\}. 
\end{align}

Note that $B_i \lambda_{ij} = b_{ij}$ and $\phi_i(p) v_{ij} = p_j$ for $j \in J^*_i(p)$.
Thus, $\partial \varphi(p)$ is equivalent to the excess supply (negative excess demand).

\paragraph{Quasi-linear Fisher market.} 
Let $\mathbf{a}_{i0}$ be an $m$-dimensional zero vector and 
$\mathbf{a}_{ij} = \frac{1}{\mathbf{v_{ij}}}$ for $j \in [m]$. 
Let $b_{i0} = 1$ if $b_{ij} = 0$ for $j \in [m]$.
Let $\phi^q_i(p) = \min\left\{ \min_{j \in [m]} \frac{p_j}{v_{ij}}, 1 \right\} = \min_{j = 0, 1, 2, \ldots, m} \mathbf{a}_{ij}^\top p + b_{ij}$. 
Denote $J^*_i(p) = \left\{ j^* \left|\, \mathbf{a}_{ij^*}^\top p + b_{ij^*} = \min_{j = 0, 1, 2, \ldots, m} \mathbf{a}_{ij}^\top p + b_{ij} \right. \right\}$.
Here, $0 \in J^*_i(p)$ means ``do not spend money'' is an optimal choice.
It follows that 
\begin{equation*}
    \partial \phi^q_i(p) 
    = \text{conv}\left( \left\{ \mathbf{a}_{ij} \big|\, j \in J^*_i(p) \right\} \right)
    = \left\{ \left( \frac{\lambda_{i1}}{v_{i1}}, \ldots, \frac{\lambda_{im}}{v_{im}} \right) 
    \left|\, \sum_{j=1}^m \lambda_{ij} \leq 1, \lambda_{ij} \geq 0; \lambda_{ij} = 0 \text{ if } j \notin J^*_i(p) \right. \right\}. 
\end{equation*}

Because $\varphi^q = \sum_j p_j - \sum_i B_i \log{\phi_i(p)}$, by the chain rule of subgradient methods we have 
\begin{align}
    \partial \varphi^q(p) 
    &= \mathbf{1} - \sum_i \frac{B_i}{\phi^q_i(p)} \partial \phi^q_i(p) \nonumber \\ 
    &= \mathbf{1} - \sum_i \left\{ \frac{B_i}{\phi^q_i(p)} \left( \frac{\lambda_1}{v_{i1}}, \ldots, \frac{\lambda_m}{v_{im}} \right) \left|\, \sum_{j=1}^m \lambda_j \leq 1, \lambda_j \geq 0, \lambda_j = 0 \text{ if } j \notin J^*_i(p) \right. \right\}. 
\end{align}

Note that $B_i \lambda_{ij} = b_{ij}$ and $\phi^q_i(p) v_{ij} = p_j$ for $j \in [m] \cap J^*_i(p)$.
Thus, $\partial \varphi^q(p)$ is equivalent to the excess supply (negative excess demand).

\section{Proofs in Section~\ref{sec:conv-analysis}}
\label{app:sec:proofs-conv-analysis}

\paragraph{Proof of~\cref{lem:lower-bounds-prices-upper-bounds-demands-upper-bounds-prices,lem:lower-bounds-prices-tightness}.}

\begin{customlem}{1}
    Let $\lbp$ be an $m$-dimensional vector where 
    \begin{equation}
        \lbp_j = \frac{B_{\min}}{4m} \max_i  \frac{v_{ij}}{\norm{v_i}_\infty}, \quad \mbox{ for all } j \in [m]. 
    \end{equation}
    Assume that we adjust prices in the LFM with~\cref{eq:general-linear-ttm}, 
    starting from an initial price vector $p^0 \geq \lbp$, 
    with any stepsize $\eta < \frac{1}{2m} \min_j \lbp_j$. 
    Then, we have for all $t \geq 0$, 
    \begin{itemize}
        \item[(i)] $p^t_j \geq \lbp_j - 2m\eta$ for all $j \in [m]$; 
        \item[(ii)] $\En{z(p^t)} \leq \frac{B}{\min_j \tilde{p}_j - 2m\eta} + m$; 
        \item[(iii)] $p^t_j \leq \Big( 1 + \frac{\eta}{\tilde{p}_j - 2m\eta} \Big) B$ for all $j \in [m]$.
    \end{itemize}
\end{customlem} 
\begin{proof}
    \begin{itemize}
        \item[(i)]
        To prove~\cref{lem:lower-bounds-prices-upper-bounds-demands-upper-bounds-prices} $(i)$, we first introduce some notations. 
        Without loss of generality, we let $\pmb{b}(j)$ be the buyer with the minimum index such that
        \begin{equation}
            \pmb{b}(j) \in \Argmax_i \frac{v_{i j}}{\|v_i\|_\infty}.
            \label{eq:def-buyer}
        \end{equation}
        Intuitively, buyer $\pmb{b}(j)$ is the buyer who values item $j$ the most 
        relative to their largest valuation.
        We use $b_{ij}^t$ to denote the spend from buyer $i$ for item $j$ at the $t$-th iteration. 
        If $b_{ij}^t > 0$, we say item $j$ receives spend $b_{ij}^t$ from buyer $i$. 
        We denote $\kappa = \frac{B_{\min}}{4m}$ for simplicity.
    
        Then, we can show the following two facts.
        \begin{fact}
            For each item $j \in [m]$, 
            if $p_j \leq \tilde{p}_j$, 
            then for any item $j' \neq j$ which is demanded by buyer $\pmb{b}(j)$, 
            we have 
            $p_{j'} \leq \lbp_{j'}$. 
            \label{fact:j-prime}
        \end{fact}
            \begin{proof}
            Since buyer $\pmb{b}(j)$ demands item $j'$, we have $p_{j'} \leq \frac{v_{\pmb{b}(j) j'}}{v_{\pmb{b}(j) j}} p_j$. 
            Then, it follows that 
            \begin{align} 
                p_{j'} \leq \frac{v_{\pmb{b}(j) j'}}{v_{\pmb{b}(j) j}} p_j \leq \frac{v_{\pmb{b}(j) j'}}{v_{\pmb{b}(j) j}} \lbp_j 
                &\stackrel{\text{\cref{eq:def-tilde-p}}}{=} \frac{v_{\pmb{b}(j) j'}}{v_{\pmb{b}(j) j}} \kappa \max_i \frac{v_{i j}}{\norm{v_i}_\infty} \nonumber \\ 
                &\stackrel{\text{\cref{eq:def-buyer}}}{=} \frac{v_{\pmb{b}(j) j'}}{v_{\pmb{b}(j) j}} \kappa \frac{v_{\pmb{b}(j) j}}{\norm{v_{\pmb{b}(j)}}_\infty} \nonumber \\ 
                &\;\;\, = \kappa \frac{v_{\pmb{b}(j) j'}}{\norm{v_{\pmb{b}(j)}}_\infty}
                \leq \kappa \max_i \frac{v_{i j'}}{\norm{v_i}_\infty} \stackrel{\text{\cref{eq:def-tilde-p}}}{=} \lbp_{j'}. 
                \label{eq:low-price-rule}
            \end{align}
            \end{proof}

        \begin{fact} 
            Let $t_j \geq 0$ be any time step. 
            For any item $j$, if $p^t_j \leq \tilde{p}_j$ for $t_j \leq t \leq t_j + 2m - 1$, 
            and $p^{t_j}_{j'} \geq \tilde{p}_{j'} - 2m\eta$ for all $j' \neq j$. 
            Then, 
            \begin{equation}
                \sum_{t=t_j}^{t_j+2m - 1} b_{\pmb{b}(j) j'}^t \leq 2 B_{\pmb{b}(j)} \quad \text{for all } j' \neq j.
            \end{equation}
            \label{fact:others-wont-get-too-much} 
        \end{fact} 
            \begin{proof}
                We prove this fact by contradiction. 
                Suppose that, at time step $\hat{t} \leq t_j + 2m - 1$, 
                item $j'$ receives more than $2 B_{\pmb{b}(j)}$ spend from buyer $\pmb{b}(j)$ cumulatively for the first time. 
                Then we have $\sum_{t=t_j}^{\hat{t}-1} b_{\pmb{b}(j) j'}^t \leq 2 B_{\pmb{b}(j)} < \sum_{t=t_j}^{\hat{t}} b_{\pmb{b}(j) j'}^t$. 
                Since item $j'$ receives a spend of at most $B_{\pmb{b}(j)}$ from buyer $\pmb{b}(j)$ per iteration, 
                we have $\sum_{t=t_j}^{\hat{t} - 1} b_{\pmb{b}(j) j'}^t > B_{\pmb{b}(j)}$. 
                By~\cref{fact:j-prime} and $p^t_j \leq \tilde{p}_j$ for $t_j \leq t \leq t_j + 2m - 1$, 
                we know $p^t_{j'} \leq \lbp_{j'}$ for $t_j \leq t \leq t_j + 2m - 1$ such that $b^t_{\pmb{b}(j) j'} > 0$. 
                Thus, 
                $\sum_{t=t_j}^{\hat{t}-1} \frac{b^t_{\pmb{b}(j) j'}}{p^t_{j'}} 
                \geq \frac{1}{{\tilde{p}_{j'}}} \sum_{t=t_j}^{\hat{t}-1} b^t_{\pmb{b}(j) j'} 
                \geq \frac{4m}{B_{\pmb{b}(j)}} \sum_{t=t_j}^{\hat{t}-1} b^t_{\pmb{b}(j) j'} > 4m$,
                because $\lbp_{j'} \leq \frac{B_{\min}}{4m} \leq \frac{B_{\pmb{b}(j)}}{4m}$. 
                This leads to 
                \begin{equation}
                    p^{\hat{t}}_{j'} \geq p^{t_j}_{j'} + \eta \sum_{t = t_j}^{\hat{t} - 1} \left( \frac{b^t_{\pmb{b}(j) j'}}{p^t_{j'}} - 1 \right) 
                                    > p^{t_j}_{j'} + 4m\eta - (2m - 1)\eta 
                                    = p^{t_j}_{j'} + (2m + 1) \eta
                                    \geq \lbp_{j'} + \eta. 
                \end{equation} 
                This contradicts $p^{\hat{t}}_{j'} \leq \lbp_{j'}$, 
                which follows from~\cref{fact:j-prime} and $b^{\hat{t}}_{\pmb{b}(j) j'} > 0$. 
            \end{proof}
    
        Now we are ready to prove \cref{lem:lower-bounds-prices-upper-bounds-demands-upper-bounds-prices}.       
        We assume there exists a $T = \min\big\{ t \,\big\vert\, p^t_j < \tilde{p}_j - 2m\eta \text{ for some } j \big\} < \infty$, 
        otherwise we are done. In words, $T$ is the first time that some item $j$'s price goes below $\tilde{p}_j - 2m\eta$.
        For each item $j$ such that $p^T_j < \tilde{p}_j - 2m\eta$, 
        we let $t_j < T$ be the last time that its price went from above $\tilde{p}_j$ to below $\tilde{p}_j$. 
        It follows that $p^t_j \leq \tilde{p}_j$ for $t_j \leq t \leq T$.
        Since the price decreases by at most $\eta$ per iteration,
        $T \geq t_j + 2m$, and thus $p^t_j \leq \tilde{p}_j$ for $t_j \leq t \leq t_j + 2m$. 
        Note that, by definition of $T$, $p^{t_j}_{j'} \geq \tilde{p}_{j'} - 2m\eta$ for all $j' \neq j$.
        
        By~\cref{fact:others-wont-get-too-much} and $p^t_j \leq \tilde{p}_j$ for $t_j \leq t \leq t_j + 2m$, 
        $\sum_{j' \neq j} \sum_{t = t_j}^{t_j + 2m - 1} b^t_{\pmb{b}(j) j'} \leq 2(m-1) B_{\pmb{b}(j)}$. 
        Since from $t_j$ to $t_j + 2m - 1$, 
        the total spend from buyer $\pmb{b}(j)$ equals $2mB_{\pmb{b}(j)}$, 
        $\sum_{t = t_j}^{t_j + 2m - 1} b^t_{\pmb{b}(j) j} \geq 2 B_{\pmb{b}(j)}$. 
        Thereby, we have  
        \begin{align}
            p^{t_j + 2m}_j \geq p^{t_j}_j + \eta \sum_{t = t_j}^{t_j + 2m - 1} \Big( \frac{b^t_{\pmb{b}(j) j}}{p^t_j} - 1 \Big) 
                        &\geq p^{t_j}_j + \eta \sum_{t = t_j}^{t_j + 2m - 1} \frac{b^t_{\pmb{b}(j) j}}{\tilde{p}_j} - 2m \eta \nonumber \\ 
                        &\geq p^{t_j}_j + 8m \eta - 2m \eta \nonumber \\ 
                        &\geq \lbp_j + (6m-1)\eta, 
            \label{eq:conclusion-1}
        \end{align}
        where the first inequality follows because total demand across buyers is lower bounded by demand from buyer $\pmb{b}(j)$, 
        the second inequality follows from $p^t_j \leq \tilde{p}_j$ for $t_j \leq t \leq t_j + 2m$, 
        the third inequality follows from $\tilde{p}_j \leq \frac{B_{\pmb{b}(j)}}{4m}$ and $\sum_{t = t_j}^{t_j + 2m - 1} b^t_{\pmb{b}(j) j} \geq 2 B_{\pmb{b}(j)}$, 
        and the last inequality follows from $p_j^{t_j} \geq \lbp_j - \eta$ since the price decreases by at most $\eta$ per iteration. 
        
        \cref{eq:conclusion-1} contradicts $p^{t_j + 2m}_j \leq \tilde{p}_j$ and thus $T < \infty$ does not exist.
        By contradiction, we can conclude $p^t_j \geq \tilde{p}_j - 2m\eta \text{ for all } j$. 
        This completes the proof of~\cref{lem:lower-bounds-prices-upper-bounds-demands-upper-bounds-prices} $(i)$.    
        \item[(ii)]
        A direct consequence of \cref{lem:lower-bounds-prices-upper-bounds-demands-upper-bounds-prices} $(i)$ is that the demands are upper bounded, due to the lower bound on prices.
        Recall that $\sum_{i = 1}^n B_i = B$.
        At any time step $t$, 
        let $z(p^t) \in \RR^m$ be the excess demand vector, 
        then we have 
        \begin{equation}
            \En{z(p^t)} \leq \En{z(p^t)}_1 
                        = \sum_{j = 1}^m \lvert z_j(p^t) \rvert 
                        \leq \sum_{j = 1}^m \Big( \frac{\sum_{i = 1}^n b^t_{ij}}{p^t_j} + 1 \Big) 
                        \leq \frac{B}{\min_j \tilde{p}_j - 2m\eta} + m, 
        \end{equation}
        where the last inequality follows 
        by $p^t_j \geq \min_j \tilde{p}_j - 2m\eta$ for all $j \in [m]$ 
        and $\sum_j\sum_i b^t_{ij} = B$.
        \item[(iii)]
        Since the demands are upper bounded, it follows that the prices are also upper bounded.
        
        If there exists some $t$ at which $p^t_j > B$, then the price of item $j$ has to decrease at the $(t+1)$-th step because 
        $p^{t + 1}_j = p^t_j + \eta \left( x_j - 1 \right) 
        = p^t_j + \eta \Big( \frac{\sum_i b^t_{ij}}{p^t_j} - 1 \Big) 
        \leq p^t_j + \eta \Big( \frac{B}{p^t_j} - 1 \Big) < p^t_j$. 
    
        Then, we consider the largest possible value $p^t_j$ can attain, approaching from below $B$. 
        For any $p^{t-1}_j \leq B$, 
        then 
        $p^t_j = p^{t-1}_j + \eta \left( x_j - 1 \right) 
        \leq p^{t-1}_j + \eta \frac{\sum_i b^t_{ij}}{p^{t-1}_j}
        \leq B + \frac{\eta B}{p^{t-1}_j} 
        \leq \left( 1 + \frac{\eta}{\tilde{p}_j - 2m\eta} \right) B$. 
        Hence, \cref{lem:lower-bounds-prices-upper-bounds-demands-upper-bounds-prices} $(iii)$ follows.
    \end{itemize}
\end{proof}

\begin{customlem}{2}
    There is an instance for which $p^t = \lbp - (m-1)\eta$ for some item at some time step $t$. 
\end{customlem}
\begin{proof}
    Assume that there is one buyer and $m$ items. 
    This buyer has a unit budget and receives the same utility for each item. 
    Thus, $\lbp_j = \frac{1}{4m}$ for all $j$. 
    Let $p^0 = \lbp$. 
    Assume that at each iteration, the buyer always selects only one item out of all MBB items and spends her entire budget on that item. 
    Then, the price of every item, except the demanded item, will have a price decrease of $\eta$ at each iteration.
    Obviously, the last item to be demanded will have a price of $\lbp_j - (m-1)\eta$ at the $m$-th iteration.
\end{proof}

\paragraph{Strong convexity of $h$.}
\label{app:sec:strong-convexity}
Here, we compute the Hessian matrix of $h(p)$ and then determine the strong convexity modulus of $h(p)$.
First, note that 
\begin{align*}
    h(p) = \sum_j p_j - \sum_j \sum_i p^*_j x^*_{ij} \log{\left( \frac{p_j}{v_{ij}} \right)} 
    &= \sum_j p_j - \sum_j p^*_j \Big( \sum_i x^*_{ij} \log{p_j} - \sum_i x^*_{ij} \log{v_{ij}} \Big) \\
    &= \sum_j p_j - \sum_j p^*_j \Big( \log{p_j}  - \sum_i x^*_{ij} \log{v_{ij}} \Big) \\
    &= \sum_j (p_j - p^*_j \log{p_j})  - \sum_i \sum_j p^*_j x^*_{ij} \log{v_{ij}}.
\end{align*}
It follows that the Hessian matrix 
\begin{equation}
    \nabla^2 h(p) = 
    \begin{bmatrix}
        \frac{p^*_1}{p_1^2} & 0 & \cdots & 0 \\ 
        0 & \frac{p^*_2}{p_2^2} & \cdots & 0 \\ 
        \vdots & \vdots & \ddots & \vdots \\
        0 & 0 & \cdots & \frac{p^*_m}{p_m^2} 
    \end{bmatrix}
    \succeq \min_j \frac{p^*_j}{\bar{p}_j^2} \; \textbf{I}, 
\end{equation}
where $\textbf{I}$ denotes the $m$-dimensional identity matrix.
Equivalently, $h(p)$ is strongly convex with modulus $\mu = \min_j \frac{p^*_j}{\bar{p}_j^2}$.
The strong convexity modulus of $h^q(p)$ can be computed in the same way.

\paragraph{Proofs of~\cref{thm:linear-convergence} and~\cref{crl:time-complexity}.}

\begin{customthm}{1}
    Assume that we adjust prices in the linear Fisher market with~\cref{eq:general-linear-ttm}, 
    starting from an initial price vector $p^0 \geq \tilde{p}$, 
    with any stepsize $\eta < \frac{1}{2m} \min_j \tilde{p}_j$, where $\tilde{p}_j$ is defined          
    in~\cref{lem:lower-bounds-prices-upper-bounds-demands-upper-bounds-prices}. 
    Then, we have 
    \begin{equation*}
        \En{p^t - p^*}^2 
        \leq (1 - 2\eta \alpha)^t \En{p^0 - p^*}^2 + e, \quad \text{for all } t \geq 0,
    \end{equation*}
    where $\alpha$ is defined in~\cref{lem:qg} 
    and 
    \begin{equation*}
        e = \frac{\eta G^2}{2\alpha} = \frac{\eta}{2\alpha} \Big( \frac{B}{\min_j \tilde{p}_j - 2m\eta} + m \Big)^2.
    \end{equation*}
\end{customthm} 

\begin{proof}
    Let $g^t \in \partial \varphi(p^t)$. 
    Since $\partial \varphi(p^t)$ is equal to the excess supply $- z^t$, 
    by~\cref{lem:lower-bounds-prices-upper-bounds-demands-upper-bounds-prices} $(ii)$,
    we have $\En{g^t} = \En{z^t} \leq G := \frac{B}{\min_j \tilde{p}_j - 2m\eta} + m$.
    First we derive a descent inequality: 
    \begin{align}
        \En{p^{t+1} - p^*}^2 
        &= \En{\Pi_{\RR^m_{\geq 0}}\left( p^t - \eta g^t \right) - p^*}^2 \nonumber \\ 
        &\leq \En{p^t - \eta g^t - p^*}^2 \tag{non-expansiveness} \nonumber \\ 
        &= \En{p^t - p^*}^2 - 2\eta\inp{g^t}{p^t - p^*} + \eta^2 \En{g^t}^2 \nonumber \\ 
        &\leq \En{p^t - p^*}^2 - 2\eta\left( \varphi(p^t) - \varphi(p^*) \right) + \eta^2 G^2 \tag{convexity and $G$-bounded subgradients} \nonumber \\ 
        &\leq \En{p^t - p^*}^2 - 2\eta\alpha \En{p^t - p^*}^2 + \eta^2 G^2 \tag{\cref{lem:qg}} \nonumber \\ 
        &= (1 - 2\eta \alpha) \En{p^t - p^*}^2 + \eta^2 G^2. 
    \end{align} 
    Equivalently, we have 
    \begin{equation}
        \En{p^{t+1} - p^*}^2 - \frac{\eta G^2}{2\alpha} \leq (1 - 2\eta \alpha) \left( \En{p^t - p^*}^2 - \frac{\eta G^2}{2\alpha} \right). 
    \end{equation} 
    By induction, we obtain that $\En{p^t - p^*}^2 - e \leq (1 - 2\eta \alpha)^t \left( \En{p^0 - p^*}^2 - e \right) \leq (1 - 2\eta \alpha)^t \En{p^0 - p^*}^2$ where $e = \frac{\eta G^2}{2\alpha}$.
\end{proof} 

\begin{customcrl}{1}
    For a given $\epsilon > 0$, 
    let $\eta \leq 
        \min\left\{ \frac{\min_j \tilde{p}_j}{4m}, 
        \frac{2 \epsilon \min_j p^*_j}{9 B^2 (\frac{2B}{\min_j \tilde{p}_j} + m)^2} \right\}$.
    Then, starting from any $p^0 \geq \tilde{p}$ such that $\sum_{j = 1}^m p^0_j = B$, 
    {\tatonnement} with stepsize $\eta$ generates a price vector $p$ 
    such that $\En{p - p^*}^2 \leq \epsilon$ 
    in $\mathcal{O}(\frac{1}{\epsilon}\log{\frac{1}{\epsilon}})$ iterations.
    If $\epsilon \geq \frac{9 B^2 \min_j \tilde{p}_j}{8m \min_j p^*_j} (\frac{2B}{\min_j \tilde{p}_j} + m)^2$, 
    the time complexity is on the order of $\mathcal{O}(\log{\frac{1}{\epsilon}})$.
\end{customcrl}

\begin{proof}
    Since $\eta \leq \frac{\min_{j'} \tilde{p}_{j'}}{4m} \leq \frac{\tilde{p}_j}{4m}$ for all $j \in [m]$, 
    we have $\frac{\eta}{\tilde{p}_j - 2m\eta} \leq \frac{1}{2m}$ and thus $\alpha \geq \frac{2 \min_j p^*_j}{9 B^2}$.
    It follows that 
    \begin{equation}
        e = \frac{\eta}{2\alpha} \left( \frac{B}{\min_j \tilde{p}_j - 2m\eta} + m \right)^2 
        \leq \frac{9 \eta B^2}{4 \min_j p^*_j} \left( \frac{2B}{\min_j \tilde{p}_j} + m \right)^2
        \leq \frac{\epsilon}{2}, 
        \label{eq:upper-bound-epsilon}
    \end{equation}
    where the inequality follows by $\eta \leq \frac{\min_j \tilde{p}_j}{4m}$ 
    and the last equality follows by $\eta \leq \frac{2 \epsilon \min_j p^*_j}{9 B^2 (\frac{2B}{\min_j \tilde{p}_j} + m)^2}$.
    After $t = \frac{1}{2\eta\alpha} \log{\frac{8 B^2}{\epsilon}}$ iterations, we have 
    \begin{equation}
        t \log{(1 - 2\eta\alpha)} + \log{\En{p^0 - p^*}^2} 
        \leq -2\eta\alpha\cdot t + \log{4B^2} 
        = - \log{\frac{8 B^2}{\epsilon}} + \log{4B^2} 
        = \log{\frac{\epsilon}{2}}, 
    \end{equation}
    where the inequality follows because $\log(1 + x) \leq x, \;\forall\, x > -1$ and $\En{p^0 - p^*} \leq 2B$.
    Thus, by~\cref{thm:linear-convergence} $\En{p^t - p^*}^2 \leq (1 - 2\eta \alpha)^t \En{p^0 - p^*}^2 + e \leq \frac{\epsilon}{2} + \frac{\epsilon}{2} = \epsilon$. 
    
    Since $\alpha \geq \frac{2 \min_j p^*_j}{9 B^2}$, $t \leq \frac{9 B^2}{4 \eta \min_j p^*_j} \log(\frac{8 B^2}{\epsilon})$. 
    If $\epsilon$ is small such that $\frac{2 \epsilon \min_j p^*_j}{9 B^2 (\frac{2B}{\min_j \tilde{p}_j} + m)^2} \leq \frac{\min_j \tilde{p}_j}{4m}$, 
    $\eta$ is of the order $\mathcal{O}(\epsilon)$, hence 
    the time complexity is of the order $\mathcal{O}(\frac{1}{\epsilon}\log{\frac{1}{\epsilon}})$.
    Otherwise, $\eta$ is independent of $\epsilon$, hence the time complexity is of the order $\mathcal{O}(\log{\frac{1}{\epsilon}})$.
\end{proof}

\section{More Detailed Discussion on the LFM}
\label{app:sec:more-detailed-discussion}

\paragraph{Approximate equilibrium implied by~\cref{thm:linear-convergence}.}

We say a Fisher market reaches an $\epsilon$-approximate equilibrium if 
allocations $x_i\in \RR^m_{\geq 0} \; \forall \; i \in [n]$ and prices $p \in \RR^m_{\geq 0}$ satisfy 
\begin{itemize}
    \item Approximate buyer optimality: $\inp{v_i}{x_i} \geq (1 - \epsilon) \inp{v_i}{y_i}$ for all $y_i$ such that $\inp{p}{y_i} \leq (1 - \epsilon) B_i$, and $\inp{p}{x_i} \leq B_i$ for all $i \in [n]$,
    \item Approximate market clearance: $1 - \epsilon \leq \sum_{i=1}^n x_{ij} \leq 1$ for all $j \in [m]$.
\end{itemize}

We note that this last-iterate convergence in $p$ implies an approximate equilibrium. 
\begin{theorem}
    For any $p$ such that $\En{p - p^*} \leq \epsilon \min_j p^*_j$, there exists an allocation $x$ such that $(p, x)$ constitute an $\epsilon$-approximate equilibrium. 
\end{theorem}
\begin{proof}
    $\En{p - p^*} \leq \epsilon \min_j p^*_j$ implies $(1 - \epsilon) p^*_j \leq p_j \leq (1 + \epsilon) p^*_j$ for all $j$. 
    Let $x^*$ be any equilibrium allocation associated with $p^*$. 
    For any $i$, let $x_{ij} = x^*_{ij}$ if $p_j \leq p^*_j$ and $x_{ij} = (1 - \epsilon) x^*_{ij}$ otherwise.
    Then, one can check: 
    \begin{itemize}
        \item[i.] $\inp{v_i}{x_i} \geq (1 - \epsilon) \inp{v_i}{x^*_i} \geq (1 - \epsilon) \inp{v_i}{y_i}$ for all $y_i$ such that $\inp{p}{y_i} \leq (1 - \epsilon) B_i$; 
        \item[ii.] $\sum_j p_j x_{ij} \leq \sum_{j: p_j \leq p^*_j} p^*_j x^*_{ij} + \sum_{j: p_j > p^*_j} (1 + \epsilon) p^*_j \cdot (1 - \epsilon) x^*_{ij} \leq \sum_j p^*_j x^*_{ij} = B_i$; 
        \item[iii.] $1 - \epsilon \leq \sum_i (1 - \epsilon) x^*_{ij} \leq \sum_i x_{ij} \leq \sum_i x^*_{ij} = 1$.
    \end{itemize}
    This shows that $(p, x)$ constitutes an $\epsilon$-approximate equilibrium.
\end{proof}



\paragraph{Nearly-optimal last-iterate convergence rate without regularity conditions.} 
\citet{zamani2023exact} showed a nearly-optimal last-iterate convergence rate without regularity conditions. 
Assuming that the function $\varphi(p)$ has $G$-bounded subgradients (i.e. any subgradient $g$ satisfies $\En{g} \leq G$) 
and $\En{p^0 - p^*} \leq R$, 
they show that a sequence of iterates $\{ p^t \}_{t = 0}^T$ generated by the (projected) subgradient method with a constant stepsize $\frac{1}{\sqrt{T}}$ satisfies 
\begin{equation} 
    \varphi(p^T) - \varphi(p^*) \leq \frac{GR}{\sqrt{T}} \left( \frac{5}{4} + \frac{1}{4}\log{T} \right). 
    \label{eq:zamani-glineur-result}
\end{equation} 
Note that smaller constant stepsizes than $\frac{1}{\sqrt{T}}$ do not guarantee the above convergence.
As discussed before, for any $p^0 \in \RR^m_{\geq 0}$ 
such that $\sum_{j = 1}^m p^0_j = B$, we have $\En{p^0 - p^*} \leq \En{p^0 - p^*}_1 \leq \En{p^0}_1 + \En{p^*}_1 = 2B$.
Also, by~\cref{lem:lower-bounds-prices-upper-bounds-demands-upper-bounds-prices} $(ii)$, 
we have $\En{g^t} \leq \frac{B}{\min_j \tilde{p}_j - 2m\eta} + m \leq \frac{2 B}{\min_j \tilde{p}_j} + m$ for any $g^t \in \partial \varphi(p^t)$ for all $t \geq 0$ 
if $\eta \leq \frac{\min_j \tilde{p}_j}{4m}$. 
Let $G = \frac{2 B}{\min_j \tilde{p}_j} + m$ and $R = 2B$. 
By combining our results with the results in~\citet{zamani2023exact}, 
we obtain that after 
$T = \frac{G^2 R^2}{\epsilon^2}\log^2\left(\frac{G R}{\epsilon}\right)$
iterations, 
{\tatonnement} 
with stepsize 
$\eta = \frac{1}{\sqrt{T}} = \frac{\epsilon}{G R \log(\frac{G R}{\epsilon})}$ 
generates a price vector $p$ satisfying 
\begin{equation}
    \varphi(p) - \varphi(p^*) \leq \epsilon,
    \label{eq:rate-from-zamani-glineur-result}
\end{equation} 
if $\epsilon \leq G R \min\{ e^{-4}, \frac{\min_j \tilde{p}_j}{m} \}$.

Using the subgradient inequality and our subgradient bound, it is also possible to use \cref{thm:linear-convergence} to directly derive the duality gap convergence rate obtained in \cref{eq:rate-from-zamani-glineur-result}. 
Conversely, one can apply the quadratic growth condition and the rate in \cref{eq:rate-from-zamani-glineur-result} to get a rate of convergence on $\En{p^t - p^*}^2$, however, this yields a rate of $\mathcal{\tilde O}(\frac{1}{\sqrt{T}})$, which is worse than the $\mathcal{\tilde O}(\frac{1}{T})$ rate in \cref{thm:linear-convergence}. Moreover, it does not allow one to obtain the linear rate of convergence to a neighborhood of the equilibrium, as in \cref{thm:linear-convergence}. 

We show a proof for the above time complexity result.
\begin{proof} 
    Since $\epsilon \leq G R \min\{ e^{-4}, \frac{\min_j \tilde{p}_j}{m} \}$, 
    we have $\log(\frac{G R}{\epsilon}) \geq 4$ and $\frac{\epsilon}{G R} \leq \frac{\min_j \tilde{p}_j}{m}$. 
    Thus, 
    \begin{equation}
        \eta = \frac{1}{\sqrt{T}} = \frac{\epsilon}{G R \log(\frac{G R}{\epsilon})} \leq \frac{\min_j \tilde{p}_j}{4 m}. 
    \end{equation}
    By~\cref{lem:lower-bounds-prices-upper-bounds-demands-upper-bounds-prices} $(ii)$, 
    we have $\En{g^t} \leq G$ for any $g^t \in \partial \varphi(p^t)$ for all $t \geq 0$. 
    By~\cref{eq:zamani-glineur-result}, we have that after $T$ iterations, 
    \begin{align}
        \varphi(p^T) - \varphi(p^*) \leq \frac{G R}{\sqrt{T}} \left( \frac{5}{4} + \frac{1}{4} \log{T} \right) &\leq \epsilon \frac{\frac{5}{4} + \frac{1}{2} \log(\frac{G R}{\epsilon}) + \frac{1}{2} \log{\log(\frac{G R}{\epsilon})}}{\log(\frac{G R}{\epsilon})} \nonumber \\ 
        &\leq \epsilon \left( \frac{5}{16} + \frac{1}{2} + \frac{1}{2 e} \right) \leq \epsilon, 
    \end{align}
    where the third inequality follows by $\log(\frac{G R}{\epsilon}) \geq 4$ and $\frac{\log{x}}{x} \leq \frac{1}{e}$ for all $x > 0$.
\end{proof}

\paragraph{Discussion on constant versus adaptive stepsize settings.} 

In this paper, we focus on the constant stepsize setting because it aligns better with the natural economic dynamics. 
From the view of economics, the stepsize in tâtonnement corresponds to a ``speed of adjustment" parameter, 
which is a simple characteristic of the market and commonly seen as a constant. 
See, for example,~\citet{bala1992chaotic} (page 5, after Eq. (3.1)) and ~\citet{cole2008fast} (page 4, after Eq. (2)). 
If we allow for adaptive stepsizes, then there are results showing that, under a quadratic growth condition, 
the subgradient descent method will converge to the exact optimal solution 
(e.g., decaying stepsizes in~\citet{johnstone2020faster}, Polyak stepsizes in~\citet{grimmer2019general}, etc.). 
As a result, tâtonnement converges to the exact market equilibrium under such adaptive stepsize schemes.

\section{Proofs in Section~\ref{sec:ql}}
\label{app:sec:quasi-linear-proofs}

First, we show that the {\tatonnement} process in the QLFM generates a sequence of prices with lower and upper bounds, as long as the stepsize is not too large. 
Lower bounds on prices in the QLFM can be obtained almost identically to how we prove~\cref{lem:lower-bounds-prices-upper-bounds-demands-upper-bounds-prices} since QLFM behaves the same as LFM when the price is small. 
Analogous to the LFM case, the upper bounds on the excess demand and prices follow as well.
Recall that, given an item $j$, we use $\pmb{b}(j)$ to denote the buyer with the minimum index such that
\begin{equation*}
    \pmb{b}(j) \in \Argmax_i \frac{v_{i j}}{\|v_i\|_\infty}.
\end{equation*}

\begin{lemma}
    Let $v_{\min} = \min_j v_{\pmb{b}(j) j}$. 
    Let $\hat{p}$ be an $m$-dimensional price vector where 
    \begin{equation}
        \hat{p}_j = \min\left\{ \frac{B_{\min}}{4m}, \frac{v_{\min}}{2} \right\} 
        \max_i \frac{v_{ij}}{\En{v_i}_\infty}, \quad \mbox{ for all } j \in [m].
        \label{eq:def-hat-p-j-ql}
    \end{equation}
    Assume that we adjust prices in the QLFM with~\cref{eq:general-linear-ttm}, 
    starting from any initial price $p^0 \geq \hat{p}$, 
    with any stepsize $\eta < \frac{1}{2m} \min_j \hat{p}_j$. 
    Then, we have for all $t \geq 0$, 
    \begin{itemize}
        \item[(i)] $p^t_j \geq \hat{p}_j - 2m\eta$ for all $j \in [m]$; 
        \item[(ii)] $\En{z(p^t)} \leq \frac{B}{\min_j \hat{p}_j - 2m\eta} + m$; 
        \item[(iii)] $p^t_j \leq \min\{ \max_i v_{i j}, B \} + \frac{\eta B}{\hat{p}_j - 2m\eta}$ for all $j \in [m]$.
    \end{itemize}
    \label{lem:lower-bounds-prices-upper-bounds-demands-upper-bounds-prices-quasi-linear}
\end{lemma} 
\begin{proof}

    \begin{itemize}
        \item[(i)] Recall that, given an item $j$, we use $\pmb{b}(j)$ to denote the buyer with the minimum index such that
        \begin{equation*}
            \pmb{b}(j) \in \Argmax_i \frac{v_{i j}}{\|v_i\|_\infty}.
        \end{equation*}

        As in the proof of~\cref{lem:lower-bounds-prices-upper-bounds-demands-upper-bounds-prices} $(i)$, we can show two facts analogous to~\cref{fact:j-prime,fact:others-wont-get-too-much}.
        Moreover, $\hat{p}_j < v_{\min} \max_i \frac{v_{ij}}{\En{v_i}_\infty} \leq v_{\pmb{b}(j) j}$ guarantees buyer $\pmb{b}(j)$ wants to spend all her money if $p_j \leq \hat{p}_j$.
        Other parts of the proof are analogous to the proof of~\cref{lem:lower-bounds-prices-upper-bounds-demands-upper-bounds-prices} $(i)$.
        We show a complete proof as follows.

        First, analogous to the proof of~\cref{lem:lower-bounds-prices-upper-bounds-demands-upper-bounds-prices}, we can show the following two facts.
        \begin{fact}
            For any two items $j, j' \in [m]$, 
            if $p_j \leq \hat{p}_j$ and item $j'$ is demanded by buyer $j$, 
            then 
            $p_{j'} \leq \hat{p}_{j'}$. 
            \label{fact:j-prime-ql}
        \end{fact}
            \begin{proof}
            We denote $\kappa' = \min\left\{ \frac{B_{\min}}{4m}, v_{\min} \right\} = \min\left\{ \frac{\min_i B_i}{4m}, \min_j v_{\pmb{b}(j) j} \right\}$ for simplicity.
            Since buyer $j$ demands item $j'$, we have $p_{j'} \leq \frac{v_{\pmb{b}(j) j'}}{v_{\pmb{b}(j) j}} p_j$. 
            Then, it follows that 
            \begin{align} 
                p_{j'} \leq \frac{v_{\pmb{b}(j) j'}}{v_{\pmb{b}(j) j}} p_j \leq \frac{v_{\pmb{b}(j) j'}}{v_{\pmb{b}(j) j}} \hat{p}_j 
                &\stackrel{\text{by~\cref{eq:def-hat-p-j-ql}}}{=} \frac{v_{\pmb{b}(j) j'}}{v_{\pmb{b}(j) j}} \kappa' \max_i \frac{v_{i j}}{\En{v_i}_{\infty}} \nonumber \\ 
                &\stackrel{\text{by~\cref{eq:def-buyer}}}{=} \frac{v_{\pmb{b}(j) j'}}{v_{\pmb{b}(j) j}} \kappa' \frac{v_{\pmb{b}(j) j}}{\En{v_{\pmb{b}(j)}}_{\infty}} \nonumber \\ 
                &\;\;\;\,\, = \kappa' \frac{v_{\pmb{b}(j) j'}}{\En{v_{\pmb{b}(j)}}_{\infty}}
                \leq \kappa' \max_i \frac{v_{i j'}}{\En{v_i}_{\infty}} = \hat{p}_{j'}. 
                \label{eq:low-price-rule-ql}
            \end{align}
            \end{proof}

        \begin{fact} 
            Let $t_j \geq 0$ be any time step. 
            For any item $j$, if $p^t_j \leq \hat{p}_j$ for $t_j \leq t \leq t_j + 2m - 1$, 
            and $p^{t_j}_{j'} \geq \hat{p}_{j'} - 2m\eta$ for all $j' \neq j$. 
            Then, 
            \begin{equation}
                \sum_{t=t_j}^{t_j+2m - 1} b_{\pmb{b}(j) j'}^t \leq 2 B_{\pmb{b}(j)} \quad \text{for all } j' \neq j.
            \end{equation}
            \label{fact:others-wont-get-too-much-ql} 
        \end{fact} 
            \begin{proof}
                We prove this fact by contradiction. 
                Suppose that, at time step $\hat{t} \leq t_j + 2m - 1$, 
                item $j'$ receives more than $2 B_{\pmb{b}(j)}$ spend from buyer $\pmb{b}(j)$ cumulatively for the first time. 
                Then we have $\sum_{t=t_j}^{\hat{t}-1} b_{\pmb{b}(j) j'}^t \leq 2 B_{\pmb{b}(j)} < \sum_{t=t_j}^{\hat{t}} b_{\pmb{b}(j) j'}^t$. 
                Since item $j'$ receives a spend of at most $B_{\pmb{b}(j)}$ from buyer $\pmb{b}(j)$ per iteration, 
                we have $\sum_{t=t_j}^{\hat{t} - 1} b_{\pmb{b}(j) j'}^t > B_{\pmb{b}(j)}$. 
                By~\cref{fact:j-prime-ql} and $p^t_j \leq \hat{p}_j$ for $t_j \leq t \leq t_j + 2m - 1$, 
                we know $p^t_{j'} \leq \hat{p}_{j'}$ for $t_j \leq t \leq t_j + 2m - 1$ such that $b^t_{\pmb{b}(j) j'} > 0$. 
                Thus, 
                $\sum_{t=t_j}^{\hat{t}-1} \frac{b^t_{\pmb{b}(j) j'}}{p^t_{j'}} 
                \geq \sum_{t=t_j}^{\hat{t}-1} \frac{b^t_{\pmb{b}(j) j'}}{\hat{p}_{j'}}
                \geq \frac{4m}{B_j} \sum_{t=t_j}^{\hat{t}-1} b^t_{\pmb{b}(j) j'}
                > 4m$, because $\hat{p}_{j'} \leq \frac{B_{\min}}{4 m} \leq \frac{B_{\pmb{b}(j)}}{4 m}$ 
                and $\sum_{t=t_j}^{\hat{t} - 1} b_{\pmb{b}(j) j'}^t > B_{\pmb{b}(j)}$. 
                This leads to 
                \begin{equation}
                    p^{\hat{t}}_{j'} \geq p^{t_j}_{j'} + \eta \sum_{t = t_j}^{\hat{t} - 1} \big( \frac{b^t_{\pmb{b}(j) j'}}{p^t_{j'}} - 1 \big) 
                                    > p^{t_j}_{j'} + 4m\eta - (2m - 1)\eta 
                                    = p^{t_j}_{j'} + (2m + 1) \eta
                                    \geq \hat{p}_{j'} + \eta. 
                \end{equation} 
                This contradicts $p^{\hat{t}}_{j'} \leq \hat{p}_{j'}$, 
                which follows from~\cref{fact:j-prime-ql} and $b^{\hat{t}}_{\pmb{b}(j) j'} > 0$. 
            \end{proof}

        Now we are ready to prove \cref{lem:lower-bounds-prices-upper-bounds-demands-upper-bounds-prices-quasi-linear} $(i)$.
        We suppose that there exists a $T = \min\big\{ t \,\big\vert\, p^t_j < \hat{p}_j - 2m\eta \text{ for some } j \big\} < \infty$, 
        otherwise we are done. 
        In words, $T$ is the first time that some item $j$'s price goes below $\hat{p}_j - 2m\eta$.
        For each item $j$ such that $p^T_j < \hat{p}_j - 2m\eta$, 
        we let $t_j < T$ be the \emph{last} time that its price went from above $\hat{p}_j$ to below $\hat{p}_j$. 
        It follows that $p^t_j \leq \hat{p}_j$ for $t_j \leq t \leq T$.
        Since the price decreases by at most $\eta$ per iteration,
        $T \geq t_j + 2m$, and thus $p^t_j \leq \hat{p}_j$ for $t_j \leq t \leq t_j + 2m$. 
        Note that, by definition of $T$, $p^{t_j}_{j'} \geq \hat{p}_{j'} - 2m\eta$ for all $j' \neq j$.

        By~\cref{fact:others-wont-get-too-much-ql}, 
        since $p^t_j \leq \hat{p}_j$ for $t_j \leq t \leq t_j + 2m$ 
        and $p^{t_j}_{j'} \geq \hat{p}_{j'} - 2m\eta$ for all $j' \neq j$, 
        we have 
        $\sum_{j' \neq j} \sum_{t = t_j}^{t_j + 2m - 1} b^t_{\pmb{b}(j) j'} \leq 2(m-1) B_{\pmb{b}(j)}$.
        Since $p^t_j \leq \hat{p}_j \leq \frac{v_{\min}}{2} < v_{\pmb{b}(j) j}$ from $t_j$ to $t_j + 2m - 1$, 
        the total spend from buyer $\pmb{b}(j)$ equals to $2m B_{\pmb{b}(j)}$. 
        Hence, we have 
        $\sum_{t = t_j}^{t_j + 2m - 1} b^t_{\pmb{b}(j) j} \geq 2 B_{\pmb{b}(j)}$. 
        Thereby, we have  
        \begin{align}
            p^{t_j + 2m}_j \geq p^{t_j}_j + \eta \sum_{t = t_j}^{t_j + 2m - 1} \Big( \frac{b^t_{\pmb{b}(j) j}}{p^t_j} - 1 \Big) 
                           \geq p^{t_j}_j + \eta \sum_{t = t_j}^{t_j + 2m - 1} \frac{b^t_{\pmb{b}(j) j}}{\hat{p}_j} - 2m \eta 
                          \geq p^{t_j}_j + 8m \eta - 2m \eta 
                          &\geq \hat{p}_j + (6m-1)\eta, 
            \label{eq:conclusion-1-ql}
        \end{align}
        where the first inequality follows because total demand across buyers is lower bounded by demand from buyer $\pmb{b}(j)$, 
        the second inequality follows from $p^t_j \leq \hat{p}_j$ for $t_j \leq t \leq t_j + 2m$, 
        the third inequality follows from $\hat{p}_j \leq \frac{B_{\pmb{b}(j)}}{4m}$ and $\sum_{t = t_j}^{t_j + 2m - 1} b^t_{\pmb{b}(j) j} \geq 2 B_{\pmb{b}(j)}$, 
        and the last inequality follows from $p_j^{t_j} \geq \hat{p}_j - \eta$ since the price decreases by at most $\eta$ per iteration. 
        
        Therefore, \cref{eq:conclusion-1-ql} contradicts $p^{t_j + 2m}_j \leq \hat{p}_j$ and thus $T < \infty$ does not exist.
        By contradiction, we can conclude $p^t_j \geq \hat{p}_j - 2m\eta$ for all $j$ and all $t \geq 0$.

        \item[(ii)] 
        Then, the part $(ii)$ follows. 
        \begin{equation}
            \En{z(p^t)} \leq \En{z(p^t)}_1 
                        = \sum_{j = 1}^m \lvert z_j(p^t) \rvert 
                        \leq \sum_{j = 1}^m \Big( \frac{\sum_{i = 1}^n b^t_{ij}}{p^t_j} + 1 \Big) 
                        \leq \frac{B}{\min_j \hat{p}_j - 2m\eta} + m, 
        \end{equation}
        where the last inequality follows 
        because $p^t_j \geq \min_j \hat{p}_j - 2m\eta$ for all $j$ 
        and $\sum_{j = 1}^m \sum_{i = 1}^n b^t_{ij} \leq B$.
        \item[(iii)] 
        If at time step $t$, $p^t_j > \min\{ \max_i v_{i j} , B \}$, then the price of item $j$ has to decrease at the $(t+1)$-th step. This is because: $1)$ if $p^t_j > \max_i v_{i j}$, then no buyer wants to spend money on item $j$ thus the total demand on this time step has to be $0$; $2)$ if $p^t_j > B$, then $p^{t + 1}_j = p^t_j + \eta \left( x_j - 1 \right) 
        = p^t_j + \eta \Big( \frac{\sum_{i = 1}^n b^t_{ij}}{p^t_j} - 1 \Big) 
        \leq p^t_j + \eta \Big( \frac{B}{p^t_j} - 1 \Big) < p^t_j$. 

        Then, we consider the largest possible value $p^t_j$ can attain, approaching from below $\min\{ \max_i v_{i j} , B \}$. 
        For any $p^{t-1}_j \leq \min\{ \max_i v_{i j} , B \}$, 
        then 
        $p^t_j = p^{t-1}_j + \eta \left( x_j - 1 \right) 
        \leq p^{t-1}_j + \eta \frac{\sum_i b^t_{ij}}{p^{t-1}_j}
        \leq \min\{ \max_i v_{i j} , B \} + \eta \frac{B}{p^{t-1}_j} 
        \leq \min\{ \max_i v_{i j} , B \} + \frac{\eta}{\hat{p}_j - 2m\eta} B
        $. 
        Hence, the part $(iii)$ follows.
    \end{itemize}

\end{proof}

Next, we show that $\varphi^q(p)$ also satisfies QG, by considering an auxiliary function $h^q(p)$ analogous to $h(p)$ in the proof of~\cref{lem:qg}.
\begin{lemma}
    The convex function $\varphi^q(p)$
    satisfies the quadratic growth condition with modulus 
    $$\alpha = \min_j 
    \frac{p^*_j}{2 \big(\min\{ \max_i v_{i j} , B \} + \frac{\eta B}{\hat{p}_j - 2m\eta} \big)^2},$$ 
    where $p^*$ denotes the unique equilibrium price vector and $\hat{p}$ is defined in~\cref{lem:lower-bounds-prices-upper-bounds-demands-upper-bounds-prices-quasi-linear}. 
    \label{lem:qg-quasi-linear}
\end{lemma} 
\begin{proof} 
    Let $(p^*, x^*, y^*)$ be a market equilibrium in the QLFM and
    \begin{equation}
        h^q(p) = \sum_{j = 1}^m p_j - \sum_{i = 1}^n \sum_{j = 1}^m p^*_j x^*_{ij} \log{\left( \frac{p_j}{v_{ij}} \right)}.
    \end{equation} 
    By KKT conditions of the EG-like primal and dual problems for QLFM, we have $\inp{p^*}{x^*_i} + y^*_i = B_i$
    and $y^*_i > 0$ implies $\min_{k \in [m]} \frac{p_k}{v_{ik}} \geq 1$ for all $i \in [n]$.
    Thereby, we have 
    \begin{align} 
        \varphi^q(p) 
        &= \sum_{j = 1}^m p_j - \sum_{i = 1}^n \left( \sum_{j = 1}^m p^*_j x^*_{ij} + y_i^* \right) \log{ \left( \min\left\{ \min_{k \in [m]} \frac{p_k}{v_{ik}}, 1 \right\} \right) } \nonumber \\ 
        &= \sum_{j = 1}^m p_j - \sum_{i = 1}^n \sum_{j = 1}^m p^*_j x^*_{ij} \log{ \left( \min\left\{ \min_{k \in [m]} \frac{p_k}{v_{ik}}, 1 \right\} \right) } - \sum_i y_i^* \log{ \left( \min\left\{ \min_{k \in [m]} \frac{p_k}{v_{ik}}, 1 \right\} \right) }  \nonumber \\ 
        &= \sum_{j = 1}^m p_j - \sum_{i = 1}^n \sum_{j = 1}^m p^*_j x^*_{ij} \log{ \left( \min\left\{ \min_{k \in [m]} \frac{p_k}{v_{ik}}, 1 \right\} \right) } 
        \geq h^q(p).
    \end{align} 
    Also, $\varphi^q(p^*) = h^q(p^*)$ since 
    $x^*_{ij} > 0$ implies $\frac{p^*_j}{v_{ij}} = \min_{k \in [m]} \frac{p^*_k}{v_{ik}} \leq 1$. 
    Analogous to the LFM case, $\nabla_{p_j} h^q(p) = 1 - \frac{p^*_j}{p_j} \sum_{i = 1}^n x^*_{ij} = 1 - \frac{p^*_j}{p_j}$ for all $j \in [m]$, 
    thus $p = p^*$ is the minimizer of $h^q(p)$. 
    Therefore, $\varphi^q(p) - \varphi^q(p^*) \geq h^q(p) - h^q(p^*) \geq \alpha \En{p - p^*}^2$, 
    where $\alpha = \min_j \frac{p^*_j}{2 \bar{p}_j^2}$.
\end{proof}

\begin{customthm}{2}
    Assume that we adjust prices in the QLFM with~\cref{eq:general-linear-ttm}, 
    starting from any initial price $p^0 \geq \hat{p}$, with stepsize $\eta < \frac{1}{2m} \min_j \hat{p}_j$, where $\hat{p}_j$ is defined          
    in~\cref{lem:lower-bounds-prices-upper-bounds-demands-upper-bounds-prices-quasi-linear}. 
    Then, we have 
    \begin{equation}
        \En{p^t - p^*}^2 
        \leq (1 - 2\eta \alpha)^t  
            \En{p^0 - p^*}^2 + e \quad \text{for all } t \geq 0, 
    \end{equation}
    where $\alpha$ is defined in~\cref{lem:qg-quasi-linear} and 
    \begin{equation*}
        e = \frac{\eta G^2}{2\alpha} = \frac{\eta}{2\alpha} \left( \frac{B}{\min_j \hat{p}_j - 2 m\eta} + m \right)^2.
    \end{equation*}
\end{customthm} 
\begin{proof}
    Since $g^t \in \partial \varphi(p^t)$ is equal to the excess supply $- z^t$, 
    by~\cref{lem:lower-bounds-prices-upper-bounds-demands-upper-bounds-prices-quasi-linear} $(ii)$, 
    we have $\En{g^t} = \En{z^t} \leq G := \frac{B}{\min_j \hat{p}_j - 2m\eta} + m$.
    By~\cref{lem:qg-quasi-linear}, the quadratic growth parameter associated with $\varphi^q(p)$ is $\alpha = \min_j 
    \frac{p^*_j}{2 \big(\min\{ \max_i v_{i j} , B \} + \frac{\eta B}{\hat{p}_j - 2m\eta} \big)^2}$.
    The rest of the proof follows the same steps as in the proof of~\cref{thm:linear-convergence}.
\end{proof} 

\begin{customcrl}{2}
    For a given $\epsilon > 0$, 
    let 
    $\eta \leq 
        \min\{ \frac{\min_j \hat{p}_j}{4m}, 
        \frac{2 \epsilon \min_j p^*_j}{9 B^2 (\frac{2B}{\min_j \hat{p}_j} + m)^2} \}$.
    Then, starting from any $p^0 \geq \hat{p}$ such that $\sum_j p^0_j = B$, 
    {\tatonnement} with stepsize $\eta$ generates a price vector $p$ 
    such that $\En{p - p^*}^2 \leq \epsilon$ 
    in $\mathcal{O}(\frac{1}{\epsilon}\log{\frac{1}{\epsilon}})$ iterations.
    If $\epsilon \geq \frac{9 B^2 \min_j \hat{p}_j}{8m \min_j p^*_j} (\frac{2B}{\min_j \hat{p}_j} + m)^2$, 
    the time complexity is on the order of $\mathcal{O}(\log{\frac{1}{\epsilon}})$.
    \label{crl:time-complexity-ql}
\end{customcrl}
\begin{proof}
    Since $\eta \leq \frac{\min_{j'} \hat{p}_{j'}}{4m} \leq \frac{\hat{p}_j}{4m}$ for all $j \in [m]$, 
    we have $\frac{\eta}{\hat{p}_j - 2m\eta} \leq \frac{1}{2m}$ and thus 
    \begin{equation}
        \alpha \geq \min_j 
        \frac{p^*_j}{2 \big(B + \frac{\eta B}{\hat{p}_j - 2m\eta} \big)^2} \geq \frac{2 \min_j p^*_j}{9 B^2}.
    \end{equation}
    It follows that 
    \begin{equation}
        e = \frac{\eta}{2\alpha} \left( \frac{B}{\min_j \hat{p}_j - 2m\eta} + m \right)^2 
        \leq \frac{9 \eta B^2}{4 \min_j p^*_j} \left( \frac{2B}{\min_j \hat{p}_j} + m \right)^2
        = \frac{\epsilon}{2}, 
        \label{eq:upper-bound-epsilon}
    \end{equation}
    where the inequality follows by $\eta \leq \frac{\min_{j'} \hat{p}_{j'}}{4m}$ 
    and the last equality follows by 
    $\eta \leq \frac{2 \epsilon \min_j p^*_j}{9 B^2 (\frac{2B}{\min_j \hat{p}_j} + m)^2}$. 
    After $t = \frac{1}{2\eta\alpha} \log{\frac{8 B^2}{\epsilon}}$ iterations, we have 
    \begin{equation}
        t \log{(1 - 2\eta\alpha)} + \log{\En{p^0 - p^*}^2} 
        \leq -2\eta\alpha\cdot t + \log{4B^2} 
        = - \log{\frac{8 B^2}{\epsilon}} + \log{4B^2} 
        = \log{\frac{\epsilon}{2}}, 
    \end{equation}
    where the inequality follows because $\log(1 + x) \leq x, \;\forall\, x > -1$ and $\En{p^0 - p^*} \leq 2B$.
    Thus, by~\cref{thm:linear-convergence} $\En{p^t - p^*}^2 \leq (1 - 2\eta \alpha)^t \En{p^0 - p^*}^2 + e \leq \frac{\epsilon}{2} + \frac{\epsilon}{2} = \epsilon$. 
    
    Since $\alpha \geq \frac{2 \min_j p^*_j}{9 B^2}$, $t \leq \frac{9 B^2}{4 \eta \min_j p^*_j} \log(\frac{8 B^2}{\epsilon})$. 
    If $\epsilon$ is small such that $\frac{2 \epsilon \min_j p^*_j}{9 B^2 (\frac{2B}{\min_j \hat{p}_j} + m)^2} \leq \frac{\min_j \hat{p}_j}{4m}$, 
    $\eta$ is of the order $\mathcal{O}(\epsilon)$, hence 
    the time complexity is of the order $\mathcal{O}(\frac{1}{\epsilon}\log{\frac{1}{\epsilon}})$.
    Otherwise, $\eta$ is independent of $\epsilon$, hence the time complexity is of the order $\mathcal{O}(\log{\frac{1}{\epsilon}})$.
\end{proof}

\section{Numerical Experiments Details and More Results}
\label{app:sec:additional-experiments}

\paragraph{Real data description.}

For the real data experiments, we implement {\tatonnement} on a market instance used in~\citet{nan2023fast}.
Different from their unit budgets, we consider randomly generated budgets from the uniform distribution on $[0, 1]$. 
For the linear Fisher market, we normalize the budgets such that the total budget is $1$; for the quasi-linear Fisher market, we use the non-normalized budgets, as money has value outside the current market under QL setting.
We assign each item a unit supply. 

The valutions in the market instance is generated by using a movie rating dataset collected from twitter called Movietweetings~\citep{Dooms13crowdrec} (``snapshots 200K'').
Here, users are viewed as the buyers, movies as items, and ratings as valuations. 
As in~\citet{nan2023fast}, since the data in the original dataset is sparse, we remove users and movies with too few entries.
We then complete the matrix by using the matrix completion software fancyimpute~\citep{fancyimpute}.
We normalize the utilities such that the sum of utilities of each buyer is $1$.
Finally, the resulting instance has $n = 691$ buyers and $m = 632$ items.
We use random seed $1$ to generate budgets for buyers.

In LFM setting, we test the convergence of {\tatonnement} with three stepsizes: $\eta = 1 \times 10^{-6}, 2 \times 10^{-6}, 3 \times 10^{-6}$; 
in QLFM setting, we test the convergence with three stepsizes: $\eta = 1 \times 10^{-7}, 2 \times 10^{-7}, 3 \times 10^{-7}$.
We run $8000$ and $70000$ iterations for LFM and QLFM, respectively, to guarantee {\tatonnement} converges to a small neighborhood of the equilibrium price.

\paragraph{More detailed description of synthetic data experiments.}

In the synthetic data experiments, we consider five different distributions of utilities:
\begin{itemize}
    \item \textbf{Uniform}: $v$ is generated from the uniform distribution on $[0,1]$.
    \item \textbf{Log-normal}: $v$ is generated from the log-normal distribution associated with the standard normal distribution $\mathcal{N}(0, 1)$.
    \item \textbf{Exponential}: $v$ is generated from the exponential distribution with the scale parameter $1$.
    \item \textbf{Truncated normal}: $v$ is generated from the truncated normal distribution associated with $\mathcal{N}(0, 1)$, and truncated at ${10}^{-3}$ and $10$ standard deviations.
    \item \textbf{Uniform integers}: $v$ is generated from the uniform distribution on $\{1,\ldots,100\}$.
\end{itemize}

For each distribution, we generate four instances of different sizes: 
$(n, m) = (10, 20), (20, 40), (30, 60), (40, 80)$, where $n$ is the number of buyers and $m$ is the number of items.
We use random seed $1$ to generate all budgets and valuation matrices.
For all synthetic instances, 
we test the convergence of {\tatonnement} with three stepsizes: $\eta = 2 \times 10^{-5}, 4 \times 10^{-5}, 6 \times 10^{-5}$.
We run $200000$ iterations for each instance to guarantee {\tatonnement} converges to a small neighborhood of the equilibrium price.
Note that we only show a part of iterations in the plots for better visualization.
To more clearly show linear convergence, we zoom in on the ``straight line'' portion of the plots (the initial part of all iterations) in each figure.

\paragraph{More details on the experiments.}

In our experiments, we choose the MBB item with the smallest index when there is more than one MBB item.
In all instances, we run {\tatonnement} starting with $p^0_j = \frac{1}{m}$ for all $j \in [m]$, where $m$ is the number of items.
We use the squared error norm to measure the convergence because it is consistent with our theoretical results.
All experiments are conducted on a personal computer using Python 3.9.12.
For our hyperparameter settings, we tried around $10$ stepsizes ($1 \times 10^{-7} - 1 \times 10^{-4}$) for each setting and chose those that can clearly show the convergence.

\begin{figure}[t]
    \centering
    \includegraphics[scale=.23]{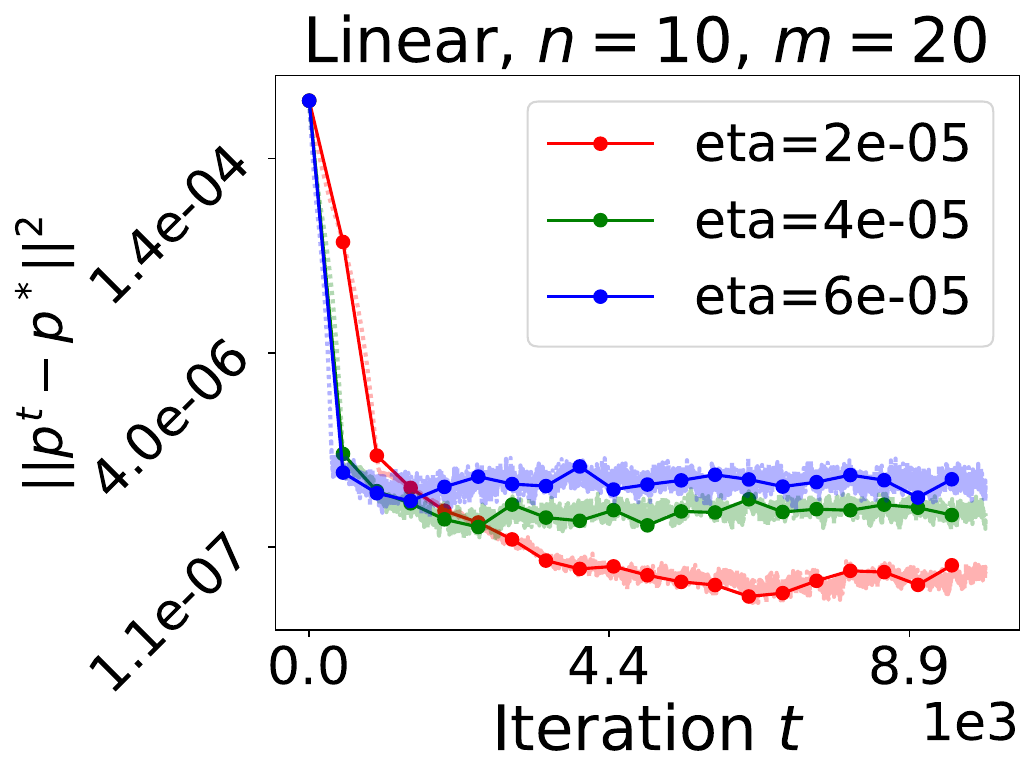}
    \includegraphics[scale=.23]{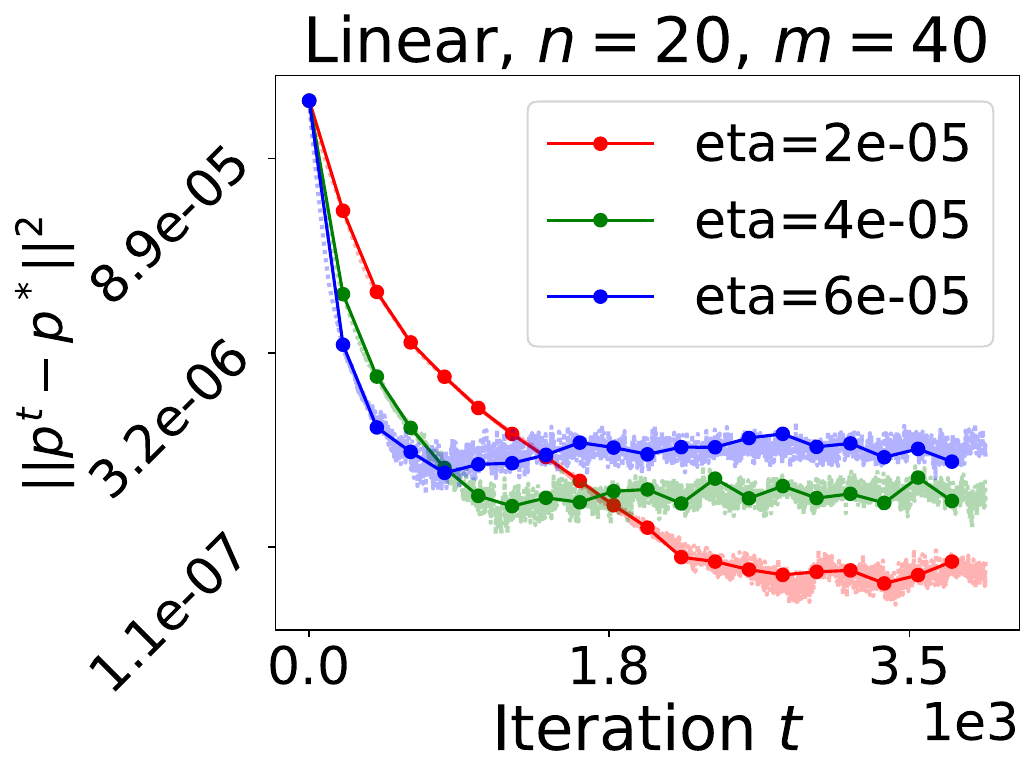}
    \includegraphics[scale=.23]{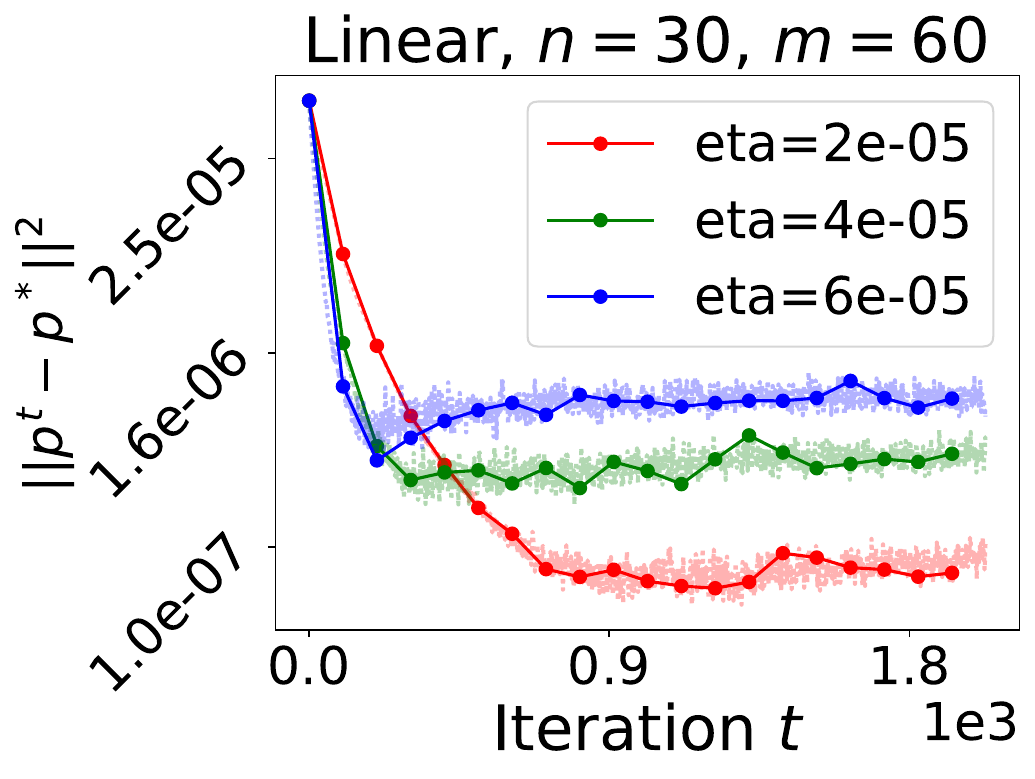}
    \includegraphics[scale=.23]{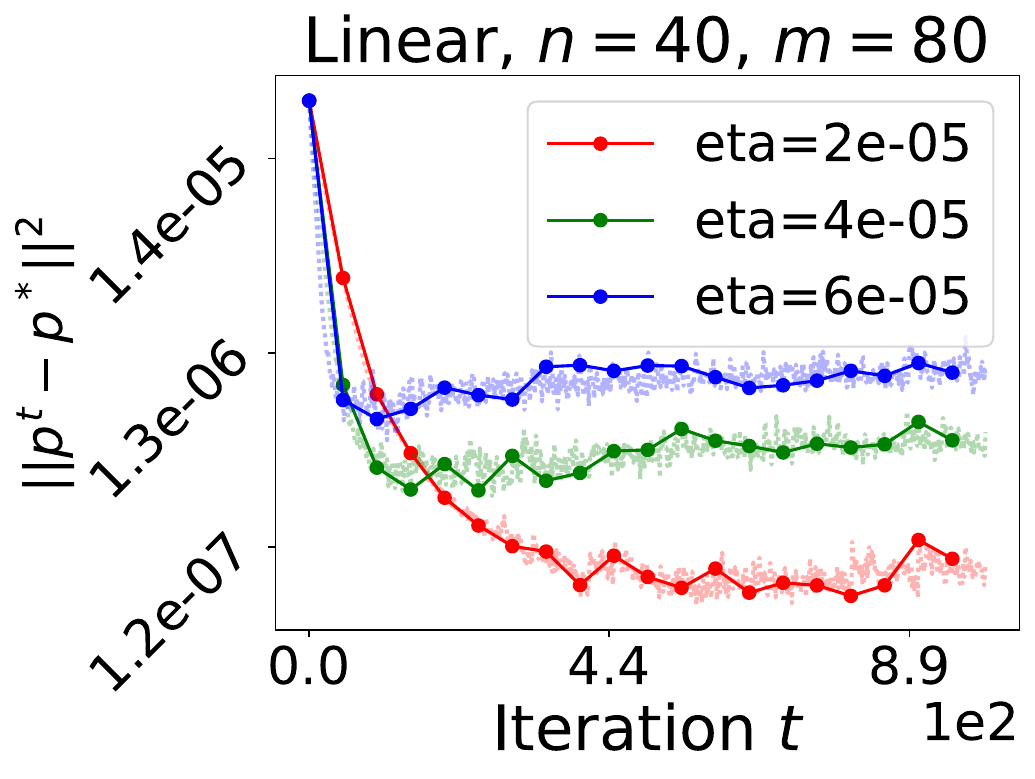}

    \includegraphics[scale=.23]{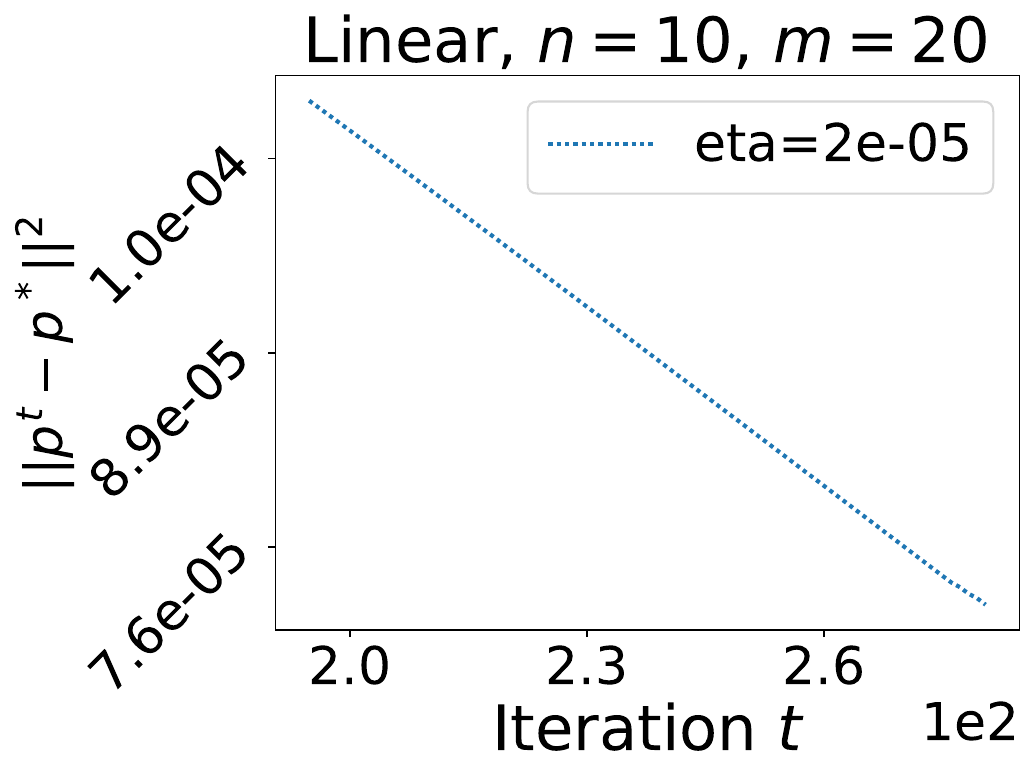}
    \includegraphics[scale=.23]{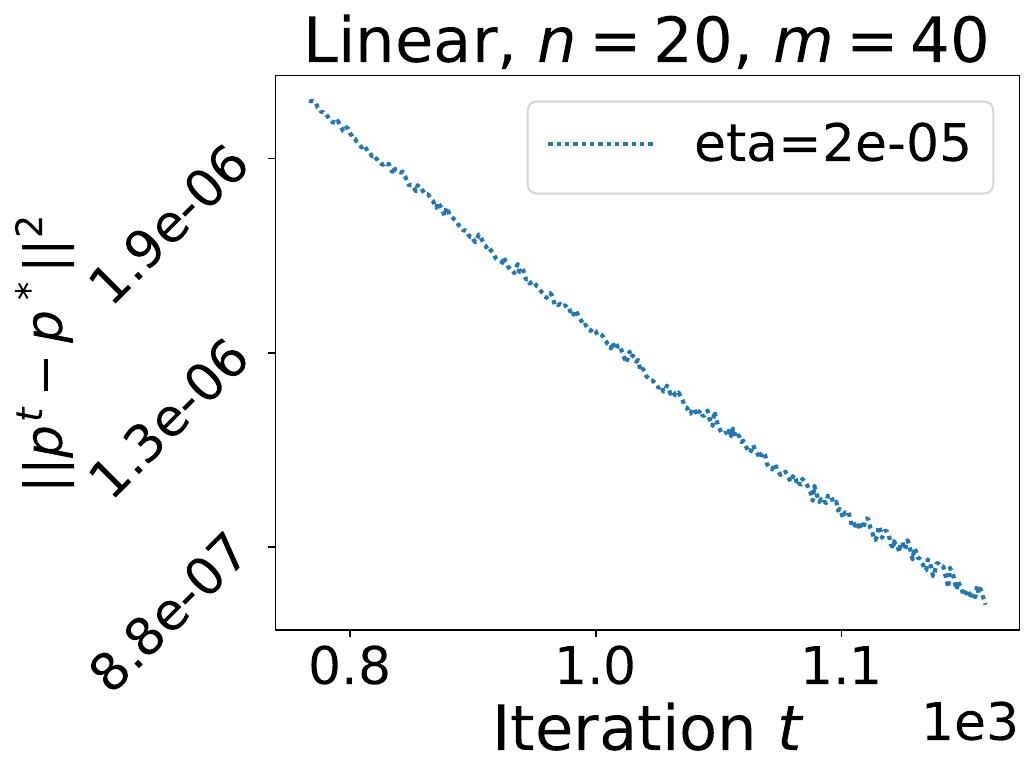}
    \includegraphics[scale=.23]{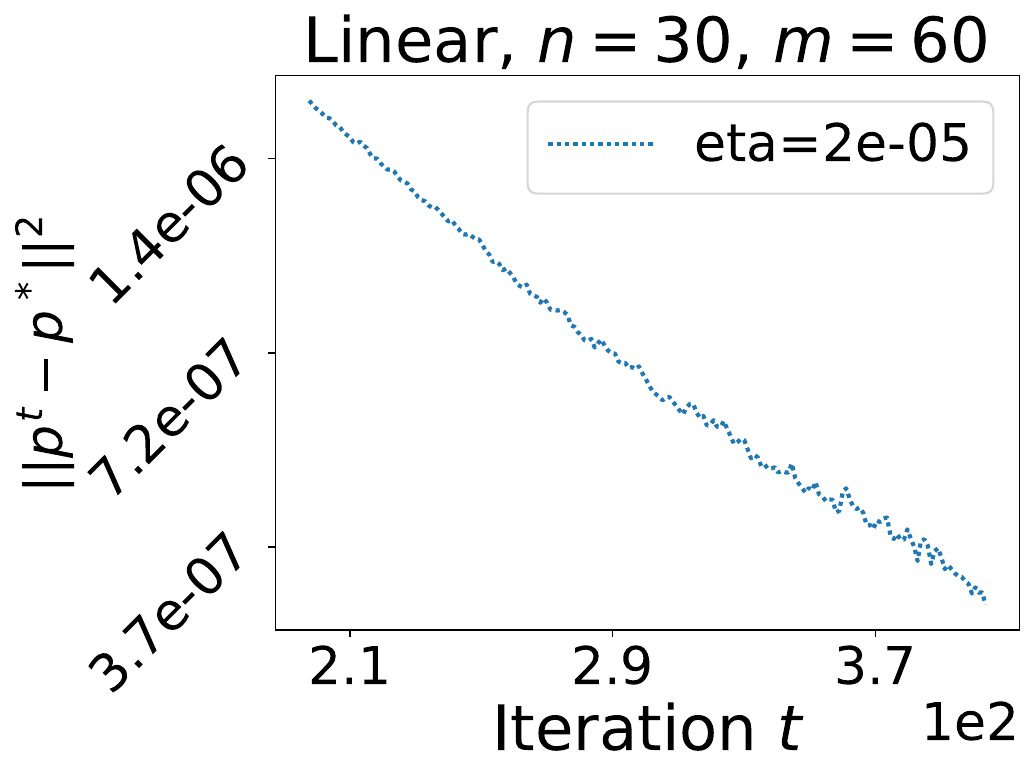}
    \includegraphics[scale=.23]{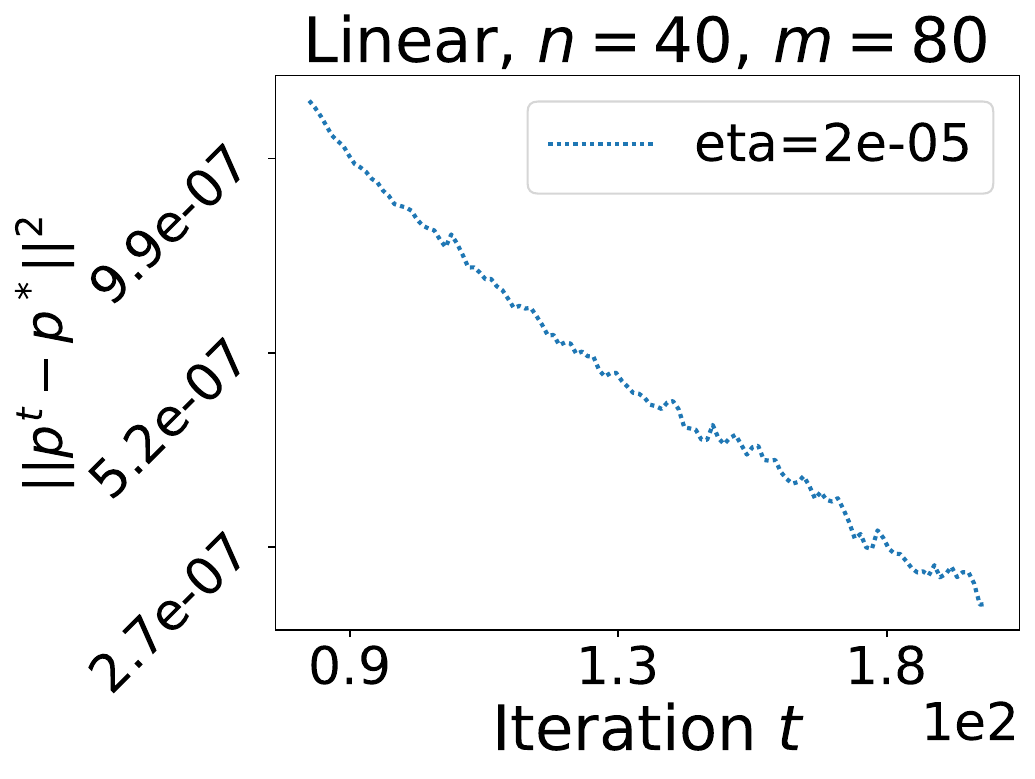}
    \caption{Convergence of squared error norms on random generated instances ($v$ is generated from the uniform distribution $[0,1)$) of different sizes under linear utilities. 
    }
    \label{fig:single-instance-uniform-error-norms-linear}
\end{figure}

\begin{figure}[t]
    \centering
    \includegraphics[scale=.23]{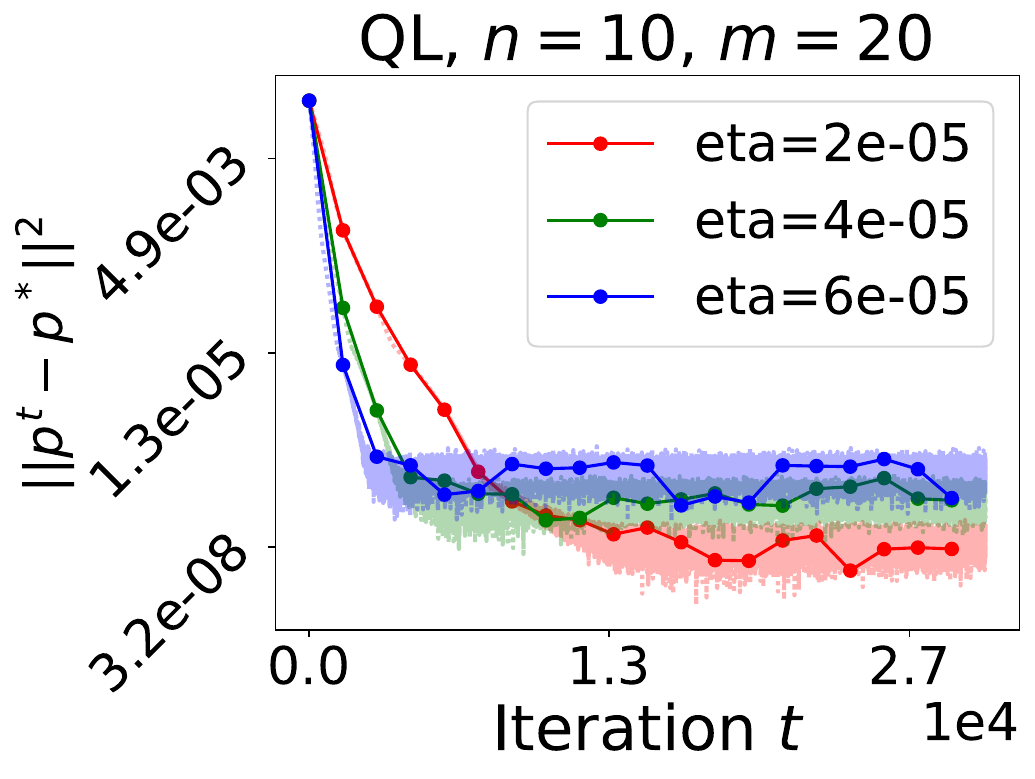}
    \includegraphics[scale=.23]{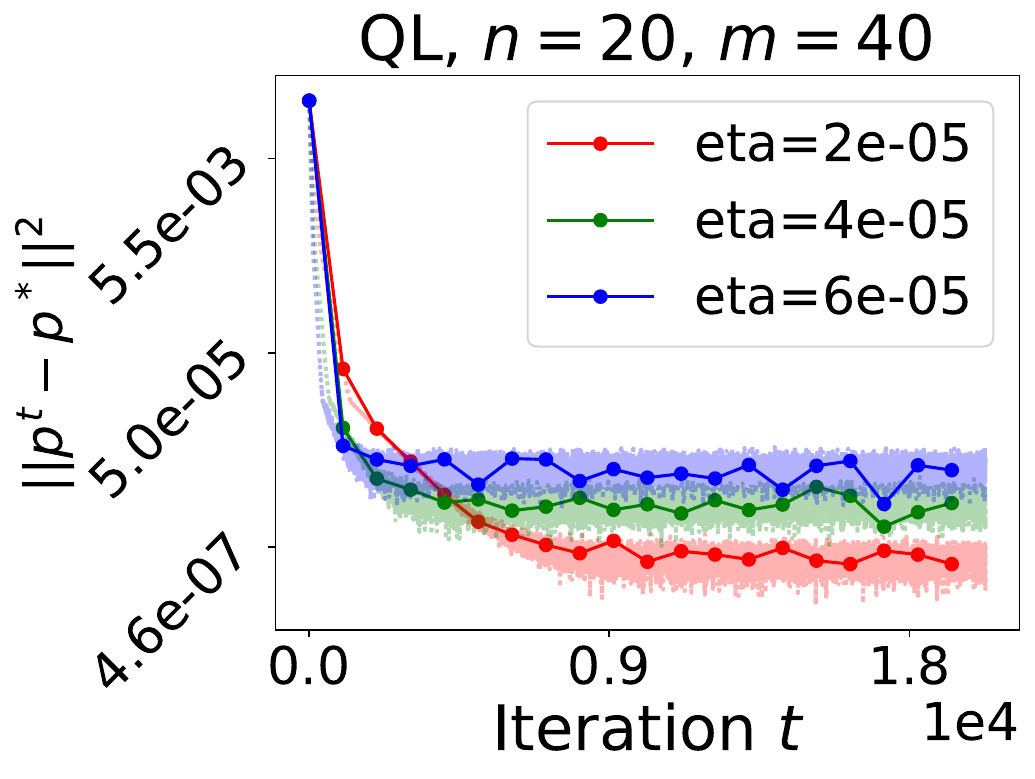}
    \includegraphics[scale=.23]{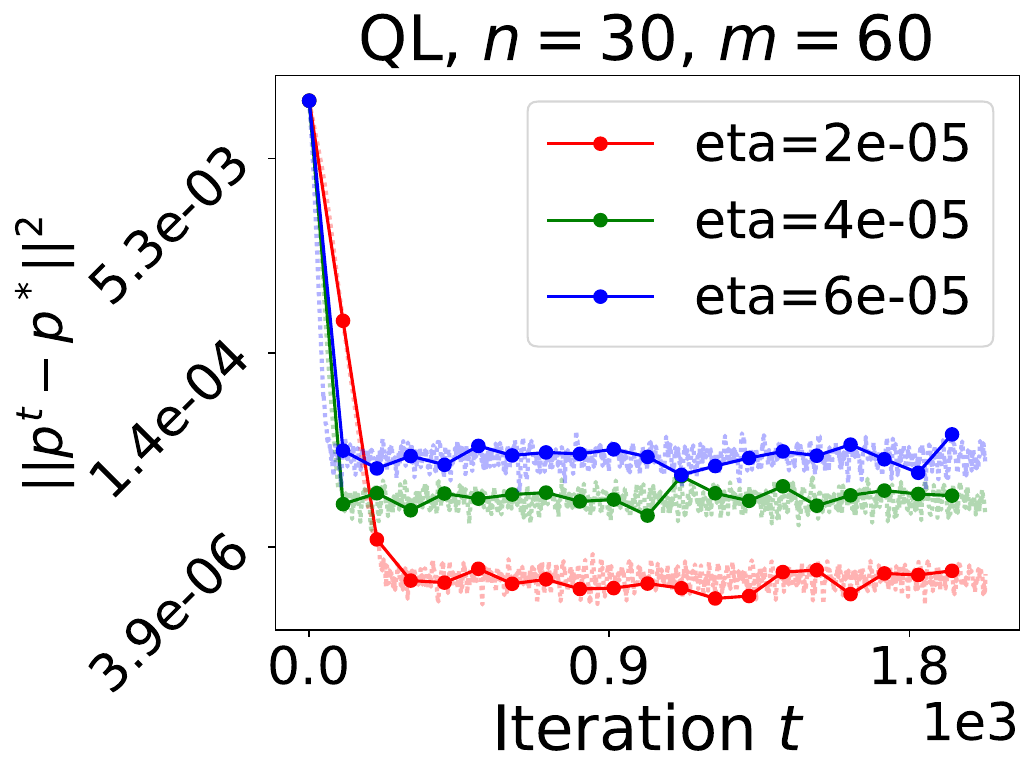}
    \includegraphics[scale=.23]{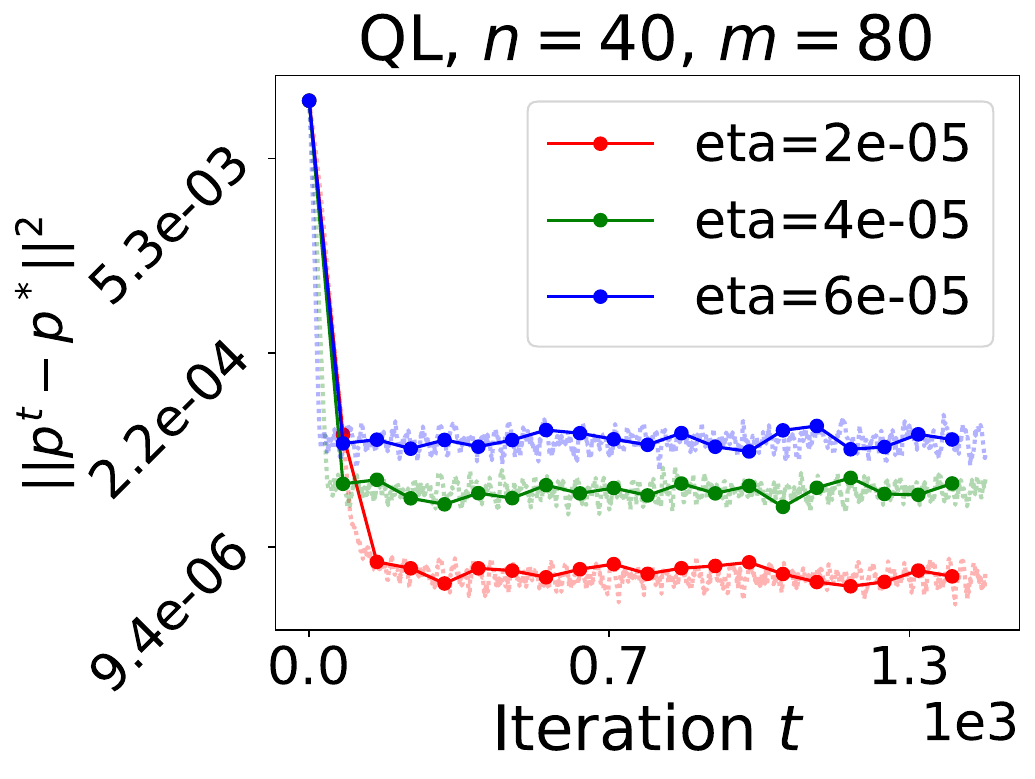}

    \includegraphics[scale=.23]{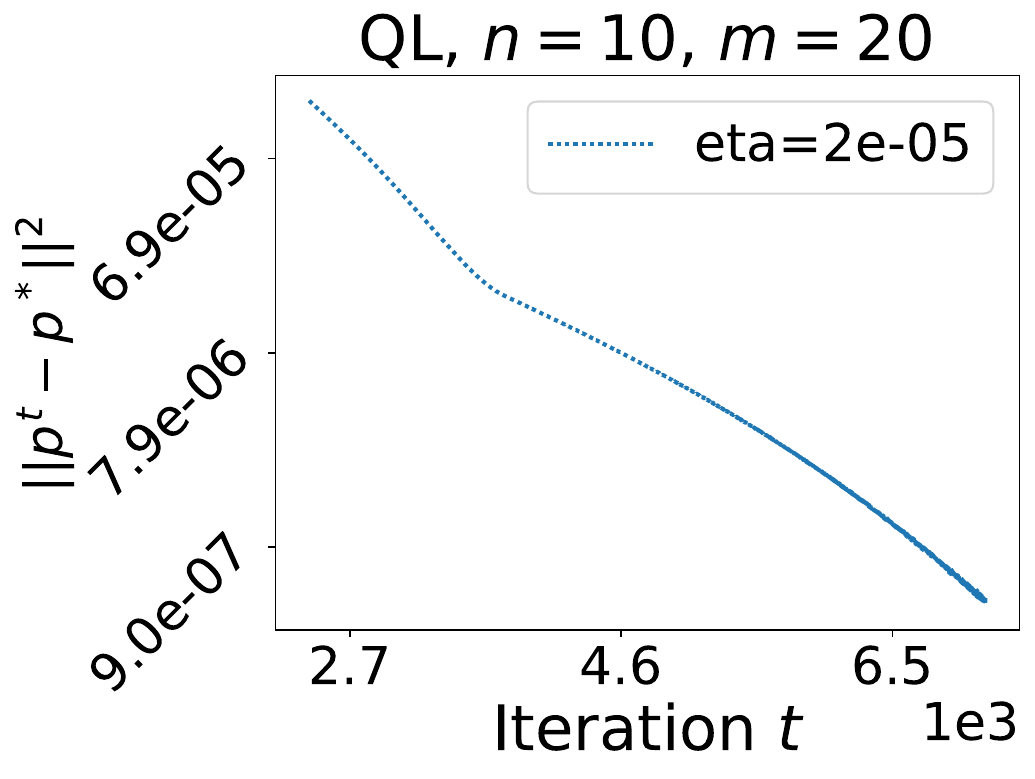}
    \includegraphics[scale=.23]{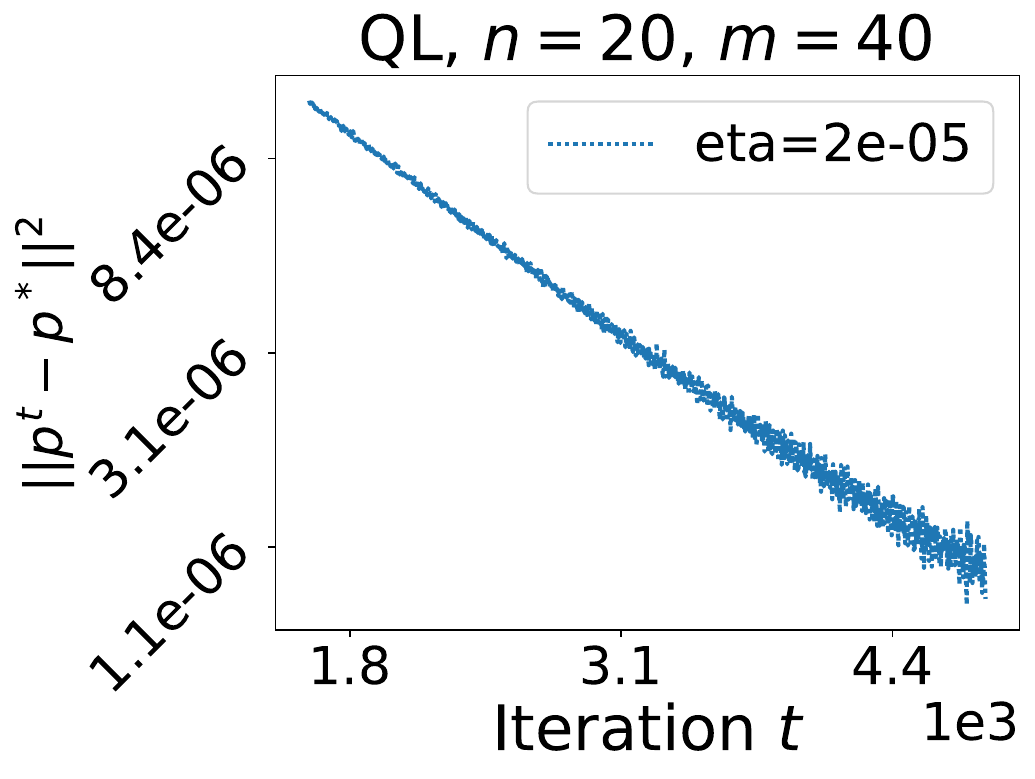}
    \includegraphics[scale=.23]{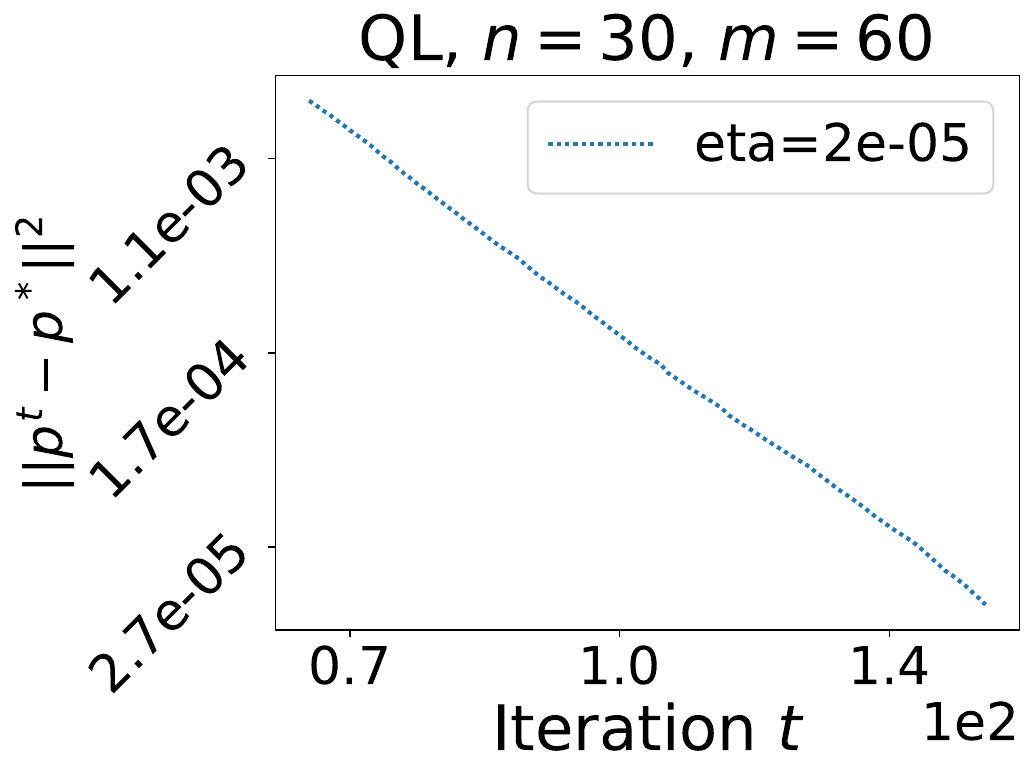}
    \includegraphics[scale=.23]{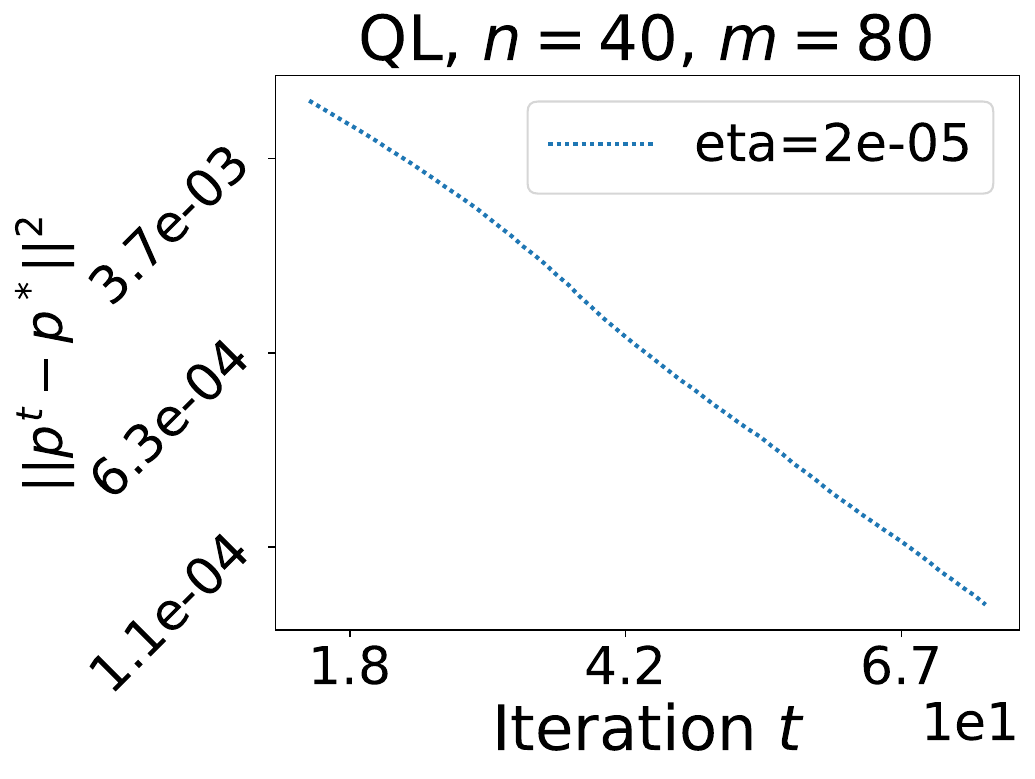}
    \caption{Convergence of squared error norms on random generated instances ($v$ is generated from the uniform distribution $[0,1)$) of different sizes under quasi-linear utilities. 
    }
    \label{fig:single-instance-uniform-error-norms-ql}
\end{figure}

\begin{figure}[t]
    \centering
    \includegraphics[scale=.23]{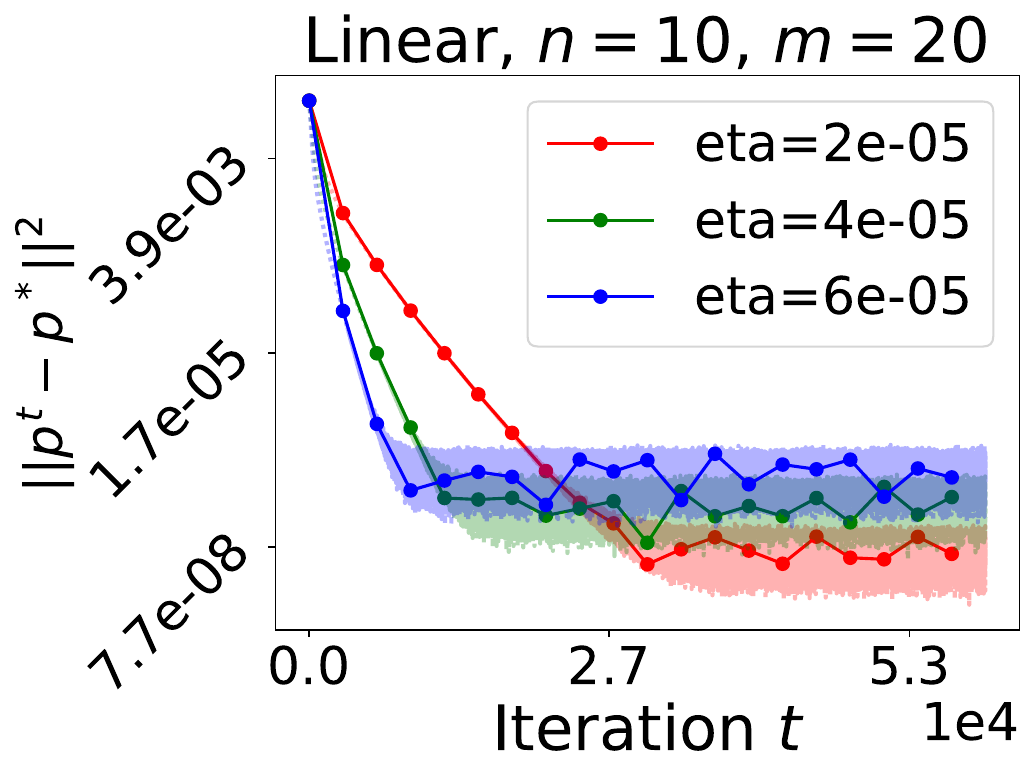}
    \includegraphics[scale=.23]{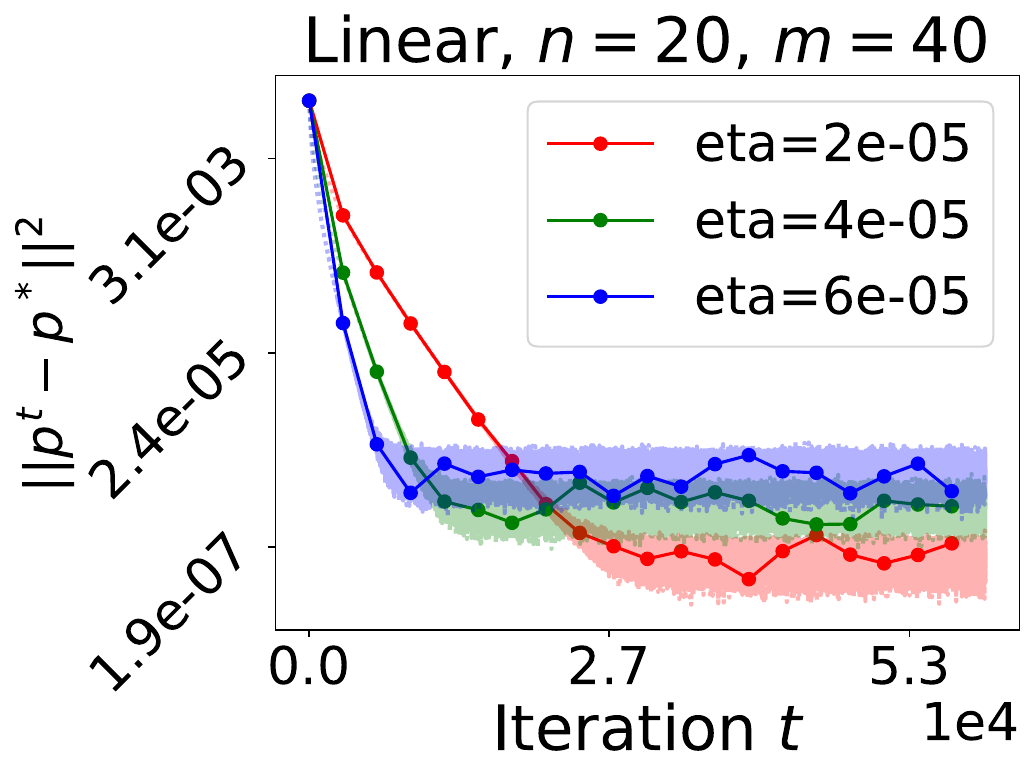}
    \includegraphics[scale=.23]{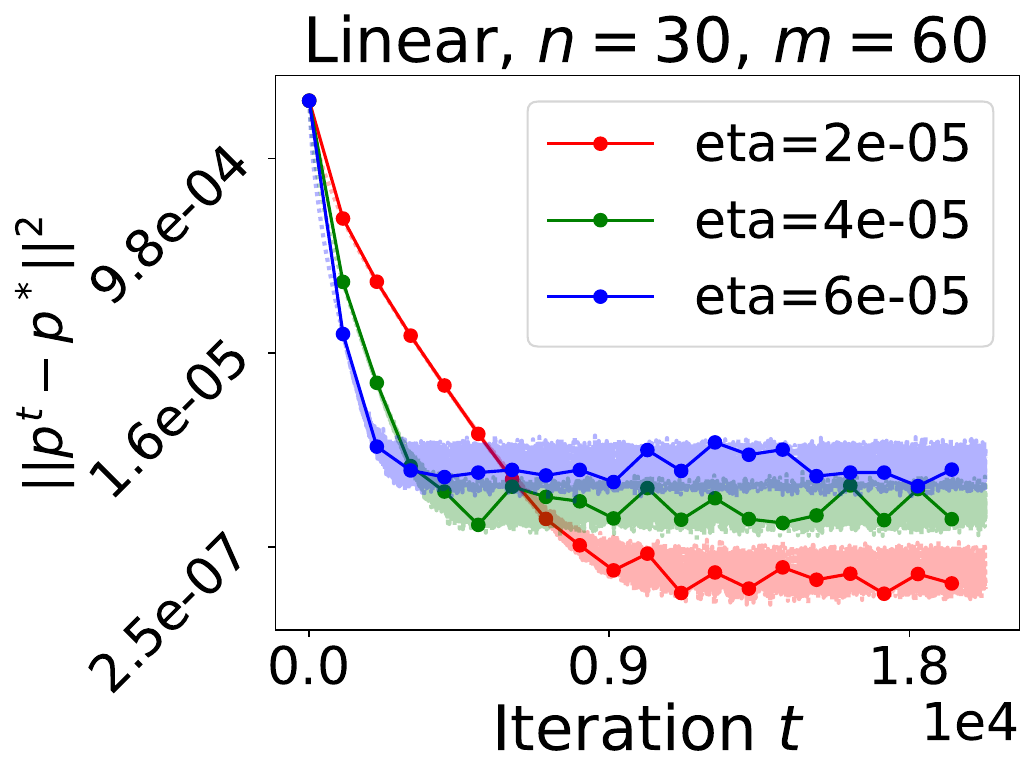}
    \includegraphics[scale=.23]{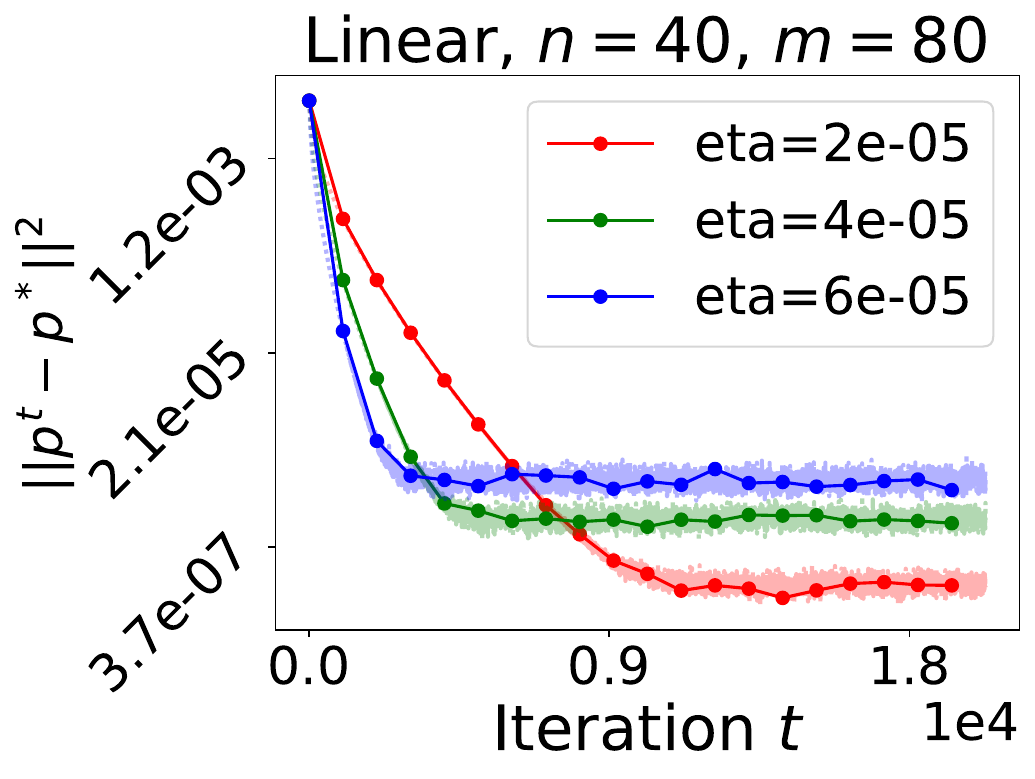}

    \includegraphics[scale=.23]{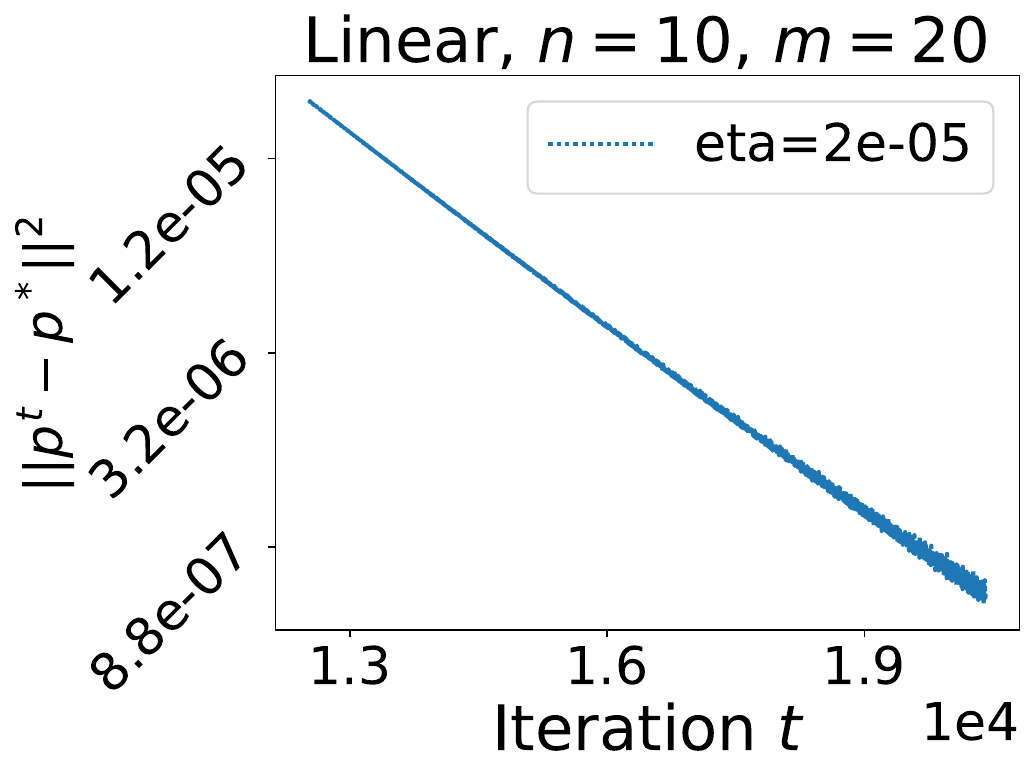}
    \includegraphics[scale=.23]{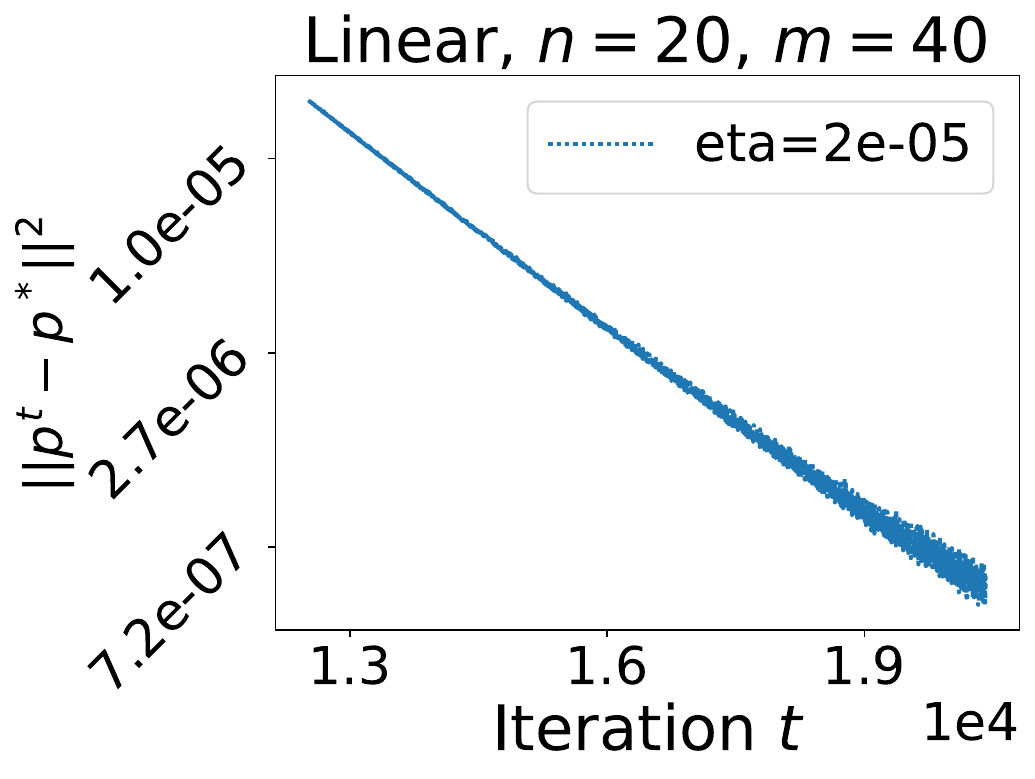}
    \includegraphics[scale=.23]{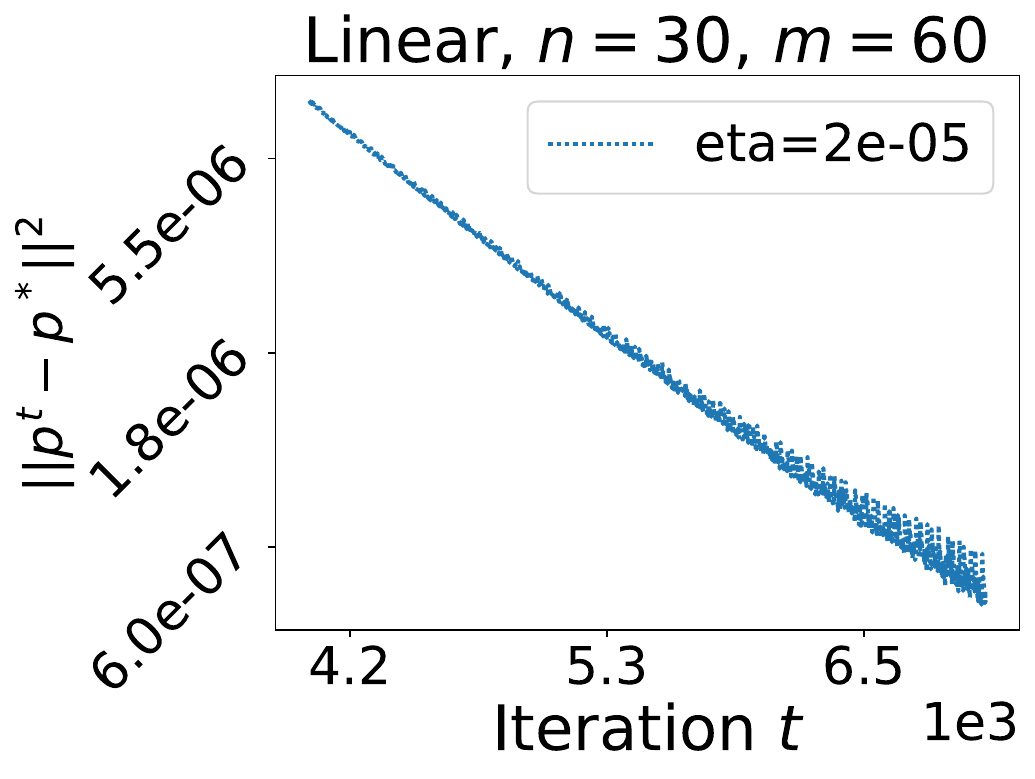}
    \includegraphics[scale=.23]{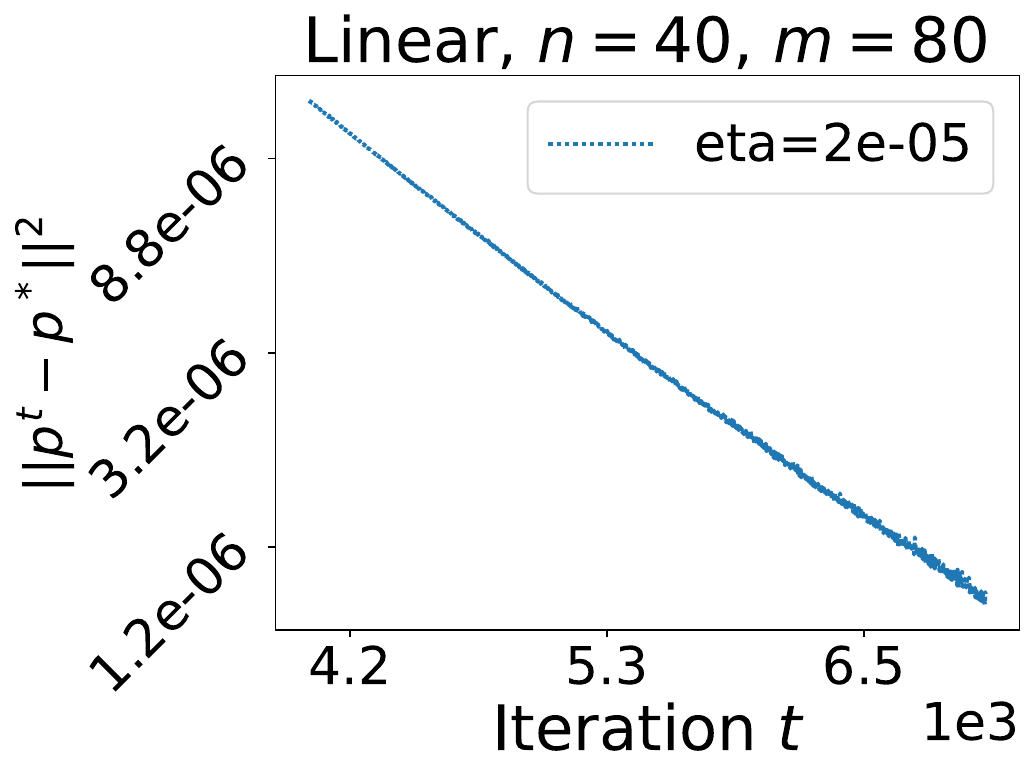}
    \caption{Convergence of squared error norms on random generated instances ($v$ is generated from the log-normal distribution associated with the standard normal distribution $\mathcal{N}(0, 1)$) of different sizes under linear utilities. 
    }
    \label{fig:single-instance-lognormal-error-norms-linear}
\end{figure}

\begin{figure}[t]
    \centering
    \includegraphics[scale=.23]{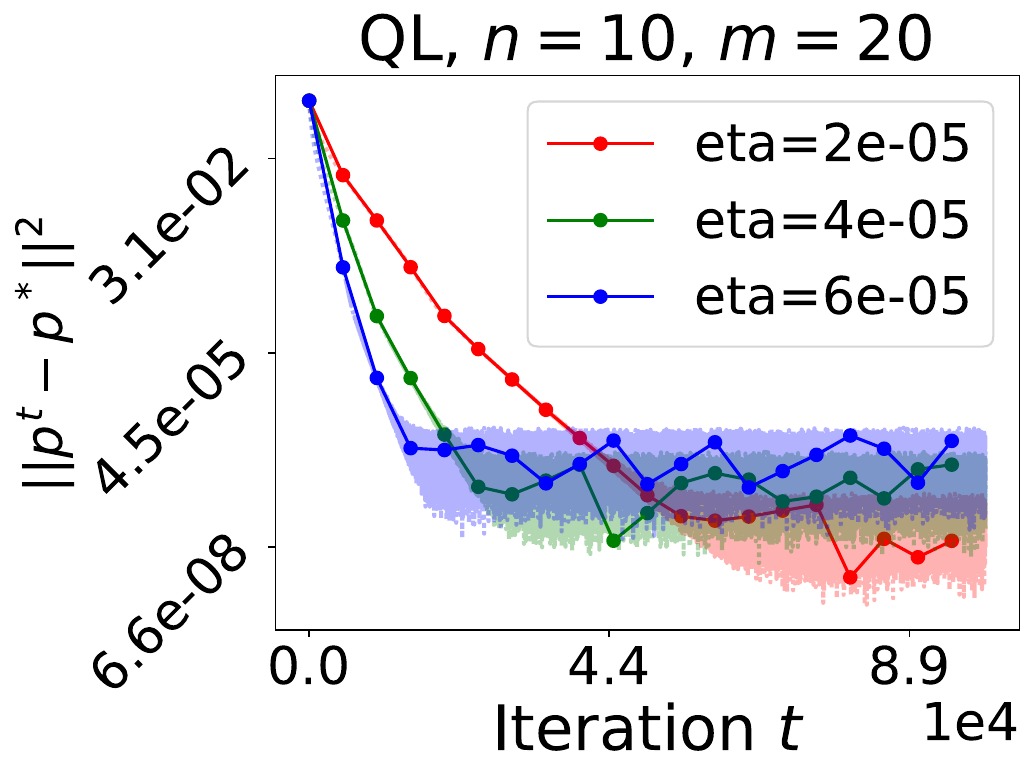}
    \includegraphics[scale=.23]{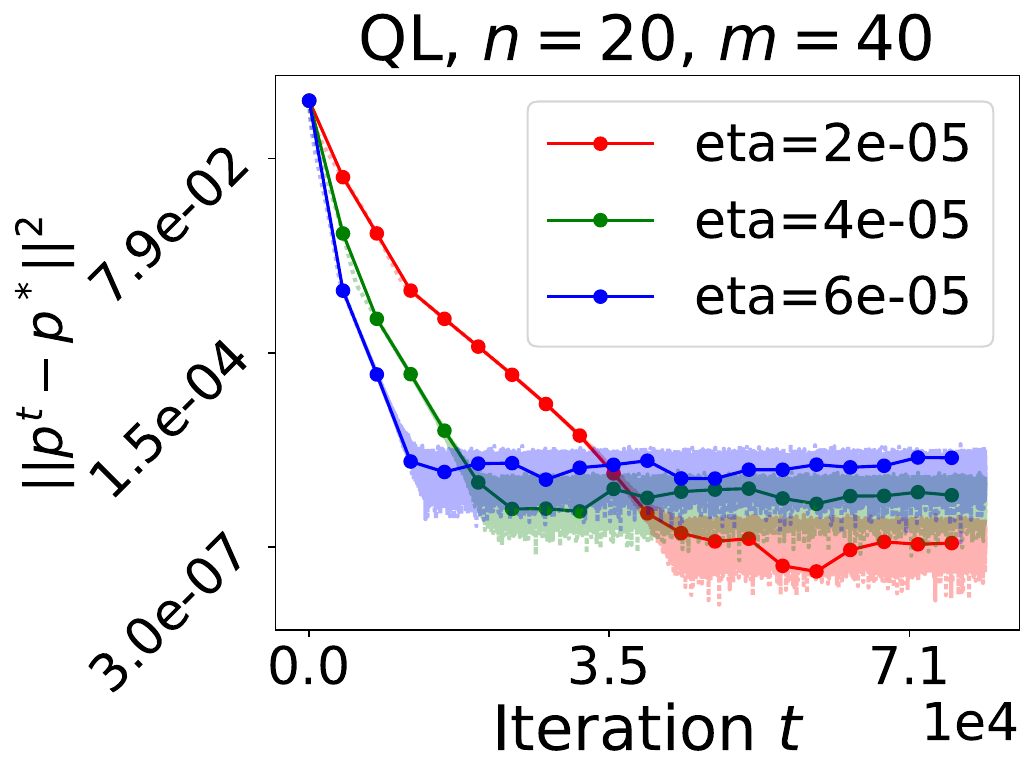}
    \includegraphics[scale=.23]{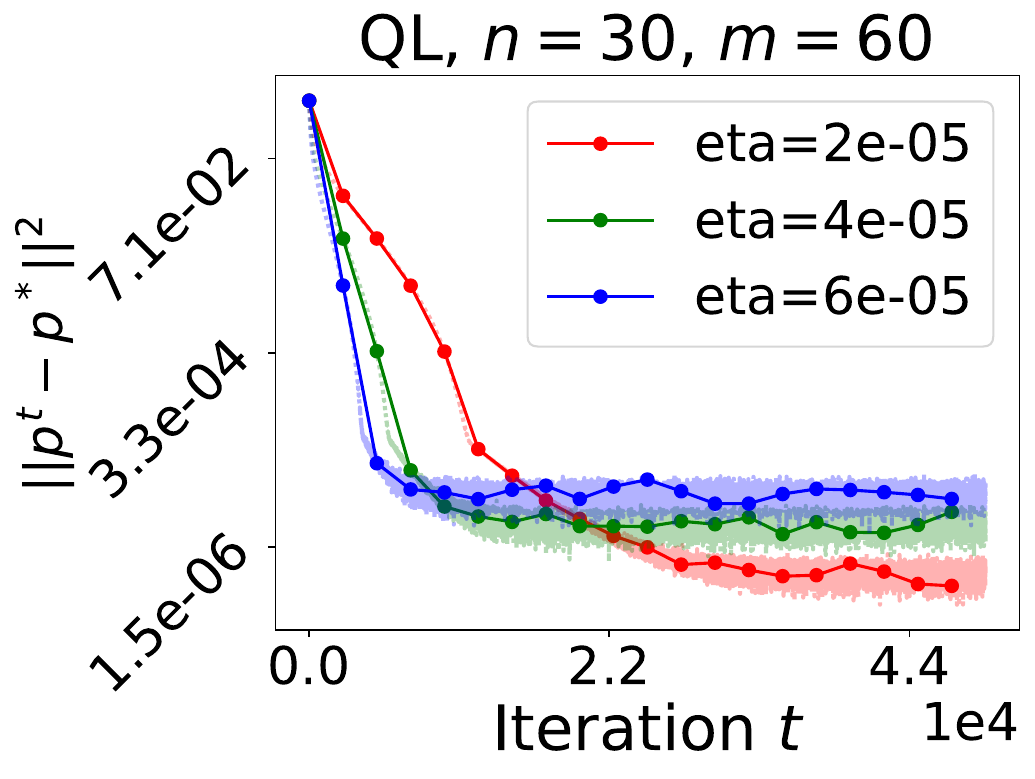}
    \includegraphics[scale=.23]{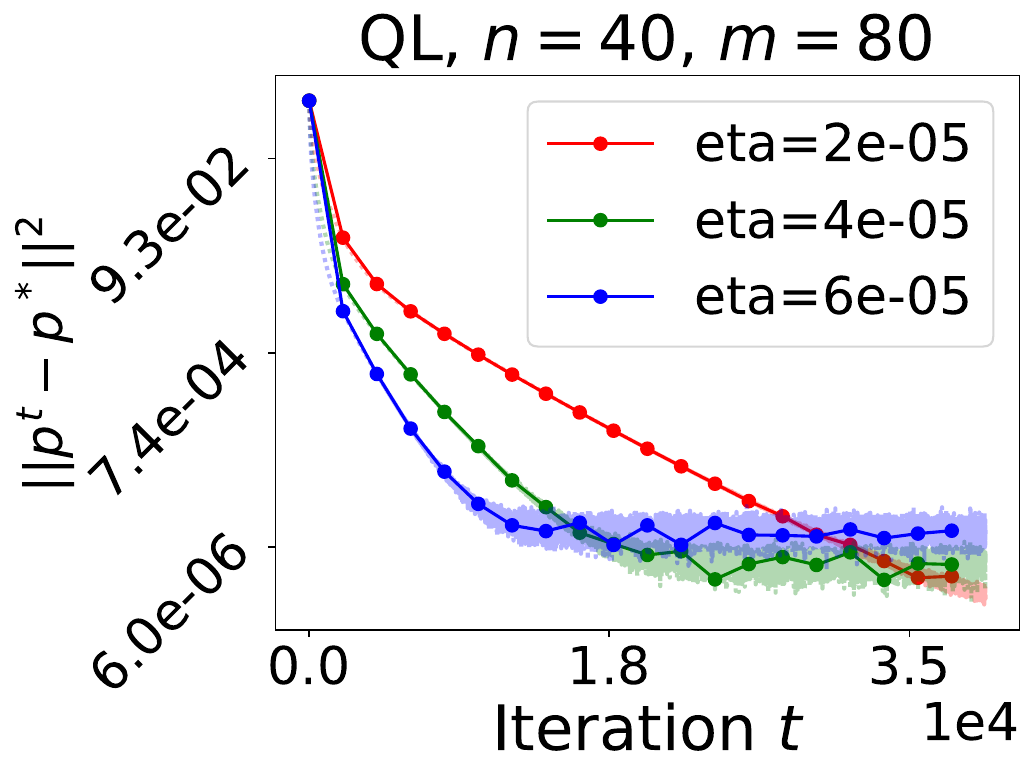}

    \includegraphics[scale=.23]{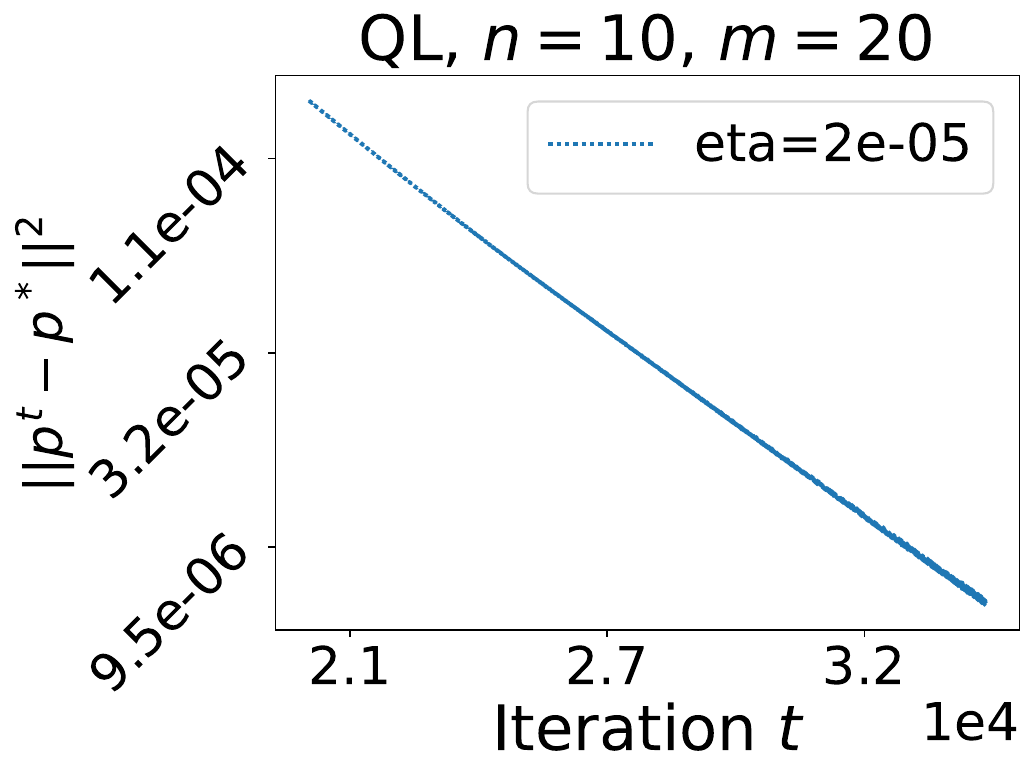}
    \includegraphics[scale=.23]{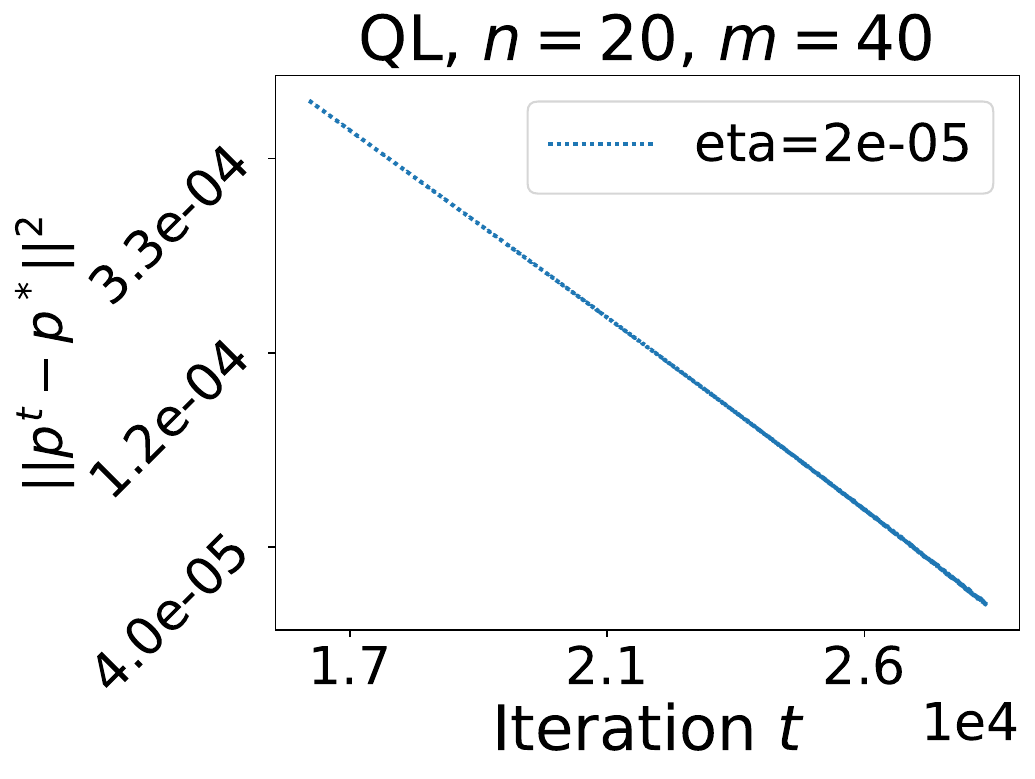}
    \includegraphics[scale=.23]{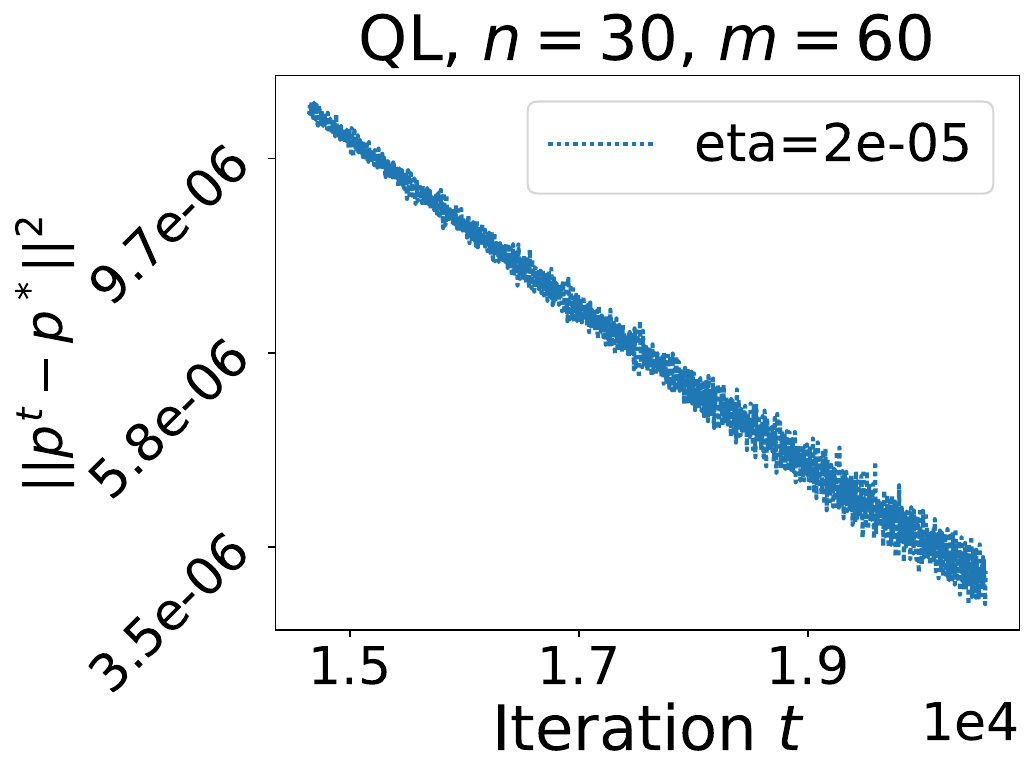}
    \includegraphics[scale=.23]{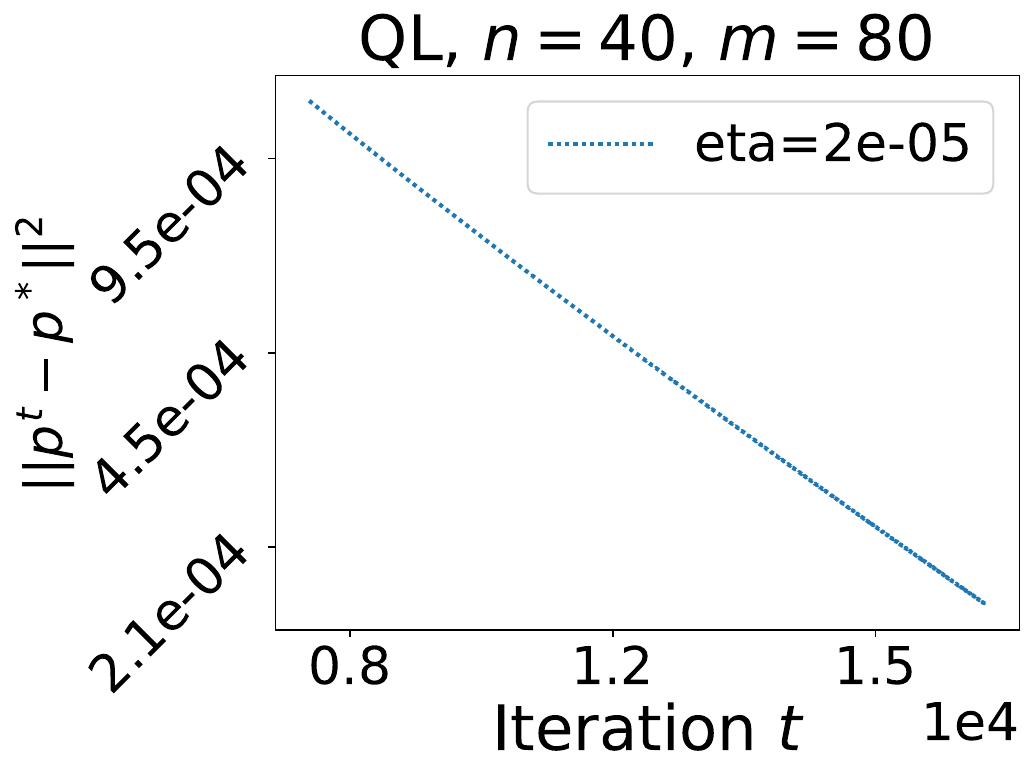}
    \caption{Convergence of squared error norms on random generated instances ($v$ is generated from the log-normal distribution associated with the standard normal distribution $\mathcal{N}(0, 1)$) of different sizes under quasi-linear utilities. 
    }
    \label{fig:single-instance-lognormal-error-norms-ql}
\end{figure}

\begin{figure}[t]
    \centering
    \includegraphics[scale=.23]{plots/linear/single-instance/exponential/error-norm-sq-vs-t-different-stepsizes-1-10-20.pdf}
    \includegraphics[scale=.23]{plots/linear/single-instance/exponential/error-norm-sq-vs-t-different-stepsizes-1-20-40.pdf}
    \includegraphics[scale=.23]{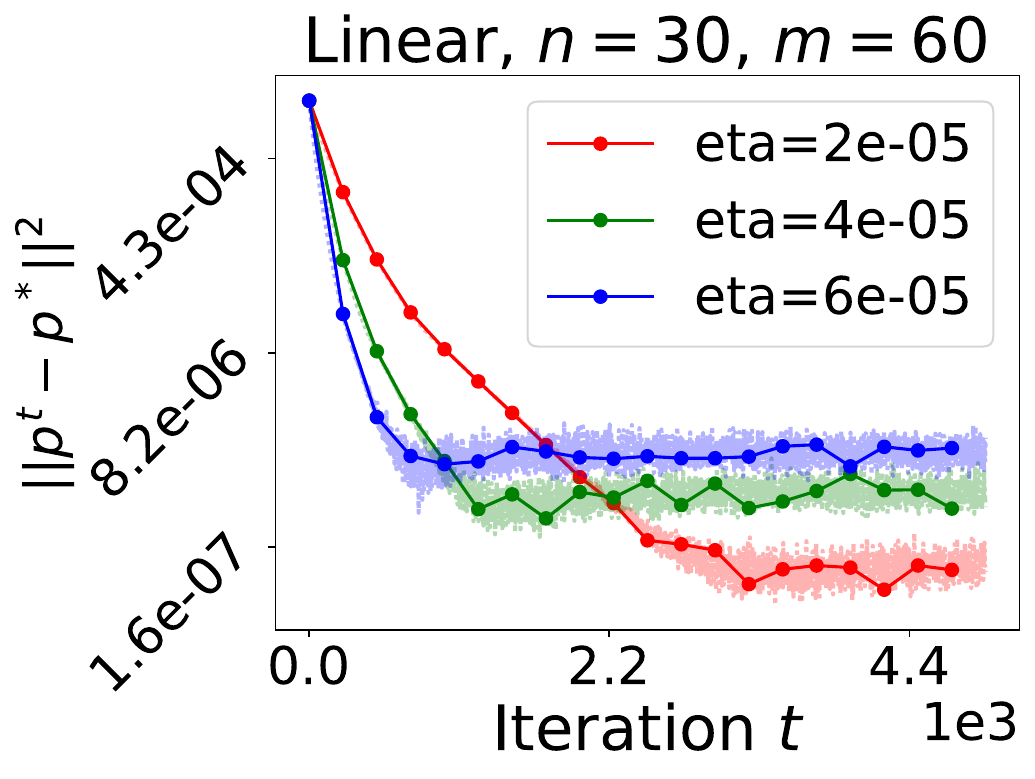}
    \includegraphics[scale=.23]{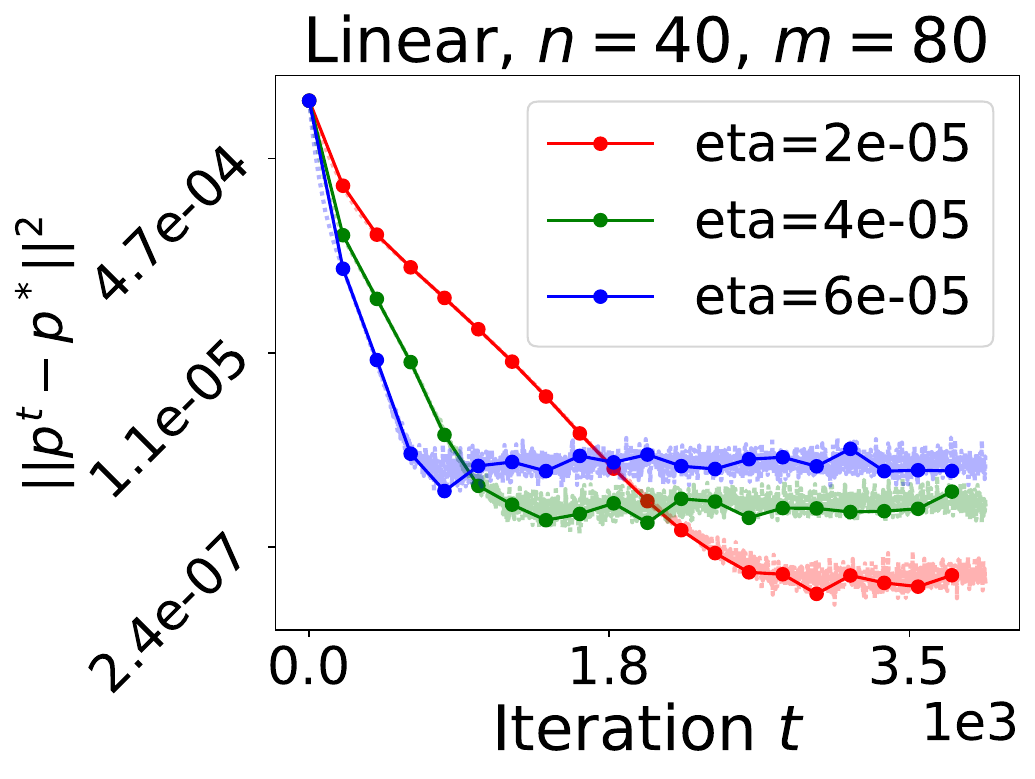}

    \includegraphics[scale=.23]{plots/linear/single-instance/exponential/error-norm-sq-vs-t-straight-line-1-10-20-2e-05.pdf}
    \includegraphics[scale=.23]{plots/linear/single-instance/exponential/error-norm-sq-vs-t-straight-line-1-20-40-2e-05.pdf}
    \includegraphics[scale=.23]{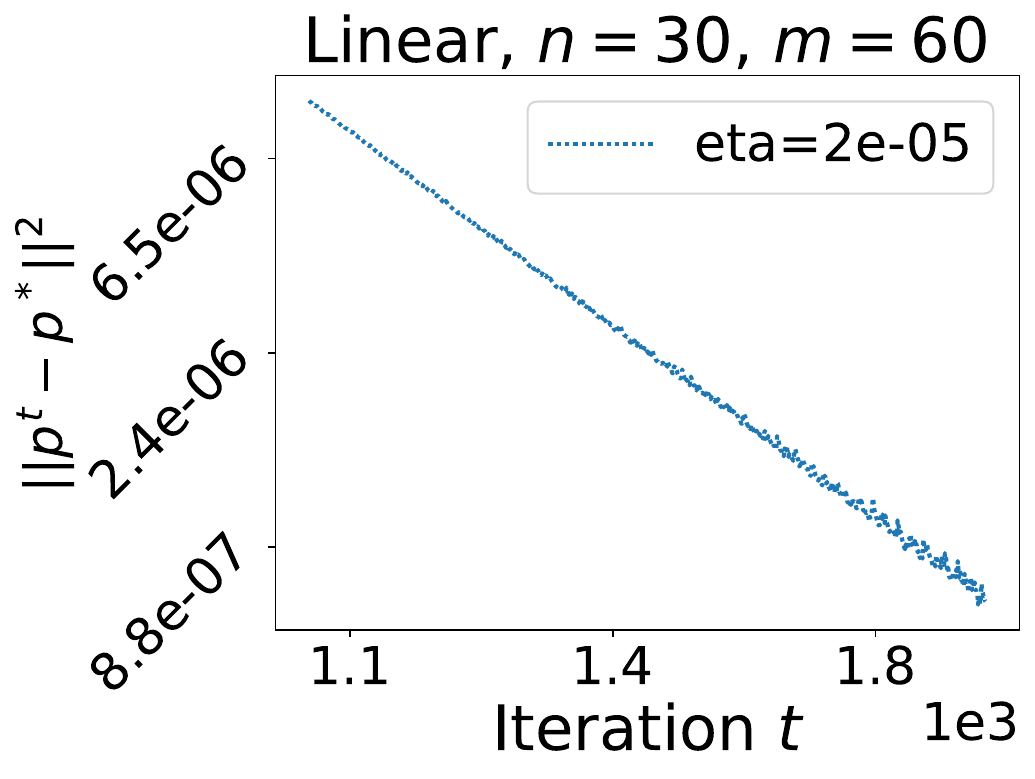}
    \includegraphics[scale=.23]{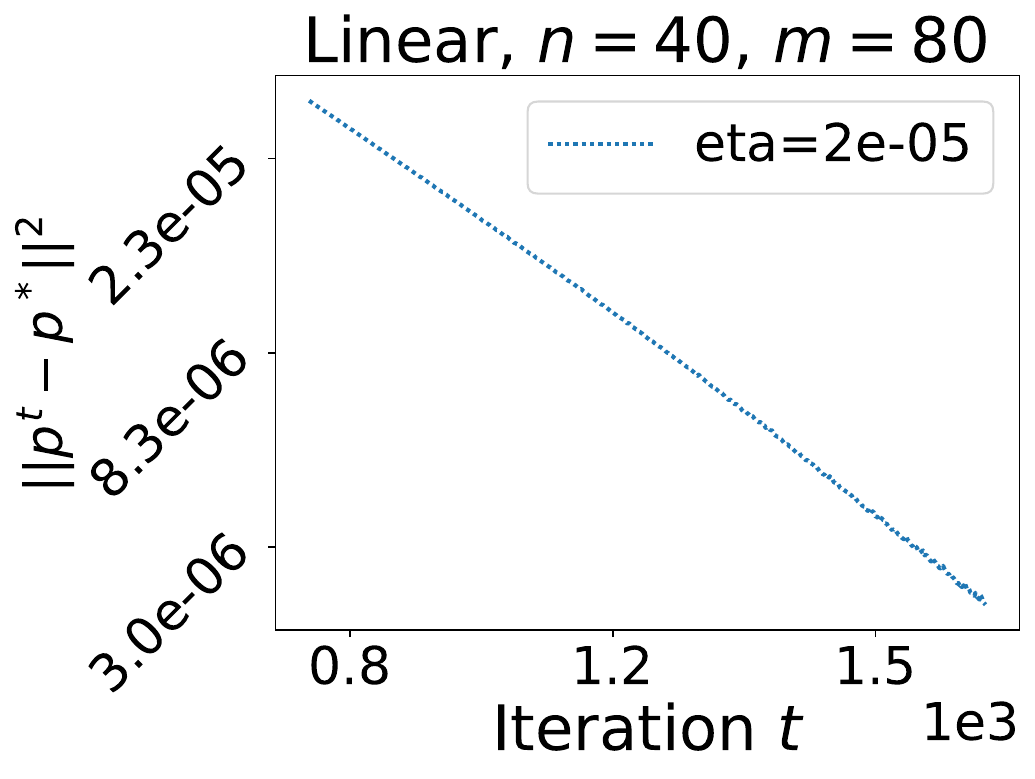}
    \caption{Convergence of squared error norms on random generated instances ($v$ is generated from the exponential distribution with the scale parameter $1$) of different sizes under linear utilities. 
    }
    \label{fig:single-instance-exponential-error-norms-linear}
\end{figure}

\begin{figure}[t]
    \centering
    \includegraphics[scale=.23]{plots/ql/single-instance/exponential/error-norm-sq-vs-t-different-stepsizes-1-10-20.pdf}
    \includegraphics[scale=.23]{plots/ql/single-instance/exponential/error-norm-sq-vs-t-different-stepsizes-1-20-40.pdf}
    \includegraphics[scale=.23]{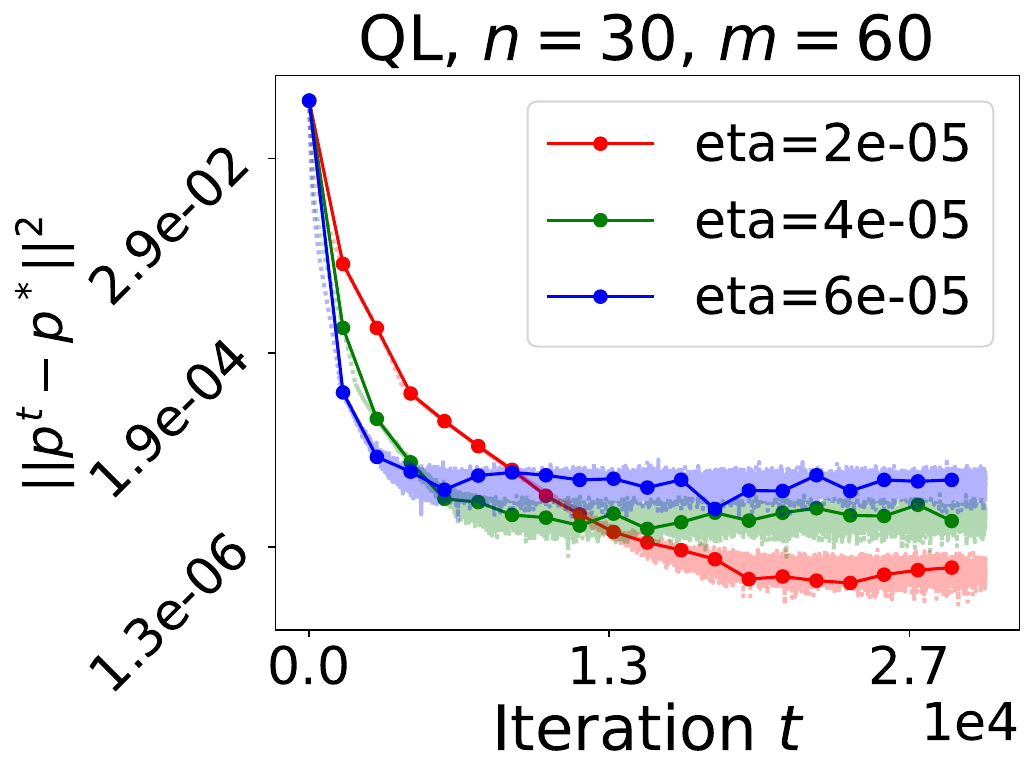}
    \includegraphics[scale=.23]{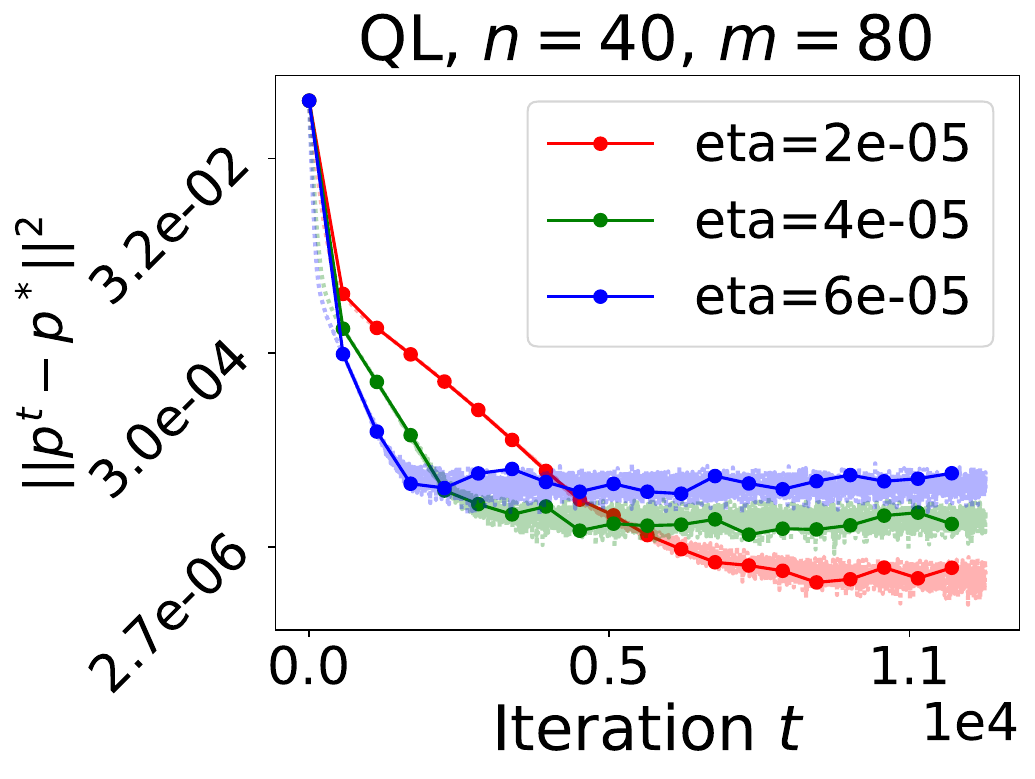}

    \includegraphics[scale=.23]{plots/ql/single-instance/exponential/error-norm-sq-vs-t-straight-line-1-10-20-2e-05.pdf}
    \includegraphics[scale=.23]{plots/ql/single-instance/exponential/error-norm-sq-vs-t-straight-line-1-20-40-2e-05.pdf}
    \includegraphics[scale=.23]{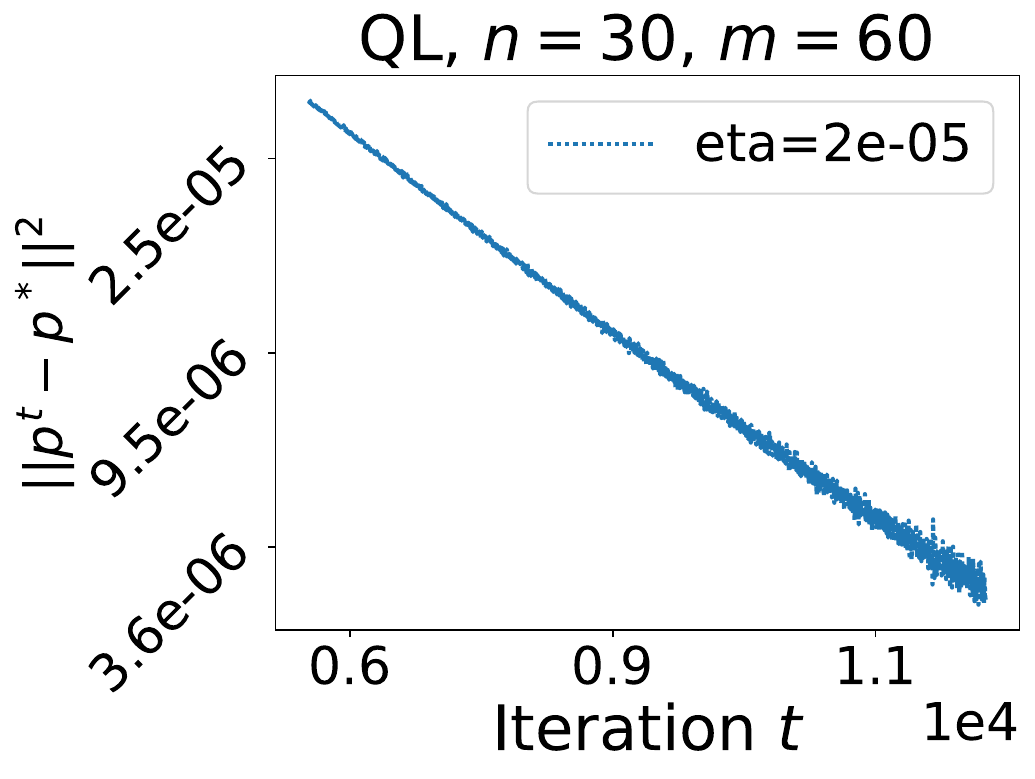}
    \includegraphics[scale=.23]{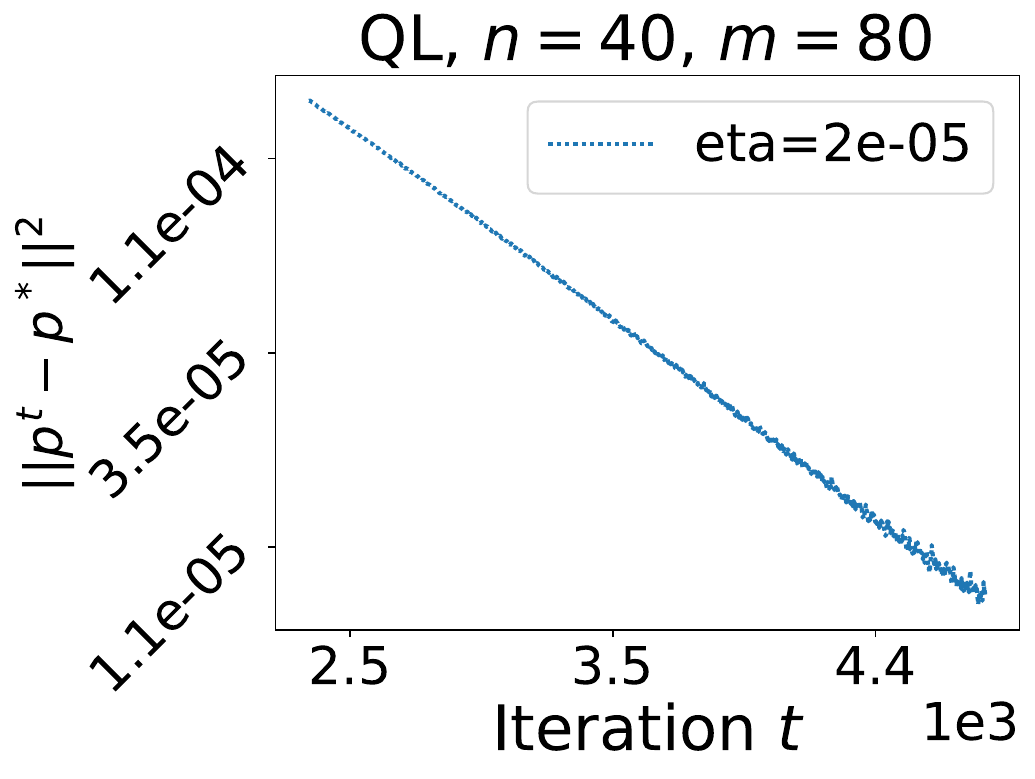}
    \caption{Convergence of squared error norms on random generated instances ($v$ is generated from the exponential distribution with the scale parameter $1$) of different sizes under quasi-linear utilities. 
    }
    \label{fig:single-instance-exponential-error-norms-ql}
\end{figure}

\begin{figure}[t]
    \centering
    \includegraphics[scale=.23]{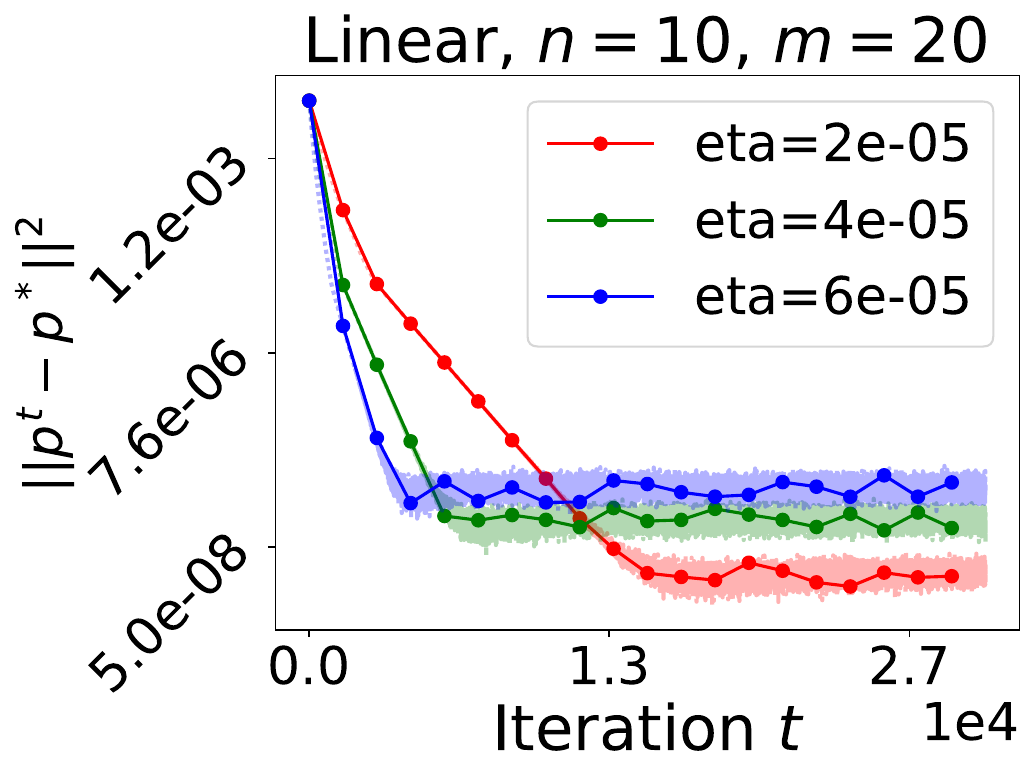}
    \includegraphics[scale=.23]{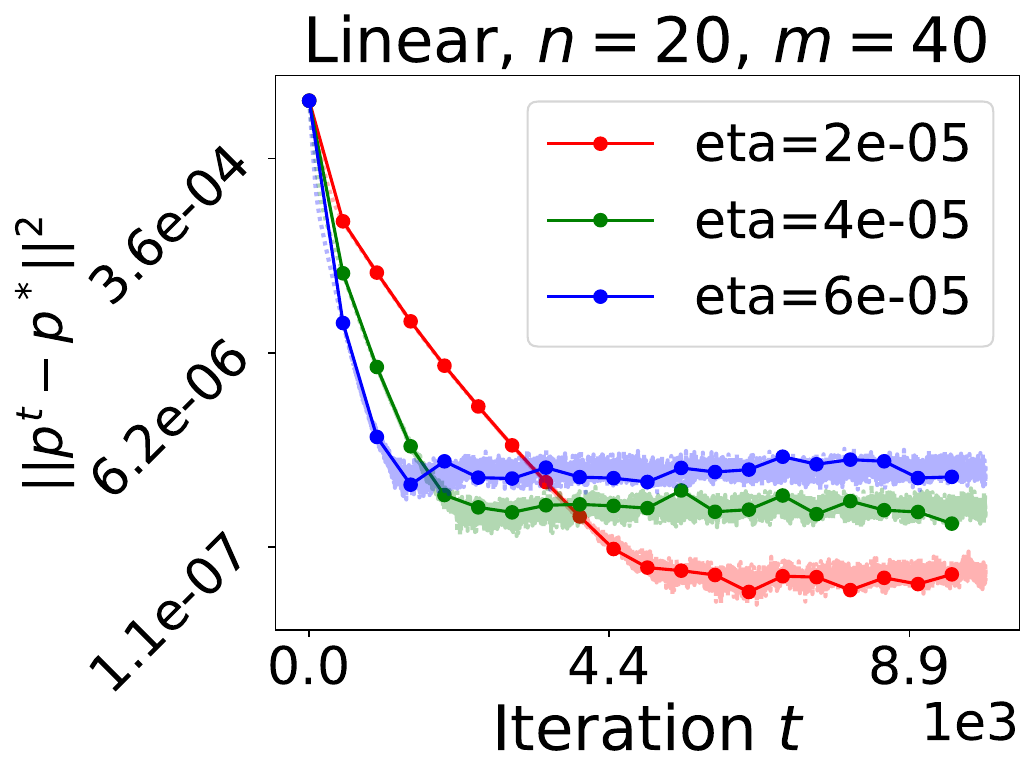}
    \includegraphics[scale=.23]{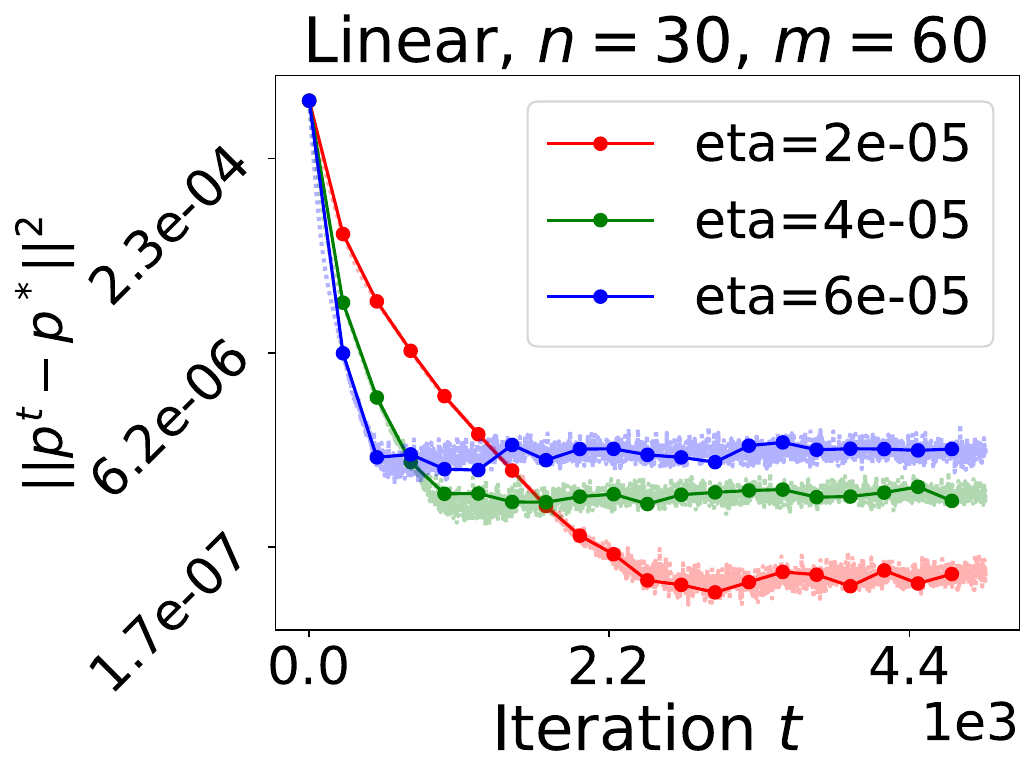}
    \includegraphics[scale=.23]{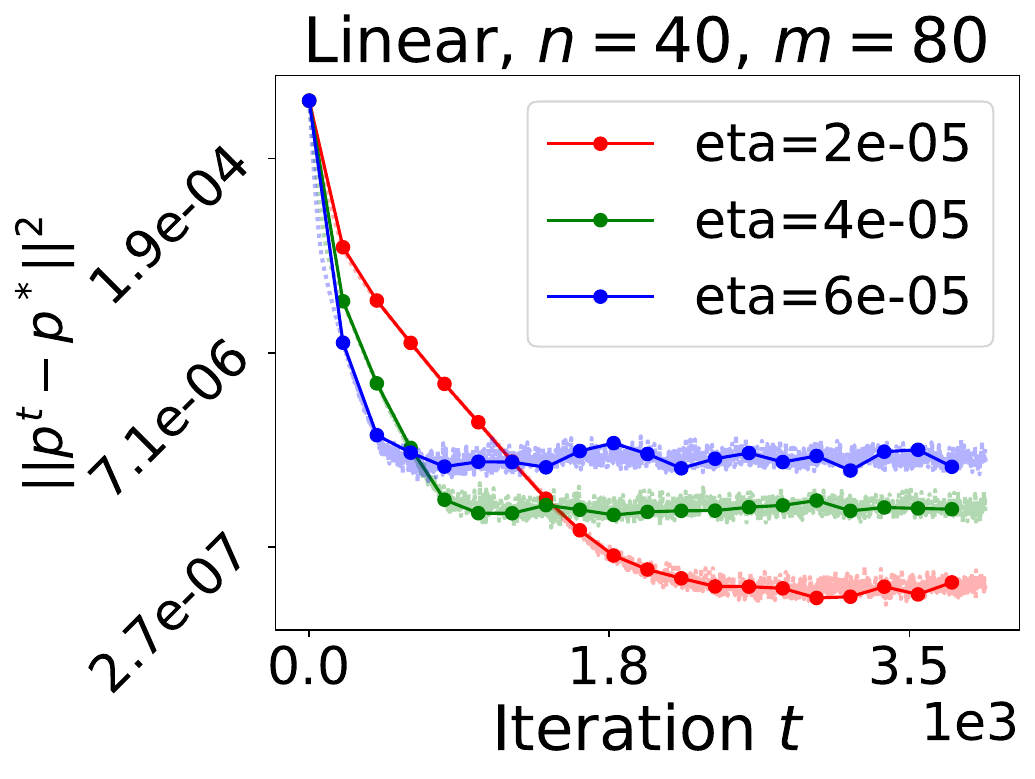}

    \includegraphics[scale=.23]{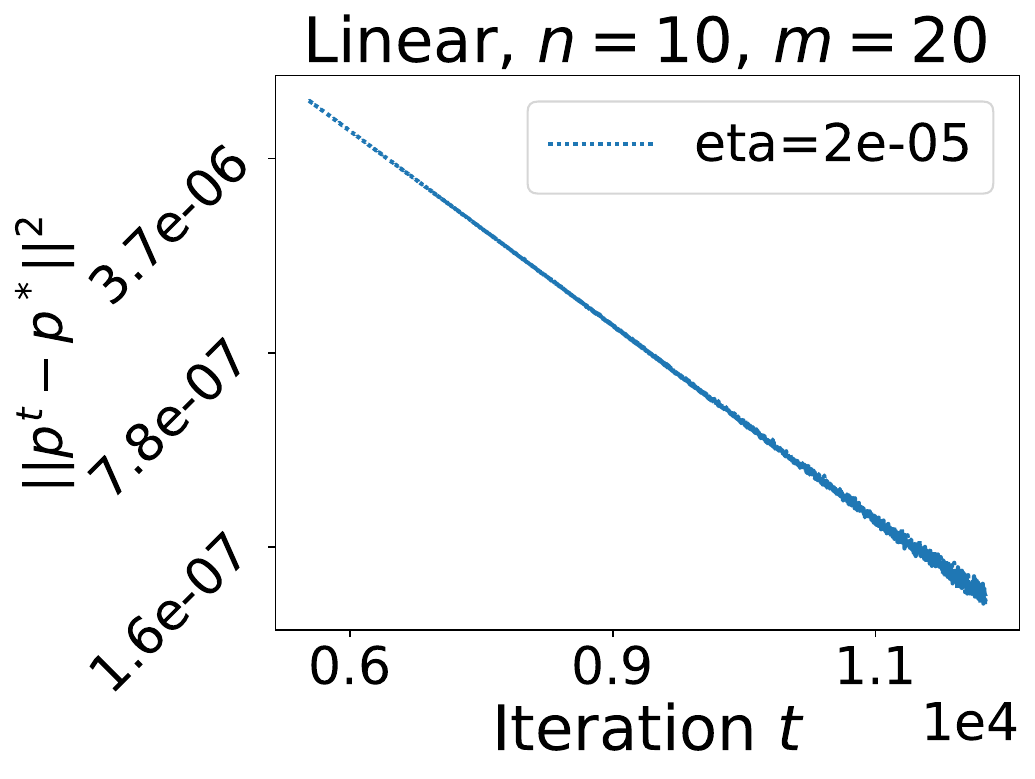}
    \includegraphics[scale=.23]{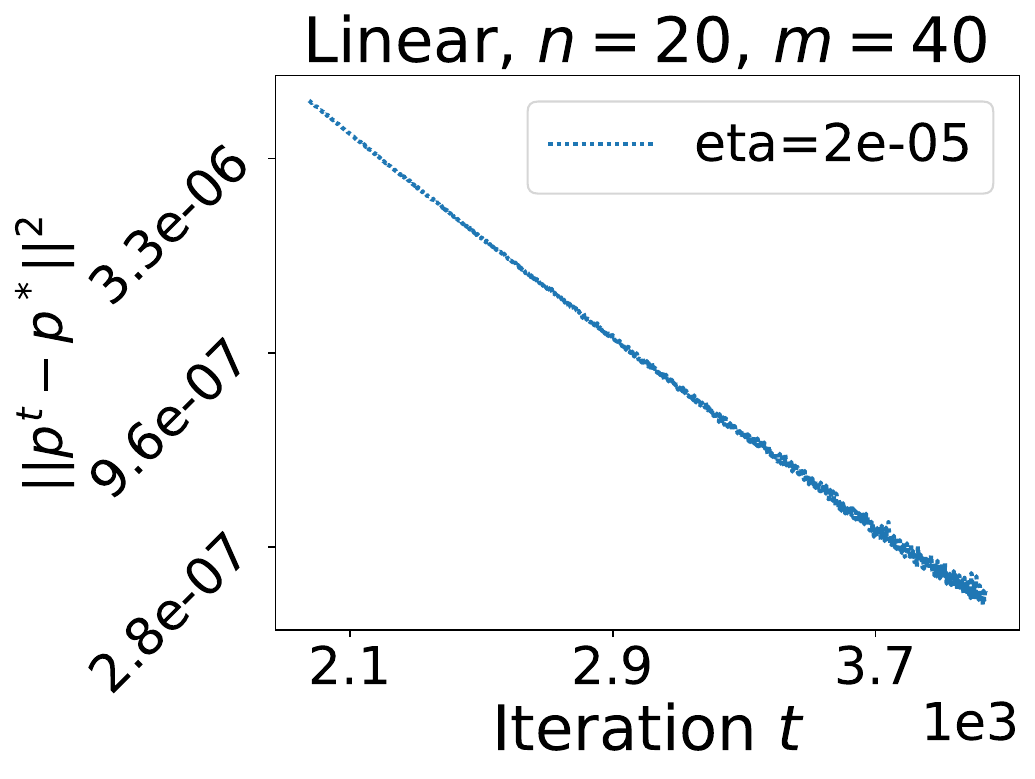}
    \includegraphics[scale=.23]{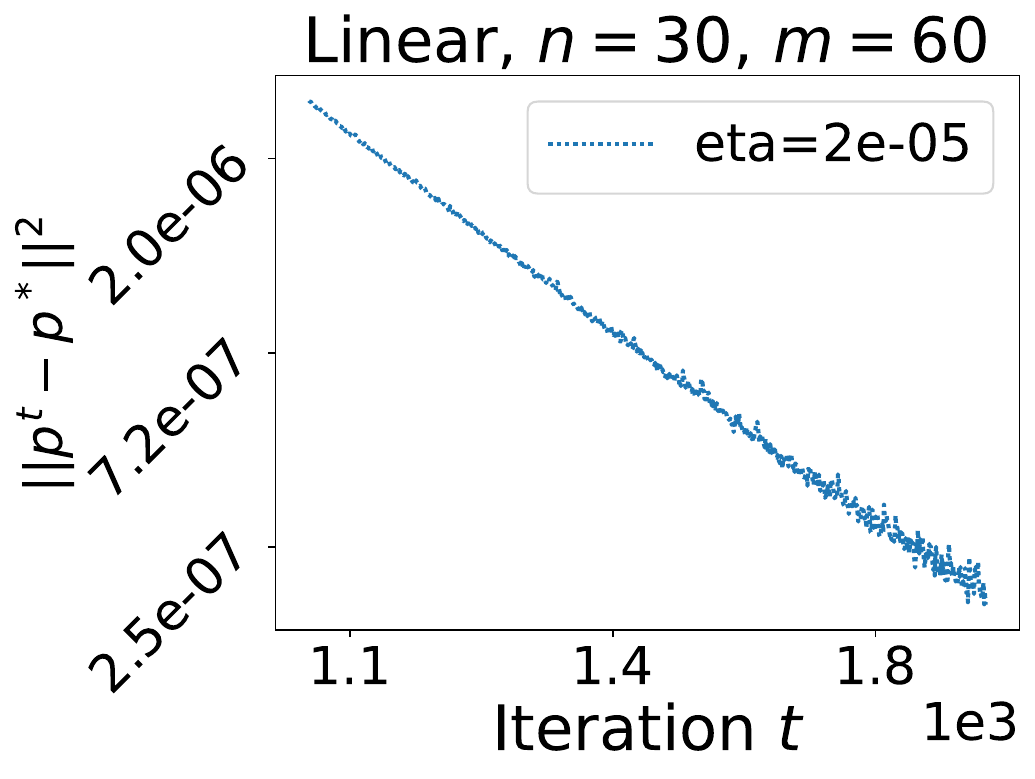}
    \includegraphics[scale=.23]{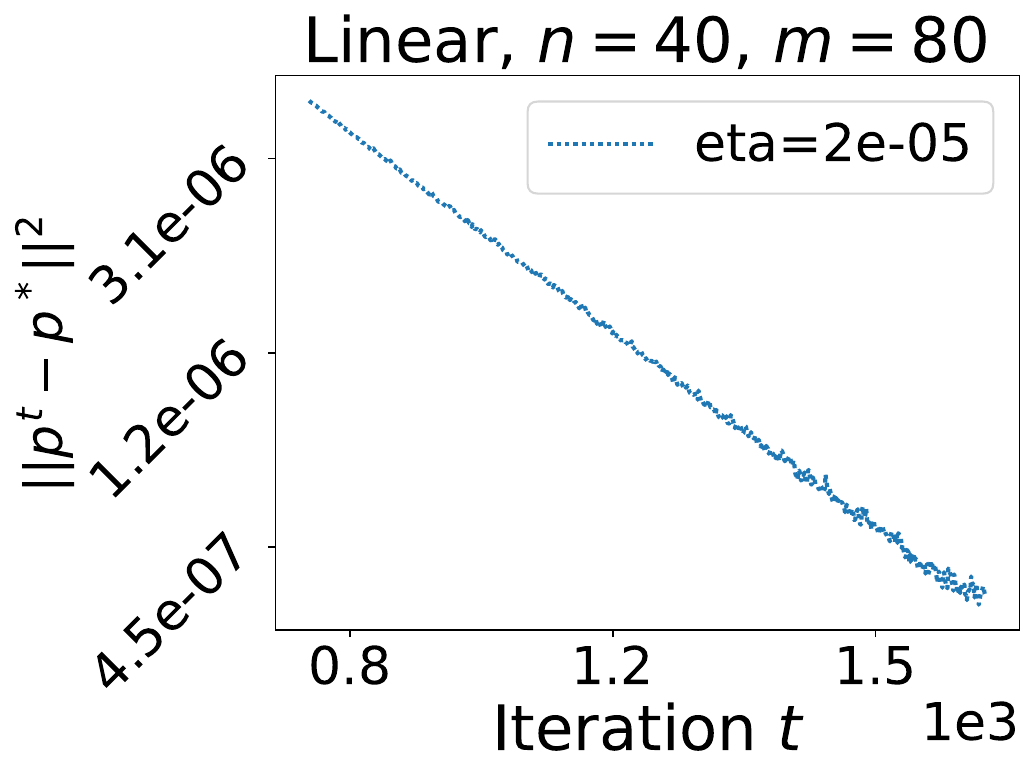}
    \caption{Convergence of squared error norms on random generated instances ($v$ is generated from the truncated normal distribution associated with $\mathcal{N}(0, 1)$, and truncated at ${10}^{-3}$ and $10$ standard deviations from $0$) of different sizes under linear utilities. 
    }
    \label{fig:single-instance-truncnorm-error-norms-linear}
\end{figure}

\begin{figure}[t]
    \centering
    \includegraphics[scale=.23]{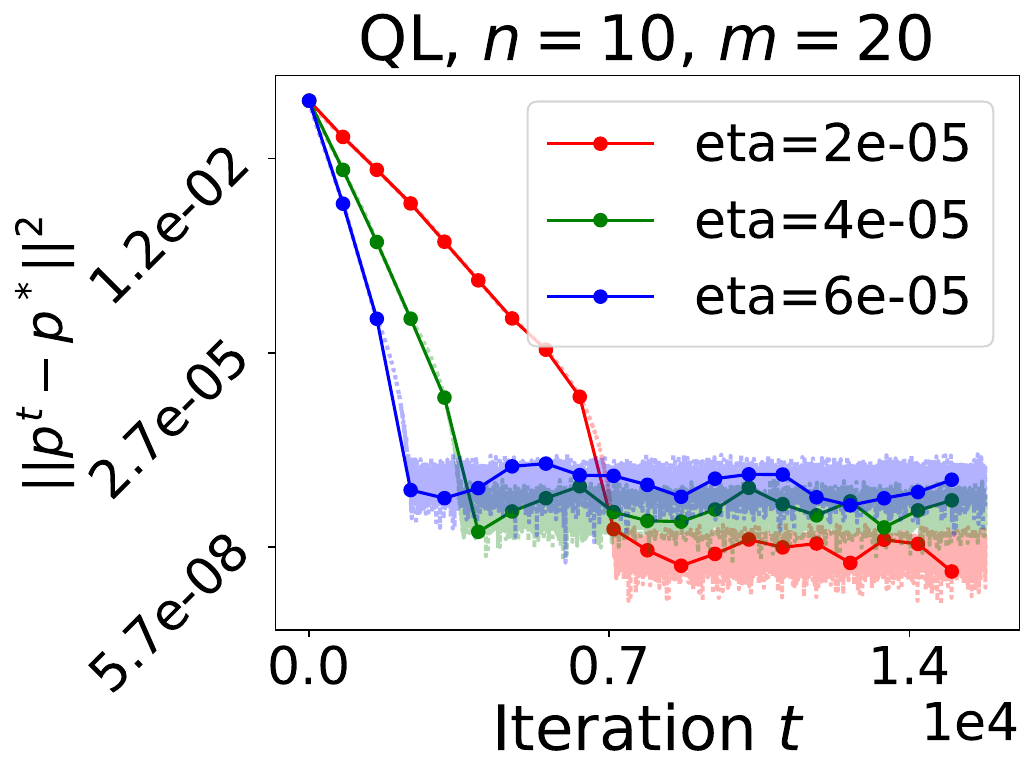}
    \includegraphics[scale=.23]{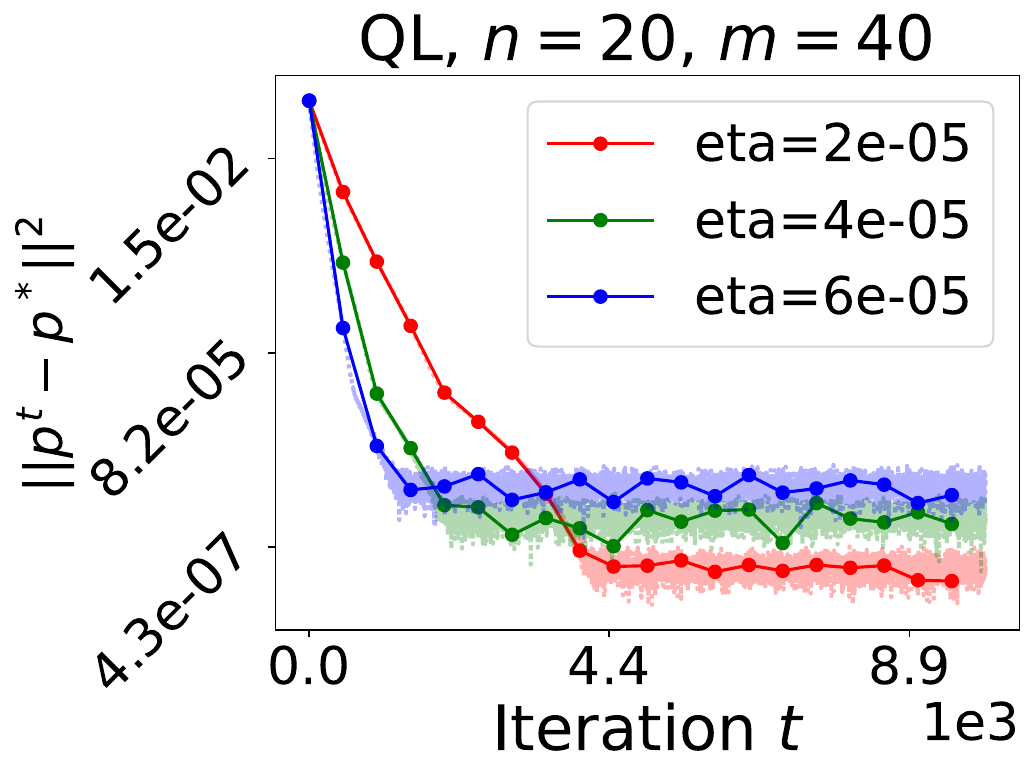}
    \includegraphics[scale=.23]{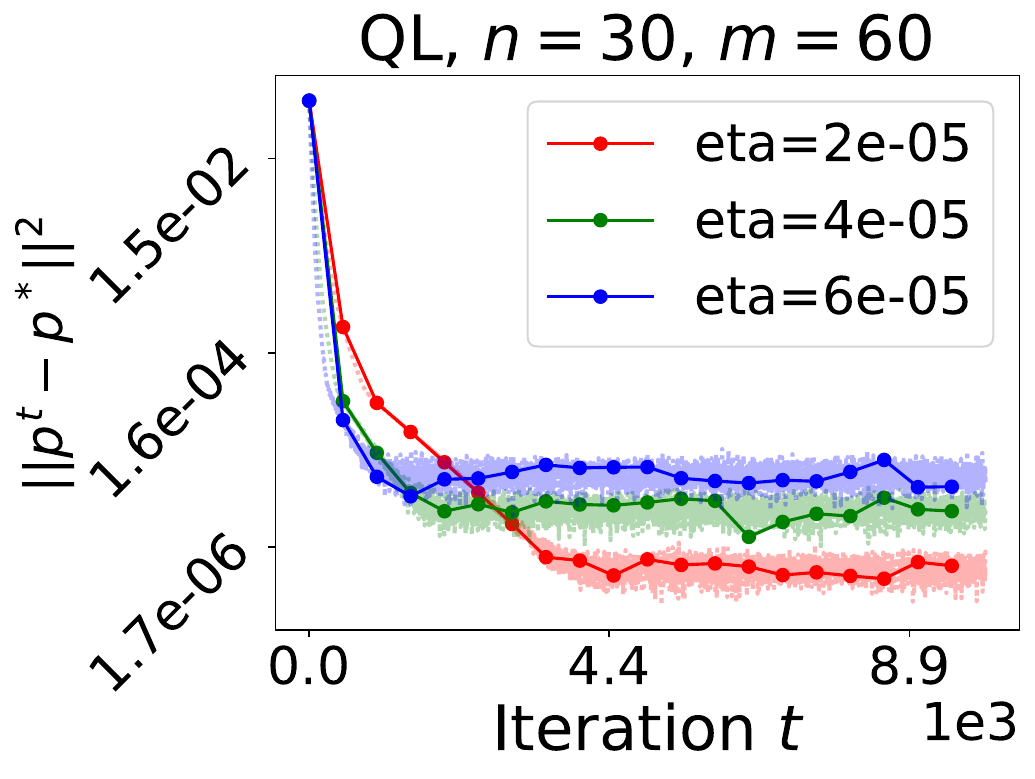}
    \includegraphics[scale=.23]{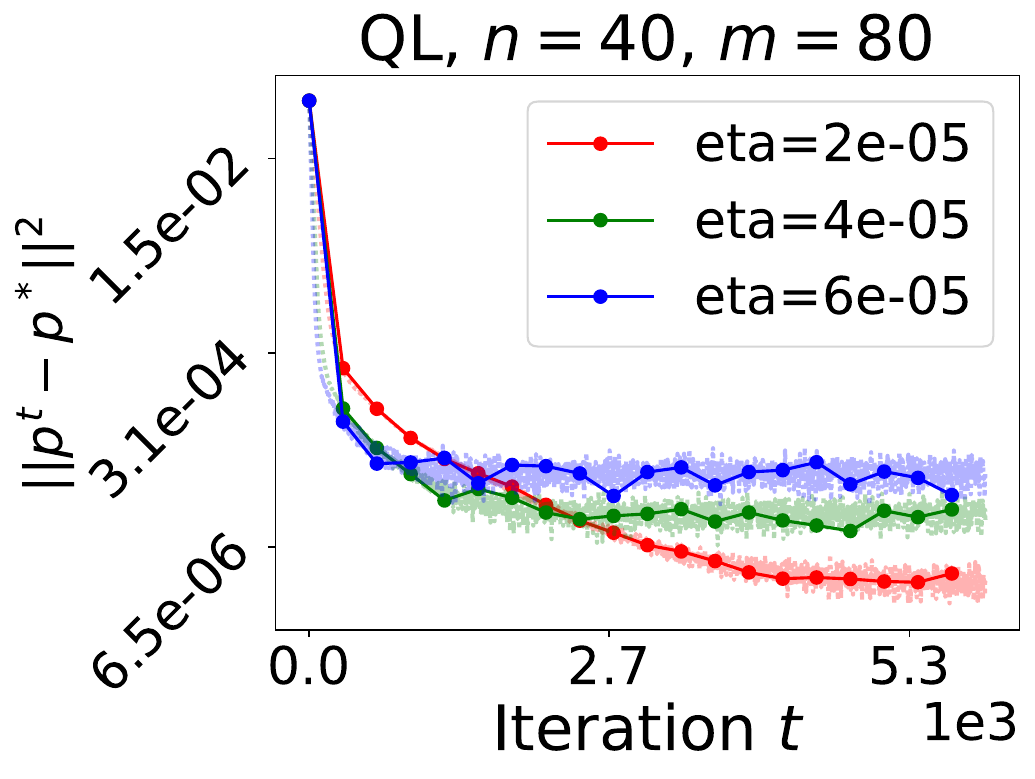}

    \includegraphics[scale=.23]{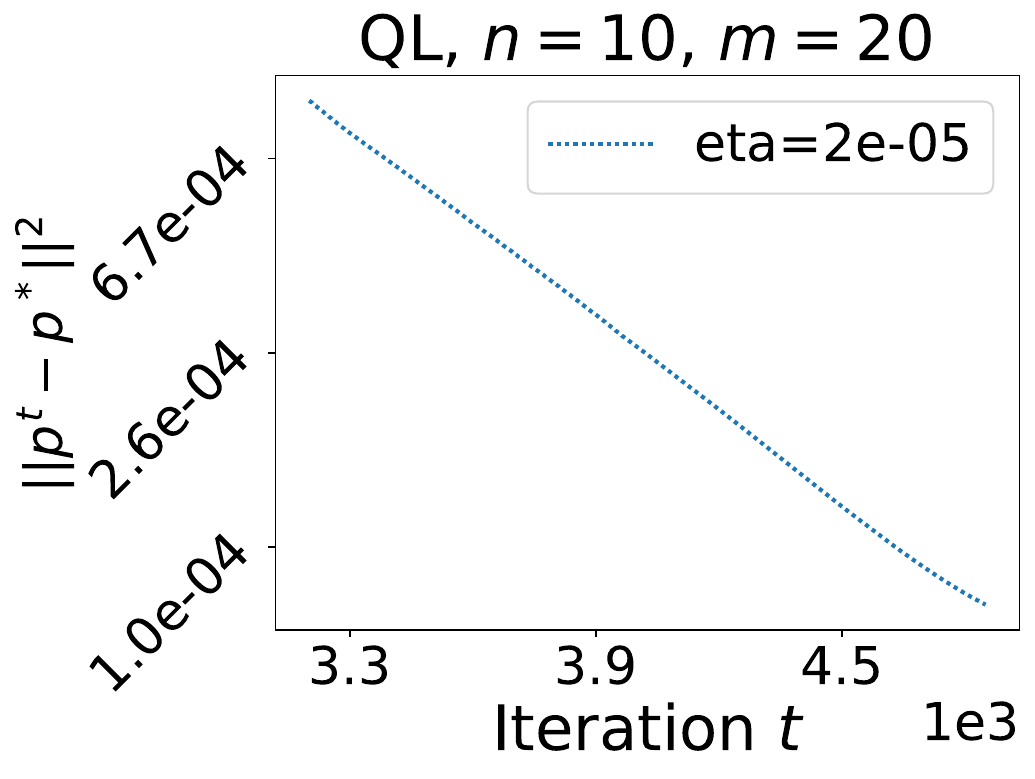}
    \includegraphics[scale=.23]{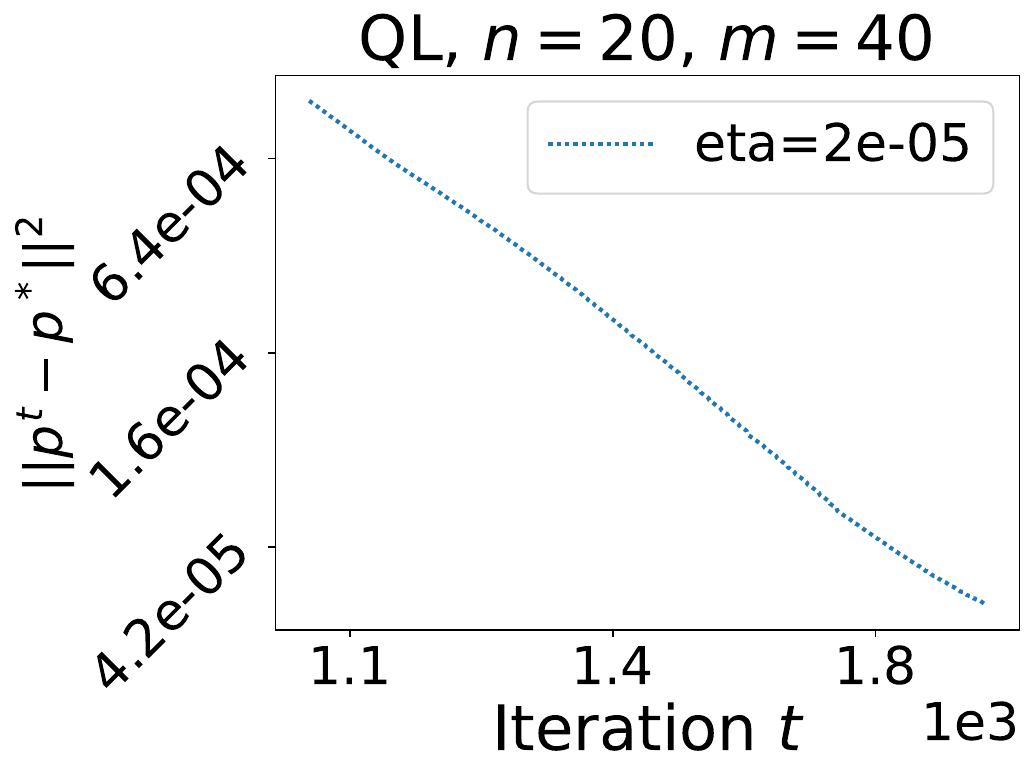}
    \includegraphics[scale=.23]{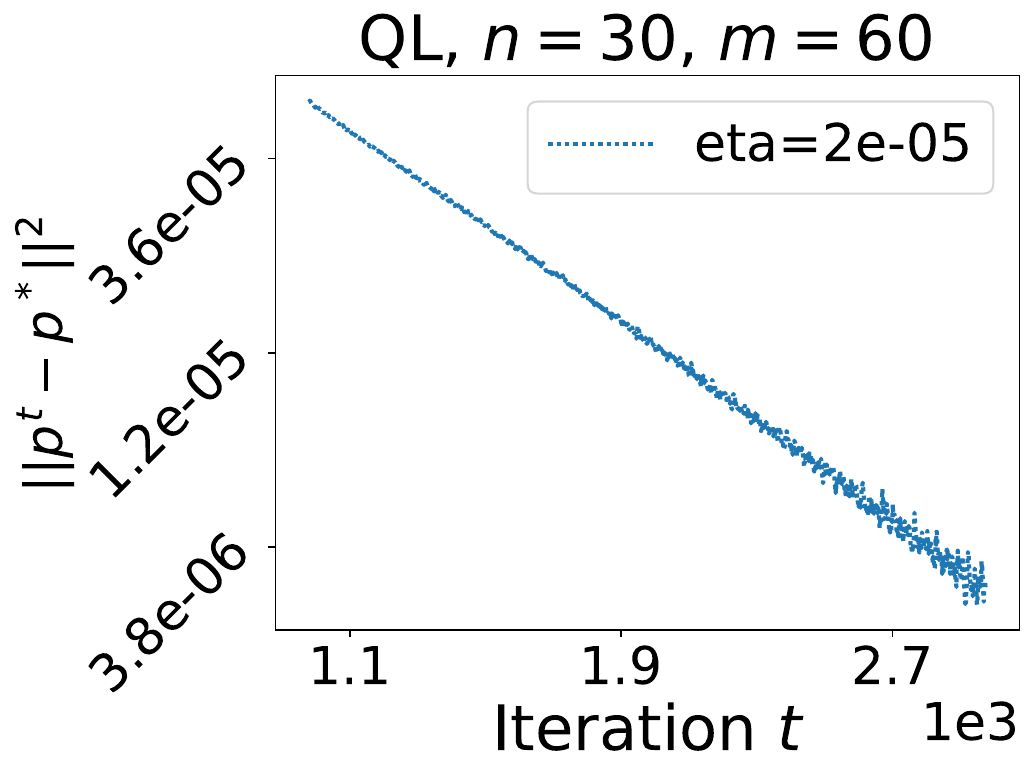}
    \includegraphics[scale=.23]{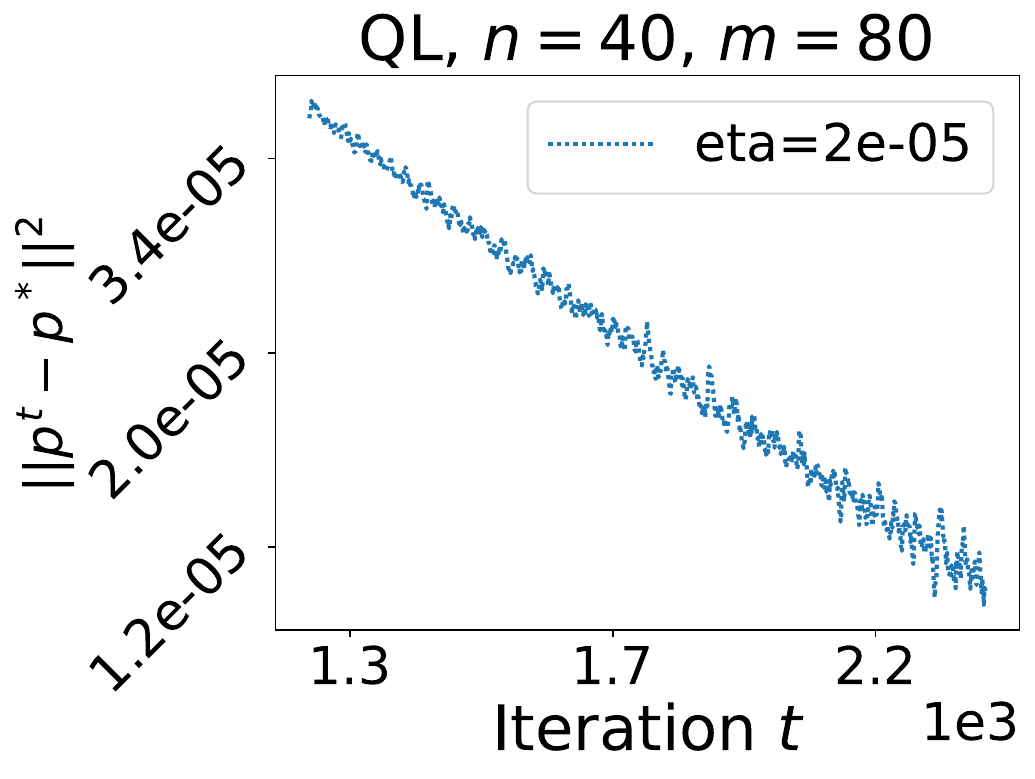}
    \caption{Convergence of squared error norms on random generated instances ($v$ is generated from the truncated normal distribution associated with $\mathcal{N}(0, 1)$, and truncated at ${10}^{-3}$ and $10$ standard deviations from $0$) of different sizes under quasi-linear utilities. 
    }
    \label{fig:single-instance-truncnorm-error-norms-ql}
\end{figure}

\begin{figure}[t]
    \centering
    \includegraphics[scale=.23]{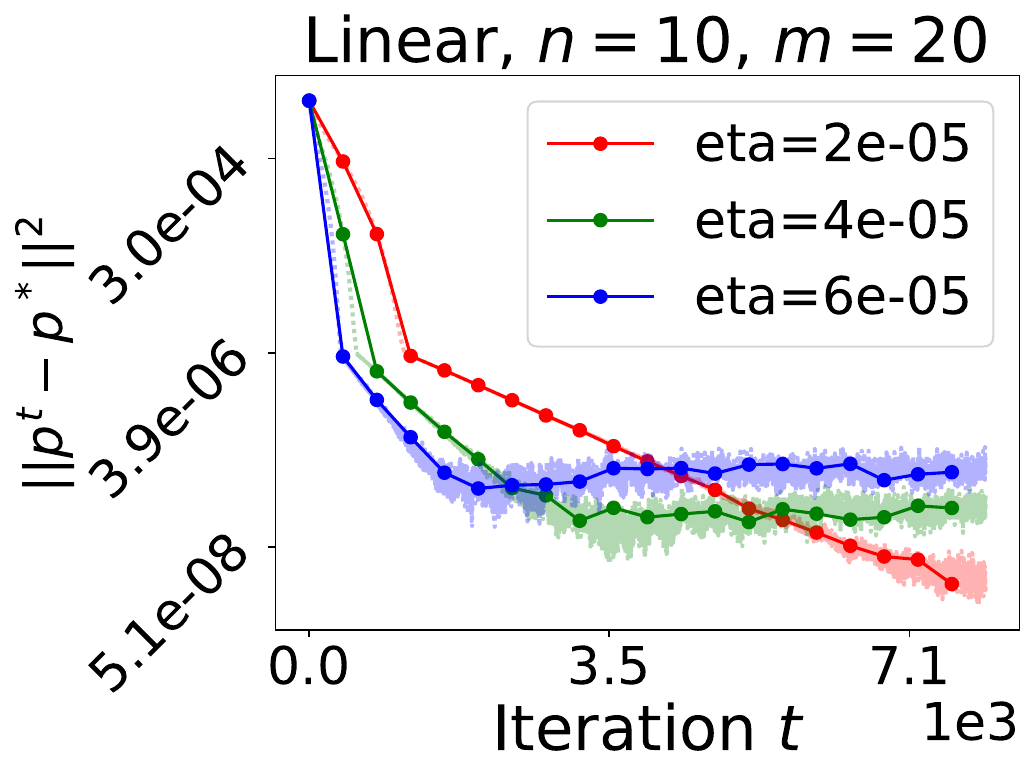}
    \includegraphics[scale=.23]{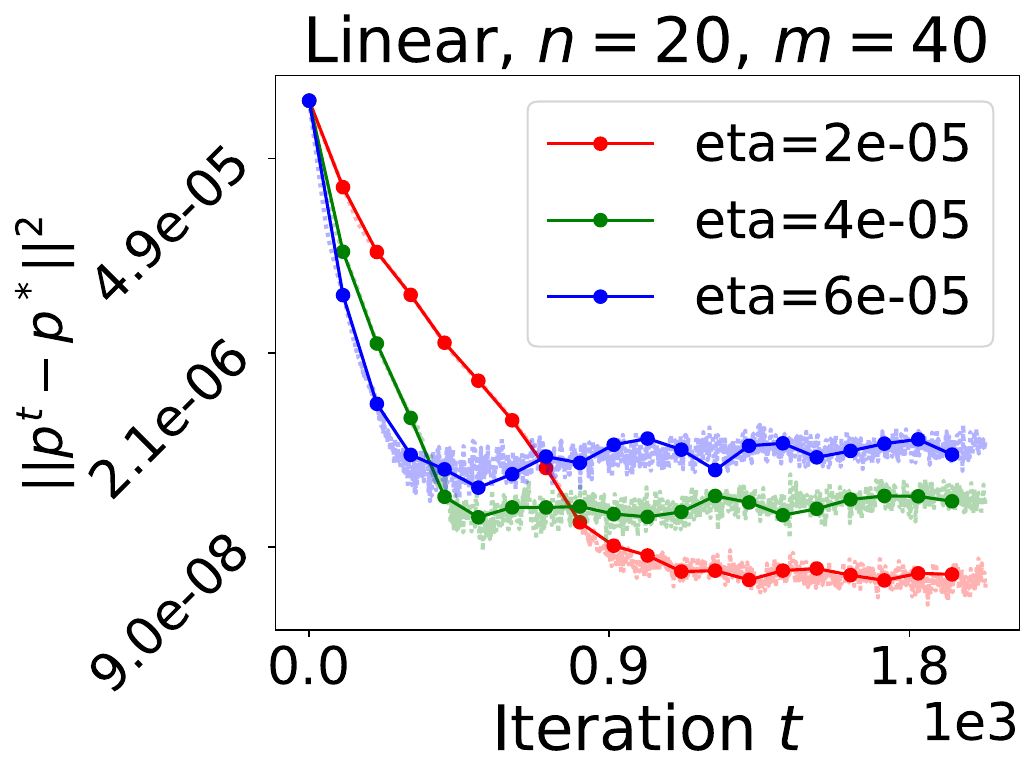}
    \includegraphics[scale=.23]{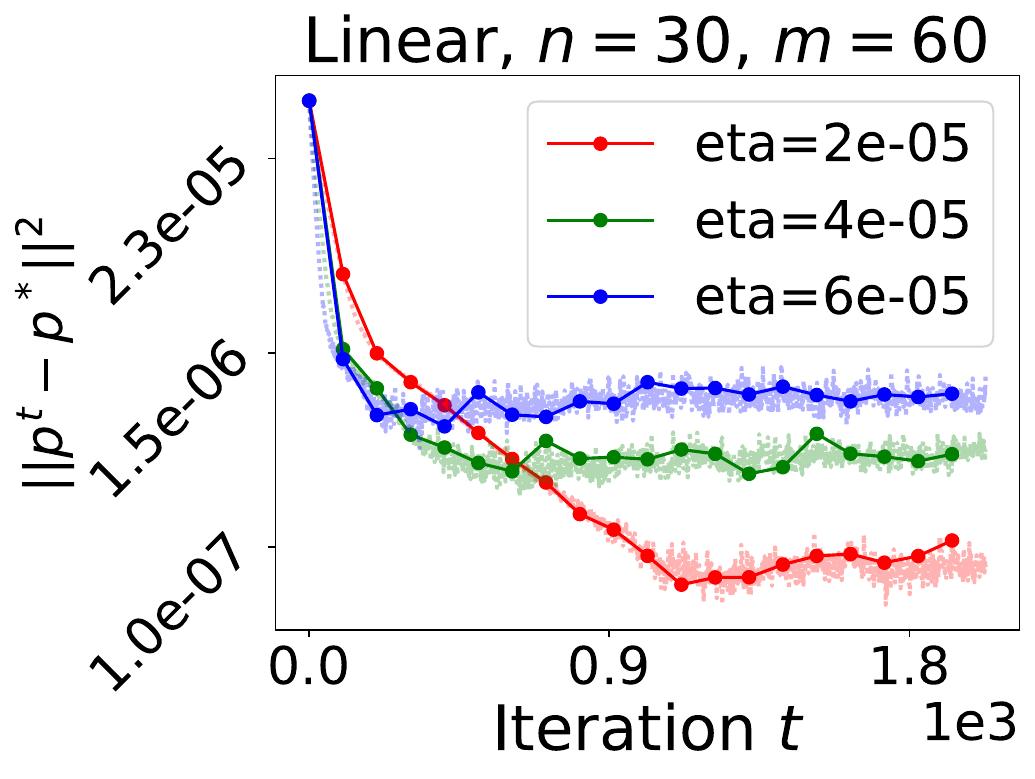}
    \includegraphics[scale=.23]{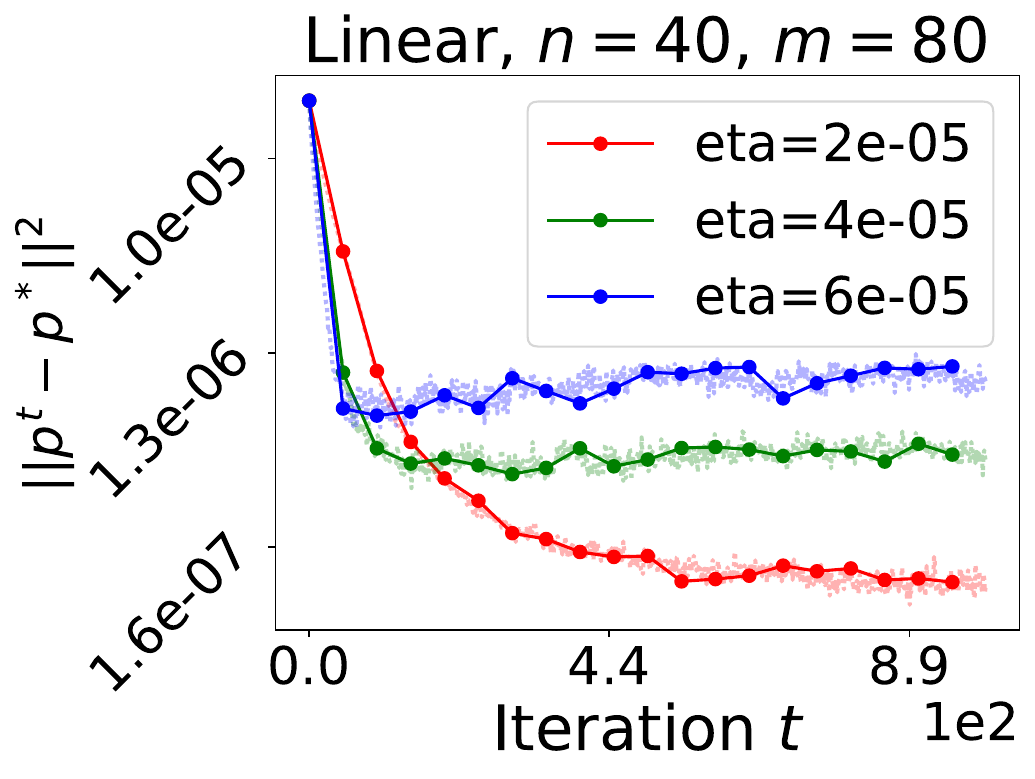}

    \includegraphics[scale=.23]{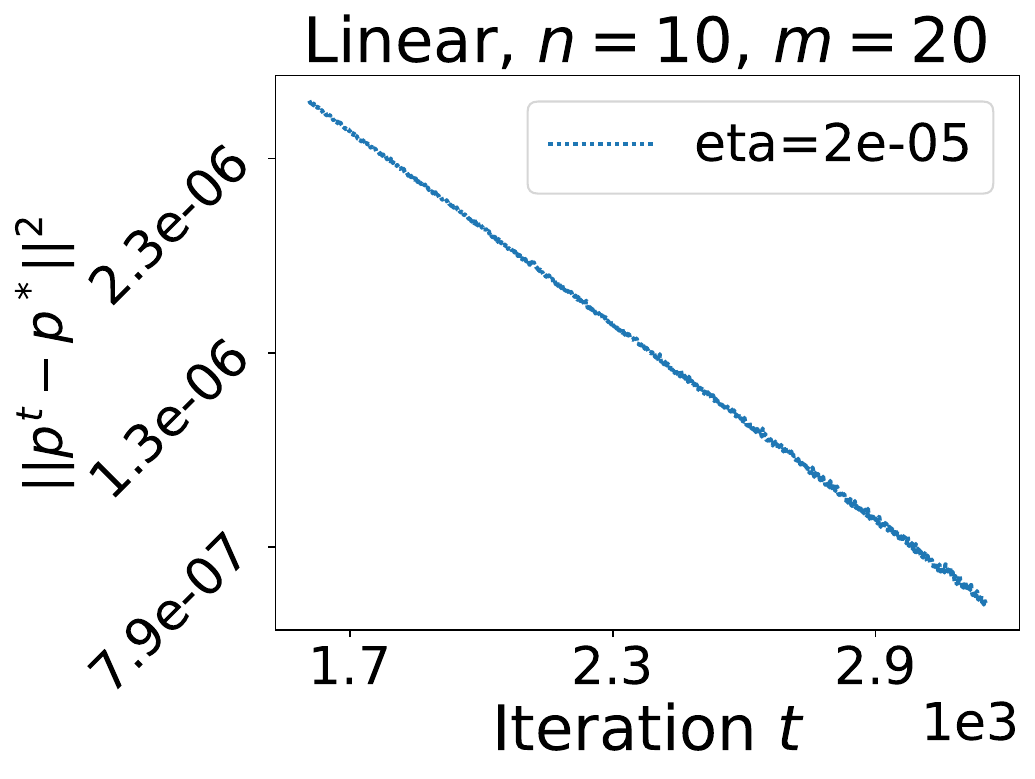}
    \includegraphics[scale=.23]{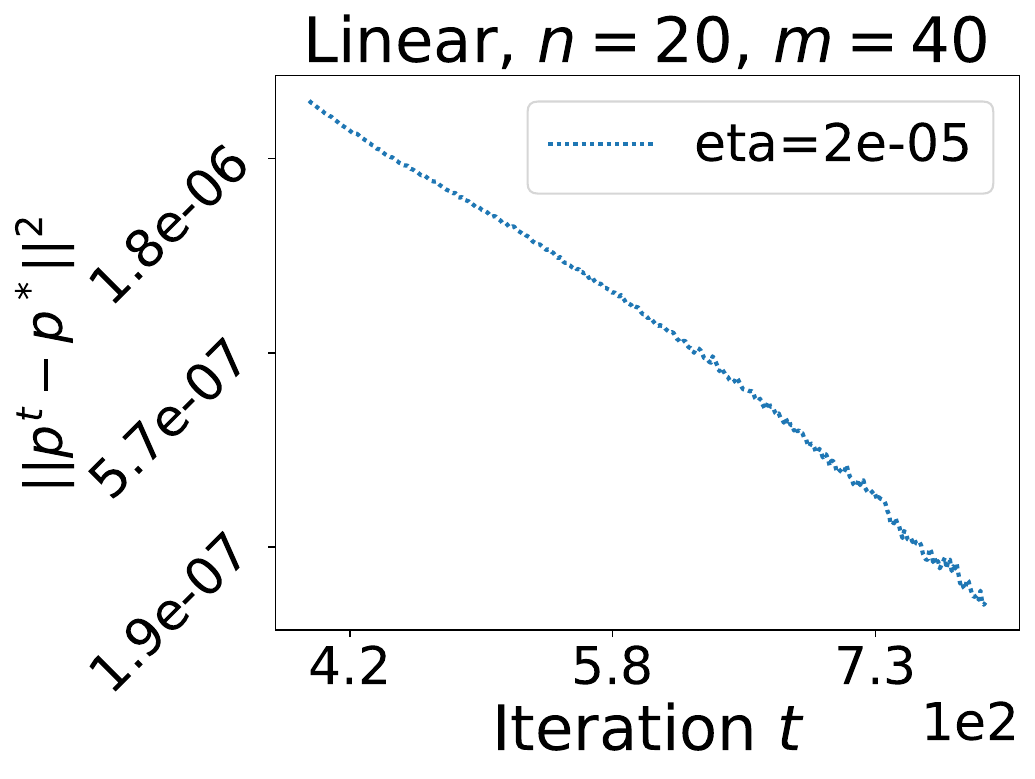}
    \includegraphics[scale=.23]{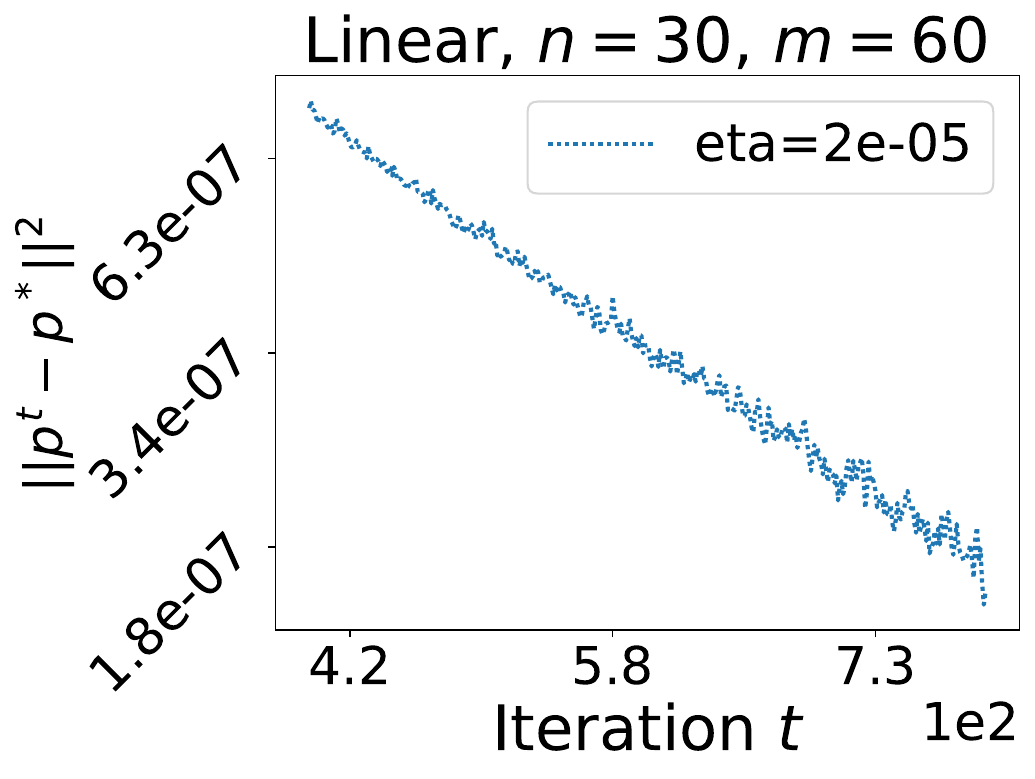}
    \includegraphics[scale=.23]{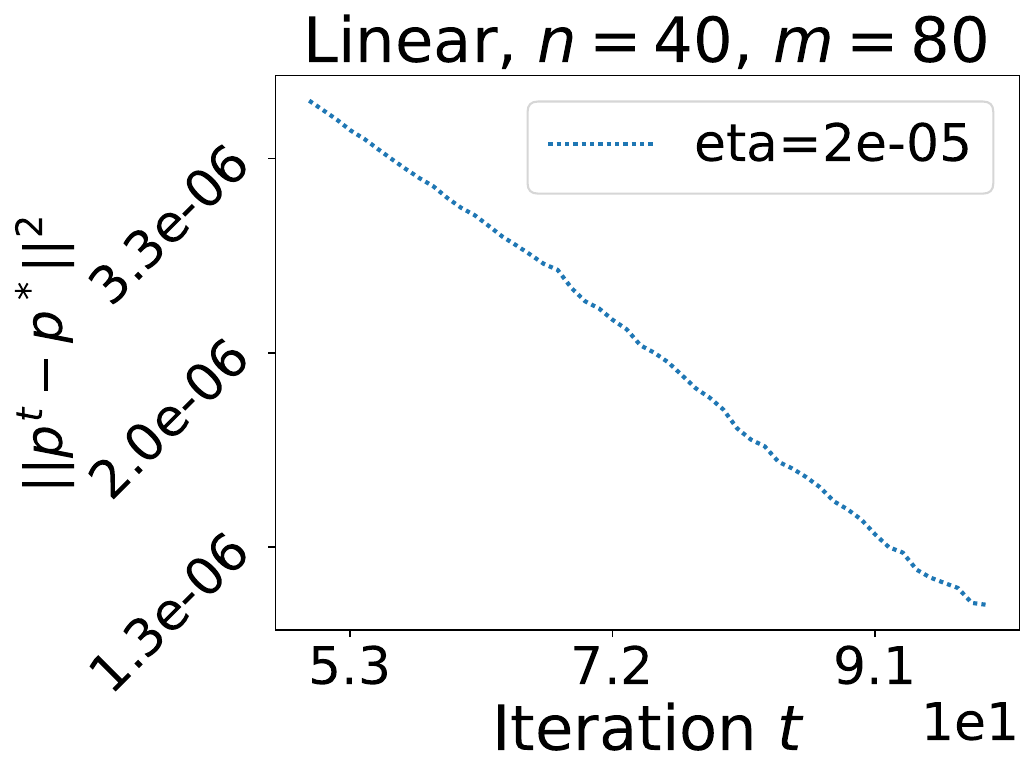}
    \caption{Convergence of squared error norms on random generated instances ($v$ is generated from the uniform random integers on $\{1,\ldots,100\}$) of different sizes under linear utilities. 
    }
    \label{fig:single-instance-randint-error-norms-linear}
\end{figure}

\begin{figure}[t]
    \centering
    \includegraphics[scale=.23]{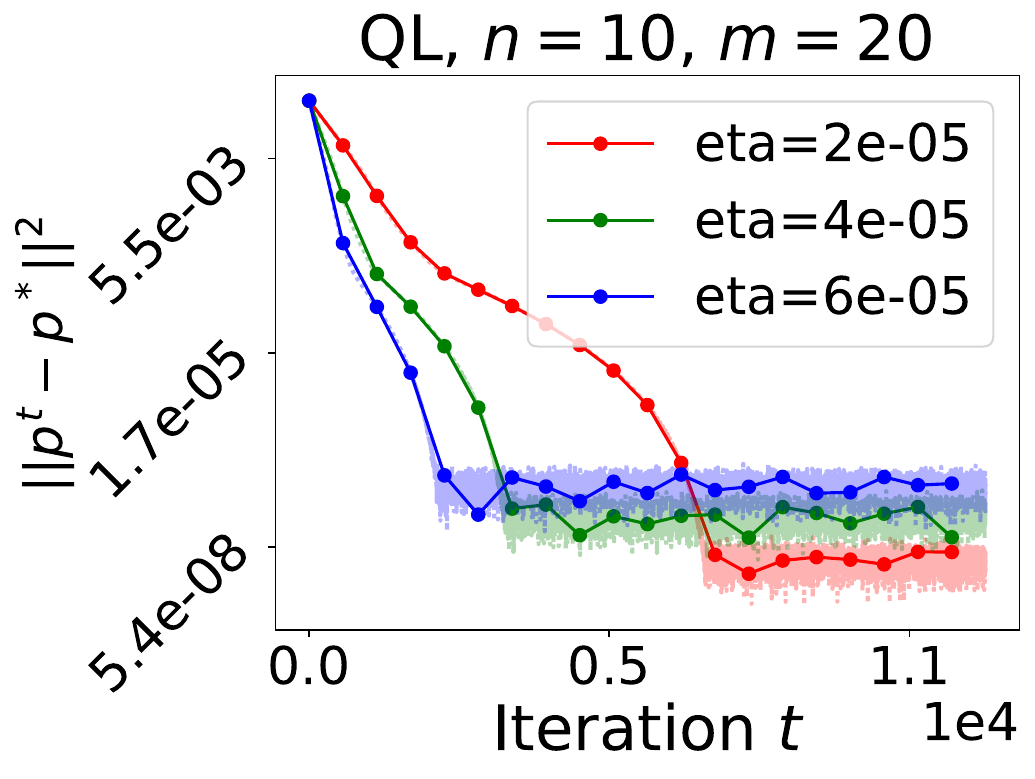}
    \includegraphics[scale=.23]{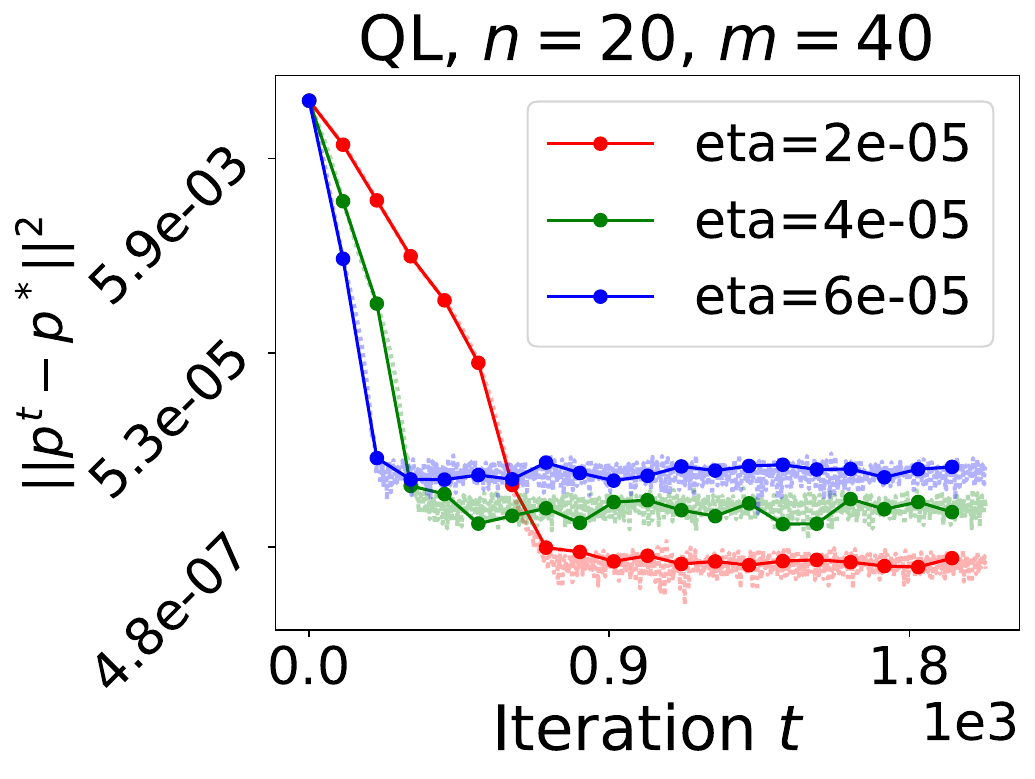}
    \includegraphics[scale=.23]{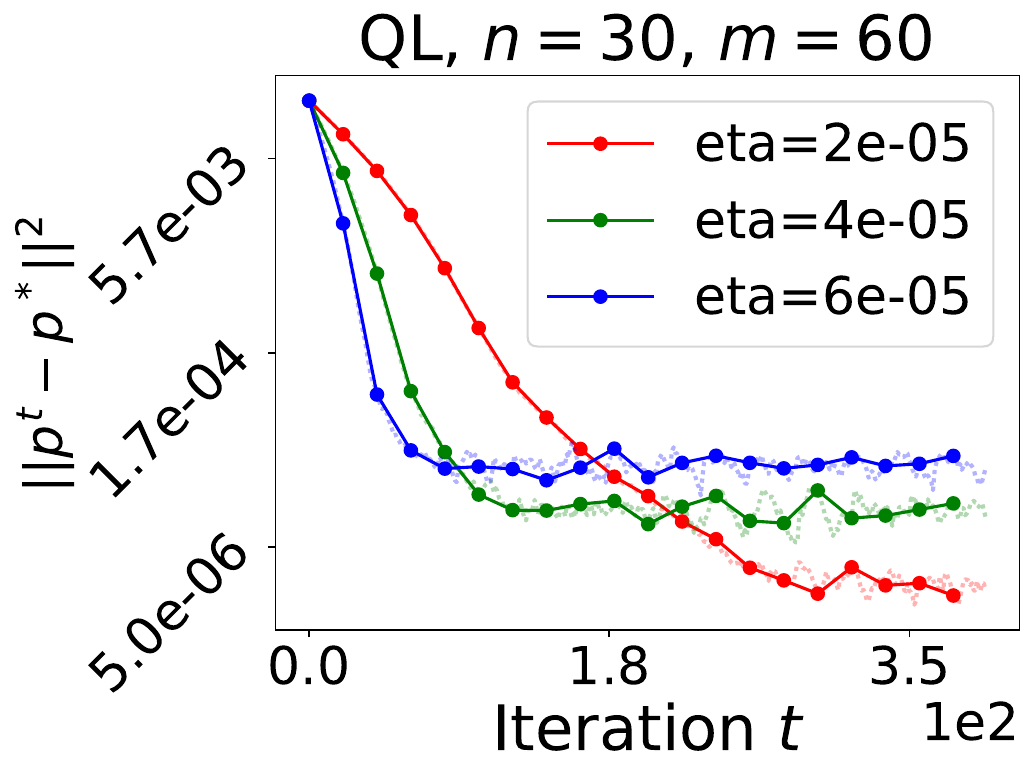}
    \includegraphics[scale=.23]{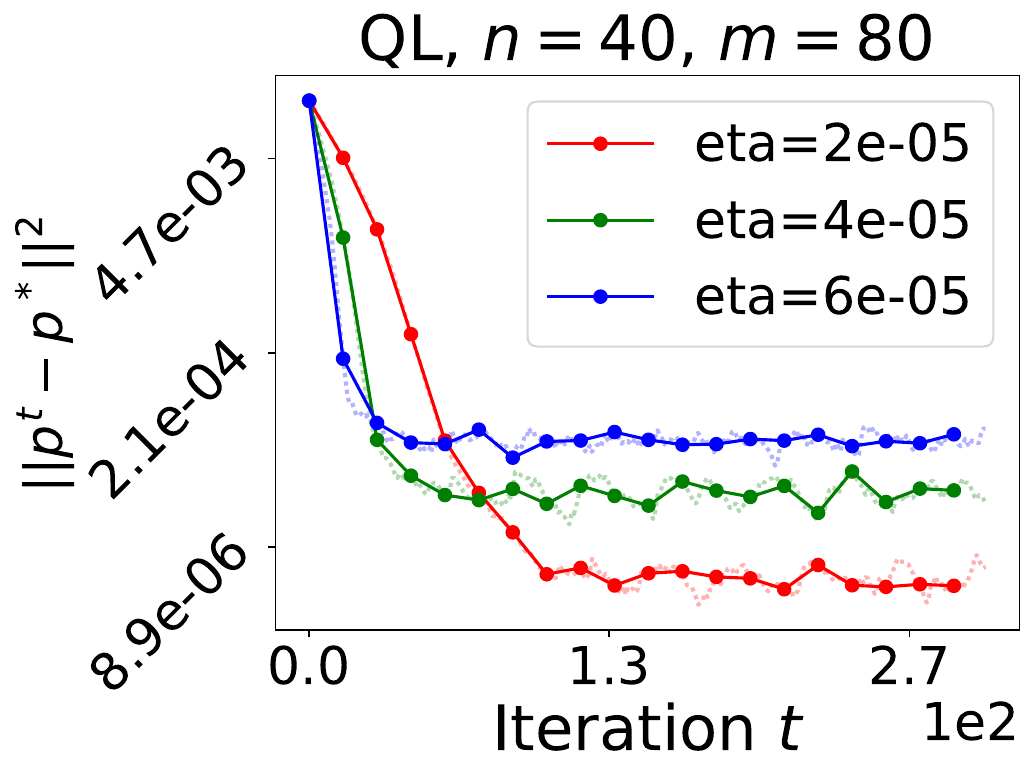}

    \includegraphics[scale=.23]{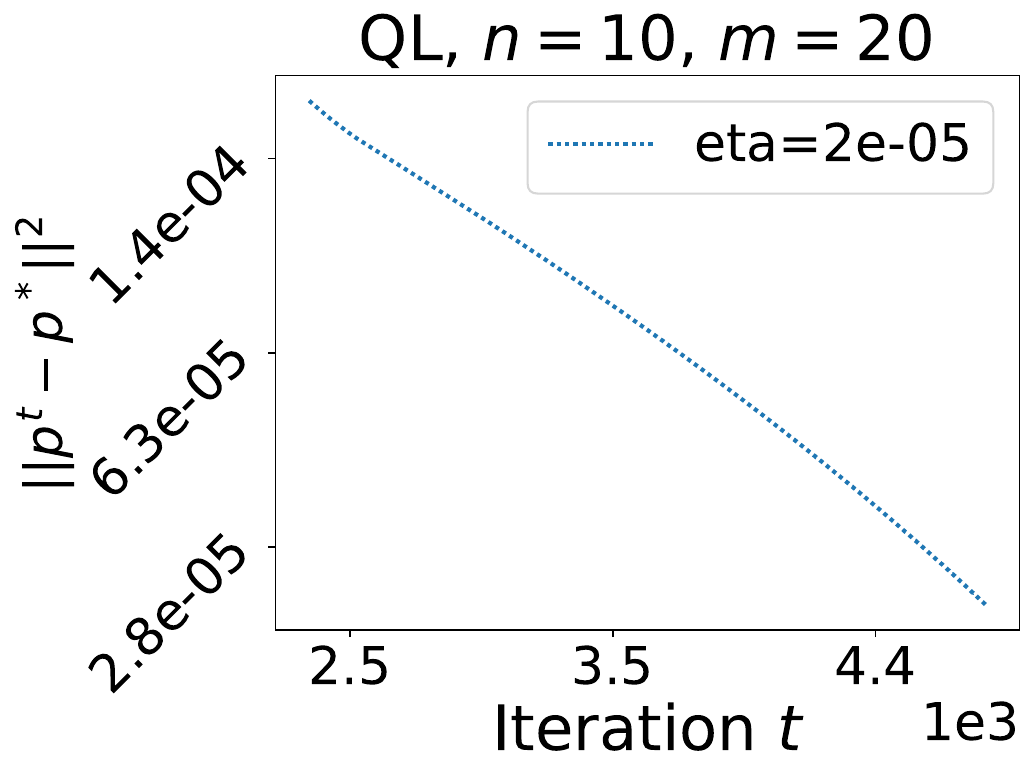}
    \includegraphics[scale=.23]{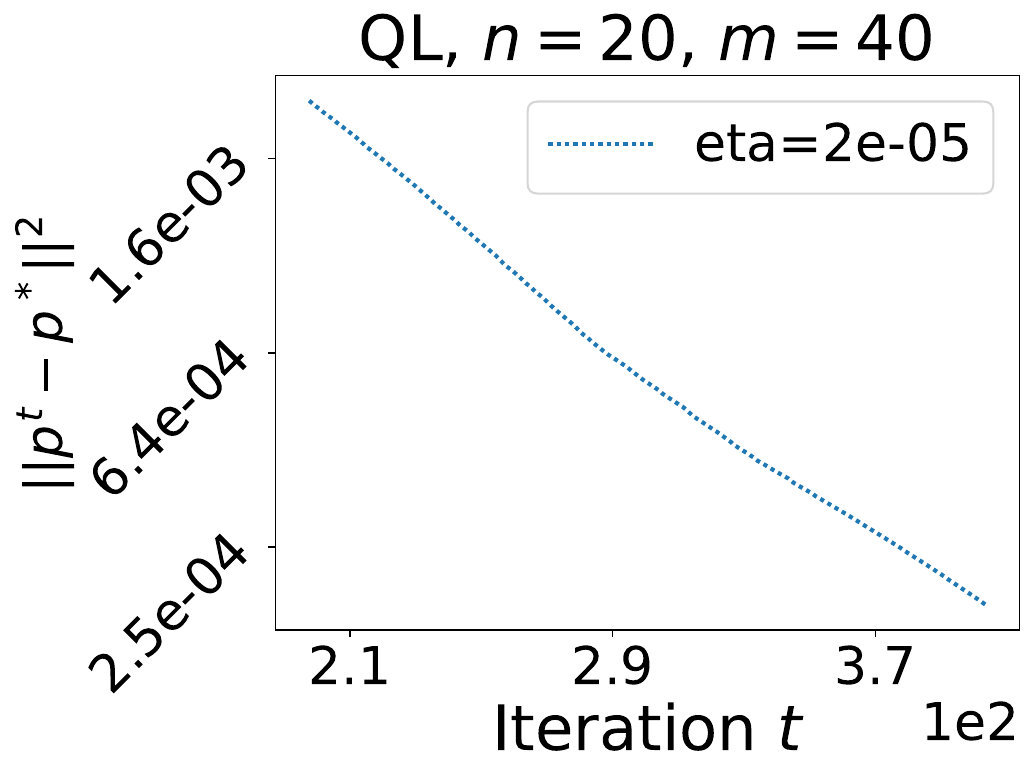}
    \includegraphics[scale=.23]{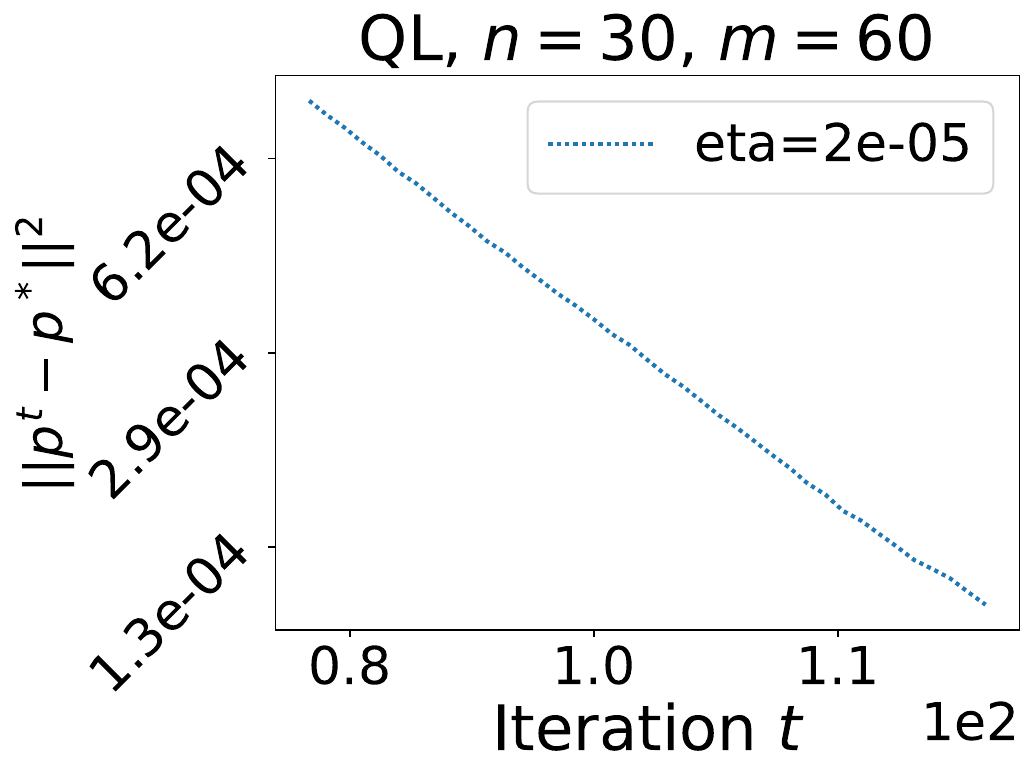}
    \includegraphics[scale=.23]{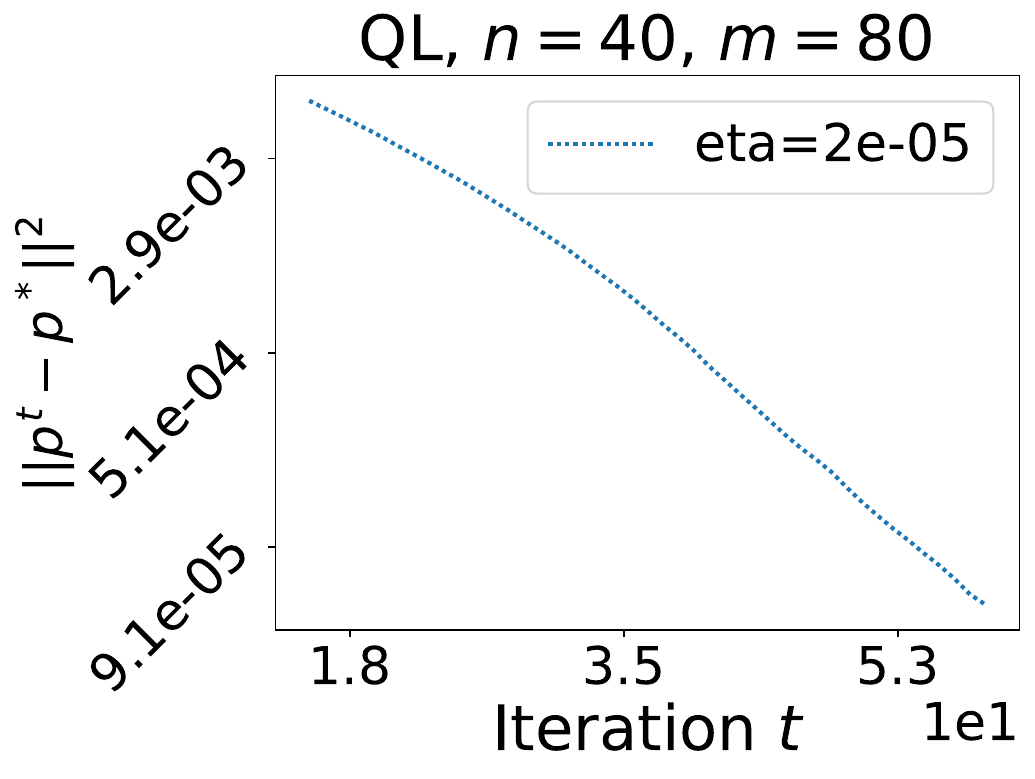}
    \caption{Convergence of squared error norms on random generated instances ($v$ is generated from the uniform random integers on $\{1,\ldots,100\}$) of different sizes under quasi-linear utilities. 
    }
    \label{fig:single-instance-randint-error-norms-ql}
\end{figure}

\section{Comparison of different types of Tâtonnement.}
\label{app:sec:compare-tatonnement}

We compare the convergence of the coordinate {\tatonnement} and the standard {\tatonnement} on the real data instance.
See~\cref{fig:compare-ttm-uniform,fig:compare-ttm-lognormal,fig:compare-ttm-exponential,fig:compare-ttm-truncnorm,fig:compare-ttm-randint} for the results on synthetic indtances 
and see~\cref{fig:compare-ttm-real} for the results on the real data instance.

\paragraph{Experimental details.} We mainly compare the convergence of 
\begin{itemize}
    \item Additive {\tatonnement}: $p^{t+1}_j = p^t_j + \eta z(p^t)$; 
    \item Multiplicative {\tatonnement}: $p^{t+1}_j = p^t_j (1 + \eta z(p^t))$; 
    \item Entropic {\tatonnement}: $p^{t+1}_j = p^t_j \exp(\eta z(p^t))$. 
\end{itemize}
Here, $\eta$ is the stepsize, and $z(p^t)$ is the excess demand vector at time $t$. 
To be fair, we use a constant stepsize $\eta$ for additive {\tatonnement} and $\eta' = \frac{m\eta}{\norm{B}_1}$ for the multiplicative and entropic {\tatonnement} to ensure the stepsize is of the similar scale as the additive {\tatonnement}. 
In particular, we use $\eta = 2 \times 10^{-5}$ for random generated instances and 
$\eta = 2 \times 10^{-6}$ for the real data instance for the linear Fisher market; 
we use $\eta = 2 \times 10^{-5}$ for random generated instances and 
$\eta = 2 \times 10^{-7}$ for the real data instance for the quasi-linear Fisher market.

\paragraph{Discussion.} As shown in the figures, the three types of {\tatonnement} have similar convergence behaviors: 
they all converge to the equilibrium price in a linear rate, and then oscillate around the equilibrium price. 
The multiplicative and entropic {\tatonnement} are almost the same, which is as expected since multiplicative {\tatonnement} is an approximation of entropic {\tatonnement} when the stepsize is small and the excess demand is bounded. 
In the linear Fisher market, 
the additive {\tatonnement} has a slightly slower convergence rate than the other two in some instances; 
in the quasi-linear Fisher market, 
the additive {\tatonnement} has a faster convergence rate than the other two in most instances.
This is an interesting observation, and we leave the more structural comparison of different types of {\tatonnement} for future work. 
Note that the convergence speed of {\tatonnement}s might change as we tune the stepsize. 

\paragraph{More variants of Tâtonnement.} As shown in~\cref{tab:comparison}, 
there are some modified versions of the above three types of {\tatonnement} in the literature. 
There are two common modifications: 
\begin{itemize}
    \item[(1)] Upper bounded the excess demand vector: $z(p^t) = \min\{z(p^t), 1\}$; 
    \item[(2)] Lower bounded the price vector with a small value: $p^{t+1}_j = \max\{p^t_j + \eta z(p^t), \epsilon\}$, where $\epsilon$ is a small value.
\end{itemize}
Via our experiments, we find that the convergence of {\tatonnement} with $(2)$ is similar to the standard {\tatonnement}. 
However, {\tatonnement} with $(1)$ usually lead to a larger final error norm.

\begin{figure}[t]
    \centering
    \includegraphics[scale=.23]{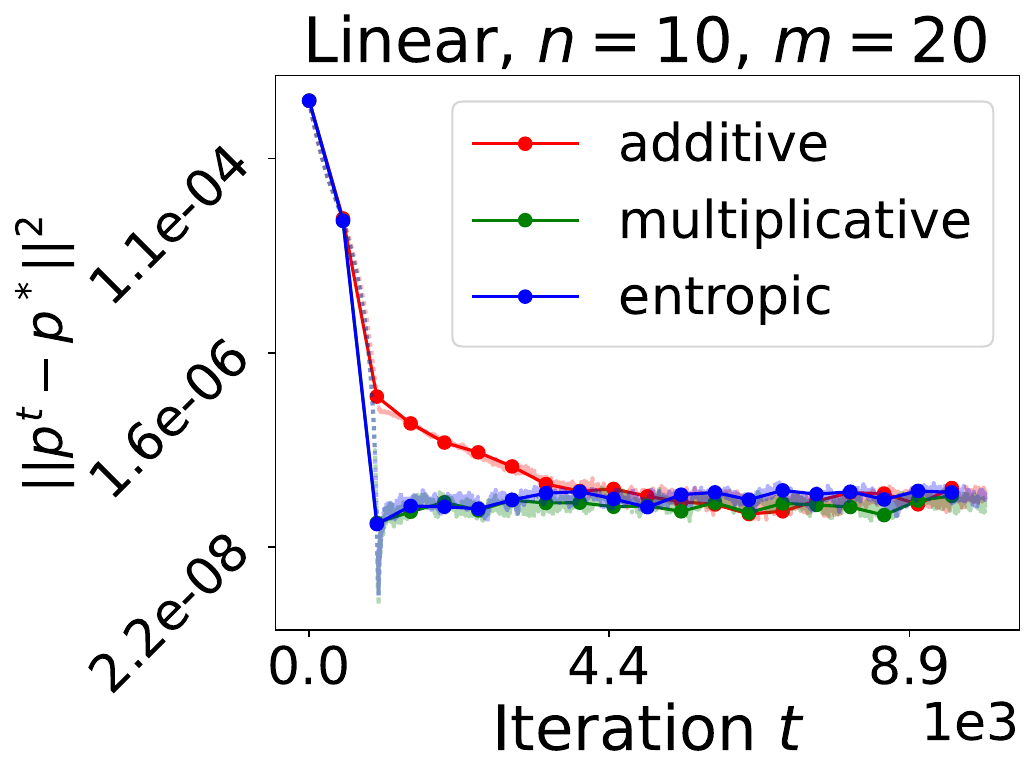}
    \includegraphics[scale=.23]{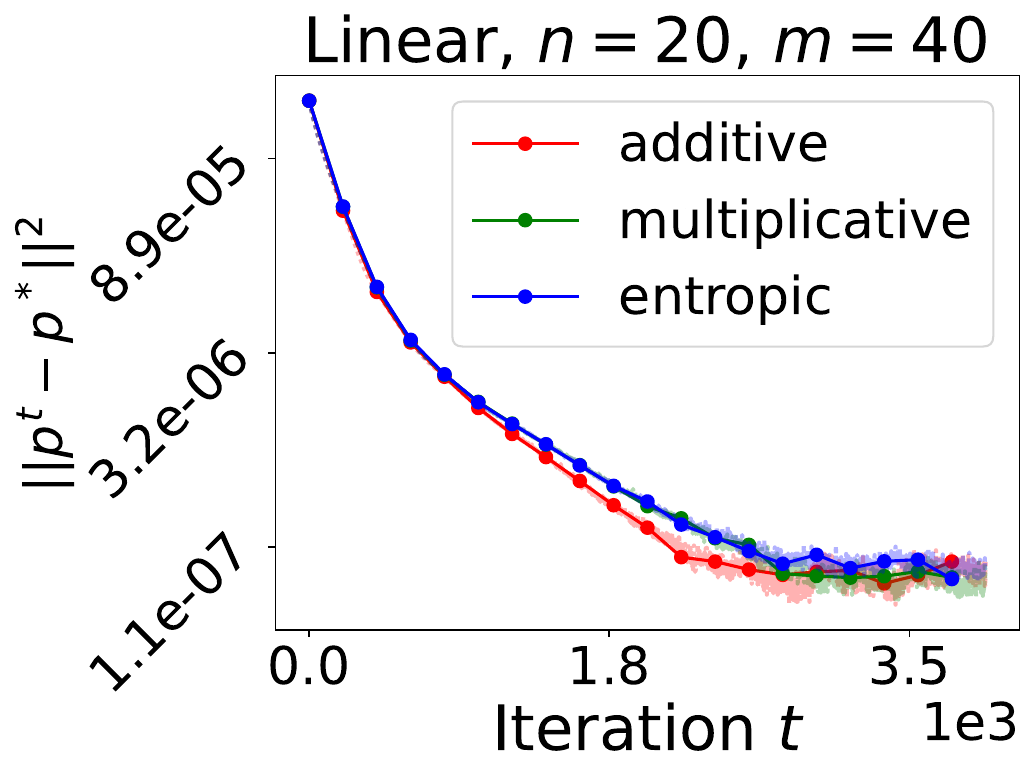}
    \includegraphics[scale=.23]{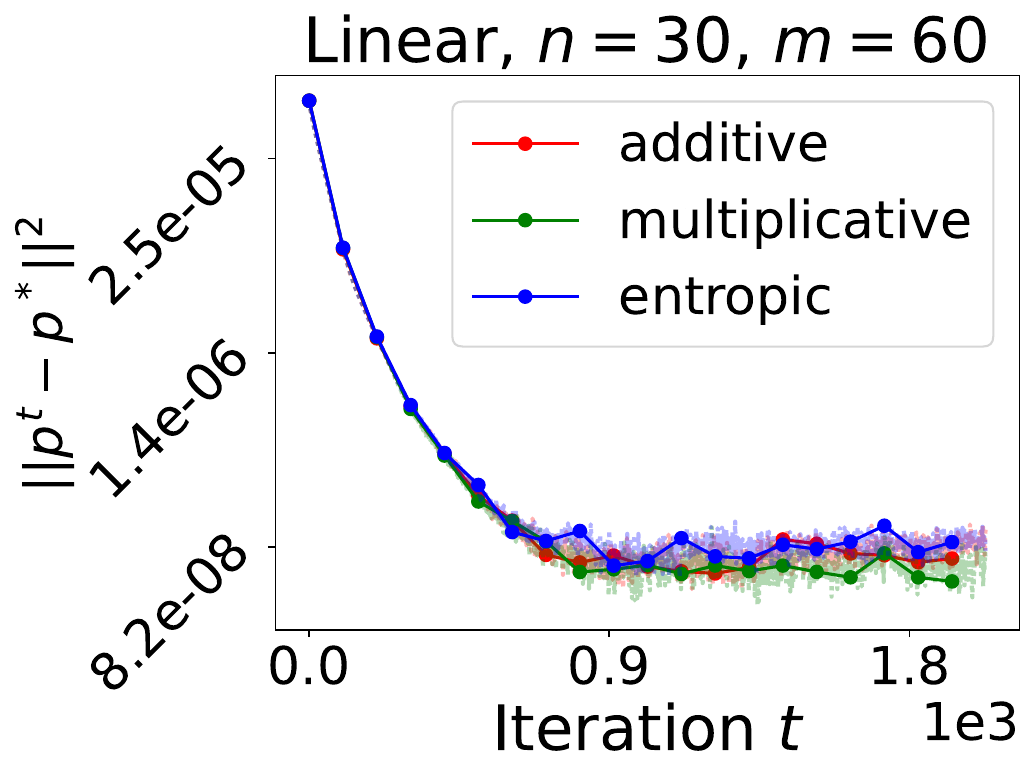}
    \includegraphics[scale=.23]{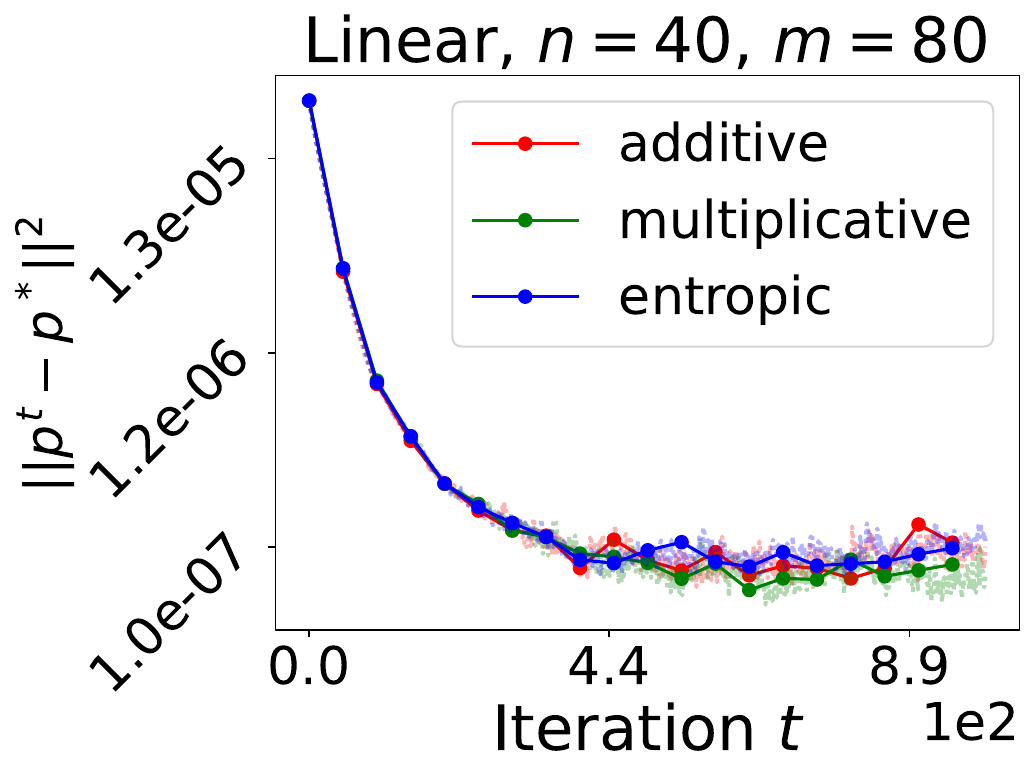}

    \includegraphics[scale=.23]{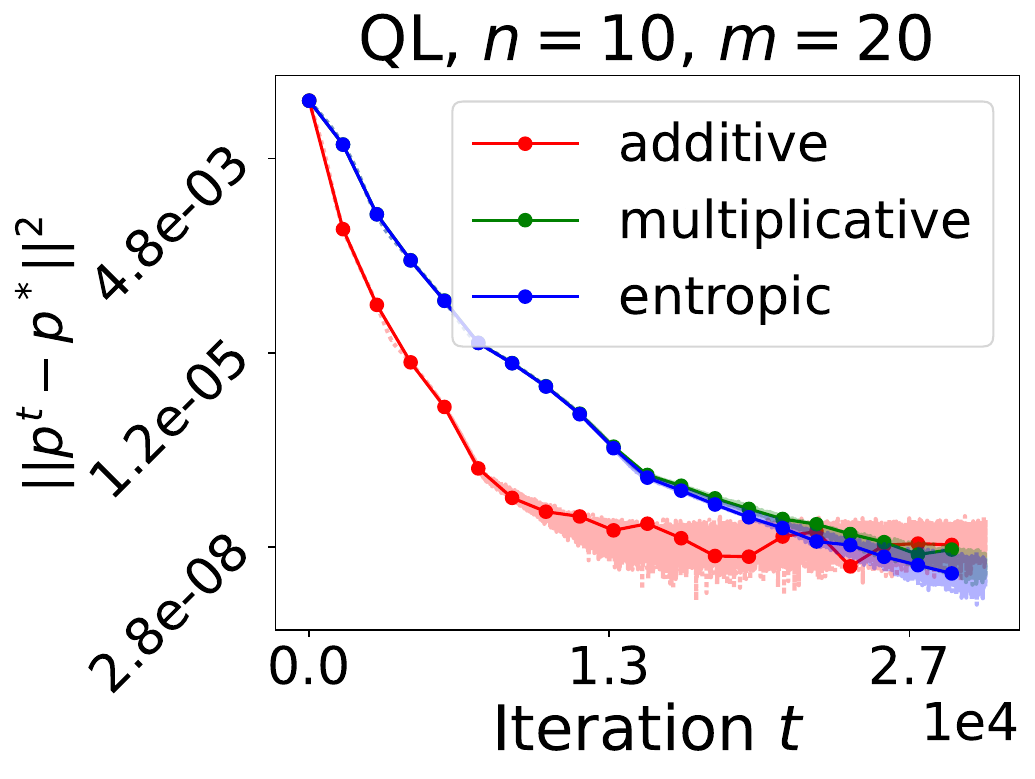}
    \includegraphics[scale=.23]{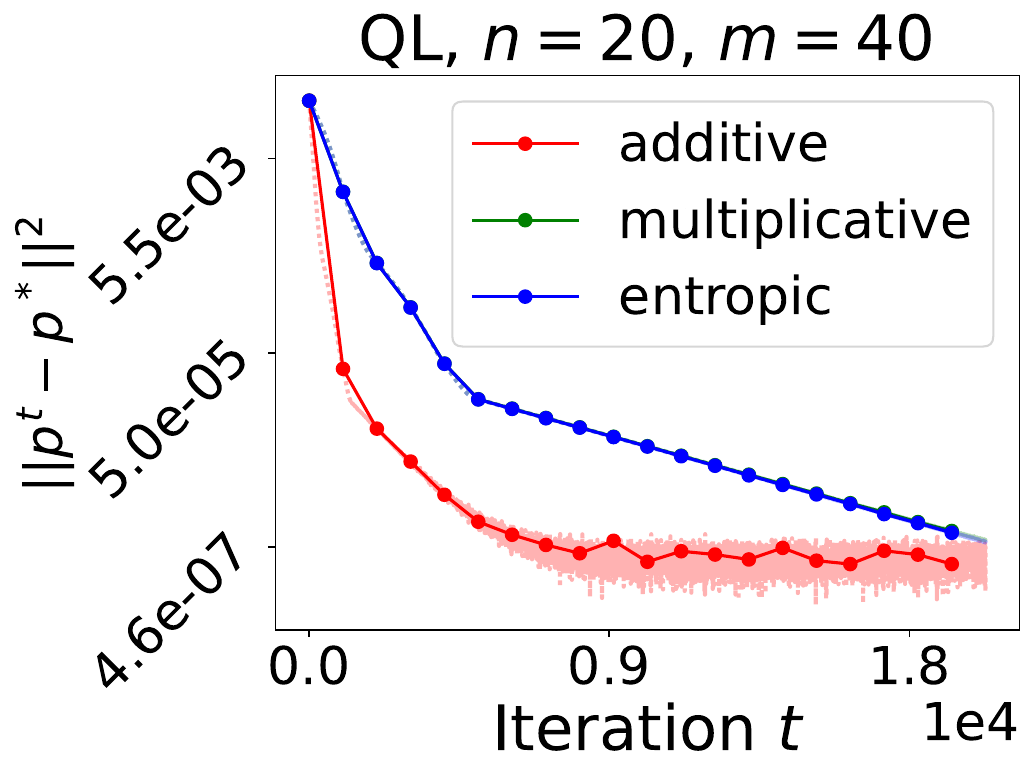}
    \includegraphics[scale=.23]{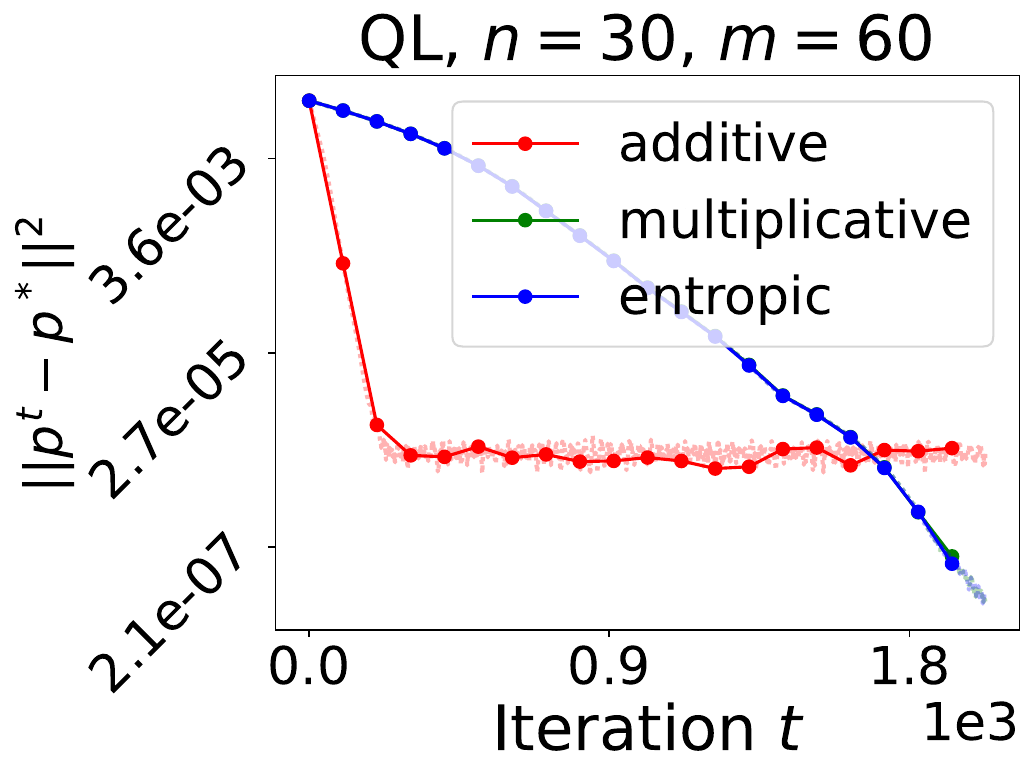}
    \includegraphics[scale=.23]{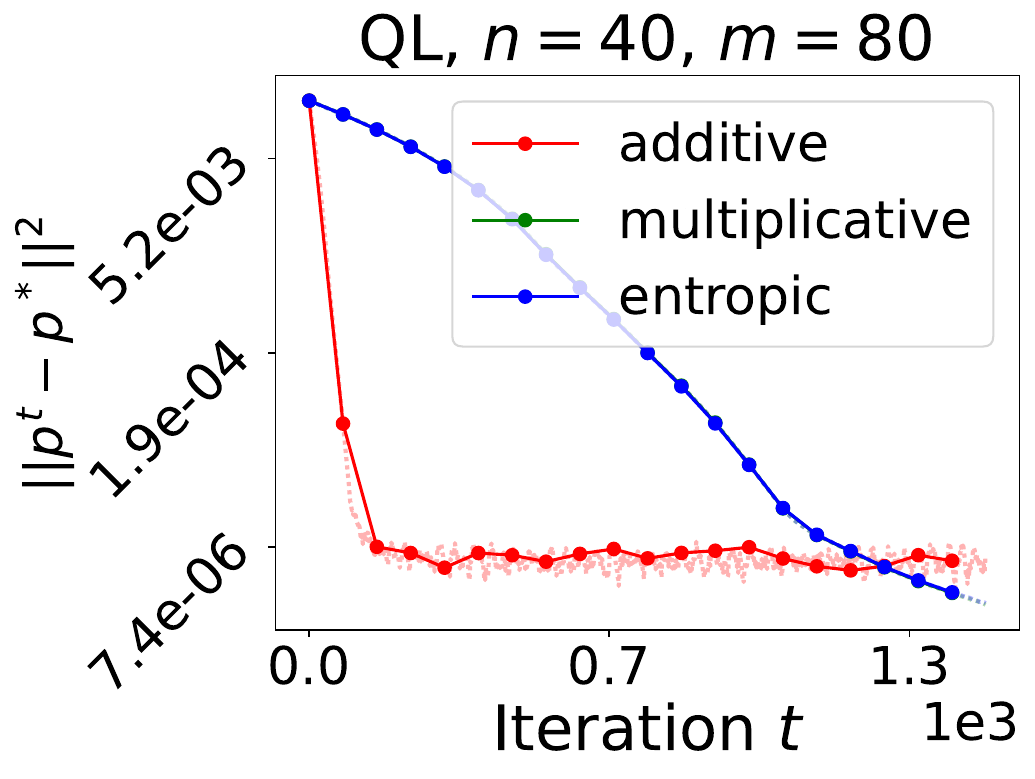}
    
    \caption{
        Comparison of different types of Tâtonnement on random generated instances ($v$ is generated from the uniform distribution $[0,1)$) of different sizes under linear utilities (upper row) and quasi-linear utilities (lower row). 
    }
    \label{fig:compare-ttm-uniform}
\end{figure}

\begin{figure}[t]
    \centering
    \includegraphics[scale=.23]{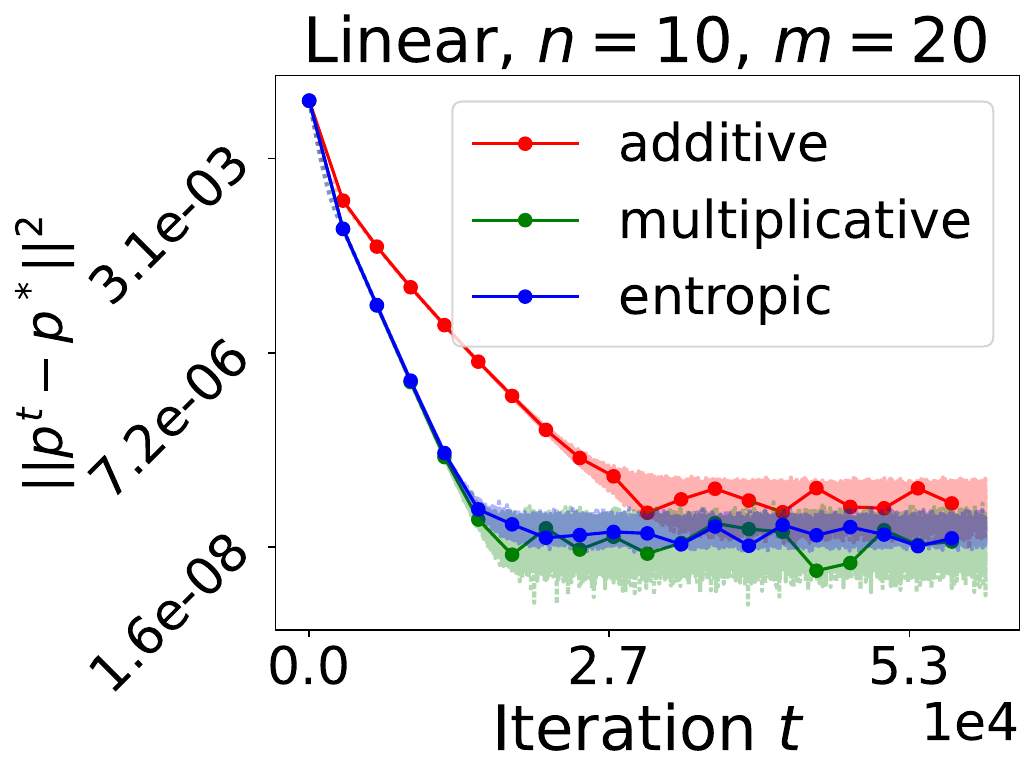}
    \includegraphics[scale=.23]{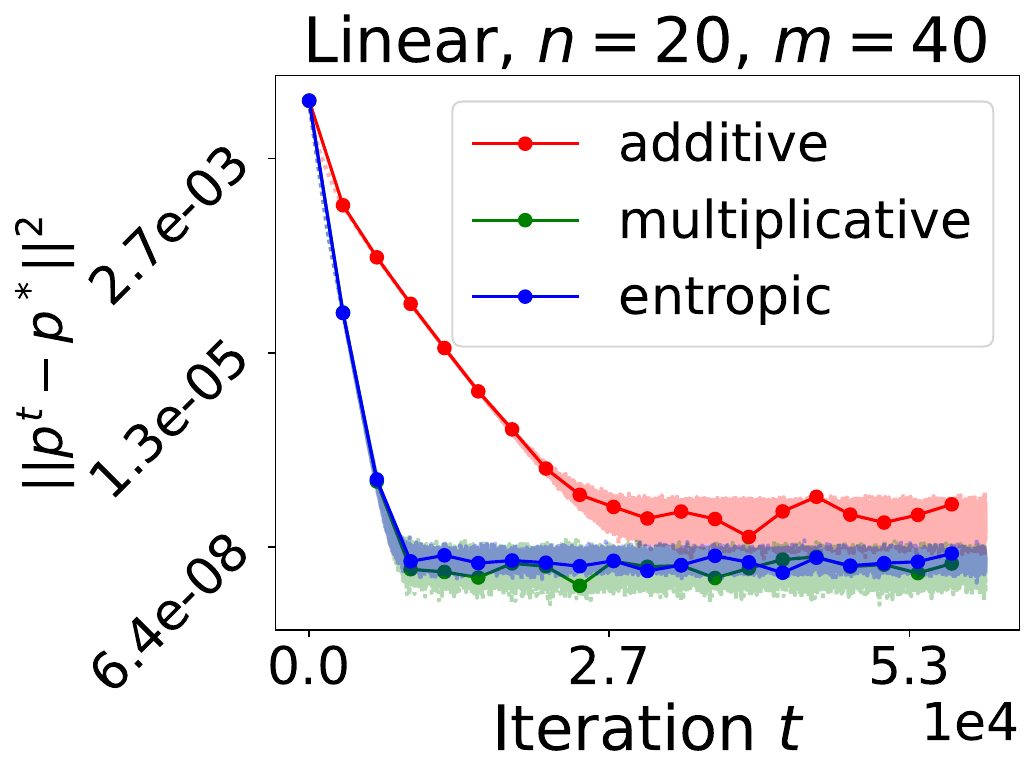}
    \includegraphics[scale=.23]{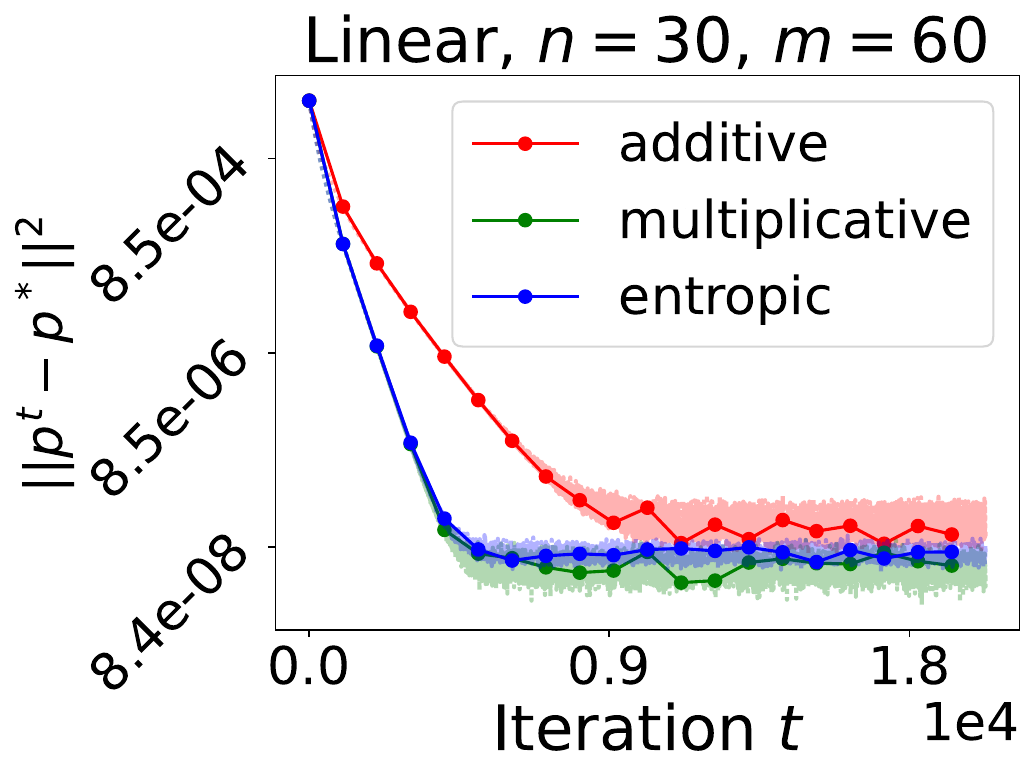}
    \includegraphics[scale=.23]{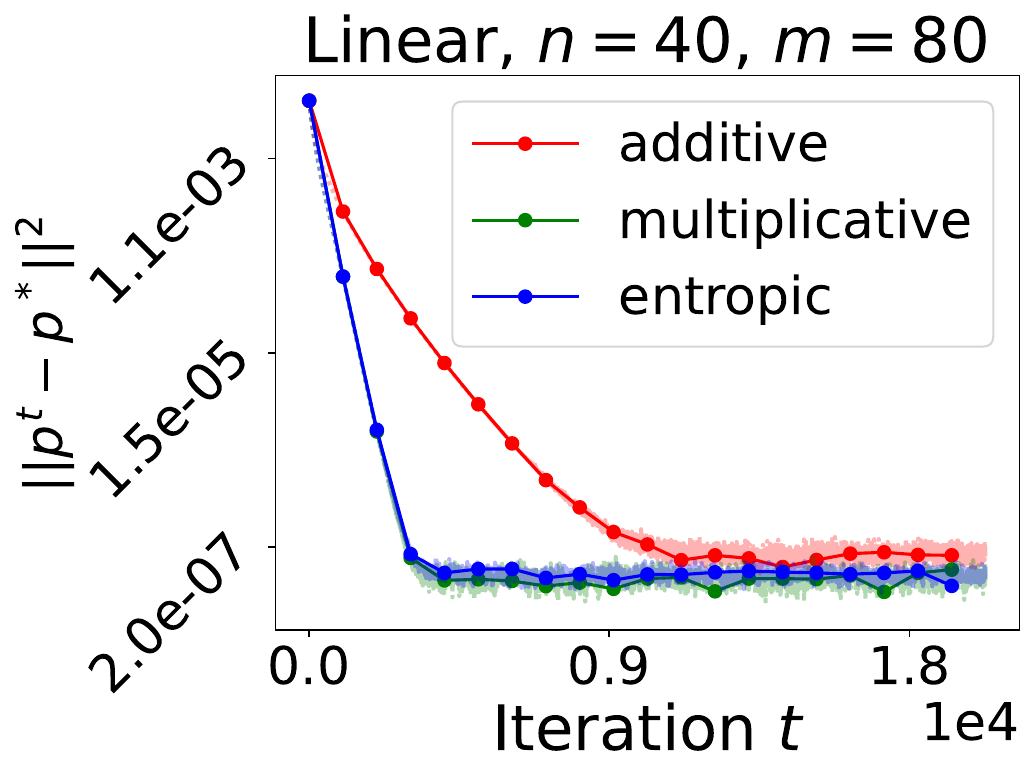}

    \includegraphics[scale=.23]{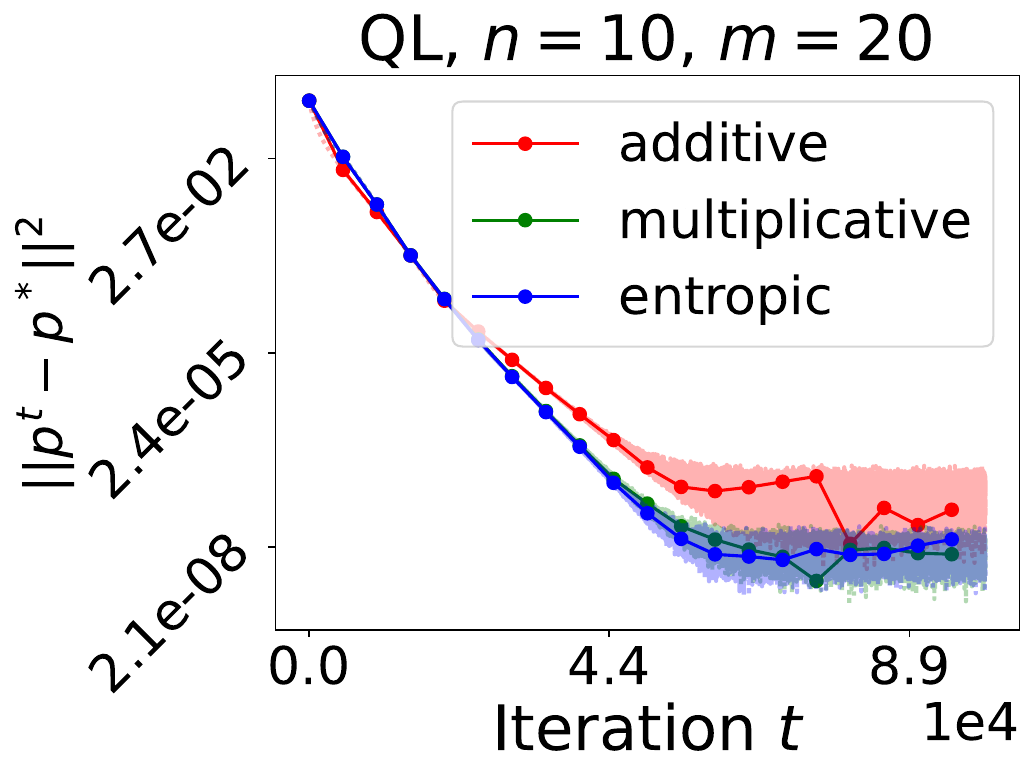}
    \includegraphics[scale=.23]{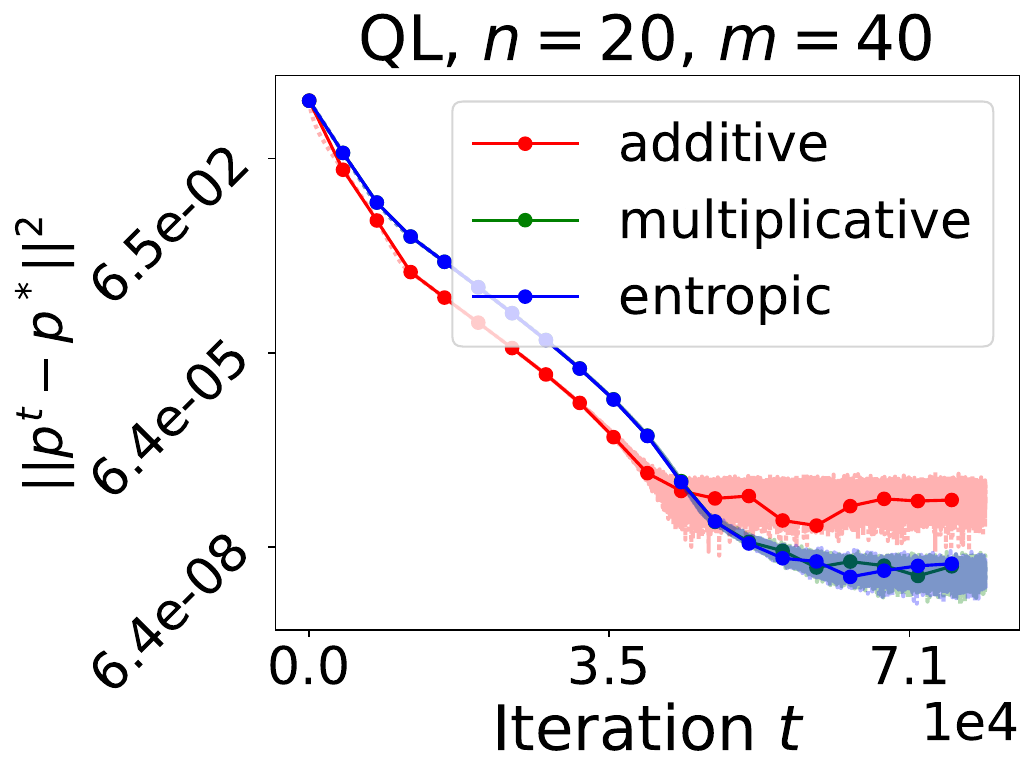}
    \includegraphics[scale=.23]{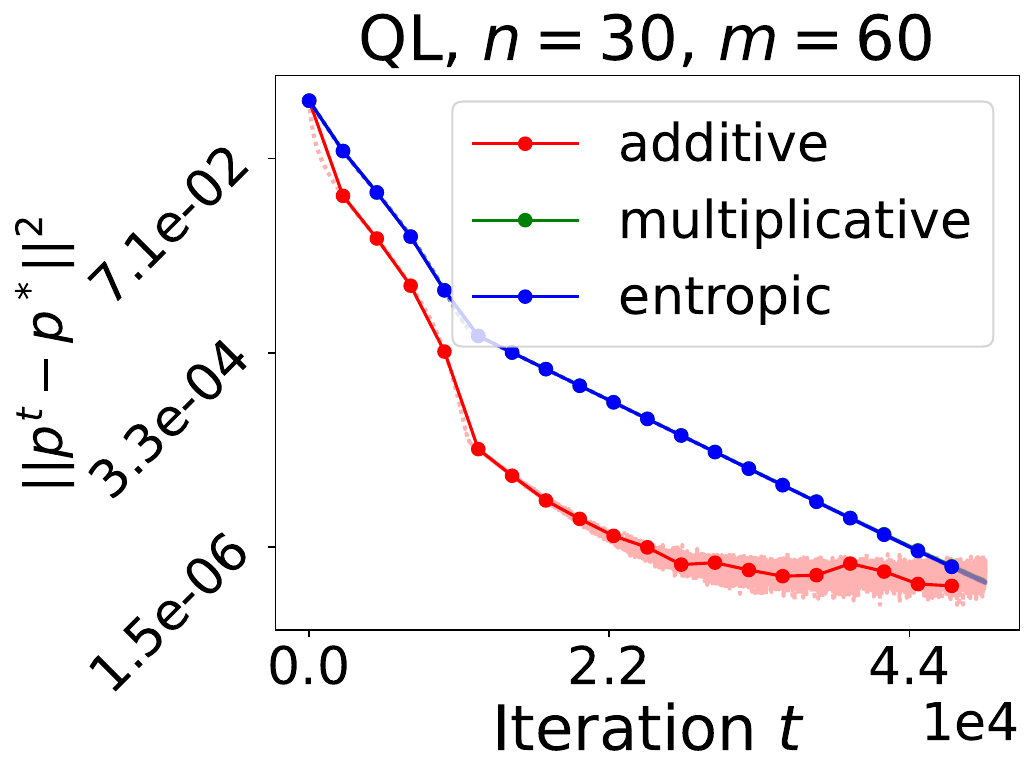}
    \includegraphics[scale=.23]{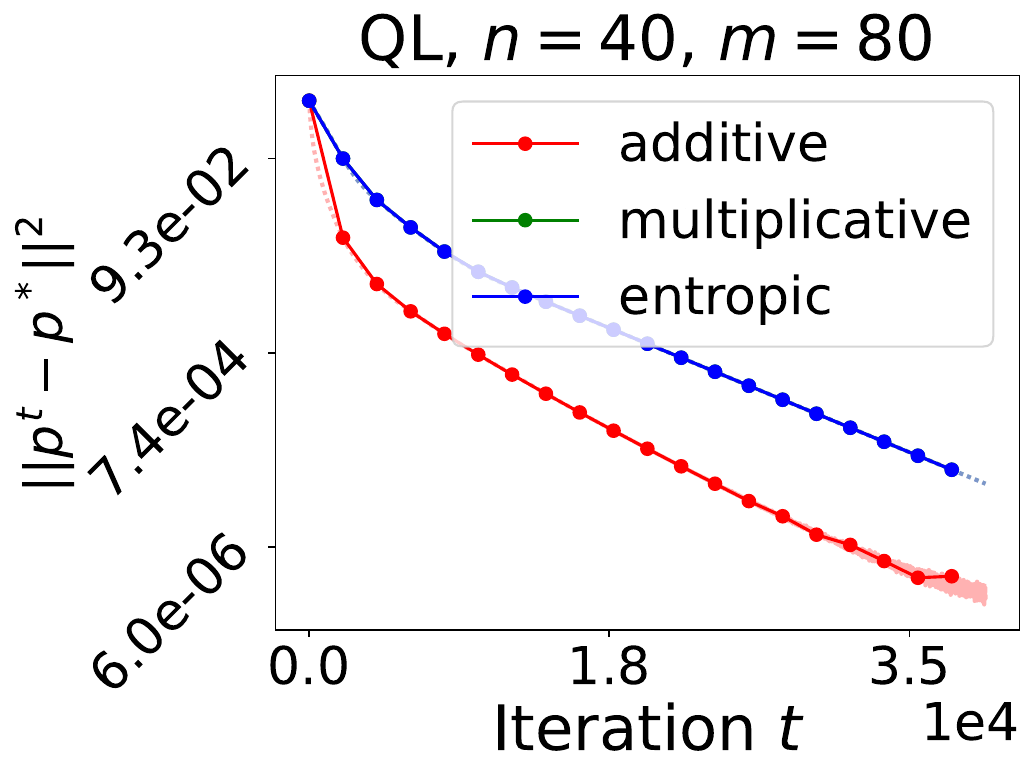}
    
    \caption{
        Comparison of different types of Tâtonnement on random generated instances ($v$ is generated from the log-normal distribution associated with the standard normal distribution $\mathcal{N}(0, 1)$) of different sizes under linear utilities (upper row) and quasi-linear utilities (lower row). 
    }
    \label{fig:compare-ttm-lognormal}
\end{figure}

\begin{figure}[t]
    \centering
    \includegraphics[scale=.23]{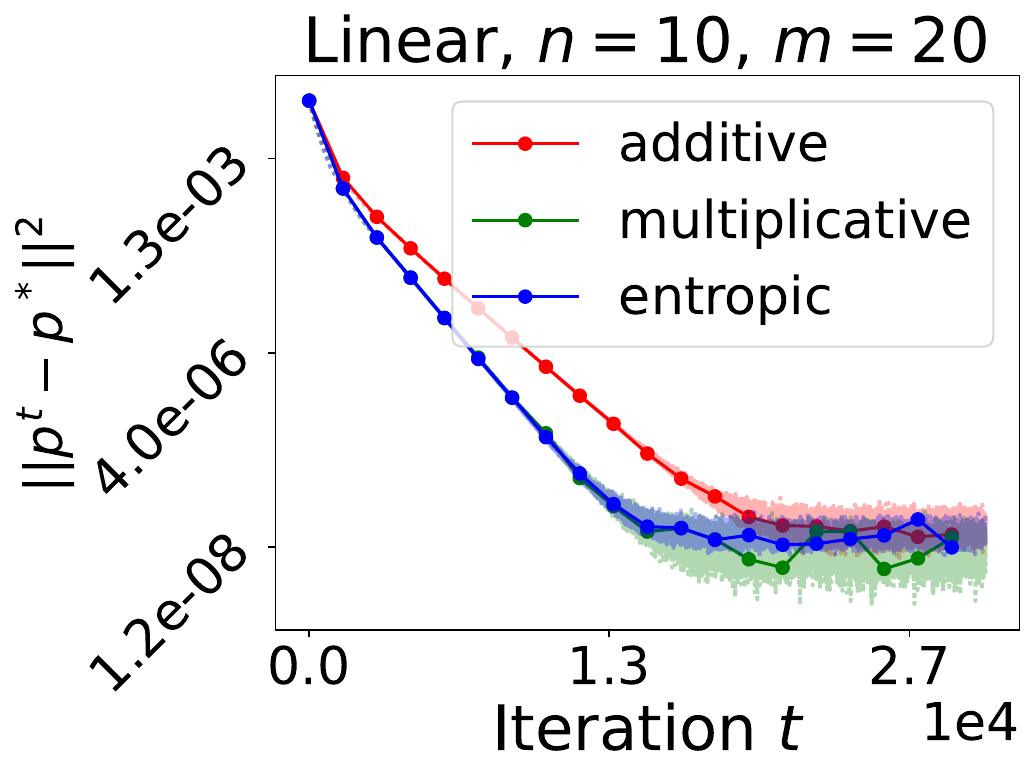}
    \includegraphics[scale=.23]{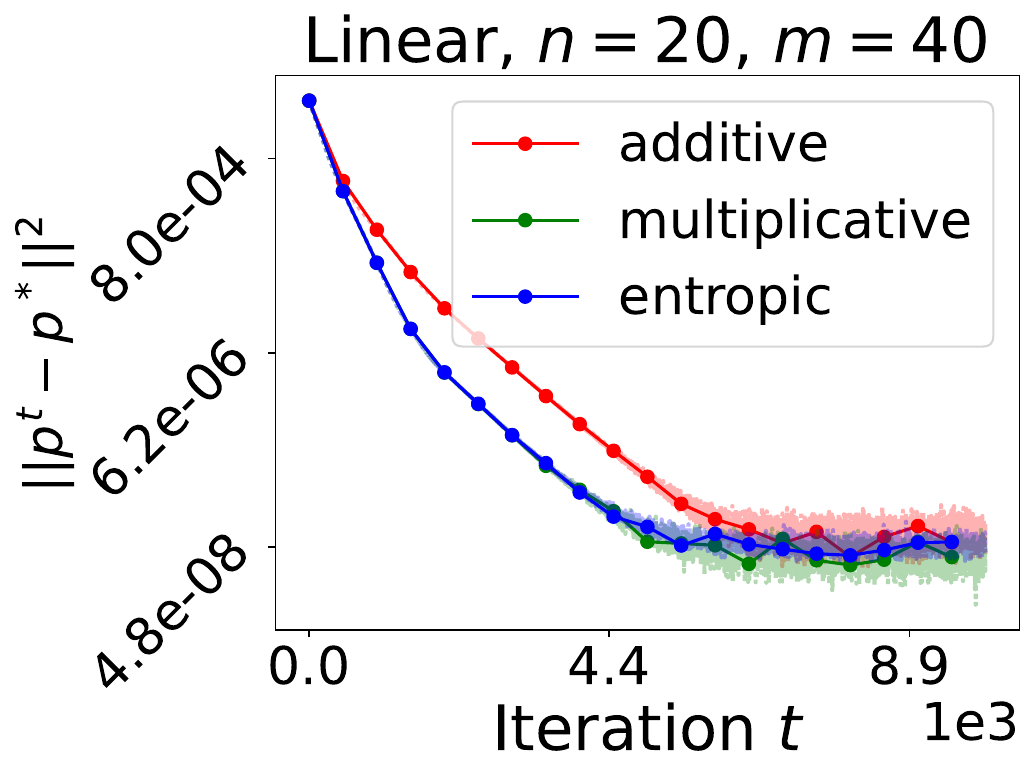}
    \includegraphics[scale=.23]{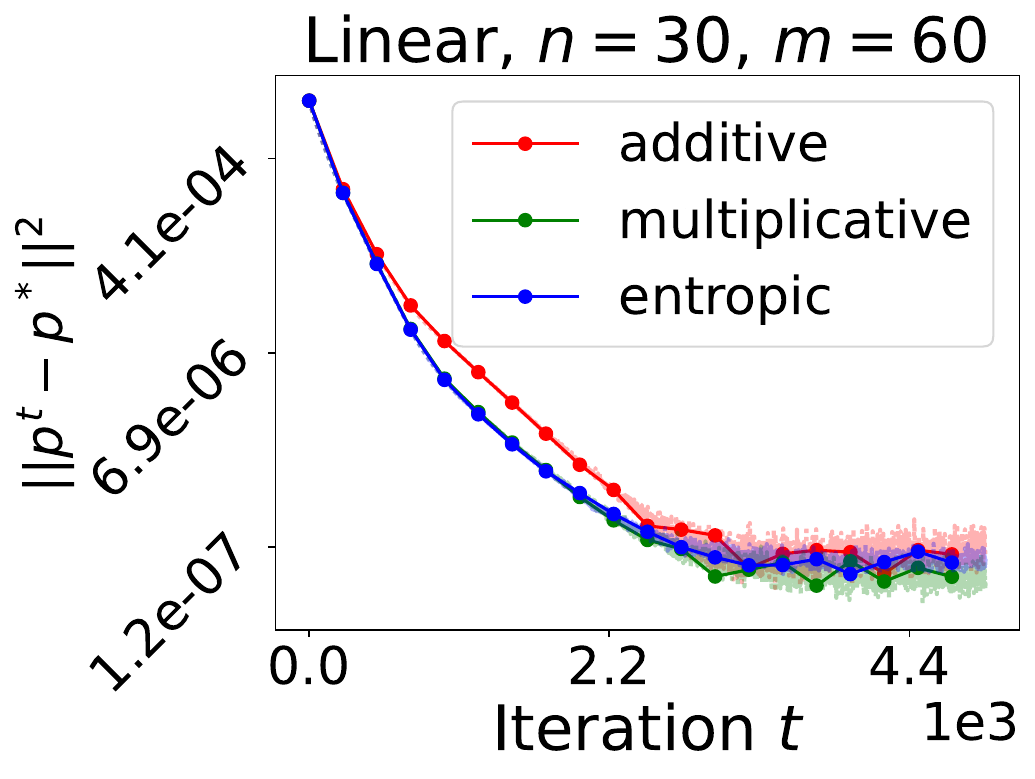}
    \includegraphics[scale=.23]{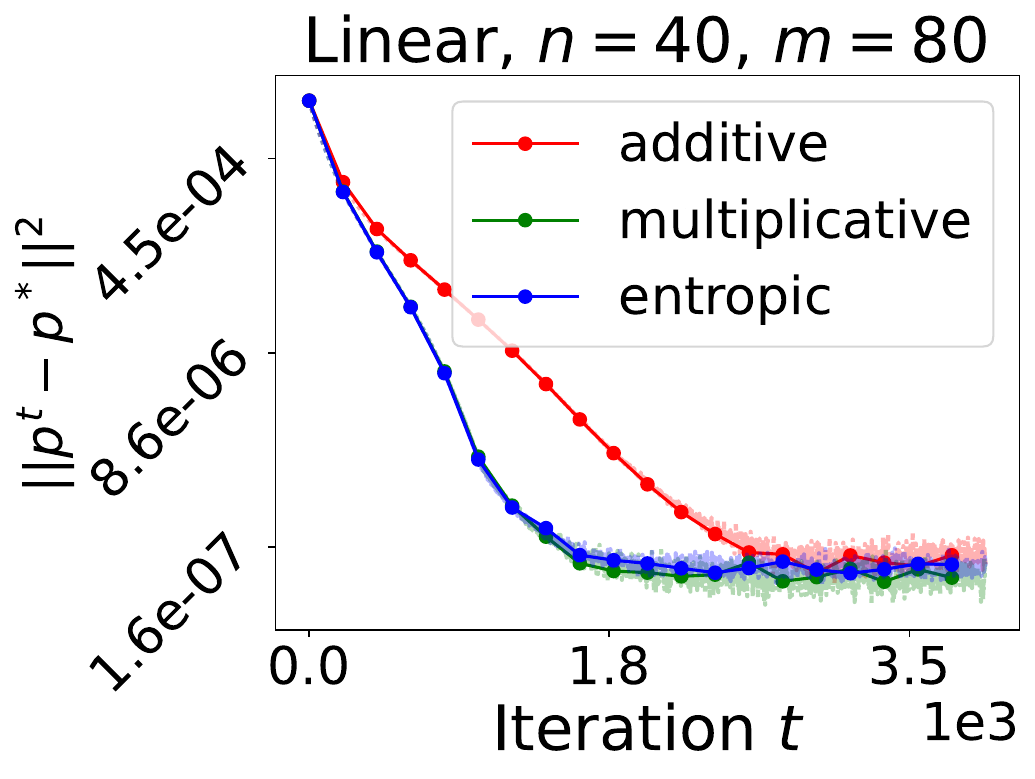}

    \includegraphics[scale=.23]{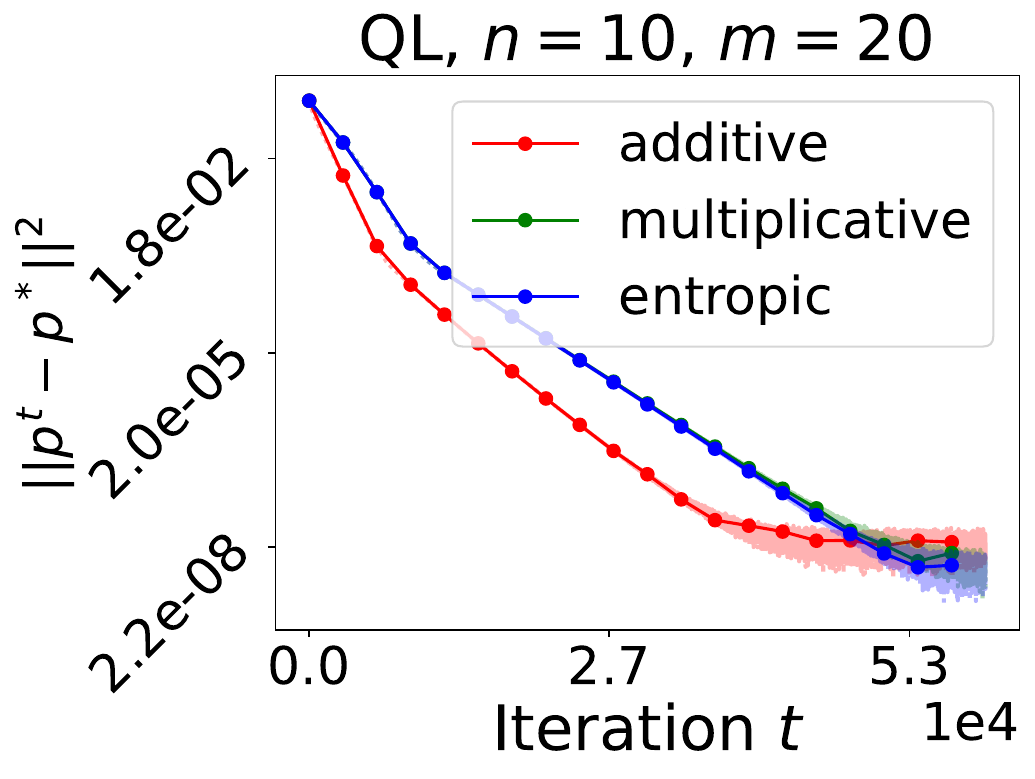}
    \includegraphics[scale=.23]{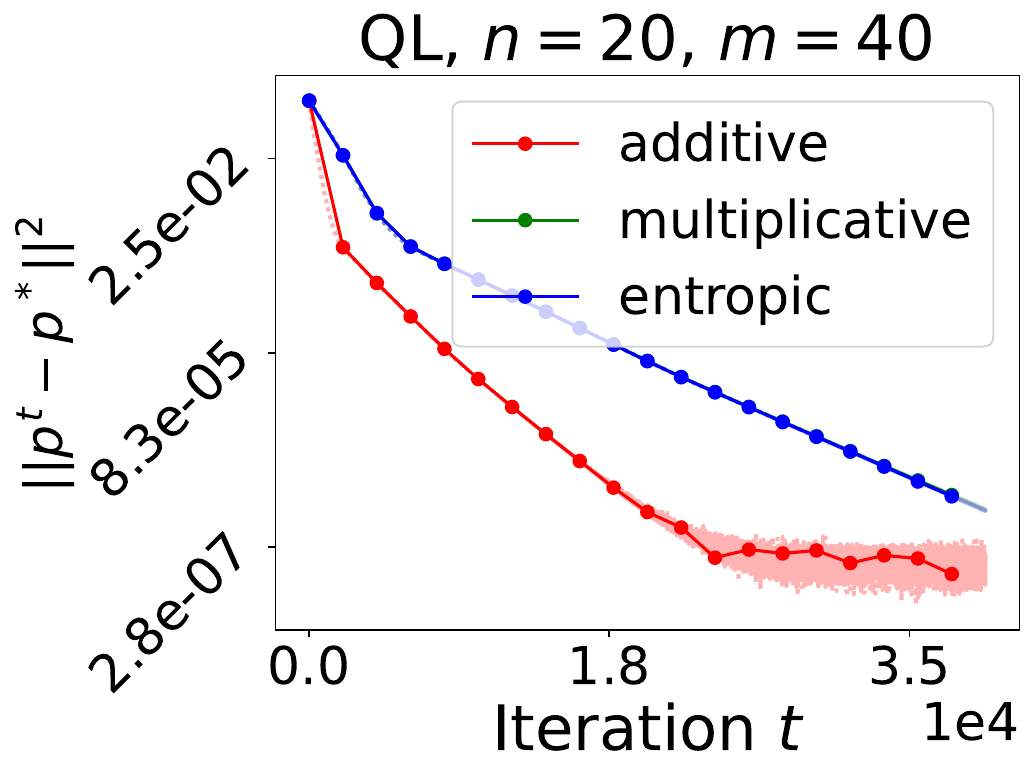}
    \includegraphics[scale=.23]{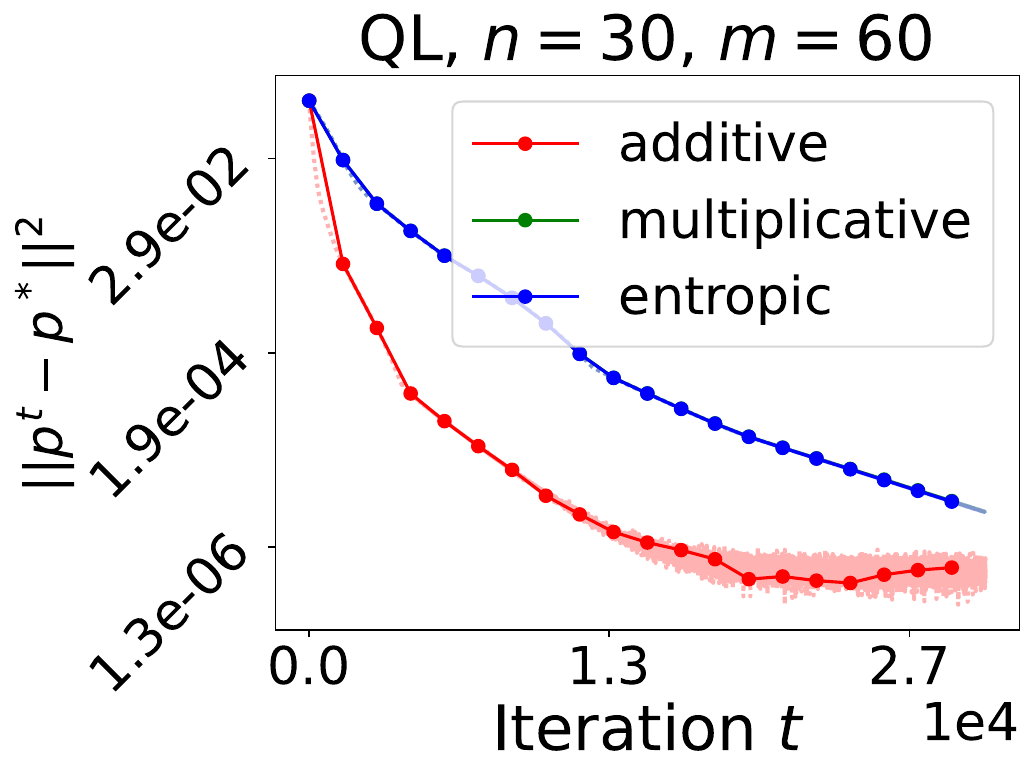}
    \includegraphics[scale=.23]{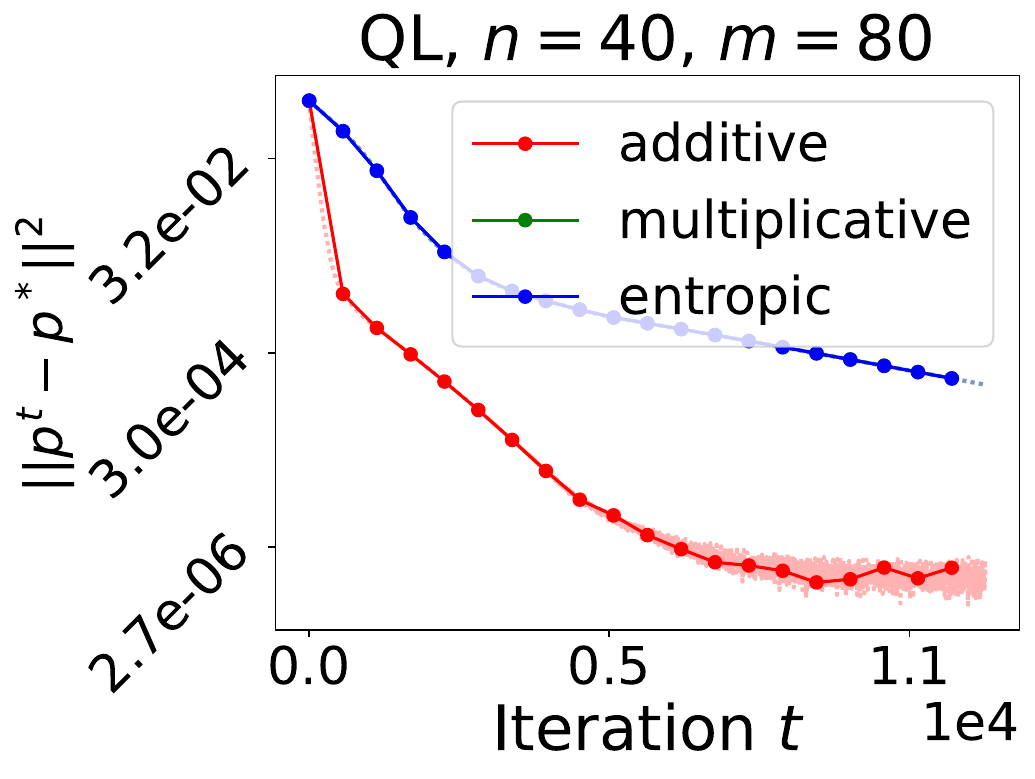}
    
    \caption{
        Comparison of different types of Tâtonnement on random generated instances ($v$ is generated from the exponential distribution with the scale parameter $1$) of different sizes under linear utilities (upper row) and quasi-linear utilities (lower row). 
    }
    \label{fig:compare-ttm-exponential}
\end{figure}

\begin{figure}[t]
    \centering
    \includegraphics[scale=.23]{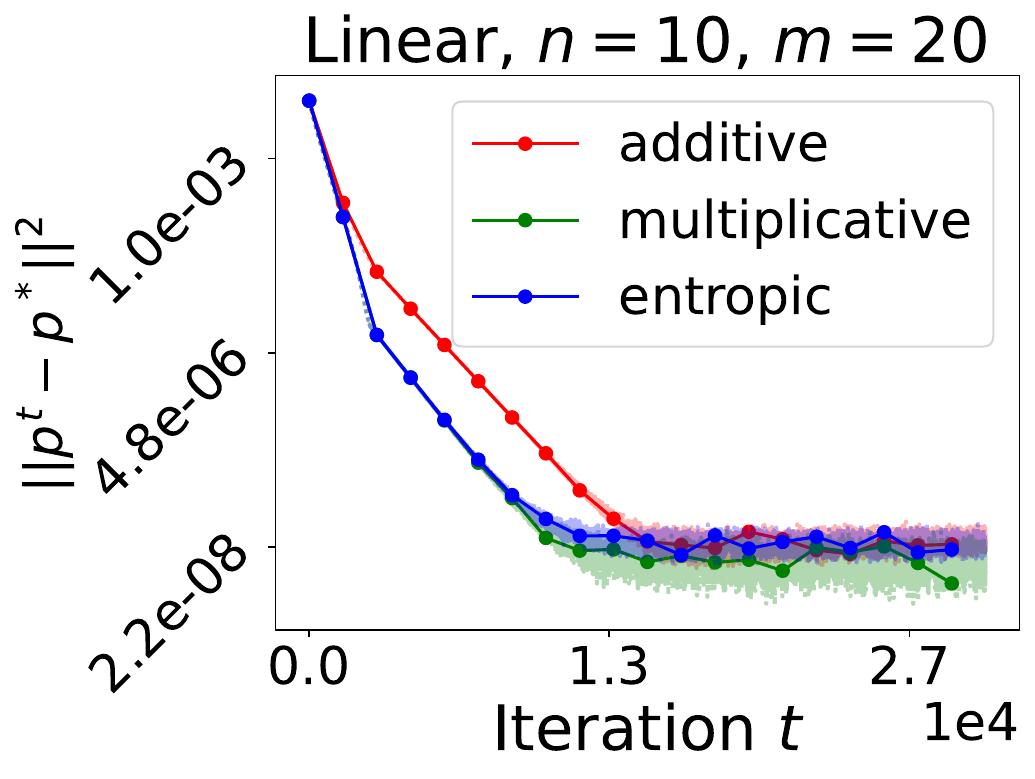}
    \includegraphics[scale=.23]{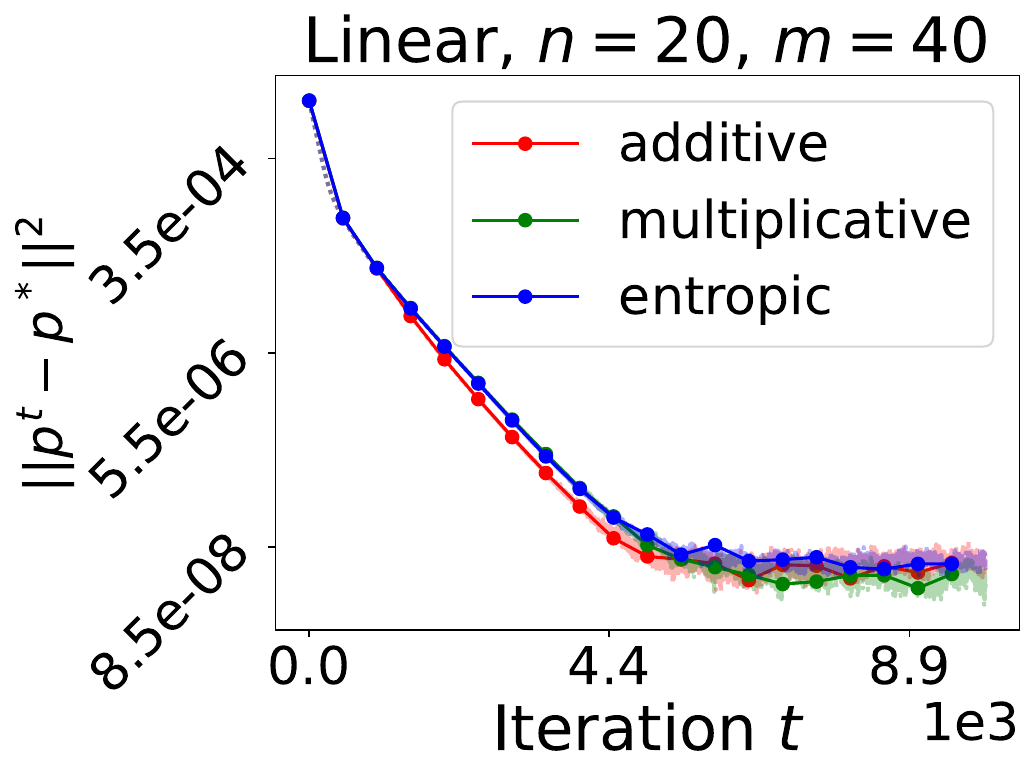}
    \includegraphics[scale=.23]{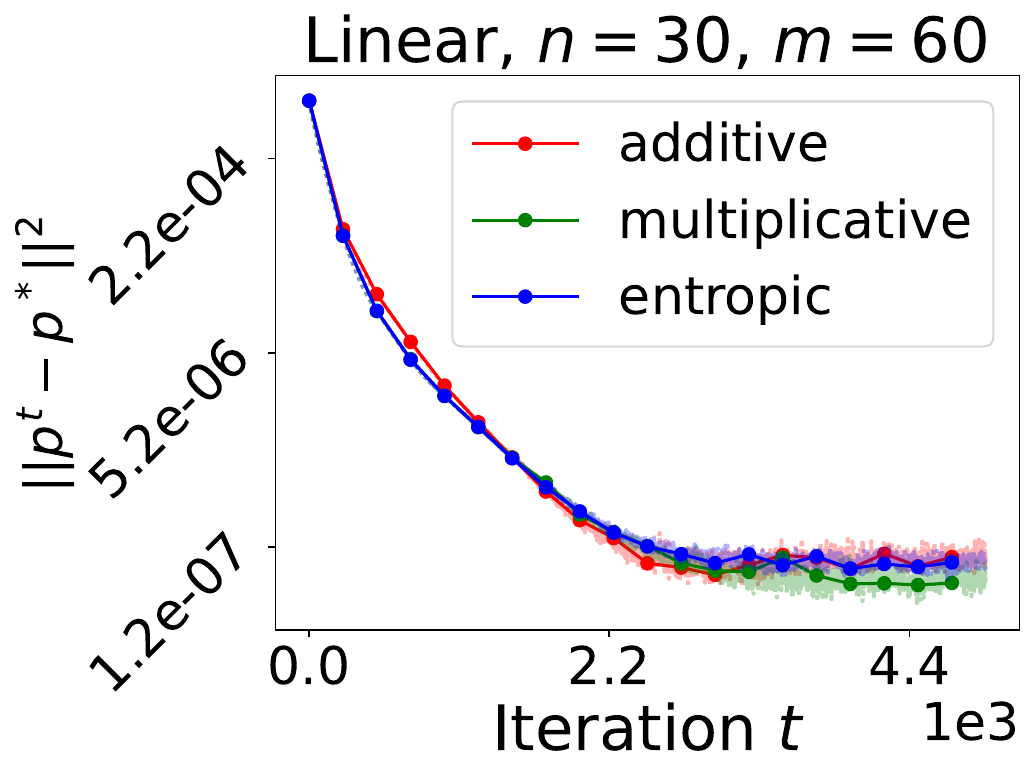}
    \includegraphics[scale=.23]{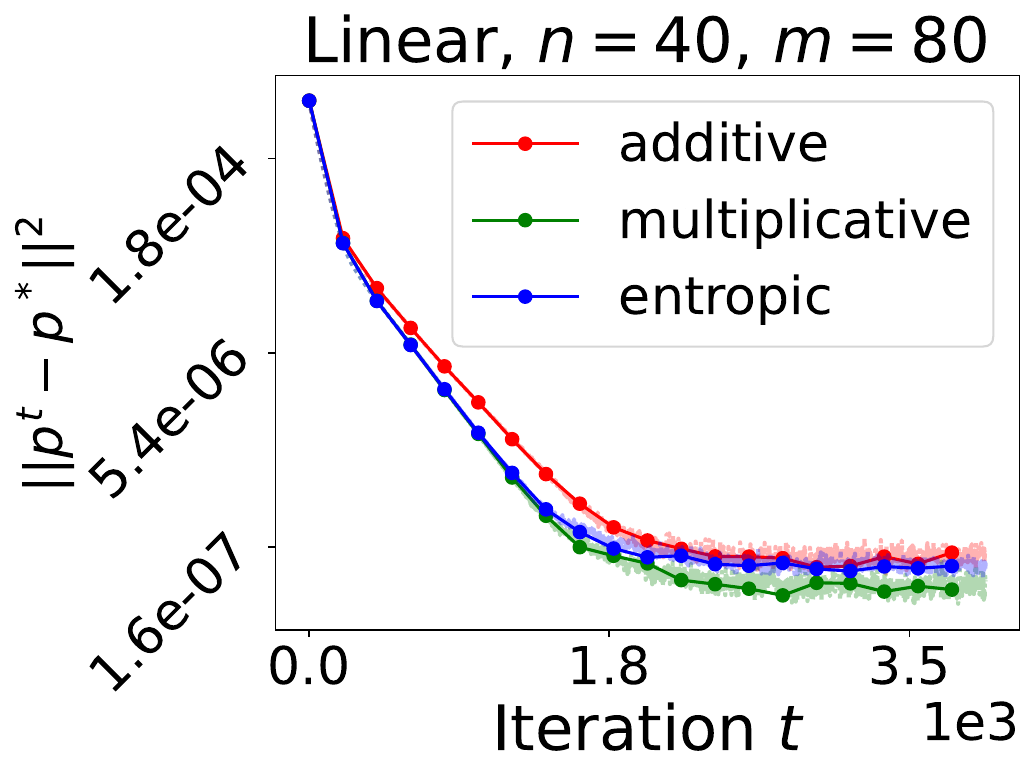}

    \includegraphics[scale=.23]{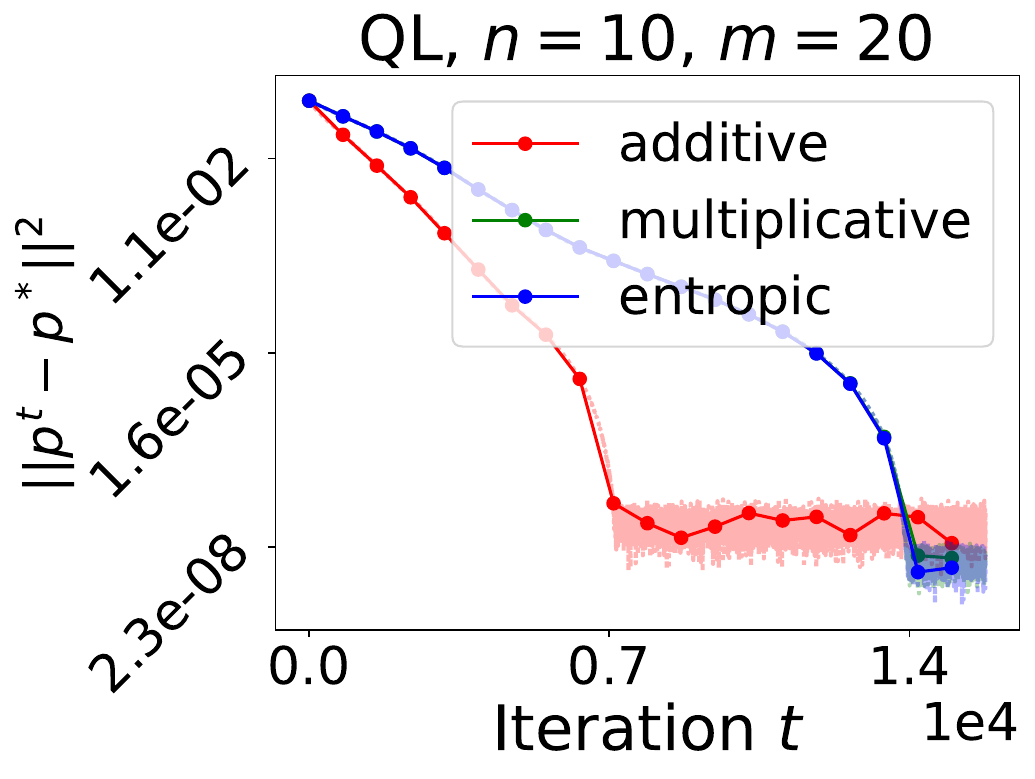}
    \includegraphics[scale=.23]{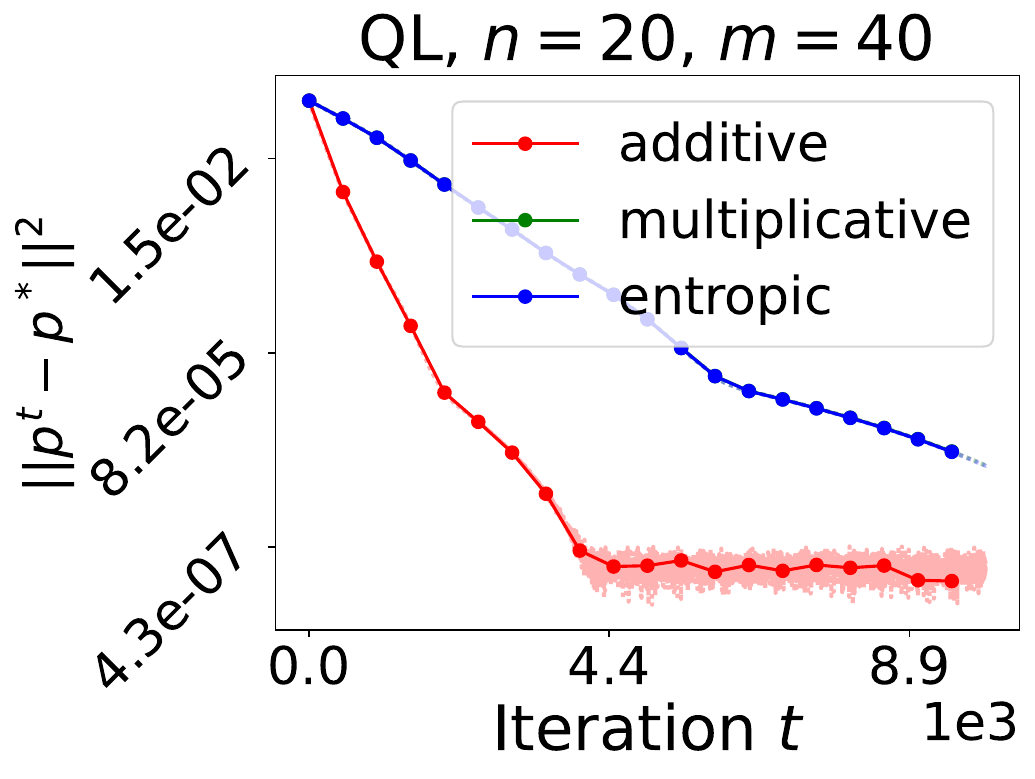}
    \includegraphics[scale=.23]{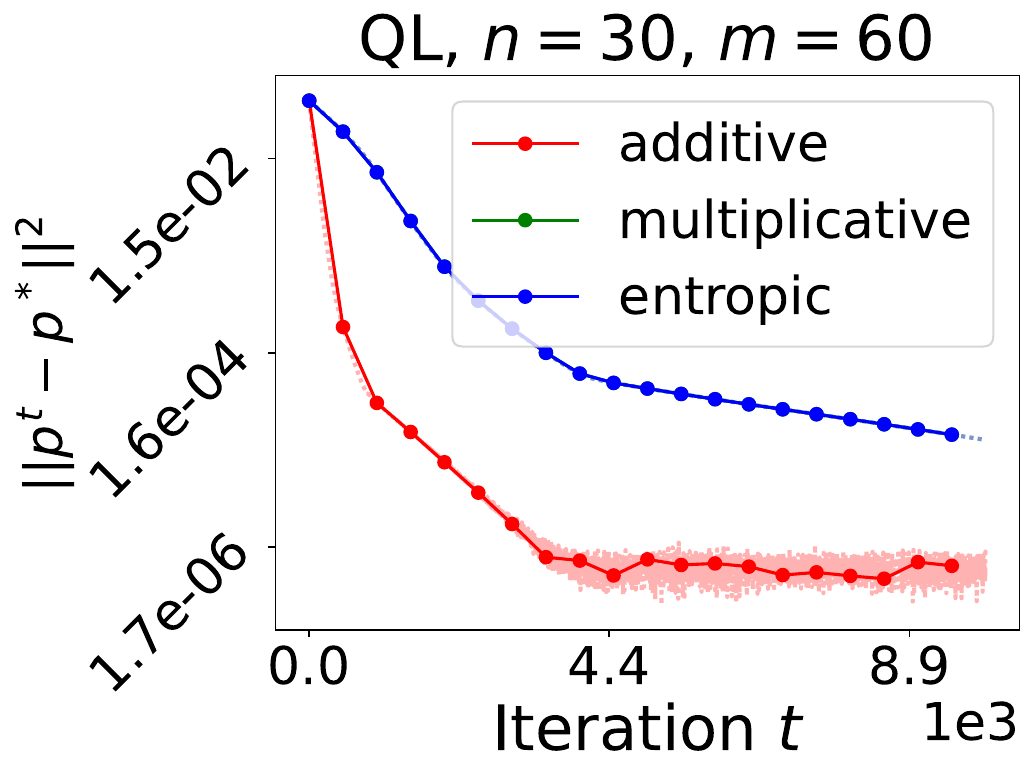}
    \includegraphics[scale=.23]{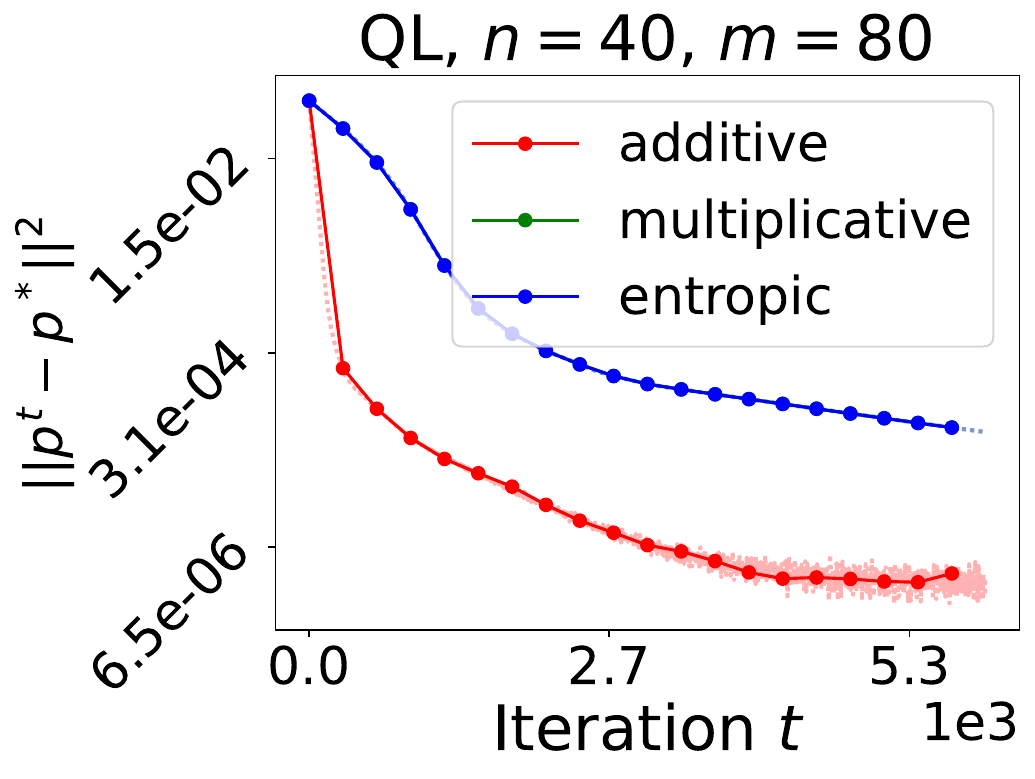}
    
    \caption{
        Comparison of different types of Tâtonnement on random generated instances ($v$ is generated from the truncated normal distribution associated with $\mathcal{N}(0, 1)$, and truncated at ${10}^{-3}$ and $10$ standard deviations from $0$) of different sizes under linear utilities (upper row) and quasi-linear utilities (lower row). 
    }
    \label{fig:compare-ttm-truncnorm}
\end{figure}

\begin{figure}[t]
    \centering
    \includegraphics[scale=.23]{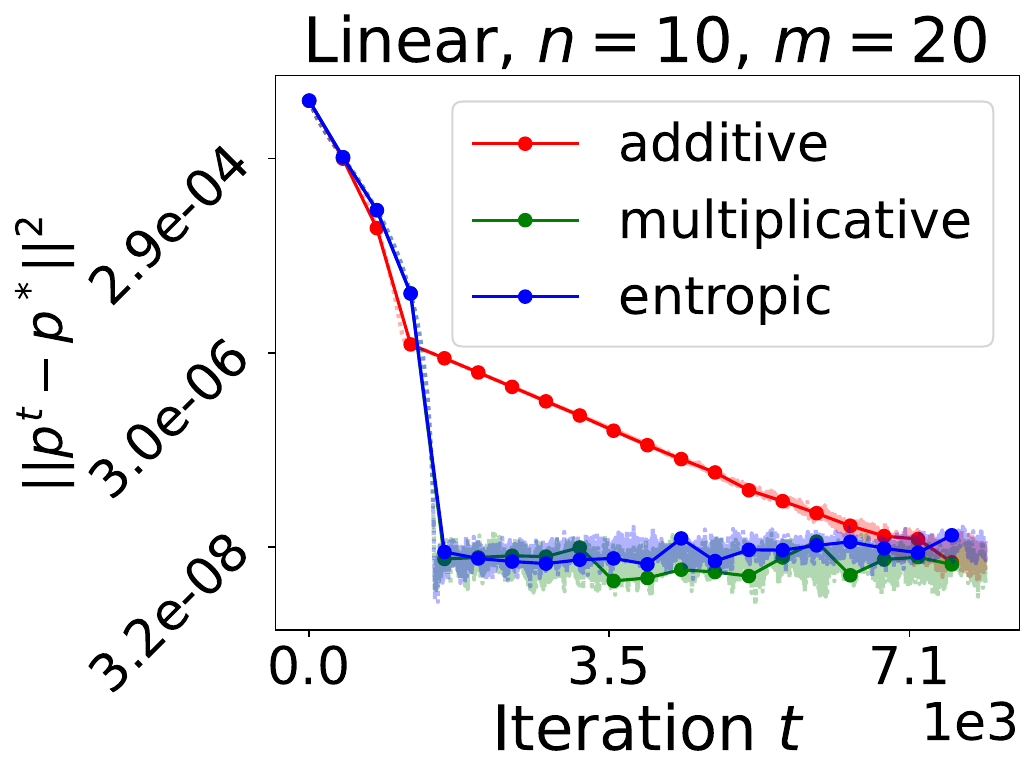}
    \includegraphics[scale=.23]{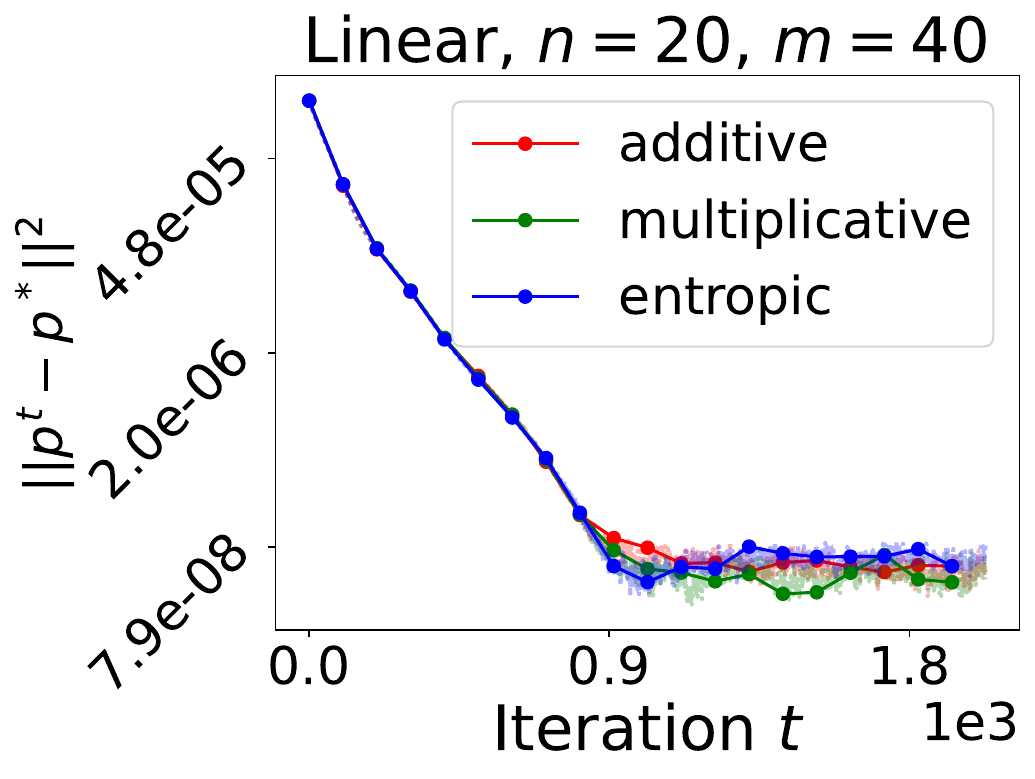}
    \includegraphics[scale=.23]{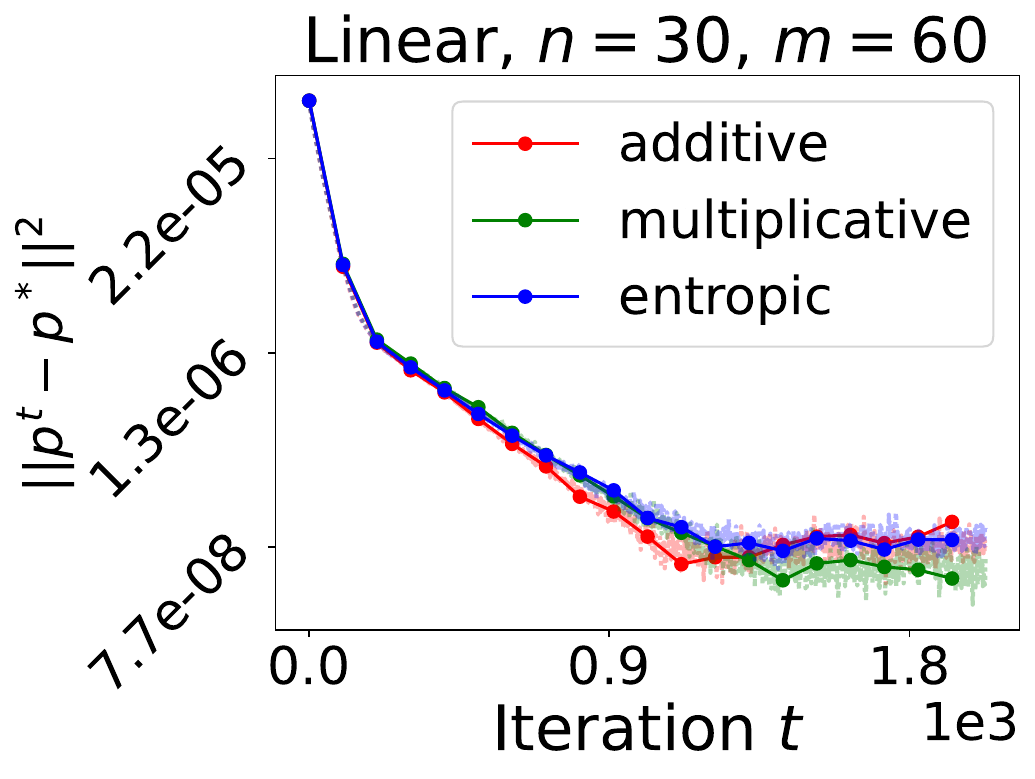}
    \includegraphics[scale=.23]{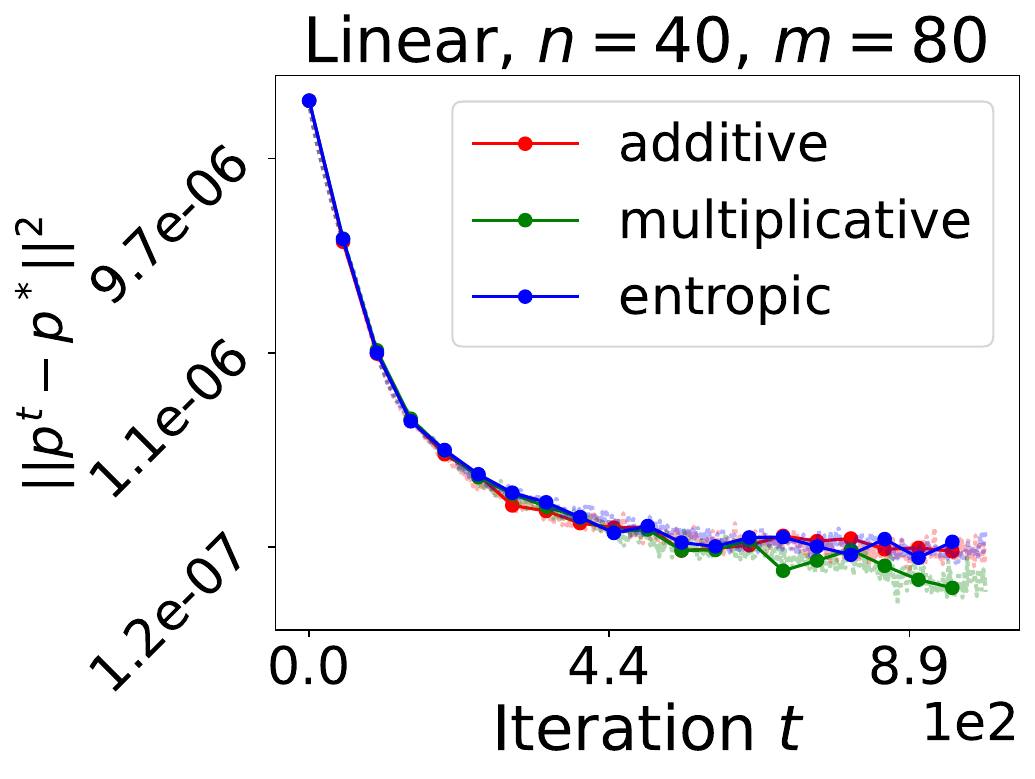}

    \includegraphics[scale=.23]{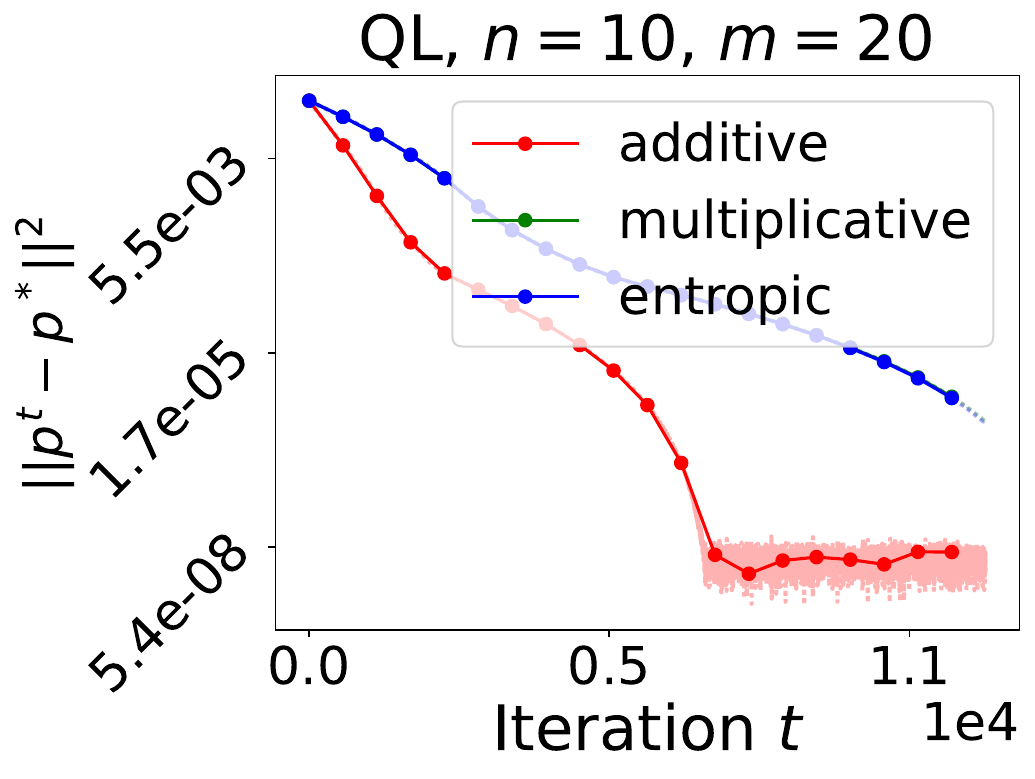}
    \includegraphics[scale=.23]{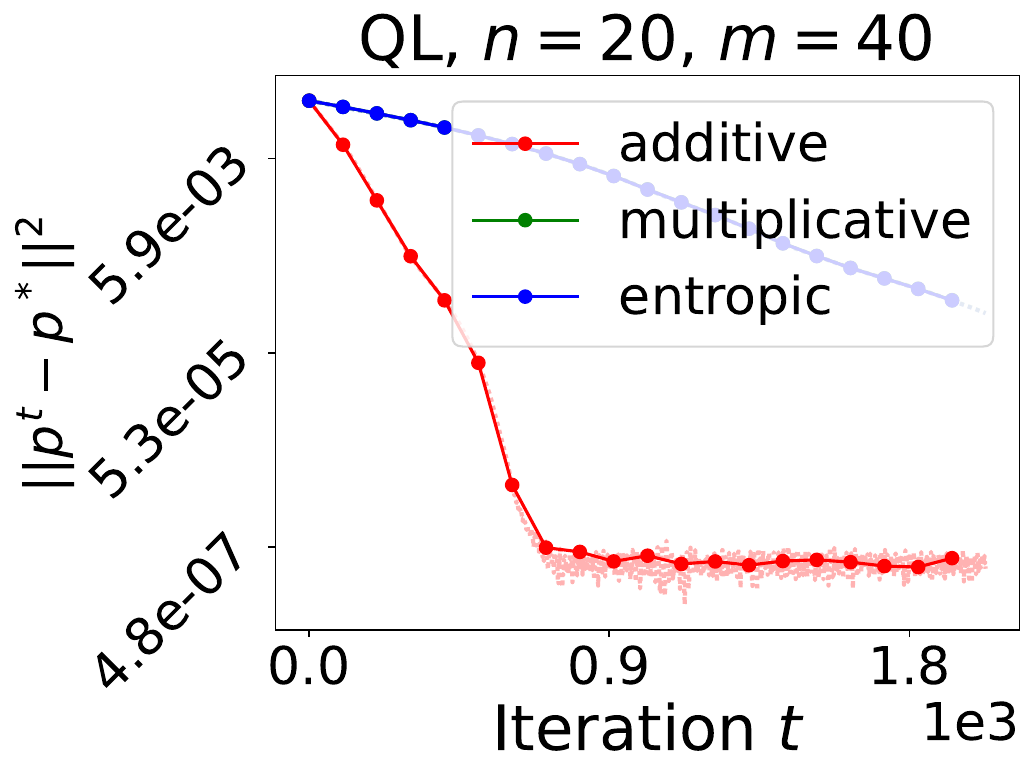}
    \includegraphics[scale=.23]{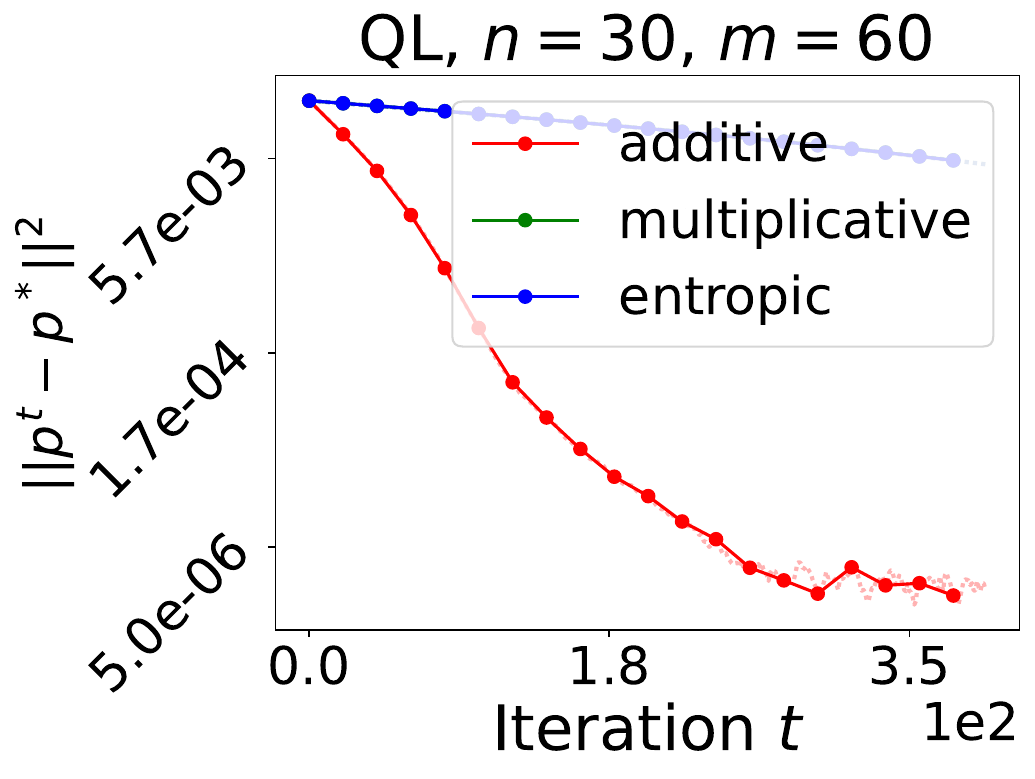}
    \includegraphics[scale=.23]{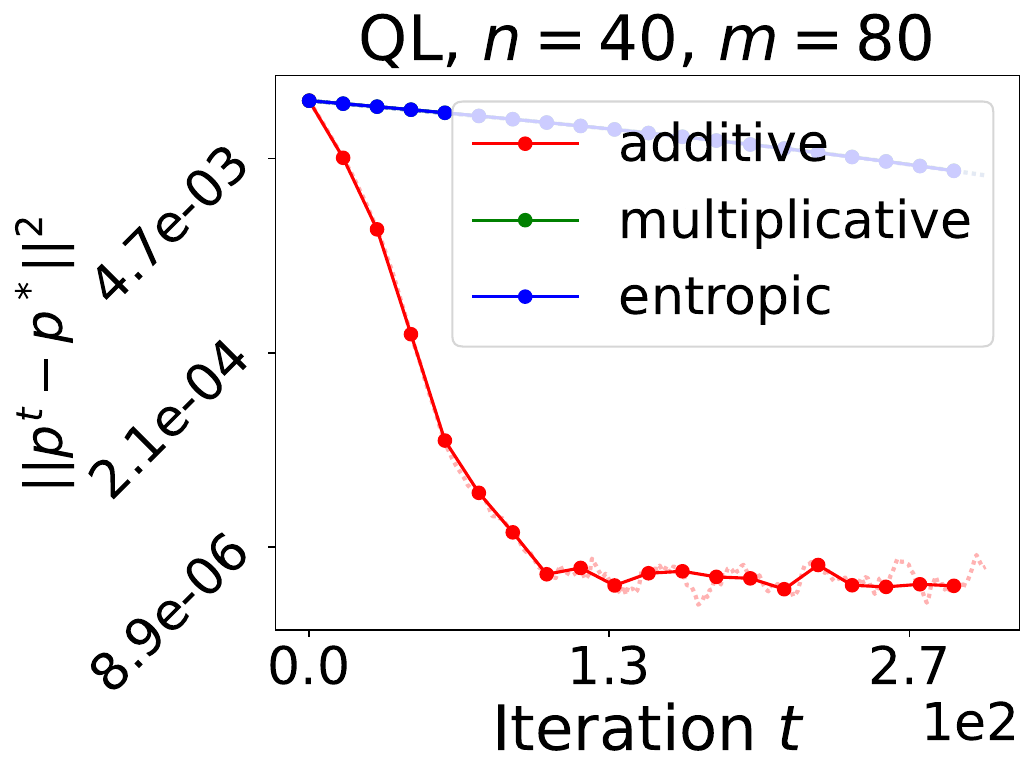}
    
    \caption{
        Comparison of different types of Tâtonnement on random generated instances ($v$ is generated from the uniform random integers on $\{1,\ldots,100\}$) of different sizes under linear utilities (upper row) and quasi-linear utilities (lower row). 
    }
    \label{fig:compare-ttm-randint}
\end{figure}

\begin{figure}[t]
    \centering
    \includegraphics[scale=.23]{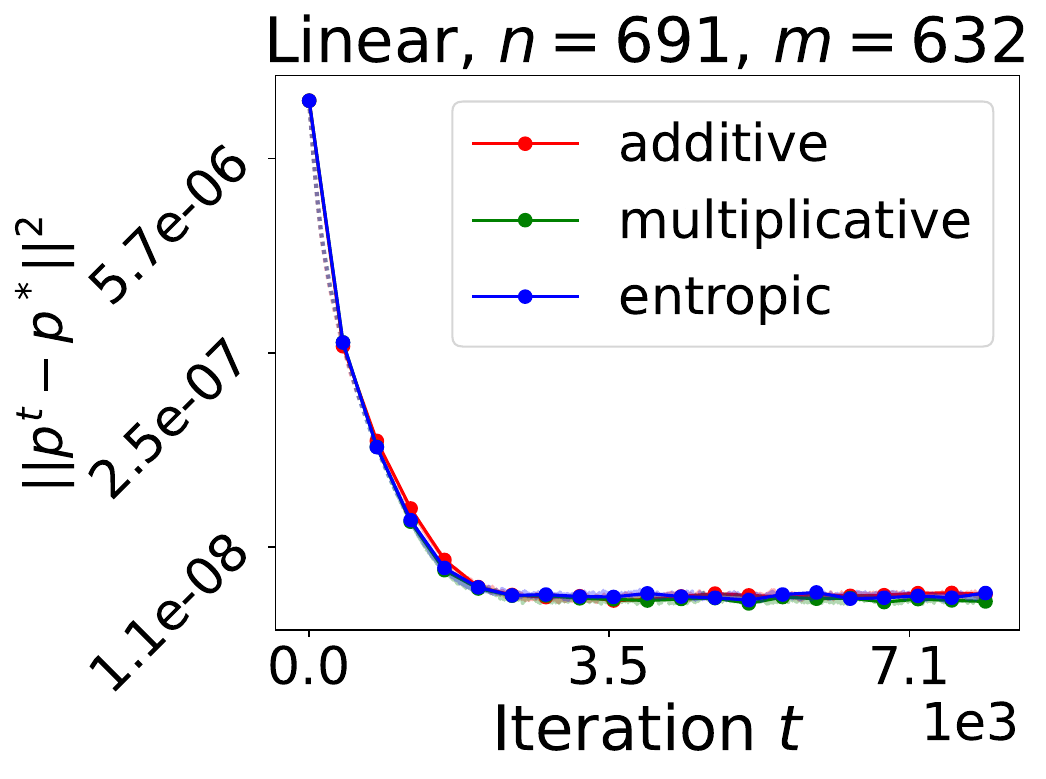}
    \hspace{3cm}
    \includegraphics[scale=.23]{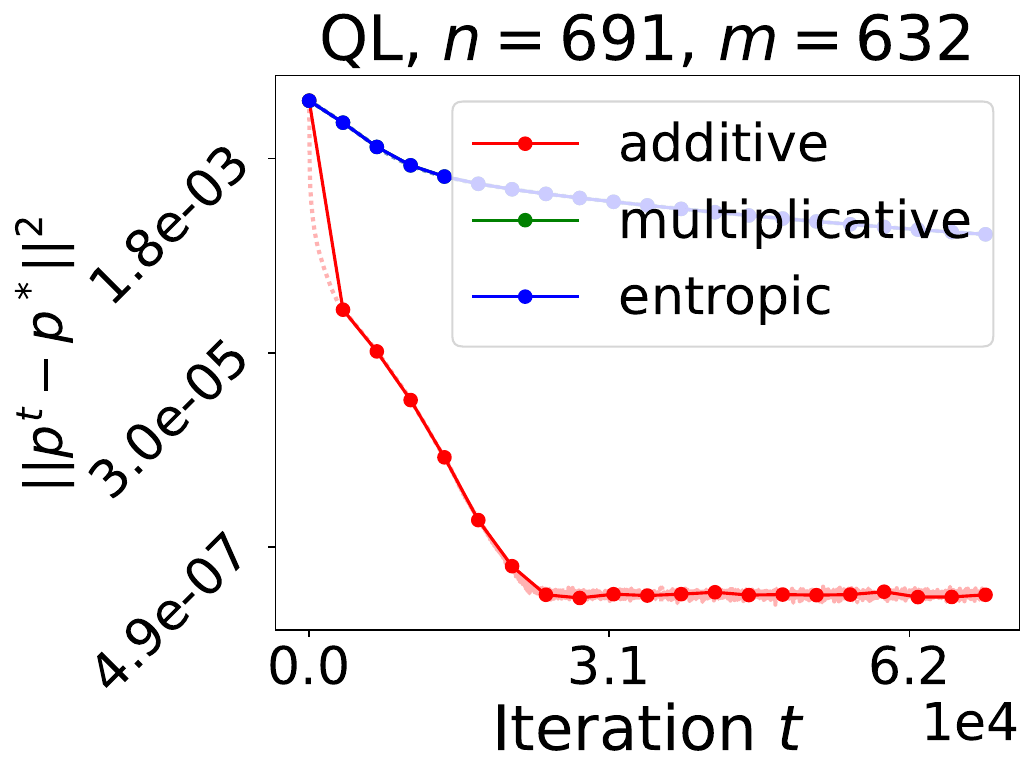}
    \caption{
        Comparison of different types of Tâtonnement on the real data instance under linear utilities (left) and quasi-linear utilities (right).
    }
    \label{fig:compare-ttm-real}
\end{figure}

\end{document}